\documentclass[11pt]{article}

\usepackage[utf8]{inputenc}
\usepackage{amssymb}
\usepackage{amsmath}
\usepackage{amsthm}
\usepackage{amsfonts}
\usepackage{graphicx}
\usepackage{tikz}
\usepackage{hyperref}
\usepackage{xspace}
\usepackage{tikz}
\usepackage{wrapfig}
\usetikzlibrary{calc,trees,positioning,arrows,chains,shapes.geometric,%
    decorations.pathreplacing,decorations.pathmorphing,shapes,%
    matrix,shapes.symbols,patterns}

\usetikzlibrary{decorations.pathmorphing}
\usetikzlibrary{shapes.misc} 
\usetikzlibrary{backgrounds}
\usepackage{transparent}
\usepackage{chapterbib}
\usepackage{titlesec}
\newtheorem{theorem}{Theorem}
\newtheorem{definition}{Definition}
\newtheorem{lemma}{Lemma}
\newtheorem{corollary}{Corollary}
\newtheorem{fact}{Fact}

\newtheorem{property}{Property}
\newtheorem{lock}{Lock}

\newcommand{\termasm}[1]{\mathcal{A}_{\Box}[{#1}]}
\newcommand{\pathassembly}[1]{\mathrm{asm}{(#1)}}
\newcommand{\asm}[1]{\pathassembly{#1}}
\newcommand{\domain}[1]{dom(#1)}

\newcommand{\N}{\mathbb{N}}

\newcommand{\vect}{\protect\overrightarrow}
\newcommand{\prodasm}[1]{\mathcal{A}[{#1}]}

\newcommand{\prodpaths}[1]{{\bf{P}}[{#1}]}

\newcommand{\vu}{\vect{v}}

\newcommand{\pos}[1]{\mathrm{pos}(#1)}
\newcommand\type[1]{\mathrm{type}(#1)}

\newcommand{\uniterm}{\gamma}

\newcommand{\gluePi}{\glueP{i}{i+1}}

\newcommand{\glueP}[2]{\mathrm{glue}(P_{#1} P_{#2})}

\newcommand{\glueN}[2]{\mathrm{glue}(\Hole_{#1} \Hole_{#2})}
\newcommand{\glueQ}[2]{\mathrm{glue}(Q_{#1} Q_{#2})}
\newcommand{\glueR}[2]{\mathrm{glue}(R_{#1} R_{#2})}

\newcommand{\glueS}[2]{\mathrm{glue}(S_{#1} S_{#2})}
\newcommand{\glueSDeux}[2]{\mathrm{glue}(S'_{#1} S'_{#2})}

\newcommand{\cano}[1]{P^{(#1)}}

\newcommand{\last}{f}
\newcommand{\lastc}{\ell}

\newcommand{\main}[1]{m_{#1}}
\newcommand{\second}[1]{b_{#1}}
\newcommand{\bord}[1]{b^{#1}}
\newcommand{\inter}[1]{\mathcal{I}_{#1}}
\newcommand{\es}{\mathcal{E}}
\newcommand{\area}[1]{\mathcal{P}_{#1}}
\newcommand{\Hole}{H}

\newcommand{\cork}[1]{l^{#1}}
\newcommand{\prodHole}[1]{{\bf{P}}[{#1}]}
\newcommand{\rev}[1]{\text{rev}({#1})}
\newcommand{\posS}{a}
\newcommand{\proG}{p}
\newcommand{\rightF}{R}
\newcommand{\leftF}{L}
\newcommand{\enSh}{e}
\newcommand{\goB}{g}
\newcommand{\hid}{h}
\newcommand{\hidDeux}{h'}

\newcommand{\colUn}{c}
\newcommand{\lastUn}{f}
\newcommand{\mainUn}[1]{m_{#1}}
\newcommand{\secondUn}[1]{b_{#1}}
\newcommand{\bordUn}[1]{b^{#1}}

\newcommand{\esUn}{\mathcal{E}}

\newcommand{\areaRes}[1]{\mathcal{R}_{#1}}
\newcommand{\areaResUn}[1]{\mathcal{R}_{#1}}
\newcommand{\areaResDeux}[1]{\mathcal{R}'_{#1}}

\newcommand{\rightFUn}{R}

\newcommand{\goBUn}{g}

\newcommand{\decoUn}{t}

\newcommand{\colDeux}{c'}
\newcommand{\lastDeux}{f'}

\newcommand{\posSDeux}{a'}

\newcommand{\rightFDeux}{R'}
\newcommand{\leftFDeux}{L'}

\newcommand{\goBDeux}{g'}

\newcommand{\newA}{A'}
\newcommand{\newS}{S'}
\newcommand{\newProG}{a'}

\newcommand{\proS}{p}
\newcommand{\proNDeux}{p'}

\newcommand{\lastcDeux}{\ell'}
\newcommand{\shieldUn}{S}
\newcommand{\shieldDeux}{S'}

\newcommand{\proSLast}[1]{p}
\newcommand{\proNLast}[1]{{n_{#1}}}

\newcommand{\Bzero}{\mathcal{B}_0}
\newcommand{\BzeroValue}{w_{\sigma}-2|T|-1}
\newcommand{\Bone}{\mathcal{B}_1}
\newcommand{\BoneValue}{e_{\sigma+|T|+1.5}}
\newcommand{\Btwo}{\mathcal{B}_2}
\newcommand{\BtwoValue}{e_{\sigma+5|T|+1.5}}
\newcommand{\Bthree}{\mathcal{B}_3}

\newcommand{\Bfour}{\mathcal{B}_4}

\newcommand{\Bfinal}{\mathcal{B}_5}
\newcommand{\BfinalValue}{7|\sigma|+58|T|+30}
\newcommand{\BfinalExtremal}{3|\sigma|+29|T|+15}
\newcommand{\BfinalTheorem}{7|\sigma|+58|T|+30}

\definecolor{xgreen}{RGB}{0,128,0}
\definecolor{xblue}{RGB}{0,204,255}
\definecolor{xorange}{RGB}{255,165,0}
\definecolor{xred}{RGB}{255,0,0}

\newcommand{\draws}[2]{\draw[fill=black] (#1.16,#2.16) rectangle (#1.84,#2.84);}

\usepackage{authblk}

\definecolor{Jvert}{RGB}{225, 221, 114}
\definecolor{Olive}{RGB}{168, 198, 108}
\definecolor{Foret}{RGB}{27, 101, 53}
\definecolor{Jvert}{RGB}{225, 221, 114}

\definecolor{R0B0}{RGB}{0, 0, 0}
\definecolor{R1B0}{RGB}{40, 0, 0}
\definecolor{R2B0}{RGB}{80, 0, 0}
\definecolor{R3B0}{RGB}{120, 0, 0}
\definecolor{R4B0}{RGB}{160, 0, 0}
\definecolor{R5B0}{RGB}{200, 0, 0}

 \definecolor{R0B1}{RGB}{0, 0, 40}
\definecolor{R1B1}{RGB}{40, 0, 40}
\definecolor{R2B1}{RGB}{80, 0, 40}
\definecolor{R3B1}{RGB}{120, 0, 40}
\definecolor{R4B1}{RGB}{160, 0, 40}
\definecolor{R5B1}{RGB}{200, 0, 40}

\definecolor{R0B2}{RGB}{0, 0, 80}
\definecolor{R1B2}{RGB}{40, 0, 80}
\definecolor{R2B2}{RGB}{80, 0, 80}
\definecolor{R3B2}{RGB}{120, 0, 80}
\definecolor{R4B2}{RGB}{160, 0, 80}
\definecolor{R5B2}{RGB}{200, 0, 80}

\definecolor{R0B3}{RGB}{0, 0, 120}
\definecolor{R1B3}{RGB}{40, 0, 120}
\definecolor{R2B3}{RGB}{80, 0, 120}
\definecolor{R3B3}{RGB}{120, 0, 120}
\definecolor{R4B3}{RGB}{160, 0, 120}
\definecolor{R5B3}{RGB}{200, 0, 120}

\definecolor{R0B4}{RGB}{0, 0, 160}
\definecolor{R1B4}{RGB}{40, 0, 160}
\definecolor{R2B4}{RGB}{80, 0, 160}
\definecolor{R3B4}{RGB}{120, 0, 160}
\definecolor{R4B4}{RGB}{160, 0, 160}
\definecolor{R5B4}{RGB}{200, 0, 160}

\definecolor{R0B5}{RGB}{0, 0, 200}
\definecolor{R1B5}{RGB}{40, 0, 200}
\definecolor{R2B5}{RGB}{80, 0, 200}
\definecolor{R3B5}{RGB}{120, 0, 200}
\definecolor{R4B5}{RGB}{160, 0, 200}
\definecolor{R5B5}{RGB}{200, 0, 200}
\definecolor{R5B5}{RGB}{200, 0, 200}
\definecolor{lightblue}{rgb}{0.68, 0.85, 0.9}
\definecolor{medblue}{RGB}{0,191,255}

\newcommand{\tile}[2]{\draw[fill=white] (#1.12,#2.12) rectangle (#1.88,#2.88);}
\newcommand{\tileg}[3]{\draw[fill=green!#3!black] (#1.12,#2.12) rectangle (#1.88,#2.88);}

\newcommand{\tiler}[2]{\draw[fill=red!80!gray] (#1.12,#2.12) rectangle (#1.88,#2.88);}

\newcommand{\tileb}[2]{\draw[fill=blue!80!gray] (#1.12,#2.12) rectangle (#1.88,#2.88);}

\newcommand{\tiles}[2]{\draw[fill=black] (#1.12,#2.12) rectangle (#1.88,#2.88);}
\newcommand{\tilem}[3]{\draw[fill=magenta!#3!white] (#1.12,#2.12) rectangle (#1.88,#2.88);}
\newcommand{\tilec}[3]{\draw[fill=cyan!#3!white] (#1.12,#2.12) rectangle (#1.88,#2.88);}
\newcommand{\tiley}[2]{\draw[fill=yellow!80!black] (#1.12,#2.12) rectangle (#1.88,#2.88);}

\newcommand{\tileo}[2]{\draw[fill=Olive] (#1.12,#2.12) rectangle (#1.88,#2.88);}
\newcommand{\tilef}[2]{\draw[fill=Foret] (#1.12,#2.12) rectangle (#1.88,#2.88);}
\newcommand{\tileor}[2]{\draw[fill=orange] (#1.12,#2.12) rectangle (#1.88,#2.88);}
\newcommand{\tilelb}[2]{\draw[fill=lightblue] (#1.12,#2.12) rectangle (#1.88,#2.88);}
\newcommand{\tilemb}[2]{\draw[fill=medblue] (#1.12,#2.12) rectangle (#1.88,#2.88);}

\newcommand{\dotor}[2]{\fill[color=black] (#1.5,#2.5) circle (0.29); \fill[color=orange] (#1.5,#2.5) circle (0.22);}
\newcommand{\dotred}[2]{\fill[color=black] (#1.5,#2.5) circle (0.29); \fill[color=red!80!gray] (#1.5,#2.5) circle (0.22);}
\newcommand{\dotlb}[2]{\fill[color=black] (#1.5,#2.5) circle (0.29); \fill[color=lightblue] (#1.5,#2.5) circle (0.22);}
\newcommand{\dotmb}[2]{\fill[color=black] (#1.5,#2.5) circle (0.29); \fill[color=medblue] (#1.5,#2.5) circle (0.22);}
\newcommand{\dotb}[2]{\fill[color=black] (#1.5,#2.5) circle (0.29); \fill[color=blue] (#1.5,#2.5) circle (0.22);}
\newcommand{\dotw}[2]{\fill[color=black] (#1.5,#2.5) circle (0.29); \fill[color=white] (#1.5,#2.5) circle (0.22);}
\newcommand{\doty}[2]{\fill[color=black] (#1.5,#2.5) circle (0.29); \fill[color=yellow] (#1.5,#2.5) circle (0.22);}

\newcommand{\diay}[2]{\node [draw, scale=0.25, diamond,fill=yellow!50] at (#1.5,#2.5) {}; }
\newcommand{\diaor}[2]{\node [draw, scale=0.25, diamond,fill=orange] at (#1.5,#2.5) {}; }
\newcommand{\diar}[2]{\node [draw, scale=0.25, diamond,fill=red] at (#1.5,#2.5) {}; }
\newcommand{\dialb}[2]{\node [draw, scale=0.25, diamond,fill=lightblue] at (#1.5,#2.5) {}; }
\newcommand{\diamb}[2]{\node [draw, scale=0.25, diamond,fill=medblue] at (#1.5,#2.5) {}; }
\newcommand{\diab}[2]{\node [draw, scale=0.25, diamond,fill=blue] at (#1.5,#2.5) {}; }
\newcommand{\diag}[2]{\node [draw, scale=0.25, diamond,fill=green] at (#1.5,#2.5) {}; }

\title{A linear bound for the size of the finite terminal assembly of a directed non-cooperative tile assembly system}

\begin{document}

\author[1]{Sergiu Ivanov\thanks{sergiu.ivanov@univ-evry.fr}}
\author[1]{Damien Regnault\thanks{damien.regnault@univ-evry.fr}}
\affil[1]{IBISC Laboratory, Université Evry, Université Paris-Saclay, 23 Boulevard de France, Evry Courcouronnnes, 91037, Essonne, France}

\maketitle

\begin{abstract}
Introduced in \cite{Winf98}, the abstract tile assembly model (aTAM) is a model  of \textsc{DNA} self-assembly. Most of the studies focus on cooperative aTAM where a form of synchronization between the tiles is possible. Simulating Turing machines is achievable in this context. Few results and constructions are known for the non-cooperative case (a variant of Wang tilings \cite{Wang61} where assemblies do not need to cover the whole plane and some mismatches may occur). For example, assembly of a square of width $n$ is done with $2n-1$ tiles types whereas only $\Theta(\frac{\log(n)}{\log(\log(n))})$ are required for the cooperative case \cite{AdChGoHu01}. 

Introduced by P.-É. Meunier in \cite{meunier2015}, \emph{efficient} paths are a non-trivial construction for non-cooperative aTAM designed with $n$ different tile types and reaching a distance linearly greater than $n$. Improved in \cite{Meunier2019}, efficient paths were shown to be able to reach a distance of $n\log(n)$. Assembling them relies heavily on a form of ``non-determinism''. Indeed, the set of tiles may produce different finite terminal assemblies but they all contain the same efficient path. 
In this paper, we prove that this non-determinism is strictly necessary for assembling the efficient paths of \cite{Meunier2019}. 
More formally, we show that if the terminal assembly of a \emph{directed} non-cooperative tile assembly system (a model where only one terminal assembly is produced) is finite then its width and length are linear in the number of tiles. 
This result also implies that the construction of a square of width $n$ using $2n-1$ tiles types is asymptotically optimal. Moreover, we hope that the techniques introduced here will lead to a better comprehension of the non-directed case.

 \end{abstract}

\section{Introduction}


Modern processors align several hundred of billions of transistors on small chips. Independently, each transistor only simulates a simple Boolean function, but together they can achieve far more complex computations.
Developing a formalism to describe complex mathematical functions by iterating simple building blocks originates from the seminal work by Turing on the abstract Turing machine---the very first formal model of computing.  Studying both this abstract model and concrete computers derived from it became what is commonly known as computer science.


More recently, the DNA molecule emerged as a possible elementary component to carry computation instead of the conventional transistor. Indeed, Adleman \cite{Adleman} showed that algorithms can be implemented by mixing together cleverly designed DNA strands, creating a massively parallel model of computing. Several approaches relying on self-assembly were later implemented: in~\cite{RothOrigami} Rothemund developed a technique of \emph{DNA origami} to fold a long DNA strand (scaffold) using hundreds of shorter DNA strands (staples), resulting in specific shapes of size on the order of hundreds of nanometers. In~\cite{RoPaWi04} fractals were assembled out of DNA. In \cite{qian2011neural} artificial neural networks were simulated. In \cite{yurke2000dna} cyclic machines were experimentally implemented, using DNA both as machine material and as fuel.

Over time, several possible applications have emerged. On the one hand, DNA can be an efficient way to store large amounts of information for millions of years~\cite{pmid31068682}. On the other hand, DNA can be used to construct chemical seeds allowing to detect small concentrations of a target molecular compound~\cite{Minev2019.12.11.873349}. In \cite{Dietz2024} a DNA cage was built around the hepatitis B virus to neutralize it in vitro. The \emph{crisscross slats} approach in \cite{CrissCross2023} allows the assembly of large structures several microns in size, opening up the possibility of interaction with biological cells.


Abstractly, some approaches to computing with DNA can be seen as a population of tiles assembling into the target structure by attaching to a seed, which represents the input.  The corresponding theoretical model of computation called the \emph{abstract Tile Assembly Model} (aTAM) was introduced by Winfree in \cite{RotWin00,Winf98}.  The aTAM relies on square tiles with glues on each of the four sides. When the glues on two tiles match, the tiles bind together forming an assembly. This assembly grows until it becomes terminal, \emph{i.e.} until no more tiles can bind to it. This model can be implemented experimentally by representing tiles as small DNA complexes binding to a seed (a DNA origami) to carry out the computation, as shown in~\cite{WoodExp}. Furthermore, aTAM is able to simulate crisscross slats \cite{drake_et_al:LIPIcs.DNA.30.3}. Lots of variants of aTAM exist, and notably cooperative aTAM (or temperature $2$ aTAM) including a mechanism by which different parts of the assembly may synchronize. In this case, simulating a Turing machine is easily achieved and cooperative aTAM was shown to be intrinsically universal \cite{IUSA}. Due to the possible applications to DNA computing, studies also focus on the shapes which can be assembled. For example, \cite{SolWin07} focuses on the Kolmogorov complexity of the shapes and shows a strong connection between the complexity of a shape and the number of distinct tile types necessary to assemble it. The special case of the square was studied in deep detail in \cite{AdChGoHu01}, this shape being a kind of benchmark for computational expressiveness.


\subsection{Non-cooperative aTAM}

From this point on, we focus on non-cooperative aTAM. In this variant, there is no mechanism to synchronise different parts of the assembly beyond concurrency: the first path to grow blocks and prevents the growth of others. Non-cooperative aTAM can be seen as a variant of Wang tiles \cite{Wang61} where the computation starts from a seed and mismatches are allowed. Surprisingly, this model proved very hard to study, and very few results were known for a long time. Concerning decidability seen as the limit of what can be achieved, slight modifications are enough to simulate the Turing machine. In~\cite{Cook-2011} simulation of Turing machines was achieved in $3D$ and was done almost surely by introducing probabilities. It was also achieved by introducing negative glue weights \cite{Patitz-2011}, by considering polyomino tiles \cite{Fekete2014}, polygonal tiles~\cite{gilbert2015continuous}, detachable glues \cite{Jonoska2014}, or by performing the assembly in several steps~\cite{BMS-DNA2012a}. Nevertheless, the question is still open for non-cooperative aTAM. Several recent studies hint that it is likely decidable: non-cooperative aTAM is not intrinsically universal \cite{SODA2014,STOC2017}, it is not able to simulate the Turing machine in the classical way \cite{STOC2017}, and a $2D$ pumping lemma was published in \cite{STOC2020}. A weaker version of non-cooperative aTAM called \emph{directed} where only one final terminal assembly is allowed was shown to be decidable by combining \cite{Doty-2011} and \cite{STOC2020}.

Although all these results point towards a negative answer, if non-co\-op\-er\-a\-tive aTAM were undecidable, it would be a surprising breakthrough and could lead to surprising results. It is worth remembering that such a scenario had already occurred for Wang tiles: the domino problem was initially conjectured to be undecidable, and proving this conjecture wrong required exhibiting aperiodic tilings. In this context, aperiodicity is an unusual way of initializing computation. Subsequently, aperiodic tilings were extensively studied and became a model of crystallography.


Concerning the complexity of non-cooperative aTAM, assembling rectangles and squares in $3D$ was achieved in \cite{furcy2018optimal,furcy2019newBounds}, and the upper and lower bounds on the number of tile types almost match. In $2D$, lower bounds are known for the case in which mismatches are disallowed \cite{ManuchTemp1}. Concerning the main $2D$ model, the existence of efficient paths of size $n\log n$ was shown in \cite{Meunier2019} (improving on a previous construction in \cite{meunier2015}). An efficient path appears in every terminal assembly and only requires a tile set of size~$n$, whereas an auxiliary path may appear somewhere else in the terminal assembly. A limitation of this approach was also obtained in \cite{STOC2020}: it is impossible to construct efficient paths of arbitrary lengths.
Concerning the directed non-cooperative aTAM, the work \cite{meunier_et_al:LIPIcs.DNA.27.6}---an update of \cite{Doty-2011}---shows that only four kinds of terminal assemblies exist, one of them being a bi-periodic grid which must be hardcoded in the tile set. 

\subsection{Our contribution}

In this paper, we consider another one of this four kinds of assemblies: the finite terminal assembly. Our main result is a linear bound (in the number of tile types) on the size of finite terminal assemblies. Besides advancing the understanding of directed non-cooperative aTAM, this result has several major implications. Firstly, it partially solves a conjecture stated in \cite{RotWin00,Winf98} about the asymptotic bound on assembling a square which is now proven to be linear. Note that finding the exact multiplicative constant is still an open question. Secondly, it is another hint that directed non-cooperative aTAM is not able to perform any complex computation. Thirdly, our result highlights a fundamental difference with non-cooperative aTAM: concurrency is required to assemble the efficient paths exhibited in \cite{Meunier2019}.  Fourthly, our bound opens new perspectives for improving several previous results such as the pumping lemma in \cite{STOC2020}. Indeed, multiple major results rely on understanding the ``information flow'' in an assembly, like for example the window movie lemma stated in \cite{SODA2014}. Intuitively, a window delimits an area of the plane and a window movie is the order of placement, position, and glue types appearing along the window. The main theorem in~\cite{STOC2020} was obtained by focusing on a particular window movie and two of its critical glues. In this paper, we improve on the previous reasoning by showing that some information flowing through the window cannot be new and is therefore redundant. A similar approach has also recently been used to show that some fractals cannot be assembled in the cooperative case \cite{DBLP:conf/soda/BeckerHP25}. Using such improved tools may have consequences for other model and results.


This paper is structured as follows.
Section \ref{sec:def} introduces the definitions. Due to the length of the paper, we present in Section~\ref{sec:roadmap} the key ideas arranged in a roadmap, and explain how they articulate in the proof of the main result. 
The four following sections implement these ideas, solving many technical problems which commonly appear when studying temperature $1$. Sections \ref{sec:canon}, \ref{sec:U-turn}, and \ref{sec:decompo} are rather independent.  They start by a presentation of the new tools, then the related results are shown, and each section is concluded by explaining the place of the presented tools in the main proof.  The last Section \ref{sec:analysis} puts all arguments together thereby providing the definitive proof.

\section{Abstract tile assembly model}\label{sec:aTAM}
\label{sec:def}

As usual, let $\mathbb{R}$ be the set of real numbers, let $\mathbb{Z}$ be the set of all integers, let $\mathbb{N}$ be the set of all natural numbers including 0, and let $\N^*$ be the set of all natural numbers excluding 0. The domain of a function $f$ is denoted $\domain{f}$, and its range (or image) is denoted $f(\domain{f})$.

\subsection{General definitions and main theorem}

A \emph{tile type} is a unit square with four sides,
each consisting of a glue \emph{type} and a nonnegative integer \emph{strength}. Let  $T$  be a finite set of tile types.
The sides of a tile type are respectively called  north, east, south, and west. 

An \emph{assembly} is a partial function $\alpha:\mathbb{Z}^2\dashrightarrow T$ where $T$ is a set of tile types and the domain of $\alpha$ (denoted $\domain{\alpha}$) is connected.\footnote{Intuitively, an assembly is a positioning of unit-sized tiles, each from some set of tile types $T$, so that their centers are placed on (some of) the elements of the discrete plane $\mathbb{Z}^2$ and such that those elements of $\mathbb{Z}^2$ form a connected set of points.} 
Two tile types in an assembly are said to  {\em bind} (or \emph{interact}, or to be
\emph{stably attached}) if the glue types on their abutting sides are
equal, and have strength $\geq 1$.  An assembly $\alpha$ induces an undirected
weighted \emph{binding graph} $G_\alpha=(V,E)$, where $V=\domain{\alpha}$, and
there is an edge $\{ a,b \} \in E$ if and only if the tiles at positions $a$ and $b$ interact, and
this edge is weighted by the glue strength of that interaction.  The
assembly is said to be $\tau$-stable if every cut of $G_{\alpha}$ has weight at
least $\tau$.

Given two $\tau$-stable assemblies $\alpha$ and $\beta$, the union of $\alpha$ and $\beta$, written $\alpha\cup\beta$,  is an assembly defined if and only if and for all $p\in \domain{\alpha}\cap\domain{\beta}$, $\alpha(p)=\beta(p)$ and either at least one tile of $\alpha$ binds with a tile of $\beta$ or $\domain{\alpha}\cap\domain{\beta}\neq \emptyset$. Then, for all $p\in \domain{\alpha}$, we have $(\alpha\cup \beta)(p)=\alpha(p)$ and for all $p\in \domain{\beta}$, we have $(\alpha \cup \beta)(p)=\beta(p)$.

A \emph{tile assembly system} is a triple $\mathcal{T}=(T,\sigma,\tau)$,
where $T$ is a finite set of tile types, $\sigma$ is a $\tau$-stable assembly called the \emph{seed}, and
$\tau \in \mathbb{N}$ is the \emph{temperature}.
Throughout this article,  $\tau=1$.

Given two $\tau$-stable assemblies $\alpha$ and $\beta$, we say that $\alpha$ is a
\emph{subassembly} of $\beta$, and write $\alpha\sqsubseteq\beta$, if
$\domain{\alpha}\subseteq \domain{\beta}$ and for all $p\in \domain{\alpha}$,
$\alpha(p)=\beta(p)$.
We also write
$\alpha\rightarrow_1^{\mathcal{T}}\beta$ if we can obtain $\beta$ from
$\alpha$ by the binding of a single tile type, that is:  $\alpha\sqsubseteq \beta$, $|\domain{\beta}\setminus\domain{\alpha}|=1$ and the tile type at the position $\domain{\beta}\setminus\domain{\alpha}$ binds to $\alpha$ at that position.  We say that~$\gamma$ is
\emph{producible} from $\alpha$ if there is a (possibly empty)
sequence $\alpha_1,\alpha_2,\ldots,\alpha_n$ where $n \in \N \cup \{ \infty \} $, $\alpha= \alpha_1$ and $\alpha_n =\gamma$, such that
$\alpha_1\rightarrow_1^{\mathcal{T}}\alpha_2\rightarrow_1^{\mathcal{T}}\ldots\rightarrow_1^{\mathcal{T}}\alpha_n$.

The set of \emph{productions}, or \emph{producible assemblies}, of a tile assembly system $\mathcal{T}=(T,\sigma,\tau)$ is the set of all assemblies producible
from the seed assembly $\sigma$ and is written~$\prodasm{\mathcal{T}}$. An assembly $\alpha$ is called \emph{terminal} if there is no $\beta$ such that $\alpha\rightarrow_1^{\mathcal{T}}\beta$. The set of all terminal assemblies of $\mathcal{T}$ is denoted~$\termasm{\mathcal{T}}$. 
If there is a unique terminal assembly, \emph{i.e.} $|\termasm{\mathcal{T}}|=1$, then $\mathcal{T}$ is \emph{directed} and its unique terminal assembly is called $\uniterm$ along this paper. 

The easternmost column of a finite assembly $\alpha$ is designated by $e_\alpha=\max\{x: \exists y, (x,y) \in dom\{\alpha\}\}$, $w_\alpha$, $n_\alpha$, $s_\alpha$ are defined similarly for the westernmost, northernmost and southernmost columns of the assembly respectively. The \emph{horizontal width} of an assembly $\alpha$ is $e_\alpha-w_\alpha$ and its \emph{vertical height} $n_\alpha-s_\alpha$. The following theorem is the main contribution of this paper.

\begin{theorem}
\label{main:theorem}
Consider a directed tile assembly system $(T,\sigma,1)$, if its terminal assembly $\uniterm$ is finite then its horizontal width and vertical height are bounded by $\BfinalTheorem$.
\end{theorem}

For the rest of the article, we consider a directed tile assembly system $\mathcal{T}=(T,\sigma,1)$ and we suppose that its terminal assembly $\uniterm$ is finite. Also, we assume that the assembly grows the furthest away from the seed to the east, \emph{i.e.}:
$$e_\alpha-e_\sigma=\max\{e_\alpha-e_\sigma,n_\alpha-n_\sigma,|w_\alpha-w_\sigma|,|s_\alpha-s_\sigma|\}.$$
With this assumption it is sufficient to bound $e_\alpha-e_\sigma$ by $\BfinalExtremal$ to prove the main theorem. The proof of this result is constructive: we design an algorithm which takes as input an assembly $\uniterm$ such that $e_\uniterm>\BfinalExtremal$ and it will either:
\begin{itemize}
\item output a certificate that the tile assembly system is not directed;
\item bind a new tile to the assembly, showing that $\uniterm$ is not terminal.
\end{itemize}

By iterating this algorithm on a large enough finite assembly, we are able to make it grow arbitrarily large. Note that, by using the pumping lemma of \cite{STOC2020}, an infinite periodic path can be built when the assembly is large enough. 
We do not explicitly write the algorithm in this paper, but it can be directly derived from the constructive proofs, which are concluded by exhibiting a contradiction: either $\gamma$ is not terminal, or the tile assembly system is not directed.

\section{State of the art, roadmap and main obstacles}
\label{sec:roadmap}

Our proof starts by following the same steps as the ones done in \cite{STOC2020}. In Subsections \ref{sub:paths}, \ref{sub:vis}, \ref{sub:cut} and \ref{road:right}, we present the reasoning and the tools necessary to construct a specific assembly called the \emph{shield}. Our proof then diverges from previous approaches, as shields were complex to assemble, despite being useful in proving the pumping lemma in \cite{STOC2020}. Nevertheless, in Subsection \ref{road:newIdea}, we show that shields take a simple form when the terminal assembly of a directed tile assembly system is finite. In Subsection \ref{road:pseudo-visibility}, we explain the key idea on how to use shields in an original way to obtain the main result of this paper. Note that this section is only a roadmap: we identify several technical difficulties, called \emph{locks}, which are solved in the following sections.

\subsection{Paths}
\label{sub:paths}
Since there is no possible synchronization between the tiles in the non-cooperative case, we focus on \emph{paths} instead of assemblies. 
Formally, a path $P$ is a finite sequence of tiles $(P_i)_{0\leq i \leq |P|-1}$ such that each tile $P_i$ interacts with the next one $P_{i+1}$ of the sequence. Moreover, a path does not intersect with the seed and is simple: two different tiles of $P$ do not occupy the same position. The length of $P$ is the size $|P|-1$ of the sequence. If the first tile of the path interacts with a tile of the seed $\sigma$ then the path is \emph{producible}. The last tile of a producible path is the \emph{easternmost one} if and only if it is the only tile of $P$ on column $e_P$ with $e_P>e_{\sigma}$; moreover, a producible path is \emph{extremal} if  $e_P=e_\uniterm$ (we remind that $\uniterm$ is the unique terminal assembly). We say that a path $P$ is a \emph{subassembly of $\uniterm$} if and only if each tile of $P$ belongs to $\uniterm$. Of course, producible and extremal paths are subassemblies of $\uniterm$.

We now introduce some basic operations on paths. The \emph{concatenation} of a finite path $P$ with a path $Q$ is the concatenation $PQ$ of these two paths as sequences. $PQ$ is a path if and only if (1) the last tile of $P$ interacts with the first tile of $Q$ and (2) $P$ and $Q$ do not intersect. For a path $P = P_0 \ldots  P_i P_{i+1} \ldots P_j  \ldots $, we define the \emph{subpath} of P between indices $i$ and $j$, both included, in the following way: $P_{i,i+1,\ldots,j} = P_i P_{i+1} \ldots P_j$.
In the special case of a subpath where $i=0$ (resp. $j=|P|-1$), we say that $P_{0,1,\ldots,j}$ (resp. $P_{j,\ldots,|P|-1}$) is a \emph{prefix} (resp. \emph{suffix}) of $P$. Consider a vector $\overrightarrow{v} \in \mathbb{Z}^2$ and a tile $A$, the translation of $A$ by $\overrightarrow{v}$ is the tile $A+\overrightarrow{v}=(\pos{A}+\vect{v},\type{A})$. The translation of a path $P$ by $\vect{v}$ is $P+\vect{v}=(P_0+\vect{v})(P_1+\vect{v})\ldots(P_{|P|-1}+\vect{v})$. Consider two indices $0\leq i \leq j \leq |P|-1$, by abuse of notation we designate by $\vect{P_iP_j}$ the vector $\vect{\pos{P_i}\pos{P_j}}$. Also, we denote by $\glueP{i}{j}$ the glue used to bind tile $P_i$ with $P_{i+1}$. It is situated on a \emph{glue column} $c$ between the two tiles (and thus between two columns $c-0.5 < c <c+0.5$). The path $\rev{P}$ designates the path $P$ in reverse: for all $0\leq i \leq |P|-1$, we have $\rev{P}_i=P_{|P|-1-i}$. The assembly $\asm{P}$ is the representation of the path $P$ as an assembly (the order of the tiles is not taken into account anymore). This notation allows to merge two paths $P$ and $Q$ as $\asm{P} \cup \asm{Q}$. This operation can be done even if $P$ and $Q$ intersect and thus is more general than the concatenation.

\begin{figure}
\center
\begin{tikzpicture}[x=0.22cm,y=0.22cm]
\input{./tikz/AssPumpEx2}
\end{tikzpicture}
\caption{An extremal path $P$ and the different bounds used during the analysis. Along the different figures, the seed is always shown in black tiles. This path $P$ crosses column $\mathcal{B}_5$, leading to the contradiction implying that the terminal assembly is in fact infinite.}
\label{fig:extremal}
\end{figure}

The previous results of \cite{STOC2020} and the ones of this paper are obtained by showing that the farther away a path grows from the seed, the more it is constrained. Figure \ref{fig:extremal} shows an extremal path and several bounds. The main one $\Bfinal$ is the bound used to prove the main Theorem \ref{main:theorem}: if a path crosses this column then the terminal assembly is infinite. The bound $\Bzero$ shows that an extremal path cannot grow too far to the west of the seed, i.e. on the opposite side of the seed from its last tile.  
\begin{lock}
\label{lock:boundZero}
If $P$ is an extremal path, then $w_P\geq \Bzero =w_{\sigma}-2|T|-1$. 
\end{lock}

This result was also proven in \cite{pumpabilityLargeBound}, in a slightly different context.
We give a simplified proof of this result in Section \ref{sec:U-turn}.

The path $P$ will be cut in three parts. The first one is its beginning (in dark green in Figure~\ref{fig:extremal}) which is too close to the seed to be useful\footnote{That is because the \emph{visible glues} between the columns $B_0$ and $B_1$ will point east.  The notion of a visible glue is formally introduced in the next subsection.}. The middle of the path $P$ (in green in Figure \ref{fig:extremal}) will start on a column with index at least $\Bone=\BoneValue$ and will end on a column with index at most $\Btwo=\BtwoValue$ (the exact beginning and ending of this subpath of $P$ are explained in Figure~\ref{fig:visible} and Subsection \ref{road:right}). In this subpath of $P$, we aim to locate two indices $0\leq i \leq j \leq |P|-1$ such that $\glueP{i}{i+1}$ and $\glueP{j}{j+1}$ have the same type, allowing us to try to assemble $P_{0,\ldots,i}(P_{j+1,\ldots,|P|-1}+\vect{P_jP_i})$ or $P_{0,\ldots,j}(P_{i+1,\ldots,|P|-1}+\vect{P_iP_j})$. Note that copying and pasting the end of $P$ is not always possible due to possible intersections between the beginning of $P$ and its translated suffix. Thus, one key step explained in Subsections \ref{road:right} and \ref{road:newIdea} is to select $i$ and $j$ such that $P_{0,\ldots,i}(P_{j+1,\ldots,|P|-1}+\vect{P_jP_i})$ is always a producible path. Finally, the end of the path $P$  (in pale green in Figure~\ref{fig:extremal}) has to be large enough such that when it is copied and pasted, it will create a large obstacle, called a \emph{shield}. This shield will be used to find other indices where different parts of $P$ can be copied and pasted. Repeating this reasoning will  eventually lead to a contradiction. The role of bounds $\Bthree$ and $\Bfour$ will be explained in Subsection \ref{road:pseudo-visibility}, we do not give here their explicit definition since their values depend on the positions of $P_i$ and $P_j$.

\subsection{Visible glues of an extremal path}
\label{sub:vis}

All results of this paper rely on the notion of \emph{visible} glues introduced in \cite{STOC2017}. Consider an index $0\leq i \leq |P|-1$, $\glueP{i}{i+1}$ is \emph{horizontal} if both tiles are on the same line or \emph{vertical} if both tiles are on the same column. A horizontal glue points east if $P_i$ is west of $P_{i+1}$ otherwise it points west. The glue is visible from the south (resp. north) in $P$ if and only it is the southernmost (northernmost) glue of $P$ on its column and if there are no tiles of the seed on both adjacent columns. We designate by $l^i$ the \emph{glue ray} associated to $\gluePi$: $l^i$ is a ray starting at the position of $\gluePi$ going south (resp. north) if $\gluePi$ is visible from the south (resp. north). If the last tile of a producible path $P$ is the easternmost one (which is the case for an extremal path) then we have the following properties (see Figure~\ref{fig:visible}):

\begin{figure}
\center
\begin{tikzpicture}[x=0.22cm,y=0.22cm]
\input{./tikz/Visible}
\end{tikzpicture}
\caption{The path $P$ of Figure \ref{fig:extremal} with some visible glues: all glues visible from the north are pointing east (by Lemma \ref{lem:glue:prop1}) and $\glueP{i_1}{i_1+1}$ is an example of a glue visible from the south and pointing west. The $\glueP{i_3}{i_3+1}$ is the westernmost glue visible from the south and pointing east. This glue is on a glue column $c=\Bone-1$ in this example (by Lock \ref{lem:glue:prop3}, $c\leq \Bone$). The last glue of $P$, $\glueP{i_7}{i_7+1}$, is always visible from the south and from the north. The $\glueP{i_4}{i_4+1}$ and the $\glueP{i_5}{i_5+1}$ are visible from the south and north respectively and form an upward span, see Figure \ref{fig:span}.
The limit between the dark green and green subpaths is delimited by $\glueP{i_2}{i_2+1}$ which is the first glue of the span of $P$ on glue column $\Bone$. The limit between the green and pale green subpaths is delimited by the last glue of $P$ on glue column $\Btwo$ which happens to be $\glueP{i_6}{i_6+1}$ in this simple example.}
\label{fig:visible}
\end{figure}

\begin{fact}
\label{lem:glue:prop1}
Consider a producible path $P$ whose last tile is the easternmost one. Then up to some symmetries\footnote{Swapping the north and south glues of all tile types leads to horizontal symmetry of the producible assemblies.}, we can consider that all the glues of $P$ which are visible from the north point east.
\end{fact}

\begin{lemma}
\label{lem:glue:prop2}
Consider a producible path $P$ whose last glue is visible from the north. Then there exists a glue column $c$ such that all glues visible from the south on glue column $c'<c$ point west and all glues visible from the south on glue column $c'\geq c$ point east.
\end{lemma}

\begin{lock}
\label{lem:glue:prop3}
Consider a producible path $P$ whose last tile is the easternmost one and $0\leq i \leq |P|-1$ such that $\gluePi$ is visible from the south on glue column $\Bone$. Then $\gluePi$ points east.
\end{lock}

\begin{lemma}
\label{lem:glue:prop4}
Consider a producible path $P$ whose last glue is visible from the north and indices $0\leq i < j < |P|-1$ such that $\glueP{i}{i+1}$ and $\glueP{j}{j+1}$ are both visible from the south on glue columns $c$ and $c'$ respectively. If both glues points east (resp. west) then $c<c'$ (resp $c>c'$).
\end{lemma}

If $P$ is an extremal path, then these three lemmas and Lock \ref{lem:glue:prop3} imply that all glues which are visible from the north or south in $P$ on a glue column $c \geq \Bone$ point east. The proofs of these lemmas are, now, classic results available in \cite{STOC2020,STOC2017} and will not be proven here again. The proof of Lock \ref{lem:glue:prop3} is available in \cite{STOC2020} and relies on a complex intermediate lemma that will be simplified in Section \ref{sec:U-turn}.

\subsection{Cut, span and workspace}
\label{sub:cut}

Combining two visible glues of a producible path $P$ allows us to delimit a part of $\mathbb{Z}^2$ which is an entire vertical half-plane, plus some other parts depending on the exact shape of $P$ between the two visibles glues. When the visible glues are chosen appropriately, this area does not contain any tile of $\sigma$ or any tile of $P$ before the first visible glue. Then, we are free to edit $P$ in this area without fearing a conflict with one of its previous parts or the seed. More formally, consider two indices $0\leq s \leq |P|-2$ and $0\leq n \leq |P|-2$ such that $\glueP{s}{s+1}$ and $\glueP{n}{n+1}$ are visible from the south and from the north respectively in $P$. If $s\leq n$ (resp. $n\leq s$), $(s,n)$ (resp. $(n,s)$) is called an \emph{upward (resp. downward) cut} of $P$ because the bi-infinite curve made of  $l^s$, $P_{s,\ldots,n+1}$ and $l^n$ cuts the space in two areas, see Figure \ref{fig:span}. If $\glueP{s}{s+1}$ (resp. $\glueP{n}{n+1}$) points east then the seed and the beginning of $P$ are in the west side of the cut and cannot block the growth of new paths in the east side of the cut. The east side of the cut is called the \emph{workspace} of the cut $(s,n)$ (resp. $(n,s)$) of $P$, denoted by $\mathcal{C}$:

\begin{fact}
\label{fact:workspace}
Consider a producible path $P$, a cut $(s,n)$ of $P$ and a path $Q$ such that $Q_0$ interacts with $P_s$  and $Q$ is in the workspace $\mathcal{C}$ of the cut $(s,n)$ of $P$\footnote{Note that tiles are placed on the discrete plane $\mathbb{Z}^2$ but a workspace is a part of $\mathbb{R}^2$. For a path $Q$, this implies that $Q$ is in $\mathcal{C}$ if and only if the tiles of $Q$ are in the area but also the segments linking the different tiles (including the glues).}. Then $P_{0,\ldots,s}Q$ is producible.
\end{fact}

\begin{figure}
\center
\begin{minipage}{0.47\linewidth}
\begin{tikzpicture}[x=0.22cm,y=0.22cm]

\fill[fill=blue!30!white, draw opacity=0.8] (15,0) |- (16.5,3.5) |- (12.5,5.5) |- (16.5,9.5) |- (12.5,11.5) |- (15,15.5) |- (26,21) |- (15,0);

\draw[very thick] (4.5,9.5) -| (8.5,7.5) -| (3.5,3.5) -| (1.5,11.5) -| (10.5,5.5) -| (6.5,3.5) -| (16.5,5.5) -| (12.5,9.5) -| (16.5,11.5) -| (12.5,15.5) -| (15,15.5);

\draws{4}{9}

\tileg{5}{9}{85}
\tileg{6}{9}{85}
\tileg{7}{9}{85}
\tileg{8}{9}{85}
\tileg{8}{8}{85}
\tileg{8}{7}{85}
\tileg{7}{7}{85}
\tileg{6}{7}{85}
\tileg{5}{7}{85}
\tileg{4}{7}{85}
\tileg{3}{7}{85}
\tileg{3}{6}{85}
\tileg{3}{5}{85}
\tileg{3}{4}{85}
\tileg{3}{3}{85}
\tileg{2}{3}{85}
\tileg{1}{3}{85}
\tileg{1}{4}{85}
\tileg{1}{5}{85}
\tileg{1}{6}{85}
\tileg{1}{7}{85}
\tileg{1}{8}{85}
\tileg{1}{9}{85}
\tileg{1}{10}{85}
\tileg{1}{11}{85}
\tileg{2}{11}{85}
\tileg{3}{11}{85}
\tileg{4}{11}{85}
\tileg{5}{11}{85}
\tileg{6}{11}{85}
\tileg{7}{11}{85}
\tileg{8}{11}{85}
\tileg{9}{11}{85}
\tileg{10}{11}{85}
\tileg{10}{10}{85}
\tileg{10}{9}{85}
\tileg{10}{8}{85}
\tileg{10}{7}{85}
\tileg{10}{6}{85}
\tileg{10}{5}{85}
\tileg{9}{5}{85}
\tileg{8}{5}{85}
\tileg{7}{5}{85}
\tileg{6}{5}{85}
\tileg{6}{4}{85}
\tileg{6}{3}{85}
\tileg{7}{3}{85}
\tileg{8}{3}{85}
\tileg{9}{3}{85}
\tileg{10}{3}{85}
\tileg{11}{3}{85}
\tileg{12}{3}{85}
\tileg{13}{3}{85}
\tileg{14}{3}{85}

\tiley{15}{3}{48}
\tiley{16}{3}{48}
\tiley{16}{4}{48}
\tiley{16}{5}{48}
\tiley{15}{5}{48}
\tiley{14}{5}{48}
\tiley{13}{5}{48}
\tiley{12}{5}{48}
\tiley{12}{6}{48}
\tiley{12}{7}{48}
\tiley{12}{8}{48}
\tiley{12}{9}{48}
\tiley{13}{9}{48}
\tiley{14}{9}{48}
\tiley{15}{9}{48}
\tiley{16}{9}{48}
\tiley{16}{10}{48}
\tiley{16}{11}{48}
\tiley{15}{11}{48}
\tiley{14}{11}{48}
\tiley{13}{11}{48}
\tiley{12}{11}{48}
\tiley{12}{12}{48}
\tiley{12}{13}{48}
\tiley{12}{14}{48}
\tiley{12}{15}{48}
\tiley{13}{15}{48}
\tiley{14}{15}{48}



\path [dotted, draw, thin] (0,0) grid[step=0.22cm] (26,21);
\draw[dashed] (8,11.5) -| (8,21);

\fill (8,11.5) circle (0.16);
\node (D) at (8,22.2) {$\Bone$};
\draw[dashed] (15,15.5) -| (15,21);
\fill (15,15.5) circle (0.16);
\draw[dashed] (15,0) -| (15,3.5);
\fill (15,3.5) circle (0.16);
\node (D) at (14,2.1) {$s$};
\node (D) at (16,1) {$l^s$};
\node (D) at (16,20) {$l^n$};
\node (D) at (14,16.7) {$n$};
\node (D) at (23,18.7) {$\mathcal{C}$};
\end{tikzpicture}

\center (a) 
\end{minipage}
\begin{minipage}{0.47\linewidth}
\begin{tikzpicture}[x=0.22cm,y=0.22cm]
\fill[fill=blue!30!white, draw opacity=0.8] (15,0) |- (16.5,3.5) |- (12.5,5.5) |- (16.5,9.5) |- (12.5,11.5) |- (15,15.5) |- (26,21) |- (15,0);

\draw[very thick] (4.5,9.5) -| (8.5,7.5) -| (3.5,3.5) -| (1.5,11.5) -| (10.5,5.5) -| (6.5,3.5) -| (15,3.5);
\draw[dashed] (15,3.5) -| (16.5,5.5) -| (12.5,9.5) -| (16.5,11.5) -| (12.5,15.5) -| (15,15.5);
\draw[very thick] (15,3.5) -| (20.5,7.5) -| (19.5,10.5) -| (17.5, 7.5) -| (14.5,7.5);

\draws{4}{9}

\tileg{5}{9}{85}
\tileg{6}{9}{85}
\tileg{7}{9}{85}
\tileg{8}{9}{85}
\tileg{8}{8}{85}
\tileg{8}{7}{85}
\tileg{7}{7}{85}
\tileg{6}{7}{85}
\tileg{5}{7}{85}
\tileg{4}{7}{85}
\tileg{3}{7}{85}
\tileg{3}{6}{85}
\tileg{3}{5}{85}
\tileg{3}{4}{85}
\tileg{3}{3}{85}
\tileg{2}{3}{85}
\tileg{1}{3}{85}
\tileg{1}{4}{85}
\tileg{1}{5}{85}
\tileg{1}{6}{85}
\tileg{1}{7}{85}
\tileg{1}{8}{85}
\tileg{1}{9}{85}
\tileg{1}{10}{85}
\tileg{1}{11}{85}
\tileg{2}{11}{85}
\tileg{3}{11}{85}
\tileg{4}{11}{85}
\tileg{5}{11}{85}
\tileg{6}{11}{85}
\tileg{7}{11}{85}
\tileg{8}{11}{85}
\tileg{9}{11}{85}
\tileg{10}{11}{85}
\tileg{10}{10}{85}
\tileg{10}{9}{85}
\tileg{10}{8}{85}
\tileg{10}{7}{85}
\tileg{10}{6}{85}
\tileg{10}{5}{85}
\tileg{9}{5}{85}
\tileg{8}{5}{85}
\tileg{7}{5}{85}
\tileg{6}{5}{85}
\tileg{6}{4}{85}
\tileg{6}{3}{85}
\tileg{7}{3}{85}
\tileg{8}{3}{85}
\tileg{9}{3}{85}
\tileg{10}{3}{85}
\tileg{11}{3}{85}
\tileg{12}{3}{85}
\tileg{13}{3}{85}
\tileg{14}{3}{85}

\tilem{15}{3}{80}
\tilem{16}{3}{80}
\tilem{17}{3}{80}
\tilem{18}{3}{80}
\tilem{19}{3}{80}
\tilem{20}{3}{80}
\tilem{20}{4}{80}
\tilem{20}{5}{80}
\tilem{20}{6}{80}
\tilem{20}{7}{80}
\tilem{19}{7}{80}
\tilem{19}{8}{80}
\tilem{19}{9}{80}
\tilem{19}{10}{80}
\tilem{18}{10}{80}
\tilem{17}{10}{80}
\tilem{17}{9}{80}
\tilem{17}{8}{80}
\tilem{17}{7}{80}
\tilem{16}{7}{80}
\tilem{15}{7}{80}
\tilem{14}{7}{80}
\path [dotted, draw, thin] (0,0) grid[step=0.22cm] (26,21);
\draw[dashed] (8,11.5) -| (8,21);

\fill (8,11.5) circle (0.16);
\node (D) at (8,22.2) {$\Bone$};
\draw[dashed] (15,15.5) -| (15,21);
\draw[dashed] (15,0) -| (15,3.5);
\fill (15,3.5) circle (0.16);
\node (D) at (14,2.1) {$s$};
\node (D) at (23,18.7) {$\mathcal{C}$};

\end{tikzpicture}

\center (b) 
\end{minipage}

\caption{(a) The upward span $(s,n)$ of path $P$ from Figure \ref{fig:visible}, with $s = i_4$ and $n = i_5$. The subpath $P_{0, \dots, s}$ is shown in green, and $P_{s+1, \dots, n}$ in yellow. This span points east.
The workspace $\mathcal{C}$ of this span is drawn in blue and is delimited by $l^s$ (the glue ray of $\glueP{s}{s+1}$), $P_{s,\ldots,n+1}$ and $l^n$ (the glue ray of $\glueP{n}{n+1}$).
(b) A path $Q$ (in magenta) which is in $\mathcal{C}$ and such that $Q_0$ interacts with $P_s$. According to Fact \ref{fact:workspace}, $P_{0,\ldots,s}Q$ is producible.}
\label{fig:span}
\end{figure}

Along the paper, we will often consider the special case where the two visible glues $s$ and $n$ are on the same glue column. In this case, the cut is called the \emph{span} of $P$ on glue column $c$. The type of a span is the type of $\glueP{s}{s+1}$ (resp. $\glueP{n}{n+1}$), it points in the same direction as $\glueP{s}{s+1}$ (resp. $\glueP{n}{n+1}$), and its width is $y_{P_{n}}-y_{P_s}$.

\subsection{Classical reasoning, right-priority and shield}
\label{road:right}

The result of \cite{STOC2020} is obtained by combining two spans of the same type and direction. Indeed, consider an extremal path $P$ and two glue columns $\Bone\leq c < c' \leq \Btwo$ such that the spans $(s,n)$ and $(s',n')$ of $P$ on glue columns $c$ and $c'$ respectively have the same characteristics (they are both upward or downward spans, point east and have the same type). Then we can try to assemble $P_{0,\ldots,s}(P_{s'+1,\ldots, |P|-1}-\vect{P_sP_s'})$ or $P_{0,\ldots,s'}(P_{s+1,\ldots, |P|-1}+\vect{P_sP_s'})$. To explain what may occur when trying to assemble such paths, we use the extremal path $P$ of Figure \ref{fig:shield}(a) as an example. 

\input{./tikz/shield}

In Figure \ref{fig:shield}(a), the spans $(s,n)$ and $(s',n')$ are upward spans and $\mathcal{D}$ is the workspace of $(s',n')$. In Figure \ref{fig:shield}(b), we try to assemble the path $Q=P_{s+1,\ldots, |P|-1}+\vect{P_sP_s'}$ at the end of $P_{0,\ldots,s'}$. While doing so, the following important events occur in this example. Firstly, there exists an index $0\leq i \leq |Q|-1$ such that $Q_{0,\ldots,i}$ grows along $P_{s'+1,\ldots,|P|-1}$ by putting its tiles in the same positions, \emph{i.e.} for all $0\leq t \leq i$, we have $\pos{Q_t}=\pos{P_{s'+1+t}}$ (the tiles in orange in Figure \ref{fig:shield}(b)). Secondly, $Q$ turns right of $P_{s',\ldots,|P|-1}$ at $\pos{Q_i}$  (for a formal definition of this notion see \cite{STOC2020,STOC2017}) and stays inside $\mathcal{D}$ until it intersects with $P$ at $\pos{P_k}=\pos{Q_j}$. Thirdly, $Q$ leaves $\mathcal{D}$ after $Q_j$, \emph{i.e.} no tiles of $Q_{j+1,\ldots,|Q|-1}$ are inside $\mathcal{D}$. Fourthly, $Q$ intersects with $P_{0,\ldots,s'}$. Thus $P_{0,\ldots,s'}Q$ is not self-avoiding and is not a path but $P_{0,\ldots,s'}Q_{0,\ldots,j}$ is a path and it is producible by Fact \ref{fact:workspace}. Moreover, since the tile assembly system is directed then we have  $P_k=Q_j$ and $Q_{0,\ldots,i}=P_{s'+1,\ldots,s'+1+i}$ is the largest common prefix of $Q$ and $P_{s'+1,\ldots,|P|-1}$.

Now, we focus on the fact that $Q$ turns right of $P_{s'+1,\ldots, |P|-1}$. This event implies that $P_{s'+1,\ldots,|P|-1}-\vect{P_sP_{s'}}$ turns left of $P_{s,\ldots,|P|-1}$, see Figure \ref{fig:shield}(d). By doing so, $P_{s'+1,\ldots,|P|-1}-\vect{P_sP_{s'}}$ immediately leaves the workspace $\mathcal{C}$ of the span $(s,n)$ of $P$. In this example, $P_{s'+1,\ldots,|P|-1}-\vect{P_sP_{s'}}$ quickly collides with $P_{0,\ldots,s}$. 
 Then, in the general case we have three possibilities:

\begin{fact}
\label{fact:turn}
Consider an extremal path $P$ and two glue columns $\Bone\leq c < c' \leq \Btwo$ such that the spans $(s,n)$ and $(s',n')$ of $P$ on glue columns $c$ and $c'$ respectively have the same characteristics. Then one of the following holds:
\begin{enumerate}
\item $P_s(P_{s'+1,\ldots,|P|-1}-\vect{P_sP_{s'}})$ turns right of $P_{s,\ldots,n+1}$;
\item or $P_{s'}(P_{s+1,\ldots,|P|-1}+\vect{P_sP_{s'}})$ turns right of $P_{s',\ldots,|P|-1}$;
\item or the width of span $(s',n')$ is less or equal to the width of span $(s,n)$. 
\end{enumerate}
\end{fact} 

When the first case of Fact \ref{fact:turn} occurs, a suffix of $(P_{s'+1,\ldots,|P|-1}-\vect{P_sP_{s'}})$ will be used as a shield. 
For the third case of Fact \ref{fact:turn}, it was proven in \cite{STOC2020} that in this case the path $P$ is pumpable or fragile which means that either the tile assembly system is not directed or that the terminal assembly is infinite. These events cannot occur in our setting. We do not explain this argument in the roadmap. 

In Subsection \ref{road:newIdea}, we will consider a specific setting where the second case of Fact \ref{fact:turn} also leads to a contradiction. To do so, consider the tile $P_k=Q_j$ of Figure \ref{fig:shield}:  it is the only intersection between $Q_{i+1,\ldots,|Q|-1}$ and $P_{s'+1+i,\ldots,|P|-1}$ in this simple example. Since the tile assembly system is directed, it is possible to ``merge'' these two paths together. Indeed, $R=P_{0,\ldots,s'}Q_{0,\ldots,j}P_{k+1,\ldots,|P|-1}$ is a producible extremal path, see Figure \ref{fig:shield}(c). In \cite{STOC2020}, this operation is repeated several times to build a path called a \emph{shield}. From this point on, our reasoning differs from \cite{STOC2020} and we introduce a setting where shields take a very simple form. Indeed, $P$ and $R$ are both extremal paths. Moreover, they share the common prefix $P_{0,\ldots, s'+1+i}$. Nevertheless the path $R$ turns right of $P$ at tile $P_i$. Then, we say that $R$ is more \emph{right-priority} than $P$. This notion of right-priority implies an order on extremal paths. Moreover, since we consider a finite terminal assembly, there is a finite number of extremal paths. Thus, it is possible to choose a \emph{rightmost-priority} one. For such a path, it is not possible to be in a case like in Figure~\ref{fig:shield}(c) where there is an intersection between $P_{s'+i+1,\ldots,|P|-1}$ and $Q_{i+1,\ldots,|Q|-1}$. 

\subsection{Shield in the directed case with a finite terminal assembly}
\label{road:newIdea}

Our aim is to work in a setting where only the first case of Fact \ref{fact:turn} is possible. To achieve this goal, we consider an extremal path $P$, two glue columns $\Bone \leq c < c' \leq \Btwo$, and we let $(s,n)$ be the span of $P$ on glue column $c$ and $(s',n')$ the span of $P$ on glue column $c'$, for an example see Figure \ref{fig:TwoSpan}. 

\begin{figure}
\center
\begin{tikzpicture}[x=0.22cm,y=0.22cm]

\draw[very thick] (4.5,9.5) -| (8.5,7.5) -| (3.5,3.5) -| (1.5,11.5) -| (10.5,5.5) -| (6.5,3.5) -| (16.5,5.5) -| (12.5,9.5) -| (16.5,11.5) -| (12.5,15.5) -| (18.5,17.5) -| (16.5,19.5) -| (29.5,17.5) -| (20.5,15.5) -| (59.5,15.5);

\draws{4}{9}

\tileg{5}{9}{85}
\tileg{6}{9}{85}
\tileg{7}{9}{85}
\tileg{8}{9}{85}
\tileg{8}{8}{85}
\tileg{8}{7}{85}
\tileg{7}{7}{85}
\tileg{6}{7}{85}
\tileg{5}{7}{85}
\tileg{4}{7}{85}
\tileg{3}{7}{85}
\tileg{3}{6}{85}
\tileg{3}{5}{85}
\tileg{3}{4}{85}
\tileg{3}{3}{85}
\tileg{2}{3}{85}
\tileg{1}{3}{85}
\tileg{1}{4}{85}
\tileg{1}{5}{85}
\tileg{1}{6}{85}
\tileg{1}{7}{85}
\tileg{1}{8}{85}
\tileg{1}{9}{85}
\tileg{1}{10}{85}
\tileg{1}{11}{85}
\tileg{2}{11}{85}
\tileg{3}{11}{85}
\tileg{4}{11}{85}
\tileg{5}{11}{85}
\tileg{6}{11}{85}
\tileg{7}{11}{85}
\tileg{8}{11}{80}
\tileg{9}{11}{85}
\tileg{10}{11}{85}
\tileg{10}{10}{85}
\tileg{10}{9}{85}
\tileg{10}{8}{85}
\tileg{10}{7}{85}
\tileg{10}{6}{85}

\tileg{10}{5}{85}
\tileg{9}{5}{85}
\tileg{8}{5}{85}
\tileg{7}{5}{85}
\tileg{6}{5}{85}
\tileg{6}{4}{85}
\tileg{6}{3}{85}
\tileg{7}{3}{85}
\tileg{8}{3}{85}
\tileg{9}{3}{85}
\tileg{10}{3}{85}
\tileg{11}{3}{85}
\tileg{12}{3}{85}
\tileg{13}{3}{85}
\tileg{14}{3}{85}

\tiley{15}{3}{48}
\tiley{16}{3}{48}
\tiley{16}{4}{48}
\tiley{16}{5}{48}
\tiley{15}{5}{48}
\tiley{14}{5}{48}
\tiley{13}{5}{48}
\tiley{12}{5}{48}
\tiley{12}{6}{48}
\tiley{12}{7}{48}
\tiley{12}{8}{48}
\tiley{12}{9}{48}
\tiley{13}{9}{48}
\tiley{14}{9}{48}
\tiley{15}{9}{48}
\tiley{16}{9}{48}
\tiley{16}{10}{48}
\tiley{16}{11}{48}
\tiley{15}{11}{48}
\tiley{14}{11}{48}
\tiley{13}{11}{48}
\tiley{12}{11}{48}
\tiley{12}{12}{48}
\tiley{12}{13}{48}
\tiley{12}{14}{48}
\tiley{12}{15}{48}
\tiley{13}{15}{48}
\tiley{14}{15}{48}
\tiley{15}{15}{48}
\tiley{16}{15}{48}
\tiley{17}{15}{48}
\tiley{18}{15}{48}
\tiley{18}{16}{48}
\tiley{18}{17}{48}
\tiley{17}{17}{48}
\tiley{16}{17}{48}
\tiley{16}{18}{48}

\tiley{16}{18}{48}
\tiley{16}{19}{48}
\tileb{17}{19}{48}
\tileb{18}{19}{48}
\tileb{19}{19}{48}
\tileb{20}{19}{48}
\tileb{20}{19}{48}
\tileb{21}{19}{48}
\tileb{22}{19}{48}
\tileb{23}{19}{48}
\tileb{24}{19}{48}
\tileb{25}{19}{48}
\tileb{26}{19}{48}
\tileb{27}{19}{48}
\tileb{28}{19}{48}
\tileb{29}{19}{48}
\tileb{29}{18}{48}
\tileb{29}{17}{48}
\tileb{28}{17}{48}
\tileb{27}{17}{48}
\tileb{26}{17}{48}
\tileb{25}{17}{48}
\tileb{24}{17}{48}
\tileb{23}{17}{48}
\tileb{22}{17}{48}
\tileb{21}{17}{48}
\tileb{20}{17}{48}
\tileb{20}{16}{48}
\tileb{20}{15}{48}
\tileb{21}{15}{48}
\tileb{22}{15}{48}
\tileb{23}{15}{48}
\tileb{24}{15}{48}
\tileb{25}{15}{48}
\tileb{26}{15}{48}
\tileb{27}{15}{48}
\tileb{28}{15}{48}
\tileb{28}{15}{48}
\tileb{29}{15}{48}
\tileb{30}{15}{48}
\tileb{31}{15}{48}
\tileb{32}{15}{48}
\tileb{33}{15}{48}
\tileb{34}{15}{48}
\tileb{35}{15}{48}
\tileb{36}{15}{48}
\tileb{37}{15}{48}
\tileb{38}{15}{48}
\tileb{39}{15}{48}
\tileb{40}{15}{48}
\tileb{41}{15}{48}
\tileb{42}{15}{48}
\tileb{43}{15}{48}
\tileb{44}{15}{48}
\tileb{45}{15}{48}
\tileb{46}{15}{48}
\tileb{47}{15}{48}
\tileb{48}{15}{48}
\tileb{49}{15}{48}
\tileb{50}{15}{48}
\tileb{51}{15}{48}
\tileb{52}{15}{48}
\tileb{53}{15}{48}
\tileb{54}{15}{48}
\tileb{55}{15}{48}
\tileb{56}{15}{48}
\tileb{57}{15}{48}
\tileb{58}{15}{48}
\tileb{59}{15}{78}

\path [dotted, draw, thin] (0,0) grid[step=0.22cm] (60,21);

\draw [dashed] (15,0) -| (15,3.5);
\draw [dashed] (15,15.5) -| (15,21);

\node (D) at (15,-2) {$c$};

\fill (14.5,3.5) circle (0.16);
\node (D) at (14,2) {$P_{s}$};
\fill (14.5,15.5) circle (0.16);
\node (D) at (14,17) {$P_{n}$};

\draw [dashed] (17,0) -| (17,15.5);
\draw [dashed] (17,19.5) -| (17,21);
\node (D) at (17,-1.8) {$c'$};

\fill (16.5,15.5) circle (0.16);
\node (D) at (16,14) {$P_{s'}$};
\fill (16.5,19.5) circle (0.16);
\node (D) at (16,21) {$P_{n'}$};
\end{tikzpicture}
\caption{The path $P$ of Figure \ref{fig:extremal}. We consider two upward spans $(s,n)$ and $(s',n')$ of $P$ on columns $\Bone\leq c <c' \leq \Btwo$ respectively. Both spans point east and we suppose that $\glueP{s}{s+1}=\glueP{s'}{s'+1}$.}
\label{fig:TwoSpan}
\end{figure}

To eliminate the second case of Fact \ref{fact:turn} where $P_{s'}(P_{s+1,\ldots,|P|-1}+\vect{P_sP_{s'}})$ turns right of $P_{s',\ldots,|P|-1}$, we consider that $P$ is not an arbitrary chosen path but that $P$ is chosen wisely to have the following property.

\begin{definition}
\label{prop:canonicalPath}
Consider an extremal path $P$, a glue column $\Bone\leq c \leq \Btwo$ and the span $(s,n)$ of $P$ on glue column $c$. If $(s,n)$ is an upward (resp. downward) span then consider the set of extremal paths $\mathcal{P}$ such that a path $Q$ belongs to $\mathcal{P}$ if and only if $P_{0,\ldots,s+1}$ is a prefix of $Q$ and $\glueQ{s}{s+1}$ is visible from the south (resp. north) in $Q$. Then $P$ is a \emph{good} path for glue column $c$ if and only if it is the rightmost (resp. leftmost) priority path of~$\mathcal{P}$.
\end{definition}

\begin{figure}
\center
\begin{tikzpicture}[x=0.22cm,y=0.22cm]

\fill[fill=yellow!85!black, draw opacity=0.5] (17,0) |- (18.5,15.5) |- (16.5,17.5) |- (17,19.5) |- (62,21) |- (16,0);
\fill[fill=yellow!90!black, draw opacity=0.5] (17,0) |- (18.5,15.5) |- (16.5,17.5) |- (29.5,19.5) |- (20.5,17.5) |- (59.5,15.5) |- (17,0);

\draw[very thick] (16.5,15.5) -| (18.5,17.5) -| (16.5,19.5) -| (29.5,17.5) -| (20.5,15.5) -| (59.5,15.5);
\draw[very thick] (16.5,15.5) -| (18.5,10.5) -| (27.5,7.5) -| (34.5,11.5) -| (39,11.5);

%
%

\tileg{16}{15}{85}
\tileor{17}{15}{80}
\tileor{18}{15}{80}
\tiley{18}{16}{48}
\tiley{18}{17}{48}
\tiley{17}{17}{48}
\tiley{16}{17}{48}
\tiley{16}{18}{48}

\tiley{16}{18}{48}
\tiley{16}{19}{48}
\tileb{17}{19}{48}
\tileb{18}{19}{48}
\tileb{19}{19}{48}
\tileb{20}{19}{48}
\tileb{20}{19}{48}
\tileb{21}{19}{48}
\tileb{22}{19}{48}
\tileb{23}{19}{48}
\tileb{24}{19}{48}
\tileb{25}{19}{48}
\tileb{26}{19}{48}
\tileb{27}{19}{48}
\tileb{28}{19}{48}
\tileb{29}{19}{48}
\tileb{29}{18}{48}
\tileb{29}{17}{48}
\tileb{28}{17}{48}
\tileb{27}{17}{48}
\tileb{26}{17}{48}
\tileb{25}{17}{48}
\tileb{24}{17}{48}
\tileb{23}{17}{48}
\tileb{22}{17}{48}
\tileb{21}{17}{48}
\tileb{20}{17}{48}
\tileb{20}{16}{48}
\tileb{20}{15}{48}
\tileb{21}{15}{48}
\tileb{22}{15}{48}
\tileb{23}{15}{48}
\tileb{24}{15}{48}
\tileb{25}{15}{48}
\tileb{26}{15}{48}
\tileb{27}{15}{48}
\tileb{28}{15}{48}
\tileb{28}{15}{48}
\tileb{29}{15}{48}
\tileb{30}{15}{48}
\tileb{31}{15}{48}
\tileb{32}{15}{48}
\tileb{33}{15}{48}
\tileb{34}{15}{48}
\tileb{35}{15}{48}
\tileb{36}{15}{48}
\tileb{37}{15}{48}
\tileb{38}{15}{48}
\tileb{39}{15}{48}
\tileb{40}{15}{48}
\tileb{41}{15}{48}
\tileb{42}{15}{48}
\tileb{43}{15}{48}
\tileb{44}{15}{48}
\tileb{45}{15}{48}
\tileb{46}{15}{48}
\tileb{47}{15}{48}
\tileb{48}{15}{48}
\tileb{49}{15}{48}
\tileb{50}{15}{48}
\tileb{51}{15}{48}
\tileb{52}{15}{48}
\tileb{53}{15}{48}
\tileb{54}{15}{48}
\tileb{55}{15}{48}
\tileb{56}{15}{48}
\tileb{57}{15}{48}
\tileb{58}{15}{48}
\tileb{59}{15}{78}

\tilem{18}{14}{80}
\tilem{18}{13}{80}
\tilem{18}{12}{80}
\tilem{18}{11}{80}
\tilem{18}{10}{80}
\tilem{19}{10}{80}
\tilem{20}{10}{80}
\tilem{21}{10}{80}
\tilem{22}{10}{80}
\tilem{23}{10}{80}
\tilem{24}{10}{80}
\tilem{25}{10}{80}
\tilem{26}{10}{80}
\tilem{27}{10}{80}
\tilem{27}{9}{80}
\tilem{27}{8}{80}
\tilem{27}{7}{80}
\tilem{28}{7}{80}
\tilem{29}{7}{80}
\tilem{30}{7}{80}
\tilem{31}{7}{80}
\tilem{32}{7}{80}
\tilem{33}{7}{80}
\tilem{34}{7}{80}
\tilem{34}{8}{80}
\tilem{34}{9}{80}
\tilem{34}{10}{80}
\tilem{34}{11}{80}
\tilem{35}{11}{80}
\tilem{36}{11}{80}
\tilem{37}{11}{80}
\tilem{38}{11}{80}

\path [dotted, draw, thin] (14,0) grid[step=0.22cm] (62,21);

\draw [dashed] (17,0) -| (17,15.5);
\draw [dashed] (17,19.5) -| (17,21);

\fill (16.5,15.5) circle (0.16);
\node (D) at (16,14) {$P_{s'}$};
\fill (16.5,19.5) circle (0.16);
\node (D) at (16,21) {$P_{n'}$};

\draw [dashed] (59.5,0) -| (59.5,15.5);
\fill (59.5,15.5) circle (0.16);
\node (D) at (59,17) {$P_{|P|-1}$};

\node (D) at (50,19) {$\mathcal{D}^-$};
\node (D) at (50,3) {$\mathcal{D}^+$};

\node (D) at (40,11.5) {\huge ?};

\fill (18.5,15.5) circle (0.16);

\end{tikzpicture}
\caption{Only the end $P_{s',\ldots,|P|-1}$ of the path $P$ of Figure \ref{fig:TwoSpan} is drawn. The workspace $\mathcal{D}$ of the span $(s',n')$ is in yellow. This workspace is cut in two areas $\mathcal{D}^-$ and $\mathcal{D}^+$ by using the ray starting in $\pos{P_{|P|-1}}$ and going south. A path $Q$ starting at tile $P_{s'+1}$ is shown in orange and magenta. This path turns right of $P_{s',\ldots,|P|-1}$ and it cannot leave $\mathcal{D}^+$ without intersecting the column $e_P+1$, $l^{s'}$ or $P_{s'+1,\ldots, |P|-1}$. The tiles in orange are shared between $Q$ and $P_{s+1,\ldots,|P|-1}$.}
\label{fig:TwoSpanPush}
\end{figure}

After the roadmap, we will use a slightly weaker version of this property and we will consider \emph{canonical} paths instead of good paths. Nevertheless, we do not enter into this level of details yet. In particular, we explain how to obtain a good/canonical path for a given glue column $c$ only in Section~\ref{sec:canon}. 

\begin{lock}
\label{lock:can}
When the span $(s,n)$ of a path $P$ on glue column $\Bone\leq c \leq \Btwo$ is considered, the path $P$ must be a good/canonical path for glue column $c$.
\end{lock}

For the moment, assume that the extremal path $P$ is a good path for glue column $c'$. If the second case of Fact \ref{fact:turn} occurs, $P_{s'}(P_{s+1,\ldots,|P|-1}+\vect{P_sP_{s'}})$ turns right of $P_{s',\ldots,|P|-1}$ and a contradiction can be found. Indeed, let $Q=P_{s+1,\ldots,|P|-1}+\vect{P_sP_{s'}}$ and suppose that $P_{0,\ldots,s'}Q$ is a producible path. Since $c < c'$ then $x_{P_s}<x_{P_{s'}}$ and $e_Q>e_P=e_\uniterm$. This is a contradiction with the fact that $P$ is extremal, since $P_{0,\dots,s'} Q$ goes further east than $P$ by $c' - c$ columns. Then by Fact \ref{fact:workspace}, $Q$ must leave the workspace $\mathcal{D}$ of the span $(s',n')$ of $P$ but if $Q$ turns right of $P_{s',\ldots,|P|-1}$, all cases lead to contradictions (see Figure \ref{fig:TwoSpanPush}):
\begin{itemize}
\item if $Q$ crosses $l^{s'}$: then $Q-\vect{P_sP_{s'}}=P_{s+1,\ldots,|P|-1}$ intersects $l^{s'}-\vect{P_sP_{s'}}=l^s$ contradicting the fact that $\glueP{s}{s+1}$ is visible from the south;
\item if $Q$ intersects $P_{s'+1,\ldots,|P|-1}$: then either $P$ intersects itself or $P$ is not a good path for glue column $c'$ because a more right-priority path exists (by a reasoning similar to the one of Figure \ref{fig:shield}(c));
\item if $Q$ crosses $l^{n'}$: then to reach this ray requires either to cross $P_{s'+1,\ldots,|P|-1}$ which is the previous case or to put a tile on column $x_{|P|-1}+1=e_\uniterm+1$ which contradicts that $P$ is extremal. 
\end{itemize}
Then the second case of Fact \ref{fact:turn} cannot occur if $P$ is a good path for column $c'$, meaning that only the first case is possible, in which $P_{s',\dots,|P|-1} - \vect{P_s P_{s'}}$ must turn right of $P_{s,\ldots,n+1}$.
Similarly, assume that $P$ is also a good path for glue column $c$ and consider the workspace $\mathcal{C}$ of the span $(s,n)$ of $P$ (see Figure~\ref{fig:SpanPull}). This workspace can be cut in two areas $\mathcal{C}^+$ and $\mathcal{C}^-$ by the ray starting at $\pos{P_{|P|-1}}$ and going south. The path $P_{s'+1,\ldots,|P|-1}-\vect{P_sP_{s'}}$ can be written as $CS$ where $C$ is the largest common prefix between  $P_{s'+1,\ldots,|P|-1}-\vect{P_sP_{s'}}$ and $P_{s+1,\ldots,n}$, and S is the rest of the path $P_{s'+1,\ldots,|P|-1}-\vect{P_sP_{s'}}$. The path $S$ is a called a \emph{shield} of span $(s,n)$ and by a similar reasoning $S$ is inside the workspace $\mathcal{C}^+$ and does not insect with $P$. Indeed, the path $S$ cannot intersect with $P$ (since $P$ is simple and a good path for glue column~$c$), cannot reach column $x_{P_{|P|-1}}+1$ (since $P$ is extremal) or cross $l^s$ (since $P_{s',\ldots,|P|-1}$ cannot intersect $l^{s'}$). Note that the exact definition of a shield is more involved and is given in Section \ref{sec:analysis}. Nevertheless, the definition of $S$ as a shield is correct in the context of good paths.

\begin{figure}
\center
\begin{tikzpicture}[x=0.22cm,y=0.22cm]

\fill[fill=blue!30!white, draw opacity=0.8] (15,0) |- (16.5,3.5) |- (12.5,5.5) |- (16.5,9.5) |- (12.5,11.5) |- (15,15.5) |- (60,21) |- (15,0);
\fill[fill=blue!40!white, draw opacity=0.8] (15,0) |- (16.5,3.5) |- (12.5,5.5) |- (16.5,9.5) |- (12.5,11.5) |- (18.5,15.5) |- (16.5,17.5) |- (29.5,19.5) |- (20.5,17.5) |- (59.5,15.5) |- (15,0);

\draw[very thick] (4.5,9.5) -| (8.5,7.5) -| (3.5,3.5) -| (1.5,11.5) -| (10.5,5.5) -| (6.5,3.5) -| (16.5,5.5) -| (12.5,9.5) -| (16.5,11.5) -| (12.5,15.5) -| (18.5,17.5) -| (16.5,19.5) -| (29.5,17.5) -| (20.5,15.5) -| (59.5,15.5);
\draw[very thick] (16.5,5.5) -| (14.5,7.5) -| (27.5,5.5) -| (18.5,3.5) -| (57.5,3.5);

\draw [dashed] (59.5,15.5) -| (59.5,0);

\draws{4}{9}

\tileg{5}{9}{85}
\tileg{6}{9}{85}
\tileg{7}{9}{85}
\tileg{8}{9}{85}
\tileg{8}{8}{85}
\tileg{8}{7}{85}
\tileg{7}{7}{85}
\tileg{6}{7}{85}
\tileg{5}{7}{85}
\tileg{4}{7}{85}
\tileg{3}{7}{85}
\tileg{3}{6}{85}
\tileg{3}{5}{85}
\tileg{3}{4}{85}
\tileg{3}{3}{85}
\tileg{2}{3}{85}
\tileg{1}{3}{85}
\tileg{1}{4}{85}
\tileg{1}{5}{85}
\tileg{1}{6}{85}
\tileg{1}{7}{85}
\tileg{1}{8}{85}
\tileg{1}{9}{85}
\tileg{1}{10}{85}
\tileg{1}{11}{85}
\tileg{2}{11}{85}
\tileg{3}{11}{85}
\tileg{4}{11}{85}
\tileg{5}{11}{85}
\tileg{6}{11}{85}
\tileg{7}{11}{85}
\tileg{8}{11}{80}
\tileg{9}{11}{85}
\tileg{10}{11}{85}
\tileg{10}{10}{85}
\tileg{10}{9}{85}
\tileg{10}{8}{85}
\tileg{10}{7}{85}
\tileg{10}{6}{85}

\tileg{10}{5}{85}
\tileg{9}{5}{85}
\tileg{8}{5}{85}
\tileg{7}{5}{85}
\tileg{6}{5}{85}
\tileg{6}{4}{85}
\tileg{6}{3}{85}
\tileg{7}{3}{85}
\tileg{8}{3}{85}
\tileg{9}{3}{85}
\tileg{10}{3}{85}
\tileg{11}{3}{85}
\tileg{12}{3}{85}
\tileg{13}{3}{85}

\tileg{14}{3}{85}
\tileor{15}{3}{48}
\tileor{16}{3}{48}
\tileor{16}{4}{48}
\tileor{16}{5}{48}
\tileor{15}{5}{48}
\tileor{14}{5}{48}

\tiley{13}{5}{48}
\tiley{12}{5}{48}
\tiley{12}{6}{48}
\tiley{12}{7}{48}
\tiley{12}{8}{48}
\tiley{12}{9}{48}
\tiley{13}{9}{48}
\tiley{14}{9}{48}
\tiley{15}{9}{48}
\tiley{16}{9}{48}
\tiley{16}{10}{48}
\tiley{16}{11}{48}
\tiley{15}{11}{48}
\tiley{14}{11}{48}
\tiley{13}{11}{48}
\tiley{12}{11}{48}
\tiley{12}{12}{48}
\tiley{12}{13}{48}
\tiley{12}{14}{48}
\tiley{12}{15}{48}
\tiley{13}{15}{48}
\tiley{14}{15}{48}
\tiley{15}{15}{48}
\tiley{16}{15}{48}
\tiley{17}{15}{48}
\tiley{18}{15}{48}
\tiley{18}{16}{48}
\tiley{18}{17}{48}
\tiley{17}{17}{48}
\tiley{16}{17}{48}
\tiley{16}{18}{48}

\tiley{16}{18}{48}
\tiley{16}{19}{48}
\tileb{17}{19}{48}
\tileb{18}{19}{48}
\tileb{19}{19}{48}
\tileb{20}{19}{48}
\tileb{20}{19}{48}
\tileb{21}{19}{48}
\tileb{22}{19}{48}
\tileb{23}{19}{48}
\tileb{24}{19}{48}
\tileb{25}{19}{48}
\tileb{26}{19}{48}
\tileb{27}{19}{48}
\tileb{28}{19}{48}
\tileb{29}{19}{48}
\tileb{29}{18}{48}
\tileb{29}{17}{48}
\tileb{28}{17}{48}
\tileb{27}{17}{48}
\tileb{26}{17}{48}
\tileb{25}{17}{48}
\tileb{24}{17}{48}
\tileb{23}{17}{48}
\tileb{22}{17}{48}
\tileb{21}{17}{48}
\tileb{20}{17}{48}
\tileb{20}{16}{48}
\tileb{20}{15}{48}
\tileb{21}{15}{48}
\tileb{22}{15}{48}
\tileb{23}{15}{48}
\tileb{24}{15}{48}
\tileb{25}{15}{48}
\tileb{26}{15}{48}
\tileb{27}{15}{48}
\tileb{28}{15}{48}
\tileb{28}{15}{48}
\tileb{29}{15}{48}
\tileb{30}{15}{48}
\tileb{31}{15}{48}
\tileb{32}{15}{48}
\tileb{33}{15}{48}
\tileb{34}{15}{48}
\tileb{35}{15}{48}
\tileb{36}{15}{48}
\tileb{37}{15}{48}
\tileb{38}{15}{48}
\tileb{39}{15}{48}
\tileb{40}{15}{48}
\tileb{41}{15}{48}
\tileb{42}{15}{48}
\tileb{43}{15}{48}
\tileb{44}{15}{48}
\tileb{45}{15}{48}
\tileb{46}{15}{48}
\tileb{47}{15}{48}
\tileb{48}{15}{48}
\tileb{49}{15}{48}
\tileb{50}{15}{48}
\tileb{51}{15}{48}
\tileb{52}{15}{48}
\tileb{53}{15}{48}
\tileb{54}{15}{48}
\tileb{55}{15}{48}
\tileb{56}{15}{48}
\tileb{57}{15}{48}
\tileb{58}{15}{48}
\tileb{59}{15}{78}

\tilec{14}{6}{48}
\tilec{14}{7}{48}
\tilec{15}{7}{48}
\tilec{16}{7}{48}
\tilec{17}{7}{48}
\tilec{18}{7}{48}
\tilec{18}{7}{48}
\tilec{19}{7}{48}
\tilec{20}{7}{48}
\tilec{21}{7}{48}
\tilec{22}{7}{48}
\tilec{23}{7}{48}
\tilec{24}{7}{48}
\tilec{25}{7}{48}
\tilec{26}{7}{48}
\tilec{27}{7}{48}
\tilec{27}{6}{48}
\tilec{27}{5}{48}
\tilec{26}{5}{48}
\tilec{25}{5}{48}
\tilec{24}{5}{48}
\tilec{23}{5}{48}
\tilec{22}{5}{48}
\tilec{21}{5}{48}
\tilec{20}{5}{48}
\tilec{19}{5}{48}
\tilec{18}{5}{48}
\tilec{18}{4}{48}
\tilec{18}{3}{48}
\tilec{19}{3}{48}
\tilec{20}{3}{48}
\tilec{21}{3}{48}
\tilec{22}{3}{48}
\tilec{23}{3}{48}
\tilec{24}{3}{48}
\tilec{25}{3}{48}
\tilec{26}{3}{48}
\tilec{26}{3}{48}
\tilec{27}{3}{48}
\tilec{28}{3}{48}
\tilec{29}{3}{48}
\tilec{30}{3}{48}
\tilec{31}{3}{48}
\tilec{32}{3}{48}
\tilec{33}{3}{48}
\tilec{34}{3}{48}
\tilec{35}{3}{48}
\tilec{36}{3}{48}
\tilec{37}{3}{48}
\tilec{38}{3}{48}
\tilec{39}{3}{48}
\tilec{40}{3}{48}
\tilec{41}{3}{48}
\tilec{42}{3}{48}
\tilec{43}{3}{48}
\tilec{44}{3}{48}
\tilec{45}{3}{48}
\tilec{46}{3}{48}
\tilec{47}{3}{48}
\tilec{48}{3}{48}
\tilec{49}{3}{48}
\tilec{50}{3}{48}
\tilec{51}{3}{48}
\tilec{52}{3}{48}
\tilec{53}{3}{48}
\tilec{54}{3}{48}
\tilec{55}{3}{48}
\tilec{56}{3}{48}
\tilec{57}{3}{48}

\path [dotted, draw, thin] (0,0) grid[step=0.22cm] (60,21);

\draw [dashed] (15,0) -| (15,3.5);
\draw [dashed] (15,15.5) -| (15,21);

\fill (14.5,3.5) circle (0.16);
\node (D) at (14,2) {$P_{s}$};
\fill (14.5,15.5) circle (0.16);
\node (D) at (14,17) {$P_{n}$};

\node (D) at (50,19) {$\mathcal{C}^-$};
\node (D) at (35,10.5) {$\mathcal{C}^+$};

\fill (16.5,15.5) circle (0.16);
\node (D) at (17,14) {$P_{s'}$};
\end{tikzpicture}
\caption{We consider the path $P$ of the Figure \ref{fig:span}, the workspace $\mathcal{C}$ of its span $(s,n)$ is in blue and is cut in two areas $\mathcal{C}^+$ and $\mathcal{C}^-$ by the ray starting at $\pos{P_{|P|-1}}$ and going south. The path $P_{s'+1,\ldots,|P|-1}-\vect{P_sP_{s'}}$ can be written as $CS$: the tiles of $C$ are in common between $P_{s+1,\ldots, |P|-1}$ and $P_{s'+1,\ldots,|P|-1}-\vect{P_sP_{s'}}$ and they are drawn in orange while the tiles of $S$ are drawn in cyan. The path $S$ is called a shield, it does not intersect $P$ and it is in $\mathcal{C}^+$.}
\label{fig:SpanPull}
\end{figure}

\subsection{Pseudo-visibility}
\label{road:pseudo-visibility}

We consider the span $(s,n)$ of a path $P$ on a glue column $\Bone\leq c \leq \Btwo$. In this subsection, we explain how a shield $S$ of this span can be used to find a \emph{protected} and {pseudo-visible} glue on glue column $c$, and we list the assumptions done to obtain this result. Such a glue will have properties similar to those of a visible glue. Then, when Fact \ref{fact:turn} is used on two spans with the same characteristics, either a shield or a contradiction is obtained. If a shield is obtained then a new protected and pseudo-visible glue can be found. This glue behaves as a visible glue and will replace the one consumed by using Fact \ref{fact:turn}. We start by defining a {pseudo-visible} glue, see Figure \ref{fig:PseudoVisible}.  

\begin{definition}
\label{def:pseudo}
Consider an extremal path $P$, a glue column $\Bone\leq c \leq \Btwo$ and an index $0 \leq d \leq |P|-1$. We say that $\glueP{d}{d+1}$ is pseudo-visible from the south (resp. north) if and only if $\glueP{d}{d+1}$ is the southernmost (resp. northernmost) glue of $P_{d,\ldots,|P|-1}$ on glue column $c$ and points east.

\end{definition}

\begin{figure}
\center
\begin{tikzpicture}[x=0.22cm,y=0.22cm]


\draw[very thick] (4.5,9.5) -| (8.5,7.5) -| (3.5,3.5) -| (1.5,11.5) -| (10.5,5.5) -| (6.5,3.5) -| (16.5,5.5) -| (12.5,9.5) -| (16.5,11.5) -| (12.5,15.5) -| (18.5,17.5) -| (16.5,19.5) -| (29.5,17.5) -| (20.5,15.5) -| (59.5,15.5);
\draw[very thick] (16.5,5.5) -| (14.5,7.5) -| (27.5,5.5) -| (18.5,3.5) -| (57.5,3.5);


\draws{4}{9}

\tileg{5}{9}{85}
\tileg{6}{9}{85}
\tileg{7}{9}{85}
\tileg{8}{9}{85}
\tileg{8}{8}{85}
\tileg{8}{7}{85}
\tileg{7}{7}{85}
\tileg{6}{7}{85}
\tileg{5}{7}{85}
\tileg{4}{7}{85}
\tileg{3}{7}{85}
\tileg{3}{6}{85}
\tileg{3}{5}{85}
\tileg{3}{4}{85}
\tileg{3}{3}{85}
\tileg{2}{3}{85}
\tileg{1}{3}{85}
\tileg{1}{4}{85}
\tileg{1}{5}{85}
\tileg{1}{6}{85}
\tileg{1}{7}{85}
\tileg{1}{8}{85}
\tileg{1}{9}{85}
\tileg{1}{10}{85}
\tileg{1}{11}{85}
\tileg{2}{11}{85}
\tileg{3}{11}{85}
\tileg{4}{11}{85}
\tileg{5}{11}{85}
\tileg{6}{11}{85}
\tileg{7}{11}{85}
\tileg{8}{11}{80}
\tileg{9}{11}{85}
\tileg{10}{11}{85}
\tileg{10}{10}{85}
\tileg{10}{9}{85}
\tileg{10}{8}{85}
\tileg{10}{7}{85}
\tileg{10}{6}{85}

\tileg{10}{5}{85}
\tileg{9}{5}{85}
\tileg{8}{5}{85}
\tileg{7}{5}{85}
\tileg{6}{5}{85}
\tileg{6}{4}{85}
\tileg{6}{3}{85}
\tileg{7}{3}{85}
\tileg{8}{3}{85}
\tileg{9}{3}{85}
\tileg{10}{3}{85}
\tileg{11}{3}{85}
\tileg{12}{3}{85}
\tileg{13}{3}{85}

\tileg{14}{3}{85}
\tiley{15}{3}{48}
\tiley{16}{3}{48}
\tiley{16}{4}{48}
\tiley{16}{5}{48}
\tiley{15}{5}{48}
\tiley{14}{5}{48}

\tiley{13}{5}{48}
\tiley{12}{5}{48}
\tiley{12}{6}{48}
\tiley{12}{7}{48}
\tiley{12}{8}{48}
\tiley{12}{9}{48}
\tiley{13}{9}{48}
\tiley{14}{9}{48}
\tiley{15}{9}{48}
\tiley{16}{9}{48}
\tiley{16}{10}{48}
\tiley{16}{11}{48}
\tiley{15}{11}{48}
\tiley{14}{11}{48}
\tiley{13}{11}{48}
\tiley{12}{11}{48}
\tiley{12}{12}{48}
\tiley{12}{13}{48}
\tiley{12}{14}{48}
\tiley{12}{15}{48}
\tiley{13}{15}{48}
\tiley{14}{15}{48}

\tileb{15}{15}{48}
\tileb{16}{15}{48}
\tileb{17}{15}{48}
\tileb{18}{15}{48}
\tileb{18}{16}{48}
\tileb{18}{17}{48}
\tileb{17}{17}{48}
\tileb{16}{17}{48}
\tileb{16}{18}{48}
\tileb{16}{18}{48}
\tileb{16}{19}{48}
\tileb{17}{19}{48}
\tileb{18}{19}{48}
\tileb{19}{19}{48}
\tileb{20}{19}{48}
\tileb{20}{19}{48}
\tileb{21}{19}{48}
\tileb{22}{19}{48}
\tileb{23}{19}{48}
\tileb{24}{19}{48}
\tileb{25}{19}{48}
\tileb{26}{19}{48}
\tileb{27}{19}{48}
\tileb{28}{19}{48}
\tileb{29}{19}{48}
\tileb{29}{18}{48}
\tileb{29}{17}{48}
\tileb{28}{17}{48}
\tileb{27}{17}{48}
\tileb{26}{17}{48}
\tileb{25}{17}{48}
\tileb{24}{17}{48}
\tileb{23}{17}{48}
\tileb{22}{17}{48}
\tileb{21}{17}{48}
\tileb{20}{17}{48}
\tileb{20}{16}{48}
\tileb{20}{15}{48}
\tileb{21}{15}{48}
\tileb{22}{15}{48}
\tileb{23}{15}{48}
\tileb{24}{15}{48}
\tileb{25}{15}{48}
\tileb{26}{15}{48}
\tileb{27}{15}{48}
\tileb{28}{15}{48}
\tileb{28}{15}{48}
\tileb{29}{15}{48}
\tileb{30}{15}{48}
\tileb{31}{15}{48}
\tileb{32}{15}{48}
\tileb{33}{15}{48}
\tileb{34}{15}{48}
\tileb{35}{15}{48}
\tileb{36}{15}{48}
\tileb{37}{15}{48}
\tileb{38}{15}{48}
\tileb{39}{15}{48}
\tileb{40}{15}{48}
\tileb{41}{15}{48}
\tileb{42}{15}{48}
\tileb{43}{15}{48}
\tileb{44}{15}{48}
\tileb{45}{15}{48}
\tileb{46}{15}{48}
\tileb{47}{15}{48}
\tileb{48}{15}{48}
\tileb{49}{15}{48}
\tileb{50}{15}{48}
\tileb{51}{15}{48}
\tileb{52}{15}{48}
\tileb{53}{15}{48}
\tileb{54}{15}{48}
\tileb{55}{15}{48}
\tileb{56}{15}{48}
\tileb{57}{15}{48}
\tileb{58}{15}{48}
\tileb{59}{15}{78}

\tilec{14}{6}{48}
\tilec{14}{7}{48}
\tilec{15}{7}{48}
\tilec{16}{7}{48}
\tilec{17}{7}{48}
\tilec{18}{7}{48}
\tilec{18}{7}{48}
\tilec{19}{7}{48}
\tilec{20}{7}{48}
\tilec{21}{7}{48}
\tilec{22}{7}{48}
\tilec{23}{7}{48}
\tilec{24}{7}{48}
\tilec{25}{7}{48}
\tilec{26}{7}{48}
\tilec{27}{7}{48}
\tilec{27}{6}{48}
\tilec{27}{5}{48}
\tilec{26}{5}{48}
\tilec{25}{5}{48}
\tilec{24}{5}{48}
\tilec{23}{5}{48}
\tilec{22}{5}{48}
\tilec{21}{5}{48}
\tilec{20}{5}{48}
\tilec{19}{5}{48}
\tilec{18}{5}{48}
\tilec{18}{4}{48}
\tilec{18}{3}{48}
\tilec{19}{3}{48}
\tilec{20}{3}{48}
\tilec{21}{3}{48}
\tilec{22}{3}{48}
\tilec{23}{3}{48}
\tilec{24}{3}{48}
\tilec{25}{3}{48}
\tilec{26}{3}{48}
\tilec{26}{3}{48}
\tilec{27}{3}{48}
\tilec{28}{3}{48}
\tilec{29}{3}{48}
\tilec{30}{3}{48}
\tilec{31}{3}{48}
\tilec{32}{3}{48}
\tilec{33}{3}{48}
\tilec{34}{3}{48}
\tilec{35}{3}{48}
\tilec{36}{3}{48}
\tilec{37}{3}{48}
\tilec{38}{3}{48}
\tilec{39}{3}{48}
\tilec{40}{3}{48}
\tilec{41}{3}{48}
\tilec{42}{3}{48}
\tilec{43}{3}{48}
\tilec{44}{3}{48}
\tilec{45}{3}{48}
\tilec{46}{3}{48}
\tilec{47}{3}{48}
\tilec{48}{3}{48}
\tilec{49}{3}{48}
\tilec{50}{3}{48}
\tilec{51}{3}{48}
\tilec{52}{3}{48}
\tilec{53}{3}{48}
\tilec{54}{3}{48}
\tilec{55}{3}{48}
\tilec{56}{3}{48}
\tilec{57}{3}{48}

\path [dotted, draw, thin] (0,0) grid[step=0.22cm] (60,21);

\draw [dashed] (15,0) -| (15,3.5);
\draw [dashed] (15,15.5) -| (15,21);
\draw [thick,color=cyan] (15,7.5) -| (15,9.5);

\fill (14.5,3.5) circle (0.16);
\node (D) at (14,2) {$P_{s}$};
\fill (14.5,15.5) circle (0.16);
\node (D) at (14,17) {$P_{n}$};
\fill (14.5,9.5) circle (0.16);
\node (D) at (14,10.8) {$P_{d}$};
\fill (14.5,7.5) circle (0.16);
\node (D) at (13.2,7.5) {$S_{i}$};

\node (D) at (15,22) {$c$};


\end{tikzpicture}
\caption{We consider the path $P$ of the Figure \ref{fig:SpanPull} and its shield $S$ (in cyan). Here, $\glueP{d}{d+1}$ is the southernmost glue of $P_{d,\ldots,|P|-1}$ on glue column $c$, and it is pseudo-visible from the south in $P$. Remark that $\glueS{i}{i+1}$ is on column $c$ and the path $P$ does not intersect the segment between $\glueS{i}{i+1}$ and $\glueP{d}{d+1}$. In this case, the shield $S$ protects $\glueP{d}{d+1}$.}
\label{fig:PseudoVisible}
\end{figure}

Consider $s\leq d \leq u \leq |P|-1$ such that $\glueP{d}{d+1}$ is pseudo-visible from the south (resp. north) on glue column $c$ and $\glueP{u}{u+1}$ is visible from the north (resp. south) in $P_{d,\ldots, |P|+1}$ on glue column $c$ then $(u,d)$ is an upward (resp. downward) \emph{pseudo-span}. Consider, the workspace $\mathcal{C}'$ associated to the cut $(d,u)$ of $P_{d,\ldots, |P|-1}$. The pseudo-span $(d,u)$ is \emph{visible} if and only if, for any path $Q$ such that $Q$ is in $\mathcal{C}'$ and $Q_0$ interacts with $P_d$, $P_{0,\ldots,d}Q$ is producible. Note that any span is a visible pseudo-span but a pseudo-span is not necessarily visible, see Figures \ref{fig:PseudoSpan} and \ref{fig:NotVisible}. Intuitively, a visible pseudo-span behaves like a span. Shields will be used to find new visible pseudo-spans (in particular those which are not spans).

\begin{figure}
\center
\begin{tikzpicture}[x=0.22cm,y=0.22cm]

\fill[fill=cyan!50!white, draw opacity=0.8] (15,0) |- (16.5,9.5) |- (12.5,11.5) |- (15,15.5) |- (60,21) |- (15,0);

\draw[very thick] (14.5,9.5) -| (16.5,11.5) -| (12.5,15.5) -| (15.5,15.5);
\draw[very thick] (16.5,10.5) -| (29.5,6.5) -| (19.5,4.5) -| (16.5,4.5);


%

%

\tileg{14}{9}{85}
\tileor{15}{9}{48}
\tileor{16}{9}{48}
\tileor{16}{10}{48}
\tiley{16}{11}{48}
\tiley{15}{11}{48}
\tiley{14}{11}{48}
\tiley{13}{11}{48}
\tiley{12}{11}{48}
\tiley{12}{12}{48}
\tiley{12}{13}{48}
\tiley{12}{14}{48}
\tiley{12}{15}{48}
\tiley{13}{15}{48}
\tiley{14}{15}{48}
\tileb{15}{15}{48}

\tiler{17}{10}{48}
\tiler{18}{10}{48}
\tiler{19}{10}{48}
\tiler{20}{10}{48}
\tiler{21}{10}{48}
\tiler{22}{10}{48}
\tiler{23}{10}{48}
\tiler{24}{10}{48}
\tiler{25}{10}{48}
\tiler{26}{10}{48}
\tiler{27}{10}{48}
\tiler{28}{10}{48}
\tiler{29}{10}{48}
\tiler{29}{9}{48}
\tiler{29}{8}{48}
\tiler{29}{7}{48}
\tiler{29}{6}{48}
\tiler{28}{6}{48}
\tiler{27}{6}{48}
\tiler{26}{6}{48}
\tiler{25}{6}{48}
\tiler{24}{6}{48}
\tiler{23}{6}{48}
\tiler{22}{6}{48}
\tiler{21}{6}{48}
\tiler{20}{6}{48}
\tiler{19}{6}{48}
\tiler{19}{5}{48}
\tiler{19}{4}{48}
\tiler{18}{4}{48}
\tiler{17}{4}{48}
\tiler{16}{4}{48}

\path [dotted, draw, thin] (0,0) grid[step=0.22cm] (60,21);

\draw [dashed] (15,0) -| (15,9.5);
\draw [dashed] (15,15.5) -| (15,21);

\fill (14.5,15.5) circle (0.16);
\node (D) at (14,17) {$P_{n}$};
\fill (14.5,9.5) circle (0.16);
\node (D) at (13,9.5) {$P_{d}$};
\node (D) at (50,13.5) {$\mathcal{C}'$};
\fill (15.5,9.5) circle (0.16);
\node (D) at (15.5,8) {$Q_{0}$};
\fill (16.5,4.5) circle (0.16);
\node (D) at (16.5,3) {$Q_{|Q|-1}$};



\end{tikzpicture}
\caption{Following Figure \ref{fig:PseudoVisible}, we represent here the pseudo-span $(d,n)$ of the path $P$ with its workspace $\mathcal{C}'$ in cyan. The beginning of the path $P$ is not shown. Consider a path $Q$ such that $Q_0$ binds with $P_d$ and $Q$ is in $\mathcal{C'}$. The orange tiles are common between $Q$ and $P_{d,\ldots,n}$ while the red tiles belong only to $Q$.}
\label{fig:PseudoSpan}
\end{figure}

\begin{figure}
\center
\begin{tikzpicture}[x=0.22cm,y=0.22cm]

\fill[fill=cyan!50!white, draw opacity=0.8] (15,0) |- (16.5,9.5) |- (12.5,11.5) |- (15,15.5) |- (60,21) |- (15,0);

\draw[very thick] (4.5,9.5) -| (8.5,7.5) -| (3.5,3.5) -| (1.5,11.5) -| (10.5,5.5) -| (6.5,3.5) -| (16.5,5.5) -| (12.5,9.5) -| (16.5,10.5);
\draw[dashed] (4.5,9.5) -| (8.5,7.5) -| (3.5,3.5) -| (1.5,11.5) -| (10.5,5.5) -| (6.5,3.5) -| (16.5,5.5) -| (12.5,9.5) -| (16.5,11.5) -| (12.5,15.5) -| (15,15.5);
\draw[very thick] (16.5,10.5) -| (29.5,6.5) -| (19.5,4.5) -| (17,4.5);


\draws{4}{9}

\tileg{5}{9}{85}
\tileg{6}{9}{85}
\tileg{7}{9}{85}
\tileg{8}{9}{85}
\tileg{8}{8}{85}
\tileg{8}{7}{85}
\tileg{7}{7}{85}
\tileg{6}{7}{85}
\tileg{5}{7}{85}
\tileg{4}{7}{85}
\tileg{3}{7}{85}
\tileg{3}{6}{85}
\tileg{3}{5}{85}
\tileg{3}{4}{85}
\tileg{3}{3}{85}
\tileg{2}{3}{85}
\tileg{1}{3}{85}
\tileg{1}{4}{85}
\tileg{1}{5}{85}
\tileg{1}{6}{85}
\tileg{1}{7}{85}
\tileg{1}{8}{85}
\tileg{1}{9}{85}
\tileg{1}{10}{85}
\tileg{1}{11}{85}
\tileg{2}{11}{85}
\tileg{3}{11}{85}
\tileg{4}{11}{85}
\tileg{5}{11}{85}
\tileg{6}{11}{85}
\tileg{7}{11}{85}
\tileg{8}{11}{80}
\tileg{9}{11}{85}
\tileg{10}{11}{85}
\tileg{10}{10}{85}
\tileg{10}{9}{85}
\tileg{10}{8}{85}
\tileg{10}{7}{85}
\tileg{10}{6}{85}

\tileg{10}{5}{85}
\tileg{9}{5}{85}
\tileg{8}{5}{85}
\tileg{7}{5}{85}
\tileg{6}{5}{85}
\tileg{6}{4}{85}
\tileg{6}{3}{85}
\tileg{7}{3}{85}
\tileg{8}{3}{85}
\tileg{9}{3}{85}
\tileg{10}{3}{85}
\tileg{11}{3}{85}
\tileg{12}{3}{85}
\tileg{13}{3}{85}

\tileg{14}{3}{85}
\tiley{15}{3}{48}
\tiley{16}{3}{48}
\tiley{16}{4}{48}
\tiley{16}{5}{48}
\tiley{15}{5}{48}
\tiley{14}{5}{48}

\tiley{13}{5}{48}
\tiley{12}{5}{48}
\tiley{12}{6}{48}
\tiley{12}{7}{48}
\tiley{12}{8}{48}
\tiley{12}{9}{48}
\tiley{13}{9}{48}
\tiley{14}{9}{48}
\tileor{15}{9}{48}
\tileor{16}{9}{48}
\tileor{16}{10}{48}

\tiler{17}{10}{48}
\tiler{18}{10}{48}
\tiler{19}{10}{48}
\tiler{20}{10}{48}
\tiler{21}{10}{48}
\tiler{22}{10}{48}
\tiler{23}{10}{48}
\tiler{24}{10}{48}
\tiler{25}{10}{48}
\tiler{26}{10}{48}
\tiler{27}{10}{48}
\tiler{28}{10}{48}
\tiler{29}{10}{48}
\tiler{29}{9}{48}
\tiler{29}{8}{48}
\tiler{29}{7}{48}
\tiler{29}{6}{48}
\tiler{28}{6}{48}
\tiler{27}{6}{48}
\tiler{26}{6}{48}
\tiler{25}{6}{48}
\tiler{24}{6}{48}
\tiler{23}{6}{48}
\tiler{22}{6}{48}
\tiler{21}{6}{48}
\tiler{20}{6}{48}
\tiler{19}{6}{48}
\tiler{19}{5}{48}
\tiler{19}{4}{48}
\tiler{18}{4}{48}
\tiler{17}{4}{48}

\path [dotted, draw, thin] (0,0) grid[step=0.22cm] (60,21);

\draw [dashed] (15,0) -| (15,9.5);
\draw [dashed] (15,15.5) -| (15,21);

\fill (14.5,9.5) circle (0.16);
\node (D) at (14.5,7.9) {$P_{d}$};
\node (D) at (50,13.5) {$\mathcal{C}'$};

\draw [->] (15.5,13) -| (15.5,10.2);
\fill (15.5,9.5) circle (0.16);
\node (D) at (15.5,14) {$Q_{0}$};
\fill (17.5,4.5) circle (0.16);
\node (D) at (19.5,3) {$Q_{|Q|-2}$};

\fill (14.5,3.5) circle (0.16);
\node (D) at (14.5,1.9) {$P_{s}$};



\end{tikzpicture}
\caption{Following Figure \ref{fig:PseudoSpan}, we try to assemble $P_{0,\ldots,d}Q$ but the last tile of $Q$ conflits with $P_{0,\ldots,d}$. Hence, the pseudo-span $(d,n)$ is not visible. In this figure, only the beginning $P_{0,\dots,d}$ is shown.}
\label{fig:NotVisible}
\end{figure}

The proof of the main Theorem \ref{main:theorem} is done by showing that the larger the path is, the more constrained it is. The first proven bound, in \cite{STOC2020}, uses only two glues per glue column: the glue visible from the north and the glue visible from the south. When the proof requires more glues to find a contradiction, a new glue column is considered. Here, shields will be used to locate visible pseudo-span, allowing to work with more than two glues per glue column. Now, we explain how to locate a visible pseudo-span on a simple case where several hypotheses on the shield are assumed. These hypotheses will allow a shield to \emph{protect} a glue of $P$. The general case is far more complex and these locks will be solved later. 

Let $S$ be a shield of the span $(s,n)$ of $P$ and consider $s< d \leq |P|-1$ of $P$ such that $\glueP{d}{d+1}$ is on glue column $c$ and pointing east. If there exists $0\leq i \leq |S|-1$ such that $\glueS{i}{i+1}$ is on glue column $c$ and no glue of $P$ is on glue column $c$ between $\glueP{d}{d+1}$ and $\glueS{i}{i+1}$, then $\glueP{d}{d+1}$ is \emph{protected} by shield $S$, see Figure~\ref{fig:PseudoVisible}. Note that, we propose here a simplified definition of a protected glue which is sufficient for the roadmap. The formal definition will be given in Section~\ref{sec:analysis}. First, remark that shield $S$ must intersect glue column $c$ to be useful according to this definition.

\begin{lock}
\label{lock:shieldone}
Consider an extremal path $P$ such that $P$ is a good path on glue column $\Bone\leq c \leq \Btwo$. Let $S$ be a shield of the span $(s,n)$ of $P$ on glue column~$c$. Then, a glue of $S$ is on glue column $c$.
\end{lock}

We explain in Section \ref{sec:analysis} how to find a protected glue when a shield satisfies the hypothesis of Lock \ref{lock:shieldone}. In this roadmap, we assume that this protected glue is pseudo-visible.

\begin{lock}
\label{lock:pseudo}
Consider an extremal path $P$ and a span $(s,n)$ like in Lock \ref{lock:shieldone}. There exists $s\leq d \leq |P|-1$ such that $\glueP{d}{d+1}$ is pseudo-visible and protected by shield $S$.
\end{lock}

Moreover, we also add this condition to avoid complex cases which are discussed later.

\begin{lock}
\label{lock:shieldtwo}
Consider an extremal path $P$, a span $(s,n)$ as in Lock 4 and an index $s \leq d \leq |P| - 1$ as in Lock 5. Then $d \leq n$.
\end{lock}

It follows from Locks \ref{lock:pseudo} and \ref{lock:shieldtwo} that $(d,n)$ is a pseudo-span. According to Lock~\ref{lock:shieldone}, shield $S$ protects $\glueP{d}{d+1}$. These hypotheses are sufficient to show that the pseudo-span $(s,n)$ is visible.

\begin{lock}
\label{lock:uturn}
Consider an extremal path $P$ with a span $(s,n)$ like in Lock 4, and an index $s \leq d \leq |P| - 1$ satisfying the property in Lock 6. Then the pseudo-span $(d,n)$ of $P$ is visible. 
\end{lock}

To obtain this result, consider a path $Q$ such that $Q_0$ interacts with $P_d$ and $Q$ is inside the workspace $\mathcal{C}'$ of the pseudo-span $(d,n)$ of $P$ as shown in Figures \ref{fig:PseudoSpan} and \ref{fig:NotVisible}. In this example, $P_{0,\ldots,d}Q$ is not a path since the last tile of $Q$ intersects with $P_{0,\ldots,d}$. Note furthermore that $Q$ also intersects shield~$S$ (Figure \ref{fig:ShieldFuse}). Then, by merging parts of $P$, $Q$, and $S$, we obtain a more right-priority path as shown in Figure \ref{fig:ShieldFinal}, which contradicts the assumption that $P$ is a good path for glue column $c$. 

Another way for $Q$ to collide with $P_{0,\dots,d}$ is to avoid $S$ by turning around it. To do so, $Q$ must grow to the east and then go back to the west, see Figure \ref{fig:ShieldFinal}. This pattern is called a \emph{U-turn} and a key point of the proof of the pumping lemma \cite{STOC2020} is to prove that there is a contradiction if the U-turn is long enough. A quadratic bound was proven in \cite{STOC2020}. To obtain the main result of this work, we must improve this bound to a linear one. Solving Lock \ref{lock:uturn} is done in two steps, see Figure \ref{fig:ShieldFinal} for an illustration.

\begin{figure}
\center
\begin{tikzpicture}[x=0.22cm,y=0.22cm]

\fill[fill=cyan!50!white, draw opacity=0.8] (15,0) |- (16.5,9.5) |- (12.5,11.5) |- (15,15.5) |- (60,21) |- (15,0);

\draw[very thick] (4.5,9.5) -| (8.5,7.5) -| (3.5,3.5) -| (1.5,11.5) -| (10.5,5.5) -| (6.5,3.5) -| (16.5,5.5) -| (12.5,9.5) -| (16.5,11.5) -| (12.5,15.5) -| (18.5,17.5) -| (16.5,19.5) -| (29.5,17.5) -| (20.5,15.5) -| (59.5,15.5);
\draw[dashed] (4.5,9.5) -| (8.5,7.5) -| (3.5,3.5) -| (1.5,11.5) -| (10.5,5.5) -| (6.5,3.5) -| (16.5,5.5) -| (12.5,9.5) -| (16.5,11.5) -| (12.5,15.5) -| (15,15.5);
\draw[very thick] (16.5,10.5) -| (29.5,6.5) -| (27.5,7.5)  -| (14.5,5.5);


\draws{4}{9}

\tileg{5}{9}{85}
\tileg{6}{9}{85}
\tileg{7}{9}{85}
\tileg{8}{9}{85}
\tileg{8}{8}{85}
\tileg{8}{7}{85}
\tileg{7}{7}{85}
\tileg{6}{7}{85}
\tileg{5}{7}{85}
\tileg{4}{7}{85}
\tileg{3}{7}{85}
\tileg{3}{6}{85}
\tileg{3}{5}{85}
\tileg{3}{4}{85}
\tileg{3}{3}{85}
\tileg{2}{3}{85}
\tileg{1}{3}{85}
\tileg{1}{4}{85}
\tileg{1}{5}{85}
\tileg{1}{6}{85}
\tileg{1}{7}{85}
\tileg{1}{8}{85}
\tileg{1}{9}{85}
\tileg{1}{10}{85}
\tileg{1}{11}{85}
\tileg{2}{11}{85}
\tileg{3}{11}{85}
\tileg{4}{11}{85}
\tileg{5}{11}{85}
\tileg{6}{11}{85}
\tileg{7}{11}{85}
\tileg{8}{11}{80}
\tileg{9}{11}{85}
\tileg{10}{11}{85}
\tileg{10}{10}{85}
\tileg{10}{9}{85}
\tileg{10}{8}{85}
\tileg{10}{7}{85}
\tileg{10}{6}{85}

\tileg{10}{5}{85}
\tileg{9}{5}{85}
\tileg{8}{5}{85}
\tileg{7}{5}{85}
\tileg{6}{5}{85}
\tileg{6}{4}{85}
\tileg{6}{3}{85}
\tileg{7}{3}{85}
\tileg{8}{3}{85}
\tileg{9}{3}{85}
\tileg{10}{3}{85}
\tileg{11}{3}{85}
\tileg{12}{3}{85}
\tileg{13}{3}{85}

\tiley{14}{3}{85}
\tiley{15}{3}{48}
\tiley{16}{3}{48}
\tiley{16}{4}{48}
\tiley{16}{5}{48}
\tiley{15}{5}{48}
\tiley{14}{5}{48}

\tiley{13}{5}{48}
\tiley{12}{5}{48}
\tiley{12}{6}{48}
\tiley{12}{7}{48}
\tiley{12}{8}{48}
\tiley{12}{9}{48}
\tiley{13}{9}{48}
\tiley{14}{9}{48}
\tiley{15}{9}{48}
\tiley{16}{9}{48}
\tiley{16}{10}{48}
\tiley{16}{11}{48}
\tiley{15}{11}{48}
\tiley{14}{11}{48}
\tiley{13}{11}{48}
\tiley{12}{11}{48}
\tiley{12}{12}{48}
\tiley{12}{13}{48}
\tiley{12}{14}{48}
\tiley{12}{15}{48}
\tiley{13}{15}{48}
\tiley{14}{15}{48}
\tiley{15}{15}{48}
\tiley{16}{15}{48}
\tiley{17}{15}{48}
\tiley{18}{15}{48}
\tiley{18}{16}{48}
\tiley{18}{17}{48}
\tiley{17}{17}{48}
\tiley{16}{17}{48}
\tiley{16}{18}{48}

\tiley{16}{18}{48}
\tiley{16}{19}{48}
\tileb{17}{19}{48}
\tileb{18}{19}{48}
\tileb{19}{19}{48}
\tileb{20}{19}{48}
\tileb{20}{19}{48}
\tileb{21}{19}{48}
\tileb{22}{19}{48}
\tileb{23}{19}{48}
\tileb{24}{19}{48}
\tileb{25}{19}{48}
\tileb{26}{19}{48}
\tileb{27}{19}{48}
\tileb{28}{19}{48}
\tileb{29}{19}{48}
\tileb{29}{18}{48}
\tileb{29}{17}{48}
\tileb{28}{17}{48}
\tileb{27}{17}{48}
\tileb{26}{17}{48}
\tileb{25}{17}{48}
\tileb{24}{17}{48}
\tileb{23}{17}{48}
\tileb{22}{17}{48}
\tileb{21}{17}{48}
\tileb{20}{17}{48}
\tileb{20}{16}{48}
\tileb{20}{15}{48}
\tileb{21}{15}{48}
\tileb{22}{15}{48}
\tileb{23}{15}{48}
\tileb{24}{15}{48}
\tileb{25}{15}{48}
\tileb{26}{15}{48}
\tileb{27}{15}{48}
\tileb{28}{15}{48}
\tileb{28}{15}{48}
\tileb{29}{15}{48}
\tileb{30}{15}{48}
\tileb{31}{15}{48}
\tileb{32}{15}{48}
\tileb{33}{15}{48}
\tileb{34}{15}{48}
\tileb{35}{15}{48}
\tileb{36}{15}{48}
\tileb{37}{15}{48}
\tileb{38}{15}{48}
\tileb{39}{15}{48}
\tileb{40}{15}{48}
\tileb{41}{15}{48}
\tileb{42}{15}{48}
\tileb{43}{15}{48}
\tileb{44}{15}{48}
\tileb{45}{15}{48}
\tileb{46}{15}{48}
\tileb{47}{15}{48}
\tileb{48}{15}{48}
\tileb{49}{15}{48}
\tileb{50}{15}{48}
\tileb{51}{15}{48}
\tileb{52}{15}{48}
\tileb{53}{15}{48}
\tileb{54}{15}{48}
\tileb{55}{15}{48}
\tileb{56}{15}{48}
\tileb{57}{15}{48}
\tileb{58}{15}{48}
\tileb{59}{15}{78}

\tilec{14}{6}{48}
\tilec{14}{7}{48}
\tilec{15}{7}{48}
\tilec{16}{7}{48}
\tilec{17}{7}{48}
\tilec{18}{7}{48}
\tilec{18}{7}{48}
\tilec{19}{7}{48}
\tilec{20}{7}{48}
\tilec{21}{7}{48}
\tilec{22}{7}{48}
\tilec{23}{7}{48}
\tilec{24}{7}{48}
\tilec{25}{7}{48}
\tilec{26}{7}{48}
\tilec{27}{7}{48}
\tilec{27}{6}{48}

\tiler{17}{10}{48}
\tiler{18}{10}{48}
\tiler{19}{10}{48}
\tiler{20}{10}{48}
\tiler{21}{10}{48}
\tiler{22}{10}{48}
\tiler{23}{10}{48}
\tiler{24}{10}{48}
\tiler{25}{10}{48}
\tiler{26}{10}{48}
\tiler{27}{10}{48}
\tiler{28}{10}{48}
\tiler{29}{10}{48}
\tiler{29}{9}{48}
\tiler{29}{8}{48}
\tiler{29}{7}{48}
\tiler{29}{6}{48}
\tiler{28}{6}{48}
\tileor{27}{6}{48}

\path [dotted, draw, thin] (0,0) grid[step=0.22cm] (60,21);

\draw [dashed] (15,0) -| (15,9.5);
\draw [dashed] (15,15.5) -| (15,21);

\node (D) at (50,3.5) {$\mathcal{C}'$};

\fill (14.5,3.5) circle (0.16);
\node (D) at (14.5,1.9) {$P_{s}$};

\node (D) at (31.2,9) {$Q$};
\node (D) at (22.5,6) {$S$};

\fill (14.5,9.5) circle (0.16);
\node (D) at (14.8,10.5) {$P_{d}$};

\fill (14.5,15.5) circle (0.16);
\node (D) at (14.5,17) {$P_{n}$};


\end{tikzpicture}
\caption{We merge parts of paths $P$ and $Q$, and of shield $S$ from Figures~\ref{fig:PseudoVisible}, \ref{fig:PseudoSpan}, and \ref{fig:NotVisible}, to obtain a new path which is more right-priority than $P$, contradicting that $P$ is a good path for glue column $c$. Such merging is possible since the tile assembly system is directed. The tile common to $Q$ and $S$ is shown in orange.}
\label{fig:ShieldFuse}
\end{figure}

\begin{figure}
\center
\begin{tikzpicture}[x=0.22cm,y=0.22cm]

\fill[fill=cyan!50!white, draw opacity=0.8] (15,0) |- (16.5,9.5) |- (12.5,11.5) |- (15,15.5) |- (60,21) |- (15,0);

\draw[very thick] (4.5,9.5) -| (8.5,7.5) -| (3.5,3.5) -| (1.5,11.5) -| (10.5,5.5) -| (6.5,3.5) -| (16.5,5.5) -| (12.5,9.5) -| (16.5,11.5) -| (12.5,15.5) -| (18.5,17.5) -| (16.5,19.5) -| (29.5,17.5) -| (20.5,15.5) -| (59.5,15.5);
\draw[dashed] (4.5,9.5) -| (8.5,7.5) -| (3.5,3.5) -| (1.5,11.5) -| (10.5,5.5) -| (6.5,3.5) -| (16.5,5.5) -| (12.5,9.5) -| (16.5,11.5) -| (12.5,15.5) -| (15,15.5);
\draw[very thick] (16.5,5.5) -| (14.5,7.5) -| (27.5,5.5) -| (18.5,3.5) -| (57.5,3.5);
\draw[very thick] (16.5,10.5) -| (36.5,6.5) -| (58.5,1.5)  -| (29.5,1.5);


\draws{4}{9}

\tileg{5}{9}{85}
\tileg{6}{9}{85}
\tileg{7}{9}{85}
\tileg{8}{9}{85}
\tileg{8}{8}{85}
\tileg{8}{7}{85}
\tileg{7}{7}{85}
\tileg{6}{7}{85}
\tileg{5}{7}{85}
\tileg{4}{7}{85}
\tileg{3}{7}{85}
\tileg{3}{6}{85}
\tileg{3}{5}{85}
\tileg{3}{4}{85}
\tileg{3}{3}{85}
\tileg{2}{3}{85}
\tileg{1}{3}{85}
\tileg{1}{4}{85}
\tileg{1}{5}{85}
\tileg{1}{6}{85}
\tileg{1}{7}{85}
\tileg{1}{8}{85}
\tileg{1}{9}{85}
\tileg{1}{10}{85}
\tileg{1}{11}{85}
\tileg{2}{11}{85}
\tileg{3}{11}{85}
\tileg{4}{11}{85}
\tileg{5}{11}{85}
\tileg{6}{11}{85}
\tileg{7}{11}{85}
\tileg{8}{11}{80}
\tileg{9}{11}{85}
\tileg{10}{11}{85}
\tileg{10}{10}{85}
\tileg{10}{9}{85}
\tileg{10}{8}{85}
\tileg{10}{7}{85}
\tileg{10}{6}{85}

\tileg{10}{5}{85}
\tileg{9}{5}{85}
\tileg{8}{5}{85}
\tileg{7}{5}{85}
\tileg{6}{5}{85}
\tileg{6}{4}{85}
\tileg{6}{3}{85}
\tileg{7}{3}{85}
\tileg{8}{3}{85}
\tileg{9}{3}{85}
\tileg{10}{3}{85}
\tileg{11}{3}{85}
\tileg{12}{3}{85}
\tileg{13}{3}{85}

\tiley{14}{3}{85}
\tiley{15}{3}{48}
\tiley{16}{3}{48}
\tiley{16}{4}{48}
\tiley{16}{5}{48}
\tiley{15}{5}{48}
\tiley{14}{5}{48}

\tiley{13}{5}{48}
\tiley{12}{5}{48}
\tiley{12}{6}{48}
\tiley{12}{7}{48}
\tiley{12}{8}{48}
\tiley{12}{9}{48}
\tiley{13}{9}{48}
\tiley{14}{9}{48}
\tiley{15}{9}{48}
\tiley{16}{9}{48}
\tiley{16}{10}{48}
\tiley{16}{11}{48}
\tiley{15}{11}{48}
\tiley{14}{11}{48}
\tiley{13}{11}{48}
\tiley{12}{11}{48}
\tiley{12}{12}{48}
\tiley{12}{13}{48}
\tiley{12}{14}{48}
\tiley{12}{15}{48}
\tiley{13}{15}{48}
\tiley{14}{15}{48}
\tiley{15}{15}{48}
\tiley{16}{15}{48}
\tiley{17}{15}{48}
\tiley{18}{15}{48}
\tiley{18}{16}{48}
\tiley{18}{17}{48}
\tiley{17}{17}{48}
\tiley{16}{17}{48}
\tiley{16}{18}{48}

\tiley{16}{18}{48}
\tiley{16}{19}{48}
\tileb{17}{19}{48}
\tileb{18}{19}{48}
\tileb{19}{19}{48}
\tileb{20}{19}{48}
\tileb{20}{19}{48}
\tileb{21}{19}{48}
\tileb{22}{19}{48}
\tileb{23}{19}{48}
\tileb{24}{19}{48}
\tileb{25}{19}{48}
\tileb{26}{19}{48}
\tileb{27}{19}{48}
\tileb{28}{19}{48}
\tileb{29}{19}{48}
\tileb{29}{18}{48}
\tileb{29}{17}{48}
\tileb{28}{17}{48}
\tileb{27}{17}{48}
\tileb{26}{17}{48}
\tileb{25}{17}{48}
\tileb{24}{17}{48}
\tileb{23}{17}{48}
\tileb{22}{17}{48}
\tileb{21}{17}{48}
\tileb{20}{17}{48}
\tileb{20}{16}{48}
\tileb{20}{15}{48}
\tileb{21}{15}{48}
\tileb{22}{15}{48}
\tileb{23}{15}{48}
\tileb{24}{15}{48}
\tileb{25}{15}{48}
\tileb{26}{15}{48}
\tileb{27}{15}{48}
\tileb{28}{15}{48}
\tileb{28}{15}{48}
\tileb{29}{15}{48}
\tileb{30}{15}{48}
\tileb{31}{15}{48}
\tileb{32}{15}{48}
\tileb{33}{15}{48}
\tileb{34}{15}{48}
\tileb{35}{15}{48}
\tileb{36}{15}{48}
\tileb{37}{15}{48}
\tileb{38}{15}{48}
\tileb{39}{15}{48}
\tileb{40}{15}{48}
\tileb{41}{15}{48}
\tileb{42}{15}{48}
\tileb{43}{15}{48}
\tileb{44}{15}{48}
\tileb{45}{15}{48}
\tileb{46}{15}{48}
\tileb{47}{15}{48}
\tileb{48}{15}{48}
\tileb{49}{15}{48}
\tileb{50}{15}{48}
\tileb{51}{15}{48}
\tileb{52}{15}{48}
\tileb{53}{15}{48}
\tileb{54}{15}{48}
\tileb{55}{15}{48}
\tileb{56}{15}{48}
\tileb{57}{15}{48}
\tileb{58}{15}{48}
\tileb{59}{15}{78}

\tilec{14}{6}{48}
\tilec{14}{7}{48}
\tilec{15}{7}{48}
\tilec{16}{7}{48}
\tilec{17}{7}{48}
\tilec{18}{7}{48}
\tilec{18}{7}{48}
\tilec{19}{7}{48}
\tilec{20}{7}{48}
\tilec{21}{7}{48}
\tilec{22}{7}{48}
\tilec{23}{7}{48}
\tilec{24}{7}{48}
\tilec{25}{7}{48}
\tilec{26}{7}{48}
\tilec{27}{7}{48}
\tilec{27}{6}{48}
\tilec{27}{5}{48}
\tilec{26}{5}{48}
\tilec{25}{5}{48}
\tilec{24}{5}{48}
\tilec{23}{5}{48}
\tilec{22}{5}{48}
\tilec{21}{5}{48}
\tilec{20}{5}{48}
\tilec{19}{5}{48}
\tilec{18}{5}{48}
\tilec{18}{4}{48}
\tilec{18}{3}{48}
\tilec{19}{3}{48}
\tilec{20}{3}{48}
\tilec{21}{3}{48}
\tilec{22}{3}{48}
\tilec{23}{3}{48}
\tilec{24}{3}{48}
\tilec{25}{3}{48}
\tilec{26}{3}{48}
\tilec{26}{3}{48}
\tilec{27}{3}{48}
\tilec{28}{3}{48}
\tilec{29}{3}{48}
\tilec{30}{3}{48}
\tilec{31}{3}{48}
\tilec{32}{3}{48}
\tilec{33}{3}{48}
\tilec{34}{3}{48}
\tilec{35}{3}{48}
\tilec{36}{3}{48}
\tilec{37}{3}{48}
\tilec{38}{3}{48}
\tilec{39}{3}{48}
\tilec{40}{3}{48}
\tilec{41}{3}{48}
\tilec{42}{3}{48}
\tilec{43}{3}{48}
\tilec{44}{3}{48}
\tilec{45}{3}{48}
\tilec{46}{3}{48}
\tilec{47}{3}{48}
\tilec{48}{3}{48}
\tilec{49}{3}{48}
\tilec{50}{3}{48}
\tilec{51}{3}{48}
\tilec{52}{3}{48}
\tilec{53}{3}{48}
\tilec{54}{3}{48}
\tilec{55}{3}{48}
\tilec{56}{3}{48}
\tilec{57}{3}{48}

\tiler{17}{10}{48}
\tiler{18}{10}{48}
\tiler{19}{10}{48}
\tiler{20}{10}{48}
\tiler{21}{10}{48}
\tiler{22}{10}{48}
\tiler{23}{10}{48}
\tiler{24}{10}{48}
\tiler{25}{10}{48}
\tiler{26}{10}{48}
\tiler{27}{10}{48}
\tiler{28}{10}{48}
\tiler{29}{10}{48}
\tiler{30}{10}{48}
\tiler{31}{10}{48}
\tiler{32}{10}{48}
\tiler{33}{10}{48}
\tiler{34}{10}{48}
\tiler{35}{10}{48}
\tiler{36}{10}{48}
\tiler{36}{9}{48}
\tiler{36}{8}{48}
\tiler{36}{7}{48}
\tiler{36}{6}{48}
\tiler{37}{6}{48}
\tiler{38}{6}{48}
\tiler{39}{6}{48}
\tiler{40}{6}{48}
\tiler{41}{6}{48}
\tiler{42}{6}{48}
\tiler{43}{6}{48}
\tiler{44}{6}{48}
\tiler{45}{6}{48}
\tiler{46}{6}{48}
\tiler{47}{6}{48}
\tiler{48}{6}{48}
\tiler{49}{6}{48}
\tiler{50}{6}{48}
\tiler{51}{6}{48}
\tiler{52}{6}{48}
\tiler{53}{6}{48}
\tiler{54}{6}{48}
\tiler{55}{6}{48}
\tiler{56}{6}{48}
\tiler{57}{6}{48}
\tiler{58}{6}{48}
\tiler{58}{5}{48}
\tiler{58}{4}{48}
\tiler{58}{3}{48}
\tiler{58}{2}{48}
\tiler{58}{1}{48}
\tiler{57}{1}{48}
\tiler{56}{1}{48}
\tiler{55}{1}{48}
\tiler{54}{1}{48}
\tiler{53}{1}{48}
\tiler{52}{1}{48}
\tiler{51}{1}{48}
\tiler{50}{1}{48}
\tiler{49}{1}{48}
\tiler{48}{1}{48}
\tiler{47}{1}{48}
\tiler{46}{1}{48}
\tiler{45}{1}{48}
\tiler{44}{1}{48}
\tiler{43}{1}{48}
\tiler{42}{1}{48}
\tiler{41}{1}{48}
\tiler{40}{1}{48}
\tiler{39}{1}{48}
\tiler{38}{1}{48}
\tiler{37}{1}{48}
\tiler{36}{1}{48}
\tiler{35}{1}{48}
\tiler{34}{1}{48}
\tiler{33}{1}{48}
\tiler{32}{1}{48}
\tiler{31}{1}{48}
\tiler{30}{1}{48}
\tiler{29}{1}{48}


\path [dotted, draw, thin] (0,0) grid[step=0.22cm] (60,21);

\draw [dashed] (15,0) -| (15,9.5);
\draw [dashed] (15,15.5) -| (15,21);

\draw (30,0) -| (30,21);
\node (D) at (30,22) {$\Bthree$};
\draw (46,0) -| (46,21);
\node (D) at (46,22) {$\Bfour$};

\node (D) at (54,11.5) {$\mathcal{C}'$};



\fill (14.5,3.5) circle (0.16);
\node (D) at (14.5,1.9) {$P_{s}$};

\node (D) at (38.2,17) {$P$};

\node (D) at (38.2,9) {$Q$};
\node (D) at (22.5,2) {$S$};

\fill (14.5,9.5) circle (0.16);
\node (D) at (14.8,10.5) {$P_{d}$};

\fill (14.5,15.5) circle (0.16);
\node (D) at (14.5,17) {$P_{n}$};

\end{tikzpicture}
\caption{Since the case shown in Figure \ref{fig:ShieldFuse} leads to a contradiction, we consider another case. Here, the path $Q$ tries to collide with $P_{0,\ldots,d}$ by avoiding the shield $S$ by turning around it. To do so, $Q$ must first reach glue column $\Bfour$ and come back to glue column $\Bthree$. }
\label{fig:ShieldFinal}
\end{figure}

\begin{lock}
\label{lock:uturnOne}
$e_{P_{0,\ldots,d}}$ is bounded by $\Bthree$ whose value depends on glue column $c$, as described in Definition \ref{def:bound:threeandfour} of Subsection \ref{sec:U-turn:linear}.
\end{lock}

\begin{lock}
\label{lock:uturnTwo}
Consider a path $Q$ inside the workspace $\mathcal{C}'$ of the pseudo-span $(d,n)$ of $P$, such that $Q_0$ interacts with $P_d$. If there exists $0\leq i \leq |Q|-1$ such that $\glueQ{i}{i+1}$ is on glue column $\Bfour$ (whose value depends on glue column $c$ as detailed in Definition \ref{def:bound:threeandfour} of Subsection \ref{sec:U-turn:linear}), then there is no index $i<j \leq |Q|-1$ such that $\glueQ{j}{j+1}$ is visible in $Q$ on glue column~$\Bthree$.
\end{lock}

Informally, these two locks capture the intuition that shield $S$ is an obstacle preventing any collision between $Q$ and $P_{0,\dots,d}$, and ensures that pseudo-span $(d,n)$ is visible.

If all these arguments hold then we can consider $2|T|+1$ spans of $P$ which are located on different glue columns between $\Bone=\BoneValue$ and $\Btwo=\BtwoValue$. By the pigeonhole principle, two of these spans have the same direction and type. According to the reasoning presented so far, either there is a contradiction ($P$ is not extremal or $P$ is not a good path for a given column) or a shield can be assembled which can be used to find a visible pseudo-span. In the considered collection of $2|T|+1$ spans, one of them can be replaced by this new visible pseudo-span. Note that the new pseudo-span is narrower than the replaced one. By iterating this reasoning, we eventually find a contradiction.

\subsection{Summary of the main obstacles}
\label{road:obstacles}

To solve the different locks, we start by explaining how to build a canonical path for a glue column $\Bone\leq c \leq \Btwo$ in Section \ref{sec:canon}, solving Lock \ref{lock:can}. A canonical path is a slightly modified good path (see Definition \ref{prop:canonicalPath}). These modifications are done to deal with Lock \ref{lock:shieldtwo}. In Section \ref{sec:U-turn}, we give a linear bound on the length of U-turns, solving Locks~\ref{lock:uturnOne} and~\ref{lock:uturnTwo} (and thus Lock~\ref{lock:uturn}). Alongside the tools developed in this section, we remind the arguments of \cite{pumpabilityLargeBound} to solve Lock \ref{lock:boundZero} and the arguments of \cite{STOC2020} to solve Lock \ref{lem:glue:prop3}. Thirdly, we introduce in Section \ref{sec:decompo} the decomposition of a canonical path $P$ on glue column $c$ into arcs\footnote{This decomposition was first introduced in \cite{pumpabilityLargeBound} but was later abandoned in the final version of the paper after some simplifications.} to solve Lock \ref{lock:pseudo}. This decomposition allows to find protected glues which are not pseudo-visible. In Section \ref{sec:analysis}, we give the formal definition of a shield and put all arguments together to prove the main Theorem \ref{main:theorem} by following the reasoning of this roadmap. Finally, we describe how to deal with the special case of Lock~\ref{lock:shieldone}.

\section{Canonical paths}
\label{sec:canon}


In the roadmap, the reasoning relies on good paths (see Definition \ref{prop:canonicalPath}) but we stated that canonical paths are required in the main proof. These canonical paths are obtained by modifying a good path.  Then, we start by explaining how to find good paths (Lemma \ref{exists:good}) to solve Lock \ref{lock:can}. Afterwards, we give the definition of canonical path (Definition \ref{def:cano}) and how to find them (Lemma \ref{lem:exists:canoP}).

\begin{lemma}
\label{exists:good}
Consider a glue column $\Bone\leq c \leq \Btwo$, there exists a good path $P$ for glue column $c$.
\end{lemma}

\begin{proof}
Consider $\Bone\leq c \leq \Btwo$, the set $\mathcal{Q}$ of extremal paths and let $$w=\min\{\text{width of the span of $Q$ on glue column $c$}:Q\in \mathcal{Q}\}.$$ We consider $Q\in \mathcal{Q}$ such that the span $(s,n)$ of $Q$ on glue column $c$ is of width $w$ (note that several paths of $\mathcal{Q}$ may be of width $w$, any of them can be chosen as $Q$. Depending of this choice different good paths for column $c$ can be found) and suppose that this span is an upward (resp. downward) span. Let  $P$ be the rightmost (resp. leftmost) priority extremal path such $Q_{0,\ldots,s+1}$ is a prefix of $P$ and $\glueP{s}{s+1}$ is visible from the south (resp. north) in $P$. We claim that $P$ is a good path for glue column $c$.


\begin{figure}
\center
\begin{minipage}{0.47\linewidth}
\begin{tikzpicture}[x=0.22cm,y=0.22cm]


\draw[very thick] (1.5,10.5) -| (12.5,6.5) -| (8.5,2.5) -| (16.5,14.5) -| (8.5,18.5) -| (24.5,18.5);

\draws{1}{10}

\tileg{2}{10}{85}
\tileg{3}{10}{85}
\tileg{4}{10}{85}
\tileg{5}{10}{85}
\tileg{6}{10}{85}
\tileg{7}{10}{85}
\tileg{8}{10}{85}
\tileg{9}{10}{85}
\tileg{10}{10}{85}
\tileg{11}{10}{85}
\tileg{12}{10}{85}
\tileg{12}{9}{85}
\tileg{12}{8}{85}
\tileg{12}{7}{85}
\tileg{12}{6}{85}
\tileg{11}{6}{85}
\tileg{10}{6}{85}
\tileg{9}{6}{85}
\tileg{8}{6}{85}
\tileg{8}{5}{85}
\tileg{8}{4}{85}
\tileg{8}{3}{85}
\tileg{8}{2}{85}
\tileg{9}{2}{85}
\tileg{10}{2}{85}
\tiley{11}{2}{48}
\tiley{12}{2}{48}
\tiley{13}{2}{48}
\tiley{14}{2}{48}
\tiley{15}{2}{48}
\tiley{16}{2}{48}
\tiley{16}{3}{48}
\tiley{16}{4}{48}
\tiley{16}{5}{48}
\tiley{16}{6}{48}
\tiley{16}{7}{48}
\tiley{16}{8}{48}
\tiley{16}{9}{48}
\tiley{16}{10}{48}
\tiley{16}{11}{48}
\tiley{16}{12}{48}
\tiley{16}{13}{48}
\tiley{16}{14}{48}
\tiley{15}{14}{48}
\tiley{14}{14}{48}
\tiley{13}{14}{48}
\tiley{12}{14}{48}
\tiley{11}{14}{48}
\tiley{10}{14}{48}
\tiley{9}{14}{48}
\tiley{8}{14}{48}
\tiley{8}{15}{48}
\tiley{8}{16}{48}
\tiley{8}{17}{48}
\tiley{8}{18}{48}
\tiley{9}{18}{48}
\tiley{10}{18}{78}
\tileb{11}{18}{78}
\tileb{12}{18}{78}
\tileb{13}{18}{78}
\tileb{14}{18}{78}
\tileb{15}{18}{78}
\tileb{16}{18}{78}
\tileb{17}{18}{78}
\tileb{18}{18}{78}
\tileb{19}{18}{78}
\tileb{20}{18}{78}
\tileb{21}{18}{78}
\tileb{22}{18}{78}
\tileb{23}{18}{78}
\tileb{24}{18}{78}



\path [dotted, draw, thin] (0,0) grid[step=0.22cm] (26,21);

\draw[dashed] (11,0) -| (11,2.5);
\fill (10.5,2.5) circle (0.16);
\node (D) at (10.5,3.7) {$s$};
\node (D) at (12.2,1) {$l^s$};
\draw[dashed] (11,18.5) -| (11,21);
\fill (10.5,18.5) circle (0.16);
\node (D) at (10.5,17) {$n$};
\end{tikzpicture}

\center (a) 
\end{minipage}
\begin{minipage}{0.47\linewidth}
\begin{tikzpicture}[x=0.22cm,y=0.22cm]

\draw[very thick] (1.5,10.5) -| (12.5,6.5) -| (8.5,2.5) -| (16.5,7.5) -| (24.5,7.5);

\draws{1}{10}

\tileg{2}{10}{85}
\tileg{3}{10}{85}
\tileg{4}{10}{85}
\tileg{5}{10}{85}
\tileg{6}{10}{85}
\tileg{7}{10}{85}
\tileg{8}{10}{85}
\tileg{9}{10}{85}
\tileg{10}{10}{85}
\tiley{11}{10}{85}
\tiley{12}{10}{85}
\tiley{12}{9}{85}
\tiley{12}{8}{85}
\tiley{12}{7}{85}
\tiley{12}{6}{85}
\tiley{11}{6}{85}
\tiley{10}{6}{85}
\tiley{9}{6}{85}
\tiley{8}{6}{85}
\tiley{8}{5}{85}
\tiley{8}{4}{85}
\tiley{8}{3}{85}
\tiley{8}{2}{85}
\tiley{9}{2}{85}
\tiley{10}{2}{85}
\tileb{11}{2}{78}
\tileb{12}{2}{78}
\tileb{13}{2}{78}
\tileb{14}{2}{78}
\tileb{15}{2}{78}
\tileb{16}{2}{78}
\tileb{16}{3}{78}
\tileb{16}{4}{78}
\tileb{16}{5}{78}
\tileb{16}{6}{78}
\tileb{16}{7}{78}
\tileb{17}{7}{78}
\tileb{18}{7}{78}
\tileb{19}{7}{78}
\tileb{20}{7}{78}
\tileb{21}{7}{78}
\tileb{22}{7}{78}
\tileb{23}{7}{78}
\tileb{24}{7}{78}



\path [dotted, draw, thin] (0,0) grid[step=0.22cm] (26,21);

\draw[dashed] (11,0) -| (11,2.5);
\fill (10.5,2.5) circle (0.16);
\node (D) at (10.5,3.7) {$s$};
\node (D) at (12.2,1) {$l^s$};
\draw[dashed] (11,10.5) -| (11,21);
\fill (10.5,10.5) circle (0.16);
\node (D) at (10.2,12) {$n'$};

\end{tikzpicture}

\center (b) 
\end{minipage}

\caption{(a) A path $Q$ and its upward span $(s,n)$ on column $c$. To transform this path into a good one for column $c$, we consider the rightmost  priority path $P$ (represented in Figure \ref{fig:cano:good}(b)) which admits $Q_{0,\ldots,s+1}$ as a prefix and which does not cross~$l^s$.
(b) In this case, the span $(n',s)$ of $P$ on column $c$ is a downward span. Nevertheless, the width of this span is strictly less than the width of $Q$ on column $c$, contradicting the definition of $Q$.}
\label{fig:cano:good}
\end{figure}

By definition and since $\glueP{s}{s+1}=\glueQ{s}{s+1}$, $\glueP{s}{s+1}$ is visible from the south (resp. north) in $P$ on glue column $c$. Let $n'$ such that $\glueP{n'}{n'+1}$ is visible from the north (resp. south) on glue column $c$. 
\begin{itemize}
\item If $s\leq n'$ then $(s,n')$ is the upward (resp. downward) span of $P$ on glue column $c$ and $P$ is a good path for glue column $c$.
\item Otherwise, $s>n'$ then $(n',s)$ is a downward (resp. upward) span of $P$ on glue column~$c$, see Figure \ref{fig:cano:good}. In this case $P$ is not a good path for glue column $c$. Nevertheless, since $s>n'$ then $P_{n'}=Q_{n'}$ (since $Q_{0,\ldots,s}$ is a prefix of $P$ by definition of~$P$). Moreover, since $\glueQ{n}{n+1}$ if visible from the north (resp. south) in $Q$ then $y_{P_{n'}}<y_{Q_n}$ (resp. $y_{P_{n'}}>y_{Q_n}$). This means that the width of the span of $P$ on glue column $c$ is strictly less than the width $w$ of the span of $Q$ on glue column $c$ which contradicts the definition of $w$. 
\end{itemize}
Thus, only the first case can occur and $P$ is a good path.
\end{proof}


Solving Lock \ref{lock:can} requires only one last remark. In the example of the roadmap (Section \ref{sec:roadmap}), we considered one path which was a good path for two glue columns $\Bone\leq c <c '\leq \Btwo$. Nevertheless, the same reasoning still holds by using one good path for glue column $c$ and another good path for glue column $c'$. Thus, in the final Section \ref{sec:analysis}, we will not build a path which is good for all glue columns  $\Bone\leq c \leq \Btwo$ but we will consider a set of $4|T|+1$ paths, one for each glue column $\Bone\leq c \leq \Btwo$.

Now, the problem of Lock~\ref{lock:shieldtwo} is similar to the one occurring in the proof of Lemma \ref{exists:good} and Figure \ref{fig:cano:good}. Doing some modification on a path may switch the direction of a span. Indeed, consider a glue column $\Bone\leq c \leq \Btwo$ and a good path $P$ for glue column $c$. Consider the setting of Figure \ref{fig:cano:lock6} where $(s,n)$ is the upward span of $P$ on glue column $c$ and where a shield $S$ of the span $(s,n)$ protects a glue $n < d \leq |P|-1$. This $\glueP{d}{d+1}$ is the glue visible from the south in $P_{n,\ldots, |P|-1}$. In this case, the reasoning of the roadmap should continue with the pseudo-span $(n,d)$. This is problematic since  $(n,d)$ is a downward span and $P$ was obtained by selecting a rightmost priority path: the direction of the span and the priority of $P$ do not match anymore. To solve this problem, we introduce \emph{canonical} path.

\begin{figure}
\center
\begin{minipage}{0.47\linewidth}
\begin{tikzpicture}[x=0.22cm,y=0.22cm]


\draw[very thick] (1.5,3.5) -| (13.5,7.5) -| (6.5,18.5) -| (13.5,16.5) -| (8.5,14.5) -| (24.5,14.5);

\draws{1}{3}

\tileg{2}{3}{85}
\tileg{3}{3}{85}
\tileg{4}{3}{85}
\tileg{5}{3}{85}
\tileg{6}{3}{85}
\tileg{7}{3}{85}
\tileg{8}{3}{85}
\tileg{9}{3}{85}
\tileg{10}{3}{85}

\tiley{11}{3}{48}
\tiley{12}{3}{48}
\tiley{13}{3}{48}
\tiley{13}{4}{48}
\tiley{13}{5}{48}
\tiley{13}{6}{48}
\tiley{13}{7}{48}
\tiley{12}{7}{48}
\tiley{11}{7}{48}
\tiley{10}{7}{48}
\tiley{9}{7}{48}
\tiley{8}{7}{48}
\tiley{7}{7}{48}
\tiley{6}{7}{48}
\tiley{6}{8}{48}
\tiley{6}{9}{48}
\tiley{6}{10}{48}
\tiley{6}{11}{48}
\tiley{6}{12}{48}
\tiley{6}{13}{48}
\tiley{6}{14}{48}
\tiley{6}{15}{48}
\tiley{6}{16}{48}
\tiley{6}{17}{48}
\tiley{6}{18}{48}
\tiley{7}{18}{48}
\tiley{8}{18}{48}
\tiley{8}{18}{48}
\tiley{9}{18}{48}
\tiley{10}{18}{48}

\tileb{11}{18}{78}
\tileb{12}{18}{78}
\tileb{13}{18}{78}
\tileb{13}{17}{78}
\tileb{13}{16}{78}
\tileb{12}{16}{78}
\tileb{11}{16}{78}
\tileb{10}{16}{78}
\tileb{9}{16}{78}
\tileb{8}{16}{78}
\tileb{8}{15}{78}
\tileb{8}{14}{78}
\tileb{9}{14}{78}
\tileb{10}{14}{78}
\tileb{11}{14}{78}
\tileb{12}{14}{78}
\tileb{13}{14}{78}
\tileb{14}{14}{78}
\tileb{15}{14}{78}
\tileb{16}{14}{78}
\tileb{17}{14}{78}
\tileb{18}{14}{78}
\tileb{19}{14}{78}
\tileb{20}{14}{78}
\tileb{21}{14}{78}
\tileb{22}{14}{78}
\tileb{23}{14}{78}
\tileb{24}{14}{78}



\path [dotted, draw, thin] (0,0) grid[step=0.22cm] (26,21);

\draw[dashed] (11,0) -| (11,2.5);
\fill (10.5,3.5) circle (0.16);
\node (D) at (10.5,4.7) {$s$};
\draw[dashed] (11,18.5) -| (11,21);
\fill (10.5,18.5) circle (0.16);
\node (D) at (10.3,19.7) {$n$};
\end{tikzpicture}

\center (a) 
\end{minipage}
\begin{minipage}{0.47\linewidth}
\begin{tikzpicture}[x=0.22cm,y=0.22cm]


\draw[very thick] (1.5,3.5) -| (13.5,7.5) -| (6.5,18.5) -| (13.5,16.5) -| (8.5,14.5) -| (24.5,14.5);
\draw[very thick] (6.5,9.5) -| (20.5,9.5);

\draws{1}{3}

\tileg{2}{3}{85}
\tileg{3}{3}{85}
\tileg{4}{3}{85}
\tileg{5}{3}{85}
\tileg{6}{3}{85}
\tileg{7}{3}{85}
\tileg{8}{3}{85}
\tileg{9}{3}{85}
\tileg{10}{3}{85}

\tiley{11}{3}{48}
\tiley{12}{3}{48}
\tiley{13}{3}{48}
\tiley{13}{4}{48}
\tiley{13}{5}{48}
\tiley{13}{6}{48}
\tiley{13}{7}{48}
\tiley{12}{7}{48}
\tiley{11}{7}{48}
\tiley{10}{7}{48}
\tiley{9}{7}{48}
\tiley{8}{7}{48}
\tiley{7}{7}{48}
\tiley{6}{7}{48}
\tiley{6}{8}{48}
\tiley{6}{9}{48}
\tiley{6}{10}{48}
\tiley{6}{11}{48}
\tiley{6}{12}{48}
\tiley{6}{13}{48}
\tiley{6}{14}{48}
\tiley{6}{15}{48}
\tiley{6}{16}{48}
\tiley{6}{17}{48}
\tiley{6}{18}{48}
\tiley{7}{18}{48}
\tiley{8}{18}{48}
\tiley{8}{18}{48}
\tiley{9}{18}{48}
\tiley{10}{18}{48}

\tiler{11}{18}{78}
\tiler{12}{18}{78}
\tiler{13}{18}{78}
\tiler{13}{17}{78}
\tiler{13}{16}{78}
\tiler{12}{16}{78}
\tiler{11}{16}{78}
\tiler{10}{16}{78}
\tiler{9}{16}{78}
\tiler{8}{16}{78}
\tiler{8}{15}{78}
\tiler{8}{14}{78}
\tiler{9}{14}{78}
\tiler{10}{14}{78}
\tileb{11}{14}{78}
\tileb{12}{14}{78}
\tileb{13}{14}{78}
\tileb{14}{14}{78}
\tileb{15}{14}{78}
\tileb{16}{14}{78}
\tileb{17}{14}{78}
\tileb{18}{14}{78}
\tileb{19}{14}{78}
\tileb{20}{14}{78}
\tileb{21}{14}{78}
\tileb{22}{14}{78}
\tileb{23}{14}{78}
\tileb{24}{14}{78}



\tilec{7}{9}{48}
\tilec{8}{9}{48}
\tilec{9}{9}{48}
\tilec{10}{9}{48}
\tilec{11}{9}{48}
\tilec{12}{9}{48}
\tilec{13}{9}{48}
\tilec{14}{9}{48}
\tilec{15}{9}{48}
\tilec{16}{9}{48}
\tilec{17}{9}{48}
\tilec{18}{9}{48}
\tilec{19}{9}{48}
\tilec{20}{9}{48}

\path [dotted, draw, thin] (0,0) grid[step=0.22cm] (26,21);

\draw[dashed] (11,0) -| (11,2.5);
\fill (10.5,3.5) circle (0.16);
\node (D) at (10.5,4.7) {$s$};
\draw[color=cyan,thick] (11,9.5) -| (11,14.5);
\fill (10.5,14.5) circle (0.16);
\node (D) at (10.3,12.7) {$d$};
\draw[dashed] (11,18.5) -| (11,21);
\fill (10.5,18.5) circle (0.16);
\node (D) at (10.3,19.7) {$n$};
\end{tikzpicture}

\center (b) 
\end{minipage}

\caption{(a) A path $P$ which is a good path for a glue column $c$. Its beginning is in green, the upward span $(s,n)$ of $P$ on glue column $c$ is in yellow and the end of the path is in blue.
(b) A shield $S$ (in cyan) of the span $(s,n)$ protects a glue $d$ with $d>n$. In this case, $(n,d)$ is a downward visible pseudo-span (in red). Nevertheless, the direction of $(n,d)$ does not match the priority of $P_{n,\ldots, |P|-1}$ anymore.}
\label{fig:cano:lock6}
\end{figure}


Here as some preliminaries definitions. The index of the last glue of an extremal path $P$ on glue column~$\Bone\leq c \leq \Btwo$ is: $$ \lastc=\max\{0 \leq i \leq |P|-1: \glueP{i}{i+1} \text{ is on glue column } c\}.$$ Note that this glue must points east since $P$ is extremal. Consider an upward (resp. downward) pseudo-span $(s,n)$ of $P$ on a glue column $\Bone\leq c \leq \Btwo$, the \emph{next pseudo-visible glue} of this span is the glue visible from the south (resp. north) in $P_{n,\ldots,|P|-1}$ on glue column $c$, see Figure \ref{fig:cano:defdcompo}.

\begin{figure}
\center
\begin{tikzpicture}[x=0.22cm,y=0.22cm]


\draw[very thick] (1.5,3.5) -| (13.5,6.5) -| (2.5,21.5) -| (13.5,18.5) -| (5.5,9.5)-| (13.5,12.5) -| (8.5,15.5) -| (24.5,15.5);

\draws{1}{3}

\tileg{2}{3}{85}
\tileg{3}{3}{85}
\tileg{4}{3}{85}
\tileg{5}{3}{85}
\tileg{6}{3}{85}
\tileg{7}{3}{85}
\tileg{8}{3}{85}
\tileg{9}{3}{85}
\tileg{10}{3}{85}

\tiley{11}{3}{48}
\tiley{12}{3}{48}
\tiley{13}{3}{48}
\tiley{13}{4}{48}
\tiley{13}{5}{48}
\tiley{13}{6}{48}
\tiley{12}{6}{48}
\tiley{11}{6}{48}
\tiley{10}{6}{48}
\tiley{9}{6}{48}
\tiley{8}{6}{48}
\tiley{7}{6}{48}
\tiley{6}{6}{48}
\tiley{5}{6}{48}
\tiley{4}{6}{48}
\tiley{3}{6}{48}
\tiley{2}{6}{48}
\tiley{2}{7}{48}
\tiley{2}{8}{48}
\tiley{2}{9}{48}
\tiley{2}{10}{48}
\tiley{2}{11}{48}
\tiley{2}{12}{48}
\tiley{2}{13}{48}
\tiley{2}{14}{48}
\tiley{2}{15}{48}
\tiley{2}{16}{48}
\tiley{2}{17}{48}
\tiley{2}{18}{48}
\tiley{2}{19}{48}
\tiley{2}{20}{48}
\tiley{2}{21}{48}
\tiley{3}{21}{48}
\tiley{4}{21}{48}
\tiley{5}{21}{48}
\tiley{6}{21}{48}
\tiley{7}{21}{48}
\tiley{8}{21}{48}
\tiley{9}{21}{48}
\tiley{10}{21}{48}

\tileor{11}{21}{48}
\tileor{12}{21}{48}
\tileor{13}{21}{48}
\tileor{13}{20}{48}
\tileor{13}{19}{48}
\tileor{13}{18}{48}
\tileor{12}{18}{48}
\tileor{11}{18}{48}
\tileor{10}{18}{48}
\tileor{9}{18}{48}
\tileor{8}{18}{48}
\tileor{7}{18}{48}
\tileor{6}{18}{48}
\tileor{5}{18}{48}
\tileor{5}{17}{48}
\tileor{5}{16}{48}
\tileor{5}{15}{48}
\tileor{5}{14}{48}
\tileor{5}{13}{48}
\tileor{5}{12}{48}
\tileor{5}{11}{48}
\tileor{5}{10}{48}
\tileor{5}{9}{48}
\tileor{6}{9}{48}
\tileor{7}{9}{48}
\tileor{8}{9}{48}
\tileor{9}{9}{48}
\tileor{10}{9}{48}

\tiler{11}{9}{48}
\tiler{12}{9}{48}
\tiler{13}{9}{48}
\tiler{13}{10}{48}
\tiler{13}{11}{48}
\tiler{13}{12}{48}
\tiler{12}{12}{48}
\tiler{11}{12}{48}
\tiler{10}{12}{48}
\tiler{9}{12}{48}
\tiler{8}{12}{48}
\tiler{8}{13}{48}
\tiler{8}{14}{48}
\tiler{8}{15}{48}
\tiler{9}{15}{48}
\tiler{10}{15}{48}

\tileb{11}{15}{78}
\tileb{12}{15}{78}
\tileb{13}{15}{78}
\tileb{14}{15}{78}
\tileb{15}{15}{78}
\tileb{16}{15}{78}
\tileb{17}{15}{78}
\tileb{18}{15}{78}
\tileb{19}{15}{78}
\tileb{20}{15}{78}
\tileb{21}{15}{78}
\tileb{22}{15}{78}
\tileb{23}{15}{78}
\tileb{24}{15}{78}



\path [dotted, draw, thin] (0,0) grid[step=0.22cm] (26,24);

\draw[dashed] (11,0) -| (11,2.5);
\fill (10.5,3.5) circle (0.16);
\node (D) at (10.5,4.7) {$s$};

\fill (10.5,9.5) circle (0.16);
\node (D) at (10.5,8.2) {$d$};

\fill (10.5,15.5) circle (0.16);
\node (D) at (10.5,16.9) {$\ell$};
\node (D) at (12.2,1) {$l^s$};
\draw[dashed] (11,21.5) -| (11,24);
\fill (10.5,21.5) circle (0.16);
\node (D) at (10.3,20.2) {$n$};
\node (D) at (12.2,23) {$l^n$};
\end{tikzpicture}

\caption{An extremal path $P$ and its decomposition $(s,n,d,\ell)$ into pseudo-visible glues on glue column $c$. The index $\ell$ is the index of the last glue of $P$ on glue column $c$. $(s,n)$ (in yellow) and $(d,\ell)$ (in red) are upward pseudo-span while $(n,d)$ (in orange) is a downward pseudo-span. We have $y_{P_s}<y_{P_d}<y_{P_\ell}$ and  $y_{P_n}>y_{P_\ell}$. Moreover, the widths of the spans $(s,n)$, $(n,d)$ and $(d,\ell)$ are strictly decreasing. To be canonical for glue column $c$, $P_{0,\ldots,n}$ must be a prefix of the rightmost priority extremal path which does not cross $l^s$ and admits $P_{0,\ldots,s+1}$ as a prefix. Also, $P_{0,\ldots,\ell}$ must have a similar property with $P_{0,\ldots,d}$ and $l^d$, while  $P_{0,\ldots,d}$ must be a prefix of the leftmost priority extremal path which does not cross $l^n$ and admits $P_{0,\ldots,n+1}$ as a prefix.}
\label{fig:cano:defdcompo}
\end{figure}

\begin{definition}
\label{def:decompo}
Consider an extremal path $P$, a glue column $\Bone\leq c \leq \Btwo$. The decomposition of $P$ into pseudo-visible glues on glue column $c$ is the sequence $(u_i)_{0\leq i \leq t}$ such that:
\begin{itemize}
\item $(u_0,u_1)$ is the span of $P$ on column $c$;
\item for $1\leq i<t$, $\glueP{u_{i+1}}{u_{i+1}+1}$ is the next pseudo-visible glue of the pseudo-span $(u_{i-1},u_i)$ of $P$;
\item the sequence ends with $u_t=\lastc$ where $\lastc$ is the index of the last glue of $P$ on glue column $c$.
\end{itemize}
\end{definition}

Intuitively, if $(s,n)$ is the span of glue column $c$ of an extremal $P$  and $\ell$ is the index of the last glue of $P$ on column $c$ then the subpath $P_{s,\ldots,\ell}$ of $P$ is decomposed into a sequence of pseudo-spans $(u_i,u_{i+1})_{0\leq i <t}$, see Figure~\ref{fig:cano:defdcompo}. Consider one of these pseudo-spans $(u_i,u_{i+1})$ then it follows from the definition of the decomposition that the direction of $(u_i,u_{i+1})$ is the inverse of the previous pseudo-span $(u_{i-1},u_i)$. Also,  consider the workspace $\mathcal{C}$ of the pseudo-span $(u_i,u_{i+1})$, then $P_{u_{i+1}+1,\ldots |P|-1}$ does not leave this workspace, see Figure \ref{fig:cano:property}. This remark implies the following properties \ref{prop:decompospan:one} and \ref{prop:decompospan:two}. First, since $P_{u_{i+1}+1,\ldots, |P|-1}$ is in the east area of this cut then the westernmost tile of $P_{u_{i},\ldots, |P|-1}$ belongs to $P_{u_{i},\ldots, u_{i+1}}$.

\begin{figure}
\center
\begin{tikzpicture}[x=0.22cm,y=0.22cm]

\fill[fill=blue!30!white, draw opacity=0.8] (11,0) |- (5.5,9.5) |- (13.5,18.5) |- (11,21.5) |- (26,24) |- (11,0);

\draw[very thick] (10.5,21.5) -| (13.5,18.5) -| (5.5,9.5)-| (13.5,12.5) -| (8.5,15.5) -| (24.5,15.5);


%
\tiley{10}{21}{48}

\tileor{11}{21}{48}
\tileor{12}{21}{48}
\tileor{13}{21}{48}
\tileor{13}{20}{48}
\tileor{13}{19}{48}
\tileor{13}{18}{48}
\tileor{12}{18}{48}
\tileor{11}{18}{48}
\tileor{10}{18}{48}
\tileor{9}{18}{48}
\tileor{8}{18}{48}
\tileor{7}{18}{48}
\tileor{6}{18}{48}
\tileor{5}{18}{48}
\tileor{5}{17}{48}
\tileor{5}{16}{48}
\tileor{5}{15}{48}
\tileor{5}{14}{48}
\tileor{5}{13}{48}
\tileor{5}{12}{48}
\tileor{5}{11}{48}
\tileor{5}{10}{48}
\tileor{5}{9}{48}
\tileor{6}{9}{48}
\tileor{7}{9}{48}
\tileor{8}{9}{48}
\tileor{9}{9}{48}
\tileor{10}{9}{48}

\tiler{11}{9}{48}
\tiler{12}{9}{48}
\tiler{13}{9}{48}
\tiler{13}{10}{48}
\tiler{13}{11}{48}
\tiler{13}{12}{48}
\tiler{12}{12}{48}
\tiler{11}{12}{48}
\tiler{10}{12}{48}
\tiler{9}{12}{48}
\tiler{8}{12}{48}
\tiler{8}{13}{48}
\tiler{8}{14}{48}
\tiler{8}{15}{48}
\tiler{9}{15}{48}
\tiler{10}{15}{48}

\tileb{11}{15}{78}
\tileb{12}{15}{78}
\tileb{13}{15}{78}
\tileb{14}{15}{78}
\tileb{15}{15}{78}
\tileb{16}{15}{78}
\tileb{17}{15}{78}
\tileb{18}{15}{78}
\tileb{19}{15}{78}
\tileb{20}{15}{78}
\tileb{21}{15}{78}
\tileb{22}{15}{78}
\tileb{23}{15}{78}
\tileb{24}{15}{78}



\path [dotted, draw, thin] (0,0) grid[step=0.22cm] (26,24);


\draw[dashed] (11,0) -| (11,9.5);
\fill (10.5,9.5) circle (0.16);
\node (D) at (10.5,8.2) {$d$};

\fill (10.5,15.5) circle (0.16);
\node (D) at (10.5,16.9) {$\ell$};
\draw[dashed] (11,21.5) -| (11,24);
\fill (10.5,21.5) circle (0.16);
\node (D) at (10.3,22.7) {$n$};
\node (D) at (20.3,6.7) {$\mathcal{C}$};
\end{tikzpicture}

\caption{Following Figure \ref{fig:cano:defdcompo}, the workspace $\mathcal{C}$ of the downward span $(n,d)$ is in blue. Since this workspace is the east part of the $2D$ plane then $w_{P_{n,\ldots, d}}<w_{P_{d+1,\ldots, |P|-1}}$. Moreover, since $P_{d+1,\ldots, |P|-1}$ is inside this workspace, we have $y_{P_d}<y_{P_\ell}<y_{P_n}$ implying that the width of the span $(n,d)$ is strictly less than the width of $(d,\ell)$.}
\label{fig:cano:property}
\end{figure}

\begin{property}
\label{prop:decompospan:one}
Consider an extremal path $P$, a glue column $\Bone\leq c \leq \Btwo$ and its decomposition $(u_i)_{0\leq i \leq t}$ into pseudo-visible glues on glue column $c$. Then, for any $0\leq i <t$, we have $w_{P_{u_{i},\ldots, u_{i+1}}}<w_{P_{u_{i+1}+1,\ldots, |P|-1}}$.
\end{property}


Secondly, if $(u_i,u_{i+1})$ is an upward (resp. downward) span then $y_{P_{u_i}}<y_{P_{u_{i+2}}}<y_{P_{u_{i+1}}}$ (resp. $y_{P_{u_{i+1}}}<y_{P_{u_{i+2}}}<y_{P_{u_{i}}}$). 
This means that if $\glueP{u_0}{u_0+1}$ is visible from the south then $(u_{2i})_{0\leq i \leq \lfloor t/2 \rfloor}$ and $(u_{2i+1})_{0\leq i \leq \lfloor t/2 \rfloor}$ are increasing and decreasing sequences respectively\footnote{If $\glueP{u_0}{u_0+1}$ is visible from the north then $(u_{2i})_{0\leq i \leq \lfloor t/2 \rfloor}$ and $(u_{2i+1})_{0\leq i \leq \lfloor t/2 \rfloor}$ are  decreasing and increasing sequences respectively. In both cases, these two sequences are bounded by $y_{P_{\ell}}$.}, implying that:

\begin{property}
\label{prop:decompospan:two}
Consider an extremal path $P$, a glue column $\Bone\leq c \leq \Btwo$ and its decomposition $(u_i)_{0\leq i \leq t}$ into pseudo-visible glues on glue column $c$. Then the widths of the pseudo-spans $(u_i,u_{i+1})_{0\leq i <t}$ are strictly decreasing.
\end{property}


The definition of a good path (Definition \ref{prop:canonicalPath}) links the direction of the span and the priority of the end of the path. Now we want a similar property for each pseudo-span of the decomposition, see Figure \ref{fig:cano:defdcompo}.

\begin{definition}
\label{def:cano}
Consider an extremal path $P$, a glue column $\Bone\leq c \leq \Btwo$ and the decomposition $(u_i)_{0\leq i \leq t}$ of $P$ into pseudo-visible glues on glue column $c$. The path $P$ is \emph{canonical} for glue column $c$ if and only if for all $0\leq i <t$, if $(u_i,u_{i+1})$ is an upward (resp. downward) span then $P_{0,\ldots,u_{i+1}}$ is a prefix of the rightmost (resp. leftmost) priority extremal path with $P_{0,\ldots,u_i+1}$ as a prefix and with $\glueP{u_i}{u_i+1}$ visible from the south (resp. north).

\end{definition}

With this definition, if we need to consider a new pseudo-span in the case of Lock \ref{lock:shieldtwo}, then the direction of this pseudo-span and the priority of $P$ on this subpath will match. To finish, this section, we show that canonical paths exist.


\begin{lemma}
\label{lem:exists:canoP}
Consider a glue column $\Bone\leq c \leq \Btwo$, there exists a canonical path $P$ for glue column $c$ .
\end{lemma}

\begin{proof}
By Lemma \ref{exists:good}, there is a good path $P$ for glue column $c$. To create a canonical path, we proceed similarly as in the proof of Lemma \ref{exists:good} for each pseudo-span $(u_i,u_{i+1})$ generated by the decomposition $(u_i)_{0\leq i \leq t}$ of $P$ into pseudo-visible glues on column~$c$. 

More formally, if $P$ is a canonical path for column~$c$ then the proof ends here. Otherwise, let $1 \leq j <t$ be the smallest integer such that if $(u_j,u_{j+1})$ is an upward (resp. downward) pseudo-span then $P_{0,\ldots,u_{j+1}}$ is not a prefix of the rightmost (resp. leftmost) priority extremal path with $P_{0,\ldots,u_j+1}$ as a prefix and with $\glueP{u_j}{u_j+1}$ visible from the south (resp. north). Then, consider the rightmost (resp. leftmost) priority extremal path $Q$ such that $P_{0,\ldots,u_j+1}$ is a prefix of $Q$ and $\glueQ{u_j}{u_j+1}$ is pseudo-visible from the south (resp. north). From now on, we suppose that $(u_j,u_{j+1})$ is an upward span (the other case is symmetric). Consider the decomposition $(u'_i)_{0\leq i \leq t'}$ of $Q$ into pseudo-visible glues on glue column $c$. The next step is to prove that $u_i=u'_i$ for all $0\leq i \leq j$. To do so, since $P_{0,\ldots,u_j+1}$ is a prefix of $Q$, we must prove that all glues of $Q_{j+1,\ldots,|Q|-1}$ which are on glue column $c$ have a $y$-axis between $y_{P_{u_j}}$ and $y_{P_{u_{j+1}}}$. Consider the area $\mathcal{C}^+$ delimited by the visibility ray of $\glueP{u_j}{u_j+1}$, $P_{u_j,\ldots,|P|-1}$ and the ray starting at $\pos{|P|}-1$ and going south, see Figure \ref{fig:cano:proof}. If $Q_{u_j+1,\ldots,|Q|-1}$ is inside this area then our claim holds. By definition, $Q_{u_j+1,\ldots,|Q|-1}$ cannot leave this area by crossing through $P_{u_j+1,\ldots,|P|-1}$ or the visibility ray of $\glueP{u_j}{u_j+1}$ and it cannot reach column $e_P+1=e_\uniterm+1$. 
Now, we continue the same reasoning on path $Q$: either $Q$ is canonical for column $c$ or there is an index $j'>j$ such that the required property is not satisfied. Eventually, we found a canonical path after a number of steps bounded by the width of $(u_0,u_1)$ (since the widths of the pseudo-spans are strictly decreasing by property \ref{prop:decompospan:two}).

\begin{figure}
\center
\begin{tikzpicture}[x=0.22cm,y=0.22cm]

\fill[fill=blue!30!white, draw opacity=0.8] (11,0) |- (5.5,9.5) |- (13.5,18.5) |- (11,21.5) |- (26,24) |- (11,0);
\fill[fill=blue!40!white, draw opacity=0.8] (11,24) |- (13.5,21.5) |- (5.5,18.5) |- (13.5,9.5) |- (8.5,12.5) |- (24.5,15.5) |- (11,24);

\draw[dashed] (24.5,15.5) -| (24.5,24);

\draw[very thick] (10.5,21.5) -| (13.5,18.5) -| (5.5,9.5)-| (13.5,12.5) -| (8.5,15.5) -| (24.5,15.5);
\draw[very thick] (21,18.5) -| (17.5,19.5) -| (13.5,19.5);


%
\tiley{10}{21}{48}

\tileor{11}{21}{48}
\tileor{12}{21}{48}
\tileor{13}{21}{48}
\tileor{13}{20}{48}
\tileor{13}{19}{48}
\tileor{13}{18}{48}
\tileor{12}{18}{48}
\tileor{11}{18}{48}
\tileor{10}{18}{48}
\tileor{9}{18}{48}
\tileor{8}{18}{48}
\tileor{7}{18}{48}
\tileor{6}{18}{48}
\tileor{5}{18}{48}
\tileor{5}{17}{48}
\tileor{5}{16}{48}
\tileor{5}{15}{48}
\tileor{5}{14}{48}
\tileor{5}{13}{48}
\tileor{5}{12}{48}
\tileor{5}{11}{48}
\tileor{5}{10}{48}
\tileor{5}{9}{48}
\tileor{6}{9}{48}
\tileor{7}{9}{48}
\tileor{8}{9}{48}
\tileor{9}{9}{48}
\tileor{10}{9}{48}

\tileb{11}{9}{48}
\tileb{12}{9}{48}
\tileb{13}{9}{48}
\tileb{13}{10}{48}
\tileb{13}{11}{48}
\tileb{13}{12}{48}
\tileb{12}{12}{48}
\tileb{11}{12}{48}
\tileb{10}{12}{48}
\tileb{9}{12}{48}
\tileb{8}{12}{48}
\tileb{8}{13}{48}
\tileb{8}{14}{48}
\tileb{8}{15}{48}
\tileb{9}{15}{48}
\tileb{10}{15}{48}

\tilem{14}{19}{80}
\tilem{15}{19}{80}
\tilem{16}{19}{80}
\tilem{17}{19}{80}
\tilem{17}{18}{80}
\tilem{18}{18}{80}
\tilem{19}{18}{80}
\tilem{20}{18}{80}

\tileb{11}{15}{78}
\tileb{12}{15}{78}
\tileb{13}{15}{78}
\tileb{14}{15}{78}
\tileb{15}{15}{78}
\tileb{16}{15}{78}
\tileb{17}{15}{78}
\tileb{18}{15}{78}
\tileb{19}{15}{78}
\tileb{20}{15}{78}
\tileb{21}{15}{78}
\tileb{22}{15}{78}
\tileb{23}{15}{78}
\tileb{24}{15}{78}



\path [dotted, draw, thin] (0,0) grid[step=0.22cm] (26,24);


\draw[dashed] (11,0) -| (11,9.5);
\fill (10.5,9.5) circle (0.16);
\node (D) at (10.5,8.2) {$d$};

\node (D) at (12.3,23) {$l^n$};
\draw[dashed] (11,21.5) -| (11,24);
\fill (10.5,21.5) circle (0.16);
\node (D) at (10.3,20.2) {$n$};
\node (D) at (20.3,6.7) {$\mathcal{C^-}$};
\node (D) at (19.3,22.2) {$\mathcal{C^+}$};

\node (D) at (22,18.5) {\huge ?};

\end{tikzpicture}

\caption{Following Figure \ref{fig:cano:property}, we focus on the span $(n,d)$ of $P$. If there exists $0\leq i < n$ such that $\glueP{i}{i+1}$ is pseudo-visible from the south (resp. north) in $P$ on glue column $c$ then it is below (resp. over) $\glueP{d}{d+1}$ (resp. $\glueP{n}{n+1}$). Consider the rightmost priority extremal path $Q$ (in magenta) such that $P_{0,\ldots,n+1}$ is a prefix of $Q$ and $\glueQ{n}{n+1}$ is pseudo-visible from the north. Then, $Q_{n+1,\ldots,|Q|-1}$ is inside the area $\mathcal{C}^+$: by definition, this path cannot leave $\mathcal{C}^+$ by crossing $l^n$ and $P_{n+1,\ldots,|P|-1}$ and since $P$ and $Q$ are extremal, the last tile of $Q$ is the only one on column $e_P$.}
\label{fig:cano:proof}
\end{figure}

\end{proof}

\section{U-turn}
\label{sec:U-turn}

In this part, we develop a toolbox to explain how shields are able to prevent different parts of the path to interact with each other (see Section \ref{road:pseudo-visibility} of the roadmap). To do so, we prove first a technical Lemma \ref{Uturn:main} in Subsection~\ref{sec:sub:Uturn}. This lemma is a formal proof that only the first case of Fact \ref{fact:turn} can occur. Note that, this lemma is also a rewriting of the Shield Lemma of~\cite{STOC2020}. The proof is far simpler here, since we are working in the directed case with a finite terminal assembly and on paths with specific properties. Thus, the statement of this lemma is different from the one of \cite{STOC2020} but all the arguments are taken from this previous paper. Indeed, we use this result to solve Lock~\ref{lem:glue:prop3} (Lemma \ref{Uturn:glueWest}) following the reasoning done in \cite{STOC2020}. Similarly, we use this result to solve Lock \ref{lock:boundZero} (Lemma \ref{Uturn:boundWest}) following the unpublished reasoning done in \cite{pumpabilityLargeBound}. Then, we proceed to improve the quadratic bound of \cite{STOC2020} for U-turn into a linear one to solve Lock \ref{lock:uturnTwo} (Lemma \ref{Uturn:LinearBound} and Corollary \ref{fact:Bthree}) in Subsection~\ref{sec:U-turn:linear}. Solving Lock \ref{lock:uturnOne} (Lemma \ref{Uturn:LinearCano}) in Subsection \ref{sec:U-turn:linearCano} unfortunately requires to do reasonings similar to the ones in Lemma \ref{Uturn:main} and Lemma \ref{Uturn:LinearBound}. We think that there should exist more general versions of these Lemmas. Such a result would slightly simplify this proof and would be a powerful tool for further studies. Unfortunately, we were not able to obtain such a result here. Finally, we conclude this section by solving Lock \ref{lock:uturn} by putting all the arguments together in Lemma \ref{Uturn:Conc}. In this section only, we reason about \emph{visible tiles} and not just visible glues.

\subsection{The U-turn lemma and its corollaries}
\label{sec:sub:Uturn}

\begin{lemma}
\label{Uturn:main}
Consider a producible path $P$ such that its last tile is visible from the north, we denote by $l^P$ the ray starting at $\pos{P_{|P|-1}}$ and going north. Suppose that there exist $0\leq i \leq j \leq |P|-1$ such that $\glueP{i}{i+1}$ and $\glueP{j}{j+1}$ are both pointing east, have the same type and are visible from the south on glue columns $c$ and $c'$ respectively. Also, assume that $P$ is the rightmost-priority path with $P_{0,\ldots,i+1}$ as a prefix, with $\glueP{i}{i+1}$ visible from the south and with its last tile on tile ray $l^P$. Then, we have: $$w_{P_{i+1,\ldots,|P|-1}}\leq w_{P_{j+1,\ldots,|P|-1}}+(c-c').$$
\end{lemma}

\begin{proof}
Consider the workspace $\mathcal{C}$ which is the right side of the cut done by $l^i$, $l^P$ and $P_{i,\ldots,|P|-1}$. By Lemma \ref{lem:glue:prop4}, we have $c \leq c'$ and then the workspace $\mathcal{C}$ is cut into two parts  $\mathcal{C^-}$ and  $\mathcal{C^+}$ by the glue ray $l^j$ as shown in Figure~\ref{fig:Uturn:techincal} (a). Remark that if $(P_{j+1,\ldots,|P|-1}-\vect{P_iP_j})$ is inside $\mathcal{C}$ then $P_{0,\ldots,i}(P_{j+1,\ldots,|P|-1}-\vect{P_iP_j})$ is producible. In this case, we must have $w_{P_{i+1,\ldots,|P|-1}}\leq w_{P_{j+1,\ldots,|P|-1}}+(c-c')$ in order to have enough space inside $\mathcal{C}$ for $(P_{j+1,\ldots,|P|-1}-\vect{P_iP_j})$ to grow inside it, see Figure \ref{fig:Uturn:techincal} (b).

\begin{figure}
\center
\begin{minipage}{0.47\linewidth}
\begin{tikzpicture}[x=0.22cm,y=0.22cm]

\fill[fill=blue!30!white, draw opacity=0.8] (6,5.5) -| (9.5,9.5) -| (2.5,15.5) -| (13.5,12.5) -| (18.5,14.5) -| (20.5,19.5) -| (13.5,23) -| (26,0) -| (6,5.5) ;
\fill[fill=blue!50!white, draw opacity=0.8] (6,5.5) -| (9.5,9.5) -| (2.5,15.5) -| (13.5,12.5) -| (15,0) -| (6,5.5);

\draw[very thick] (5.5,5.5) -| (9.5,9.5) -| (2.5,15.5) -| (13.5,12.5) -| (18.5,14.5) -| (20.5,19.5) -| (13.5,19.5);


\tileg{5}{5}{85}

\tiley{6}{5}{48}
\tiley{7}{5}{48}
\tiley{8}{5}{48}
\tiley{9}{5}{48}
\tiley{9}{6}{48}
\tiley{9}{7}{48}
\tiley{9}{8}{48}
\tiley{9}{9}{48}
\tiley{8}{9}{48}
\tiley{7}{9}{48}
\tiley{6}{9}{48}
\tiley{5}{9}{48}
\tiley{4}{9}{48}
\tiley{3}{9}{48}
\tiley{2}{9}{48}
\tiley{2}{10}{48}
\tiley{2}{11}{48}
\tiley{2}{12}{48}
\tiley{2}{13}{48}
\tiley{2}{14}{48}
\tiley{2}{15}{48}
\tiley{3}{15}{48}
\tiley{4}{15}{48}
\tiley{5}{15}{48}
\tiley{6}{15}{48}
\tiley{7}{15}{48}
\tiley{8}{15}{48}
\tiley{9}{15}{48}
\tiley{10}{15}{48}
\tiley{11}{15}{48}
\tiley{12}{15}{48}
\tiley{13}{15}{48}
\tiley{13}{14}{48}
\tiley{13}{13}{48}
\tiley{13}{12}{48}
\tiley{14}{12}{48}
\tileb{15}{12}{48}
\tileb{16}{12}{48}
\tileb{17}{12}{48}
\tileb{18}{12}{48}
\tileb{18}{13}{48}
\tileb{18}{14}{48}
\tileb{19}{14}{48}
\tileb{20}{14}{48}
\tileb{20}{15}{48}
\tileb{20}{16}{48}
\tileb{20}{17}{48}
\tileb{20}{18}{48}
\tileb{20}{19}{48}
\tileb{19}{19}{48}
\tileb{18}{19}{48}
\tileb{17}{19}{48}
\tileb{16}{19}{48}
\tileb{15}{19}{48}
\tileb{14}{19}{48}
\tileb{13}{19}{48}

\path [dotted, draw, thin] (0,0) grid[step=0.22cm] (26,23);

\draw[dashed] (6,0) -| (6,5.5);
\fill (5.5,5.5) circle (0.16);
\node (D) at (5.5,6.9) {$i$};
\node (D) at (4.8,1.5) {$l^i$};
\draw[dashed] (13.5,19.5) -| (13.5,23);
\fill (13.5,19.5) circle (0.16);
\node (D) at (13.5,18) {\tiny $|P|-1$};
\draw[dashed] (15,0) -| (15,12.5);
\fill (14.5,12.5) circle (0.16);
\node (D) at (14.5,11.2) {$j$};
\node (D) at (14,1.5) {$l^j$};
\node (D) at (12.2,21.5) {$l^P$};

\node (D) at (10,1.7) {$\mathcal{C^+}$};
\node (D) at (24,1.7) {$\mathcal{C^-}$};
\end{tikzpicture}

\center (a) 
\end{minipage}
\begin{minipage}{0.47\linewidth}
\begin{tikzpicture}[x=0.22cm,y=0.22cm]

\fill[fill=blue!30!white, draw opacity=0.8] (6,5.5) -| (9.5,9.5) -| (2.5,15.5) -| (13.5,12.5) -| (18.5,14.5) -| (20.5,19.5) -| (13.5,23) -| (26,0) -| (6,5.5) ;
\fill[fill=blue!50!white, draw opacity=0.8] (6,5.5) -| (9.5,9.5) -| (2.5,15.5) -| (13.5,12.5) -| (15,0) -| (6,5.5);

\draw[very thick] (5.5,5.5) -| (9.5,9.5) -| (2.5,15.5) -| (13.5,12.5) -| (18.5,14.5) -| (20.5,19.5) -| (13.5,19.5);
\draw[very thick] (9.5,7.5) -| (11.5,12.5) -| (4.5,12.5) ;


\tileg{5}{5}{85}

\tileor{6}{5}{48}
\tileor{7}{5}{48}
\tileor{8}{5}{48}
\tileor{9}{5}{48}
\tileor{9}{6}{48}
\tileor{9}{7}{48}
\tiley{9}{8}{48}
\tiley{9}{9}{48}
\tiley{8}{9}{48}
\tiley{7}{9}{48}
\tiley{6}{9}{48}
\tiley{5}{9}{48}
\tiley{4}{9}{48}
\tiley{3}{9}{48}
\tiley{2}{9}{48}
\tiley{2}{10}{48}
\tiley{2}{11}{48}
\tiley{2}{12}{48}
\tiley{2}{13}{48}
\tiley{2}{14}{48}
\tiley{2}{15}{48}
\tiley{3}{15}{48}
\tiley{4}{15}{48}
\tiley{5}{15}{48}
\tiley{6}{15}{48}
\tiley{7}{15}{48}
\tiley{8}{15}{48}
\tiley{9}{15}{48}
\tiley{10}{15}{48}
\tiley{11}{15}{48}
\tiley{12}{15}{48}
\tiley{13}{15}{48}
\tiley{13}{14}{48}
\tiley{13}{13}{48}
\tiley{13}{12}{48}
\tiley{14}{12}{48}
\tileb{15}{12}{48}
\tileb{16}{12}{48}
\tileb{17}{12}{48}
\tileb{18}{12}{48}
\tileb{18}{13}{48}
\tileb{18}{14}{48}
\tileb{19}{14}{48}
\tileb{20}{14}{48}
\tileb{20}{15}{48}
\tileb{20}{16}{48}
\tileb{20}{17}{48}
\tileb{20}{18}{48}
\tileb{20}{19}{48}
\tileb{19}{19}{48}
\tileb{18}{19}{48}
\tileb{17}{19}{48}
\tileb{16}{19}{48}
\tileb{15}{19}{48}
\tileb{14}{19}{48}
\tileb{13}{19}{48}

\tiler{10}{7}{48}
\tiler{11}{7}{48}
\tiler{11}{8}{48}
\tiler{11}{9}{48}
\tiler{11}{10}{48}
\tiler{11}{11}{48}
\tiler{11}{12}{48}
\tiler{10}{12}{48}
\tiler{9}{12}{48}
\tiler{8}{12}{48}
\tiler{7}{12}{48}
\tiler{6}{12}{48}
\tiler{5}{12}{48}
\tiler{4}{12}{48}

\path [dotted, draw, thin] (0,0) grid[step=0.22cm] (26,23);

\draw[dashed] (6,0) -| (6,5.5);
\fill (5.5,5.5) circle (0.16);
\node (D) at (5.5,6.9) {$i$};
\node (D) at (4.8,1.5) {$l^i$};
\draw[dashed] (13.5,19.5) -| (13.5,23);
\fill (13.5,19.5) circle (0.16);
\node (D) at (13.5,18) {\tiny $|P|-1$};
\node (D) at (12.2,21.5) {$l^P$};

\draw[dashed] (15,0) -| (15,12.5);
\fill (14.5,12.5) circle (0.16);
\node (D) at (14.5,11.2) {$j$};
\node (D) at (14,1.5) {$l^j$};

\draw[->,thick] (13.5,19.5) -- (4.5,12.5);

\node (D) at (10,1.7) {$\mathcal{C^+}$};
\node (D) at (24,1.7) {$\mathcal{C^-}$};

\end{tikzpicture}

\center (b) 
\end{minipage}

\caption{Proof of Lemma \ref{Uturn:main} Part (1/2) (a) a path $P$ satisfying the hypothesis of Lemma \ref{Uturn:main}, only $P_{i,\ldots,|P|-1}$ is represented, $P_{i+1,\ldots,j}$ is in yellow while $P_{j+1,\ldots,|P|-1}$ is in blue. The workspace $\mathcal{C}$ (in blue) is cut in two areas  $\mathcal{C^-}$ and  $\mathcal{C^+}$ by $l^j$. (b) $P_{j+1,\ldots,|P|-1}-\vect{P_iP_j}$ (in orange and red) is assembled at the end of $P_{0,\ldots,i}$. The orange tiles  are in common between $P_{j+1,\ldots,|P|-1}-\vect{P_iP_j}$ and $P_{i+1,\ldots,|P|-1}$.  $P_{j+1,\ldots,|P|-1}-\vect{P_iP_j}$ is inside $\mathcal{C}$ and thus its westernmost tile (the last one in this example) is also inside $\mathcal{C}$.}
\label{fig:Uturn:techincal}
\end{figure}

Indeed, if $P_{j+1,\ldots,|P|-1}-\vect{P_iP_j}$ is a prefix of $P_{i+1,\ldots,|P|-1}$, then the result is true. Otherwise, if $P_{j+1,\ldots,|P|-1}-\vect{P_iP_j}$ turns right of $P_{i+1,\ldots,|P|-1}$ then the end of $P_{j+1,\ldots,|P|-1}-\vect{P_iP_j}$ cannot intersect with $l^i$ (since $P_{j+1,\ldots,|P|-1}$ does not intersect with $l^j$ due to the visibility of $\glueP{j}{j+1}$), $P_{i+1,\ldots,|P|-1}$ or $l^P$ (because of the priority of $P$), see Figure \ref{fig:Uturn:techincal}(b). Thus, $P_{j+1,\ldots,|P|-1}-\vect{P_iP_j}$ is inside $\mathcal{C}$ and the result is true. 

For the sake of contradiction, the last remaining case is $P_{i+1,\ldots,|P|-1}+\vect{P_iP_j}$ turning right of $P_{j+1,\ldots,|P|-1}$, see Figure \ref{fig:Uturn:techincalPart2}. By the same arguments, $P_{i+1,\ldots,|P|-1}+\vect{P_iP_j}$ is inside $\mathcal{C}^-$ since this path cannot intersect $l^j$ (since $P_{i+1,\ldots,|P|-1}$ does not intersect with $l^i$ due to the visibility of $\glueP{i}{i+1}$ )  and $P_{j+1,\ldots,|P|-1}$ or $l^P$ (because of the priority of $P$).  Then $P_{i+1,\ldots,|P|-1}+\vect{P_iP_j}$ is inside $\mathcal{C}^-$
(and in particular, $P_{i+1,\ldots,j}+\vect{P_iP_j}$ is inside $\mathcal{C}^-$). By definition $P_{i+1,\ldots,j}$ is inside $\mathcal{C}^+$. Then $P_{i+1,\ldots,j}$ and $P_{i+1,\ldots,j}+\vect{P_iP_j}$ do not intersect. By a folklore lemma (see \cite{STOC2020}, for example), $Q=\bigcup_{k\geq0} (P_{i+1,\ldots,j}+k\vect{P_iP_j})$ is an infinite self-avoiding path. If $Q$ is inside $\mathcal{C}$ then $P_{0,\ldots,i}Q$ is producible which contradicts the fact that the terminal assembly $\uniterm$ is finite. Then, the path $Q$ must leave $\mathcal{C}$. Nevertheless, $Q$ cannot cross trough $P_{j+1,\ldots,|P|-1}$ and $l^P$ (due to the priority of $P$) and if it crosses $l^j$ then consider the smallest $k > 0$ such that $P_{i+1,\ldots,j+1}+k\vect{P_iP_j}$ crosses $l^j$. If $k=1$ then $P_{i+1,\ldots,j+1}$ intersects $l^i$ which is a contradiction (since $\glueP{i}{i+1}$ is visible in $P$) otherwise $P_{i+1,\ldots,j+1}+(k-1)\vect{P_iP_j}$ intersects $l^i$. Nevertheless, reaching $\mathcal{C}^+$ from $\mathcal{C}^-$ requires to cross $l^j$ which contradicts the definition of $k$.

\begin{figure}
\center
\begin{minipage}{0.47\linewidth}
\begin{tikzpicture}[x=0.22cm,y=0.22cm]

\fill[fill=blue!30!white, draw opacity=0.8] (6,5.5) -| (9.5,9.5) -| (13.5,12.5) -| (18.5,19.5) -| (13.5,23) -| (26,0) -| (6,5.5) ;
\fill[fill=blue!50!white, draw opacity=0.8] (6,5.5) -| (9.5,9.5) -| (13.5,12.5) -| (15,0) -| (6,5.5);

\draw[very thick] (5.5,5.5) -| (9.5,9.5) -| (13.5,12.5) -| (18.5,19.5) -| (13.5,19.5);
\draw[very thick] (9.5,9.5) |- (4.5,12.5);


\tileg{5}{5}{85}

\tileor{6}{5}{48}
\tileor{7}{5}{48}
\tileor{8}{5}{48}
\tileor{9}{5}{48}
\tileor{9}{6}{48}
\tileor{9}{7}{48}
\tileor{9}{8}{48}
\tileor{9}{9}{48}
\tiley{10}{9}{48}
\tiley{11}{9}{48}
\tiley{12}{9}{48}
\tiley{13}{9}{48}
\tiley{13}{10}{48}
\tiley{13}{11}{48}
\tiley{13}{12}{48}
\tiley{14}{12}{48}
\tileb{15}{12}{48}
\tileb{16}{12}{48}
\tileb{17}{12}{48}
\tileb{18}{12}{48}
\tileb{18}{13}{48}
\tileb{18}{14}{48}
\tileb{18}{15}{48}
\tileb{18}{16}{48}
\tileb{18}{17}{48}
\tileb{18}{18}{48}
\tileb{18}{19}{48}
\tileb{17}{19}{48}
\tileb{16}{19}{48}
\tileb{15}{19}{48}
\tileb{14}{19}{48}
\tileb{13}{19}{48}

\tiler{9}{10}{48}
\tiler{9}{11}{48}
\tiler{9}{12}{48}
\tiler{8}{12}{48}
\tiler{7}{12}{48}
\tiler{6}{12}{48}
\tiler{5}{12}{48}
\tiler{4}{12}{48}

\path [dotted, draw, thin] (0,0) grid[step=0.22cm] (26,23);

\draw[dashed] (6,0) -| (6,5.5);
\fill (5.5,5.5) circle (0.16);
\node (D) at (5.5,6.9) {$i$};
\node (D) at (4.8,1.5) {$l^i$};
\draw[dashed] (13.5,19.5) -| (13.5,23);
\fill (13.5,19.5) circle (0.16);
\node (D) at (13.5,18) {\tiny $|P|-1$};
\draw[dashed] (15,0) -| (15,12.5);
\fill (14.5,12.5) circle (0.16);
\node (D) at (14.5,14) {$j$};
\node (D) at (14,1.5) {$l^j$};
\node (D) at (12.2,21.5) {$l^P$};

\node (D) at (10,1.7) {$\mathcal{C^+}$};
\node (D) at (24,1.7) {$\mathcal{C^-}$};

\draw[->, thick] (13.5,19.5) -- (4.5,12.5);

\end{tikzpicture}
\center (a) 
\end{minipage}
\begin{minipage}{0.47\linewidth}
\begin{tikzpicture}[x=0.22cm,y=0.22cm]

\fill[fill=blue!30!white, draw opacity=0.8] (6,5.5) -| (9.5,9.5) -| (13.5,12.5) -| (18.5,19.5) -| (13.5,23) -| (26,0) -| (6,5.5) ;
\fill[fill=blue!50!white, draw opacity=0.8] (6,5.5) -| (9.5,9.5) -| (13.5,12.5) -| (15,0) -| (6,5.5);

\draw[very thick] (5.5,5.5) -| (9.5,9.5) -| (13.5,12.5) -| (18.5,19.5) -| (13.5,19.5);
\draw[very thick] (18.5,16.5) -| (22.5,19.5) -| (23.5,19.5);


\tileg{5}{5}{85}

\tiley{6}{5}{48}
\tiley{7}{5}{48}
\tiley{8}{5}{48}
\tiley{9}{5}{48}
\tiley{9}{6}{48}
\tiley{9}{7}{48}
\tiley{9}{8}{48}
\tiley{9}{9}{48}
\tiley{10}{9}{48}
\tiley{11}{9}{48}
\tiley{12}{9}{48}
\tiley{13}{9}{48}
\tiley{13}{10}{48}
\tiley{13}{11}{48}
\tiley{13}{12}{48}
\tiley{14}{12}{48}
\tileor{15}{12}{48}
\tileor{16}{12}{48}
\tileor{17}{12}{48}
\tileor{18}{12}{48}
\tileor{18}{13}{48}
\tileor{18}{14}{48}
\tileor{18}{15}{48}
\tileor{18}{16}{48}
\tileb{18}{17}{48}
\tileb{18}{18}{48}
\tileb{18}{19}{48}
\tileb{17}{19}{48}
\tileb{16}{19}{48}
\tileb{15}{19}{48}
\tileb{14}{19}{48}
\tileb{13}{19}{48}

\tiler{19}{16}{48}
\tiler{20}{16}{48}
\tiler{21}{16}{48}
\tiler{22}{16}{48}
\tiler{22}{17}{48}
\tiler{22}{18}{48}
\tiler{22}{19}{48}
\tiler{23}{19}{48}

\path [dotted, draw, thin] (0,0) grid[step=0.22cm] (26,23);

\draw[dashed] (6,0) -| (6,5.5);
\fill (5.5,5.5) circle (0.16);
\node (D) at (5.5,6.9) {$i$};
\node (D) at (4.8,1.5) {$l^i$};
\draw[dashed] (13.5,19.5) -| (13.5,23);
\fill (13.5,19.5) circle (0.16);
\node (D) at (13.5,18) {\tiny $|P|-1$};
\draw[dashed] (15,0) -| (15,12.5);
\fill (14.5,12.5) circle (0.16);
\node (D) at (14.5,14) {$j$};
\node (D) at (14,1.5) {$l^j$};
\node (D) at (12.2,21.5) {$l^P$};

\node (D) at (10,1.7) {$\mathcal{C^+}$};
\node (D) at (24,1.7) {$\mathcal{C^-}$};

\draw[->, thick] (14.5,12.5) -- (23.5,19.5);

\end{tikzpicture}
\center (b) 
\end{minipage}

\vspace{+0.4em}
\begin{tikzpicture}[x=0.22cm,y=0.22cm]

\fill[fill=blue!30!white, draw opacity=0.8] (6,5.5) -| (9.5,9.5) -| (13.5,12.5) -| (18.5,19.5) -| (13.5,25) -| (31,0) -| (6,5.5) ;
\fill[fill=blue!50!white, draw opacity=0.8] (6,5.5) -| (9.5,9.5) -| (13.5,12.5) -| (15,0) -| (6,5.5);

\draw[very thick] (5.5,5.5) -| (9.5,9.5) -| (13.5,12.5) -| (18.5,19.5) -| (13.5,19.5);
\draw[very thick] (18.5,16.5) -| (22.5,19.5) -| (27.5,23.5) -| (31,23.5);


\tileg{5}{5}{85}

\tiley{6}{5}{48}
\tiley{7}{5}{48}
\tiley{8}{5}{48}
\tiley{9}{5}{48}
\tiley{9}{6}{48}
\tiley{9}{7}{48}
\tiley{9}{8}{48}
\tiley{9}{9}{48}
\tiley{10}{9}{48}
\tiley{11}{9}{48}
\tiley{12}{9}{48}
\tiley{13}{9}{48}
\tiley{13}{10}{48}
\tiley{13}{11}{48}
\tiley{13}{12}{48}
\tiley{14}{12}{48}
\tileor{15}{12}{48}
\tileor{16}{12}{48}
\tileor{17}{12}{48}
\tileor{18}{12}{48}
\tileor{18}{13}{48}
\tileor{18}{14}{48}
\tileor{18}{15}{48}
\tileor{18}{16}{48}
\tileb{18}{17}{48}
\tileb{18}{18}{48}
\tileb{18}{19}{48}
\tileb{17}{19}{48}
\tileb{16}{19}{48}
\tileb{15}{19}{48}
\tileb{14}{19}{48}
\tileb{13}{19}{48}

\tiler{19}{16}{48}
\tiler{20}{16}{48}
\tiler{21}{16}{48}
\tiler{22}{16}{48}
\tiler{22}{17}{48}
\tiler{22}{18}{48}
\tiler{22}{19}{48}
\tiler{23}{19}{48}

\tiler{24}{19}{48}
\tiler{25}{19}{48}
\tiler{26}{19}{48}
\tiler{27}{19}{48}
\tiler{27}{20}{48}
\tiler{27}{21}{48}
\tiler{27}{22}{48}
\tiler{27}{23}{48}
\tiler{28}{23}{48}
\tiler{29}{23}{48}
\tiler{30}{23}{48}

\path [dotted, draw, thin] (0,0) grid[step=0.22cm] (31,25);

\draw[dashed] (6,0) -| (6,5.5);
\fill (5.5,5.5) circle (0.16);
\node (D) at (5.5,6.9) {$i$};
\node (D) at (4.8,1.5) {$l^i$};
\draw[dashed] (13.5,19.5) -| (13.5,25);
\fill (13.5,19.5) circle (0.16);
\node (D) at (13.5,18) {\tiny $|P|-1$};
\draw[dashed] (15,0) -| (15,12.5);
\fill (14.5,12.5) circle (0.16);
\node (D) at (14.5,14) {$j$};
\node (D) at (14,1.5) {$l^j$};
\node (D) at (12.2,21.5) {$l^P$};

\node (D) at (10,1.7) {$\mathcal{C^+}$};
\node (D) at (28,1.7) {$\mathcal{C^-}$};

\draw[->, thick] (14.5,12.5) -- (23.5,19.5);

\end{tikzpicture}

(c)

\caption{Proof of Lemma \ref{Uturn:main} Part (2/2). (a) The case where $P_{j+1,\ldots,|P|-1}-\vect{P_iP_j}$ turns left of $P_{i+1,\ldots,|P|-1}$. (b) Then, $P_{i+1,\ldots,j}+\vect{P_iP_j}$ is in $\mathcal{C^-}$ and does not intersect $P_{i+1,\ldots,j}$ which is inside $\mathcal{C^+}$. (c) It is possible to assemble an infinite path which is a contradiction.}
\label{fig:Uturn:techincalPart2}
\end{figure}

\end{proof}


When the hypotheses of Lemma \ref{Uturn:main} are met, a path must create an area  near its beginning in order to allow a copy of its end to grow inside this area. In fact this lemma will be used to show that U-turn are impossible because there is not enough space to create this area, see Lemma \ref{Uturn:LinearBound}. We illustrate how to use Lemma \ref{Uturn:main} in a simple example to solve Lock \ref{lem:glue:prop3}  (see Lemma 4.2 of \cite{STOC2020}). Also note that there exists variant of Lemma \ref{Uturn:main} when the glues are visible from the north or when the last tile of $P$ is visible from the south.

\begin{lemma}
\label{Uturn:glueWest}
Consider a producible path $P$ whose last tile is the easternmost one and $0\leq i \leq |P|-1$ such that $\gluePi$ is visible from the south on glue column $\Bone$, then  $\gluePi$ points east.
\end{lemma}

\begin{proof}
For the sake of contradiction, suppose that there exists a path $P$ such that its last tile is the easternmost one and there exists $0\leq i \leq |P|-1$ such that $\gluePi$ is visible from the south on glue column $\Bone$ and points west, see Figure \ref{fig:Proof:CorollaryOne}. Then consider the path $Q$ such that $Q$ is the leftmost priority path such that $P_{0,\ldots,i+1}$ is a prefix of $Q$, $\glueQ{i}{i+1}$ is visible from the south (and thus points west) and $Q_{|Q|-1}$ is the easternmost tile of $Q$ and is on column $e_P$. Then by the pigeonhole principle, there exist $0\leq j \leq |Q|-1$ and $0\leq k \leq |Q|-1$ such that $\glueQ{j}{j+1}$ and $\glueQ{k}{k+1}$ have the same type and they both are visible from the south on glue columns $c'$ and $c$ respectively with $e_\sigma+0.5 < c < c' \leq \Bone$. Since $\glueQ{i}{i+1}$ points west then both of these glues points west by Lemma \ref{lem:glue:prop2}, and we have $i\leq j < k$ by Lemma \ref{lem:glue:prop4}. Then we can apply Lemma \ref{Uturn:main} on $Q$ with the indices $j$ and $k$ to obtain $e_{Q_{j,\ldots,|Q|-1}}\geq e_{Q_{k,\ldots,|Q|-1}}+(c'-c)$. Since $e_Q=x_{Q_{|Q|-1}}$ and $c'>c$, there is a contradiction since the previous equation implies $e_{Q_{j,\ldots,|Q|-1}}>e_Q$.

\begin{figure}
\center
\begin{minipage}{0.47\linewidth}
\begin{tikzpicture}[x=0.22cm,y=0.22cm]

\fill[fill=blue!30!white, draw opacity=0.8] (14,7.5) -| (4.5,13.5) -| (22.5,21) -| (0,0) -| (14,7.5);

\draw[very thick] (7.5,10.5) -| (15.5,7.5) -| (4.5,13.5) -| (22.5,13.5);

\tiles{7}{10}{85}
\tileg{8}{10}{85}
\tileg{9}{10}{85}
\tileg{10}{10}{85}
\tileg{11}{10}{85}
\tileg{11}{10}{85}
\tileg{12}{10}{85}
\tileg{13}{10}{85}
\tileg{14}{10}{85}
\tileg{15}{10}{85}
\tileg{15}{9}{85}
\tileg{15}{8}{85}
\tileg{15}{7}{85}
\tileg{14}{7}{85}
\tiley{13}{7}{85}
\tiley{12}{7}{85}
\tiley{11}{7}{85}
\tiley{10}{7}{85}
\tiley{9}{7}{85}
\tiley{8}{7}{85}
\tiley{7}{7}{85}
\tiley{6}{7}{85}
\tiley{5}{7}{85}
\tiley{4}{7}{85}
\tiley{4}{8}{85}
\tiley{4}{9}{85}
\tiley{4}{10}{85}
\tiley{4}{11}{85}
\tiley{4}{12}{85}
\tiley{4}{13}{85}
\tiley{5}{13}{85}
\tiley{6}{13}{85}
\tiley{7}{13}{85}
\tiley{8}{13}{85}
\tiley{9}{13}{85}
\tiley{10}{13}{85}
\tiley{11}{13}{85}
\tiley{12}{13}{85}
\tiley{13}{13}{85}
\tiley{14}{13}{85}
\tiley{15}{13}{85}
\tiley{16}{13}{85}
\tiley{17}{13}{85}
\tiley{18}{13}{85}
\tiley{19}{13}{85}
\tiley{20}{13}{85}
\tiley{21}{13}{85}
\tiley{22}{13}{85}


\path [dotted, draw, thin] (0,0) grid[step=0.22cm] (26,21);

\draw[dashed] (14,0) -| (14,7.5);
\fill (14.5,7.5) circle (0.16);
\node (D) at (14.5,6.2) {$i$};
\node (D) at (14.8,1.5) {$l^i$};

\draw[dashed] (22.5,13.5) -| (22.5,21);
\fill (22.5,13.5) circle (0.16);
\node (D) at (22.5,12) {\tiny $|P|-1$};
\node (D) at (24.2,15.5) {$l^P$};

\node (D) at (14,-1) {$\Bone$};
\node (D) at (22.5,22) {$e_\sigma$};

\node (D) at (2,19.7) {$\mathcal{C}$};
\end{tikzpicture}

\center (a) 
\end{minipage}
\begin{minipage}{0.47\linewidth}
\begin{tikzpicture}[x=0.22cm,y=0.22cm]

\fill[fill=blue!30!white, draw opacity=0.8] (14,7.5) -| (4.5,13.5) -| (22.5,21) -| (0,0) -| (14,7.5);

\draw[very thick] (7.5,10.5) -| (15.5,7.5) -| (4.5,13.5) -| (22.5,13.5);
\draw[very thick] (10.5,7.5) |- (2.5,4.5) |- (4.5,12.5);
\draw[very thick] (10.5,13.5) |- (15.5,15.5) |- (22.5,16.5);

\tiles{7}{10}{85}
\tileg{8}{10}{85}
\tileg{9}{10}{85}
\tileg{10}{10}{85}
\tileg{11}{10}{85}
\tileg{11}{10}{85}
\tileg{12}{10}{85}
\tileg{13}{10}{85}
\tileg{14}{10}{85}
\tileg{15}{10}{85}
\tileg{15}{9}{85}
\tileg{15}{8}{85}
\tileg{15}{7}{85}
\tileg{14}{7}{85}

\tileor{13}{7}{85}
\tileor{12}{7}{85}
\tileor{11}{7}{85}
\tileor{10}{7}{85}
\tiley{9}{7}{85}
\tiley{8}{7}{85}
\tiley{7}{7}{85}
\tiley{6}{7}{85}
\tiley{5}{7}{85}
\tiley{4}{7}{85}
\tiley{4}{8}{85}
\tiley{4}{9}{85}
\tiley{4}{10}{85}
\tiley{4}{11}{85}
\tileor{4}{12}{85}
\tileor{4}{13}{85}
\tileor{5}{13}{85}
\tileor{6}{13}{85}
\tileor{7}{13}{85}
\tileor{8}{13}{85}
\tileor{9}{13}{85}
\tileor{10}{13}{85}
\tiley{11}{13}{85}
\tiley{12}{13}{85}
\tiley{13}{13}{85}
\tiley{14}{13}{85}
\tiley{15}{13}{85}
\tiley{16}{13}{85}
\tiley{17}{13}{85}
\tiley{18}{13}{85}
\tiley{19}{13}{85}
\tiley{20}{13}{85}
\tiley{21}{13}{85}
\tiley{22}{13}{85}

\tiler{10}{6}{85}
\tiler{10}{5}{85}
\tiler{10}{4}{85}
\tiler{9}{4}{85}
\tiler{8}{4}{85}
\tiler{7}{4}{85}
\tiler{6}{4}{85}
\tiler{5}{4}{85}
\tiler{4}{4}{85}
\tiler{3}{4}{85}
\tiler{2}{4}{85}
\tiler{2}{5}{85}
\tiler{2}{6}{85}
\tiler{2}{7}{85}
\tiler{2}{8}{85}
\tiler{2}{9}{85}
\tiler{2}{10}{85}
\tiler{2}{11}{85}
\tiler{2}{12}{85}
\tiler{3}{12}{85}
\tiler{10}{14}{85}
\tiler{10}{15}{85}
\tiler{11}{15}{85}
\tiler{12}{15}{85}
\tiler{13}{15}{85}
\tiler{14}{15}{85}
\tiler{15}{15}{85}
\tiler{15}{16}{85}
\tiler{16}{16}{85}
\tiler{17}{16}{85}
\tiler{18}{16}{85}
\tiler{19}{16}{85}
\tiler{20}{16}{85}
\tiler{21}{16}{85}
\tiler{22}{16}{85}


\path [dotted, draw, thin] (0,0) grid[step=0.22cm] (26,21);

\draw[dashed] (9,0) -| (9,4.5);
\draw[<->, thick] (9,2.5) -- (14,2.5);
\node (D) at (11.5,1.5) {\tiny $|T|+1$};

\draw[dashed] (14,0) -| (14,7.5);
\fill (14.5,7.5) circle (0.16);
\node (D) at (14.5,6.2) {$i$};
\node (D) at (14.8,1.5) {$l^i$};

\fill (11.5,7.5) circle (0.16);
\node (D) at (11.5,8.9) {$j$};

\fill (9.5,4.5) circle (0.16);
\node (D) at (9.5,5.9) {$k$};

\draw[->, thick] (22.5,16.5) -- (24.5,19.5);
\fill (24.5,19.5) circle (0.16);

\draw[dashed] (22.5,13.5) -| (22.5,21);
\fill (22.5,13.5) circle (0.16);
\node (D) at (22.5,12) {\tiny $|P|-1$};
\node (D) at (24.2,15.5) {$l^P$};

\node (D) at (2,19.7) {$\mathcal{C}$};

\node (D) at (14,-1) {$\Bone$};
\node (D) at (22.5,22) {$e_\sigma$};

\end{tikzpicture}

\center (b) 
\end{minipage}

\caption{Proof of Lemma \ref{Uturn:glueWest} (a) A path whose $\glueP{i}{i+1}$ is visible from the south on column $\Bone$ and pointing west. The workspace $\mathcal{C}$ is in blue. (b) $P_{i+1,\ldots,|P|-1}$ is replaced to obtain the leftmost priority path $Q$ (in green, red and orange). Between columns $e_\sigma+1.5$ and $\Bone$, we look for two glues of the same type, $\glueQ{j}{j+1}$ and $\glueQ{k}{k+1}$, which are both visible from the south and pointing west. Applying Lemma \ref{Uturn:main} on $j$ and $k$ leads to a contradiction: the tile $Q_{|Q|-1}-\vect{P_jP_k}$ should be a tile of $\uniterm$ but it is on a column east of $e_\sigma$.}
\label{fig:Proof:CorollaryOne}
\end{figure}

\end{proof}


Lock \ref{lock:boundZero} is solved by the following result where Lemma \ref{Uturn:main} is applied two times.

\begin{lemma}
\label{Uturn:boundWest}
If $P$ is an extremal path, we have $w_P\geq w_{\sigma}-2|T|-1$. 
\end{lemma}

\begin{proof}
For the sake of contradiction, suppose that there exists an extremal path $P$ such that $w_P< w_{\sigma}-2|T|-1$, see Figure \ref{fig:Proof:CorollaryTwo}. Then consider $j=\min\{0\leq i \leq |P|-1 : x_{P_i}=w_P\}$ and let $Q=P_{0,\ldots,j}$ be the shortest prefix of $P$ with a tile on column $w_P$. Now consider $0 \leq s \leq |Q|-1$ and $0 \leq n \leq |Q|-1$ such that $\glueQ{s}{s+1}$ and $\glueQ{n}{n+1}$ are visible from the south and from the north respectively in $Q$ on glue column $w_\sigma-1.5-|T|$. With the same reasoning as the one of proof of Lemma \ref{Uturn:glueWest}, both of these glue points east, see Figure \ref{fig:Proof:CorollaryTwo}(a). Now consider $k=\min\{j<i\leq |P|-1 : e_P=\max\{e_Q+1,e_\sigma+1\}\}$ and let $R=P_{0,\ldots,k}$ be the shortest prefix of $P$ with $Q$ as a prefix and whose last tile is the easternmost one. In this case, if $y_{R_{|Q|}}=y_{R_{|Q|-1}}+1$ (resp. $y_{R_{|Q|}}=y_{R_{|Q|-1}}-1$) then $\glueR{s}{s+1}$ (resp. $\glueR{n}{n+1}$) is visible from the south (resp. north) in $R$ and is still pointing east. Again, we obtain a contradiction with the same reasoning, see Figure \ref{fig:Proof:CorollaryTwo}(b).

\begin{figure}
\center
\begin{minipage}{0.47\linewidth}
\begin{tikzpicture}[x=0.22cm,y=0.22cm]


\draw[very thick] (15.5,10.5) -| (20.5,14.5) -| (13.5,9.5) -| (8.5,11.5) -| (11.5,13.5) -| (3.5,13.5);

\tiles{15}{10}{85}

\tileg{16}{10}{85}
\tileg{17}{10}{85}
\tileg{18}{10}{85}
\tileg{19}{10}{85}
\tileg{20}{10}{85}
\tileg{20}{11}{85}
\tileg{20}{12}{85}
\tileg{20}{13}{85}
\tileg{20}{14}{85}
\tileg{19}{14}{85}
\tileg{18}{14}{85}
\tileg{17}{14}{85}
\tileg{16}{14}{85}
\tileg{15}{14}{85}
\tileg{14}{14}{85}
\tileg{13}{14}{85}
\tileg{13}{13}{85}
\tileg{13}{12}{85}
\tileg{13}{11}{85}
\tileg{13}{10}{85}
\tileg{13}{9}{85}
\tileg{12}{9}{85}
\tileg{11}{9}{85}
\tileg{10}{9}{85}
\tileg{9}{9}{85}
\tileg{8}{9}{85}
\tileg{8}{10}{85}
\tileg{8}{11}{85}
\tileg{9}{11}{85}
\tileg{10}{11}{85}
\tileg{11}{11}{85}
\tileg{11}{12}{85}
\tileg{11}{13}{85}

\tileg{10}{13}{85}
\tileg{9}{13}{85}
\tileg{8}{13}{85}
\tileg{7}{13}{85}
\tileg{6}{13}{85}
\tileg{5}{13}{85}
\tileg{4}{13}{85}
\tileg{3}{13}{85}



\path [dotted, draw, thin] (0,0) grid[step=0.22cm] (26,21);

\fill (3.5,13.5) circle (0.16);
\node (D) at (3.5,14.7) {$j$};

\draw[dashed] (10,0) -| (10,9.5);
\fill (10.5,9.5) circle (0.16);
\node (D) at (10.5,8.2) {$s$};
\node (D) at (10.8,1.5) {$l^s$};

\draw[dashed] (10,13.5) -| (10,21);
\fill (10.5,13.5) circle (0.16);
\node (D) at (10.7,14.7) {$n$};
\node (D) at (11,19.5) {$l^n$};

\draw[dashed] (5.5,0) -| (5.5,21);
\node (D) at (5.5,-1) {$\Bzero$};

\draw[dashed] (14,0) -| (14,21);
\node (D) at (15,-1) {$w_\sigma-1.5$};

\draw[<->, thick] (10,4.5) -- (14,4.5);
\node (D) at (12,5) {\tiny $|T|+1$};

\end{tikzpicture}

\center (a) 
\end{minipage}
\begin{minipage}{0.47\linewidth}
\begin{tikzpicture}[x=0.22cm,y=0.22cm]


\draw[very thick] (15.5,10.5) -| (20.5,14.5) -| (13.5,9.5) -| (8.5,11.5) -| (11.5,13.5) -| (3.5,6.5) -| (21.5,6.5);

\tiles{15}{10}{85}

\tileg{16}{10}{85}
\tileg{17}{10}{85}
\tileg{18}{10}{85}
\tileg{19}{10}{85}
\tileg{20}{10}{85}
\tileg{20}{11}{85}
\tileg{20}{12}{85}
\tileg{20}{13}{85}
\tileg{20}{14}{85}
\tileg{19}{14}{85}
\tileg{18}{14}{85}
\tileg{17}{14}{85}
\tileg{16}{14}{85}
\tileg{15}{14}{85}
\tileg{14}{14}{85}
\tileg{13}{14}{85}
\tileg{13}{13}{85}
\tileg{13}{12}{85}
\tileg{13}{11}{85}
\tileg{13}{10}{85}
\tileg{13}{9}{85}
\tileg{12}{9}{85}
\tileg{11}{9}{85}
\tileg{10}{9}{85}
\tileg{9}{9}{85}
\tileg{8}{9}{85}
\tileg{8}{10}{85}
\tileg{8}{11}{85}
\tileg{9}{11}{85}
\tileg{10}{11}{85}
\tileg{11}{11}{85}
\tileg{11}{12}{85}
\tileg{11}{13}{85}

\tileg{10}{13}{85}
\tileg{9}{13}{85}
\tileg{8}{13}{85}
\tileg{7}{13}{85}
\tileg{6}{13}{85}
\tileg{5}{13}{85}
\tileg{4}{13}{85}
\tileg{3}{13}{85}

\tiley{3}{12}{85}
\tiley{3}{11}{85}
\tiley{3}{10}{85}
\tiley{3}{9}{85}
\tiley{3}{8}{85}
\tiley{3}{7}{85}
\tiley{3}{6}{85}
\tiley{4}{6}{85}
\tiley{5}{6}{85}
\tiley{6}{6}{85}
\tiley{7}{6}{85}
\tiley{8}{6}{85}
\tiley{9}{6}{85}
\tiley{10}{6}{85}
\tiley{11}{6}{85}
\tiley{12}{6}{85}
\tiley{13}{6}{85}
\tiley{14}{6}{85}
\tiley{15}{6}{85}
\tiley{16}{6}{85}
\tiley{17}{6}{85}
\tiley{18}{6}{85}
\tiley{19}{6}{85}
\tiley{20}{6}{85}
\tiley{21}{6}{85}


\path [dotted, draw, thin] (0,0) grid[step=0.22cm] (26,21);

\fill (3.5,13.5) circle (0.16);
\node (D) at (3.5,14.7) {$j$};

\draw[dashed] (21.5,0) -| (21.5,6.5);
\fill (21.5,6.5) circle (0.16);
\node (D) at (22.5,6.5) {$k$};

\fill (10.5,9.5) circle (0.16);
\node (D) at (10.5,8.2) {$s$};

\draw[dashed] (10,13.5) -| (10,21);
\fill (10.5,13.5) circle (0.16);
\node (D) at (10.7,14.7) {$n$};
\node (D) at (11,19.5) {$l^n$};

\draw[dashed] (5.5,0) -| (5.5,21);
\node (D) at (5.5,-1) {$\Bzero$};

\draw[<->, thick] (6,16.5) -- (10,16.5);
\node (D) at (8,17) {\tiny $|T|+1$};

\end{tikzpicture}

\center (b) 
\end{minipage}

\caption{Proof of Lemma \ref{Uturn:boundWest} (a) The beginning $Q=P_{0,\ldots,j}$ of an extremal path $P$ in green. The tile $P_j$ is the first westernmost tile of $P$ and is on column $c < \Bzero$. The $\glueP{s}{s+1}$ and $\glueP{n}{n+1}$ are visible from the south and north respectively on glue column $w_\sigma-1,5-|T|$. Both of these glues must point west by the same reasoning as in the proof of Lemma \ref{Uturn:glueWest}, see Figure \ref{fig:Proof:CorollaryOne}. (b) The tile $P_k$ is the first tile of $P$ on column $c'=\max\{e_\sigma,e_{P_{0,\ldots,j}}\}+1$. We consider the case where $P_{j+1}$ is below $P_i$. Then, $P_{j+1,\ldots,k}$ crosses $l^s$ to reach column $c'$ but $\glueP{n}{n+1}$ is still visible from the north in $P_{0,\ldots,k}$. By the same reasoning we obtain a contradiction.}
\label{fig:Proof:CorollaryTwo}
\end{figure}

\end{proof}

\subsection{Linear bound for U-turn and values of bounds $\Bthree$ and $\Bfour$}
\label{sec:U-turn:linear}


In order to explain the properties of a shield, we give the definition of the bounds $\Bthree$ and $\Bfour$. Since these bounds will be used to analyze an extremal path which is canonical for a glue column $\Bone \leq c \leq \Btwo$, these two bounds depends of the value of $c$. Then, we explain how to improve the quadratic bound of \cite{STOC2020} for U-turn into a linear one.

\begin{definition}
\label{def:bound:threeandfour}
Consider a glue column $\Bone \leq c \leq \Btwo$, we define the functions: $$\Bthree(c)=2c-w_{\sigma}+3|T|+2 $$ and $$\Bfour(c)=c+3|\sigma|+24|T|+14$$ 
\end{definition}

\begin{fact}
\label{fact:true:Bthree}
Note that $\Bfinal$ is defined such that the maximum of $\Bfour$ is reached for $$\Bfour(\Btwo)=\Bfinal-0.5$$  
Also, for any $\Bone\leq c \leq c' \leq \Btwo$, we have $\Bfour(c)=\Bfour(c')+(c'-c)$.
\end{fact}

In the following result, we consider an area delimited by a visible glue and a visible tile. We supposed that the path does a U-turn between this glue and this tile: it grows far to the east before coming back to the west. Then, Lemma \ref{Uturn:main} is applied $|T|$ times to bring this U-turn close to the beginning the area.  

\begin{lemma}
\label{Uturn:LinearBound}
Consider a producible path $P$, a glue column $c_i$ and a column $c$. Suppose that there exists a $\glueP{i}{i+1}$ which is visible from the south (resp. north) and pointing east on glue column $c_i$ and that the last tile of $P$ is visible from the north (resp. south) on column $c$. Moreover, $P$ is the rightmost (resp. leftmost) priority path with $P_{0,\ldots,i+1}$ as a prefix, such that $\glueP{i}{i+1}$ is visible from the south (resp. north) and with its last tile visible from the north (resp. south) on column~$c$. Then, we have  $$e_P-c \leq  (c_i-w_{P_{i+1,\ldots,|P|-1}})+|T|+0.5.$$

\end{lemma}

\begin{proof}

Consider $P$, $0\leq i \leq |P|-1$, $c_i$ and $c$ which satisfies the hypothesis of the lemma, see Figure \ref{fig:Uturn:Linear}(a). 
Consider $0\leq k \leq |P|-1$ such that $\glueP{k}{k+1}$ is visible from the south on glue column $e_P-0.5$. Let $0 \leq j \leq |P|-1$ such that $\glueP{j}{j+1}$ has the same type as $\glueP{k}{k+1}$ and is the westernmost glue visible from the south in $P$ on glue column $c_1$ with $c \leq c_1 \leq e_P-0.5$, see Figure \ref{fig:Uturn:Linear}(b). By Lemmas \ref{lem:glue:prop2} and \ref{lem:glue:prop4}, we have $i \leq j \leq k $ and $\glueP{j}{j+1}$ and $\glueP{k}{k+1}$ are both pointing east. Then Lemma \ref{Uturn:main} can be applied on indices $j$ and $k$ (see Figure \ref{fig:Uturn:Linear2}) to obtain 
\begin{align}
w_{P_{j+1,\ldots,|P|-1}} & \leq w_{P_{k+1,\ldots,|P|-1}}+(c_1-(e_P-0.5))\\
& \leq c+(c_1-(e_P-0.5))
\end{align}

\begin{figure}
\center
\begin{minipage}{0.47\linewidth}
\begin{tikzpicture}[x=0.22cm,y=0.22cm]

\fill[fill=blue!30!white, draw opacity=0.8] (6,5.5) -| (11.5,7.5) -| (2.5,15.5) -| (17.5,17.5) -| (8.5,21.5) -| (19.5,19.5) -| (21.5,23.5) -| (14.5,26) -| (26,0) -|  (6,5.5);

\draw[very thick] (5.5,5.5) -| (11.5,7.5) -| (2.5,15.5) -| (17.5,17.5) -| (8.5,21.5) -| (19.5,19.5) -| (21.5,23.5) -| (14.5,23.5);



\tilef{5}{5}{85}

\tileo{14}{23}{48}
\tileo{15}{23}{48}
\tileo{16}{23}{48}
\tileo{17}{23}{48}
\tileo{18}{23}{48}
\tileo{19}{23}{48}
\tileo{20}{23}{48}
\tileo{21}{23}{48}
\tileo{21}{22}{48}
\tileo{21}{21}{48}
\tileo{21}{20}{48}
\tileo{21}{19}{48}
\tileg{20}{19}{70}
\tileg{19}{19}{70}
\tileg{19}{20}{70}
\tileg{19}{21}{70}
\tileg{18}{21}{70}
\tileg{17}{21}{70}
\tileg{16}{21}{70}
\tileg{15}{21}{70}
\tileg{14}{21}{70}
\tileg{13}{21}{70}
\tileg{12}{21}{70}
\tileg{11}{21}{70}
\tileg{10}{21}{70}
\tileg{9}{21}{70}
\tileg{8}{21}{70}
\tileg{8}{20}{70}
\tileg{8}{19}{70}
\tileg{8}{18}{70}
\tileg{8}{17}{70}
\tileg{9}{17}{70}
\tileg{10}{17}{70}
\tileg{11}{17}{70}
\tileg{12}{17}{70}
\tileg{13}{17}{70}
\tileg{14}{17}{70}
\tileg{15}{17}{70}
\tileg{16}{17}{70}
\tileg{17}{17}{70}
\tileg{17}{16}{70}
\tileg{17}{15}{70}
\tileg{16}{15}{70}
\tileg{15}{15}{70}
\tileg{14}{15}{70}
\tileg{13}{15}{70}
\tileg{12}{15}{70}
\tileg{11}{15}{70}
\tileg{10}{15}{70}
\tileg{9}{15}{70}
\tileg{8}{15}{70}
\tileg{7}{15}{70}
\tileg{6}{15}{70}
\tileg{5}{15}{70}
\tileg{4}{15}{70}
\tileg{3}{15}{70}
\tileg{2}{15}{70}
\tileg{2}{14}{70}
\tileg{2}{13}{70}
\tileg{2}{12}{70}
\tileg{2}{11}{70}
\tileg{2}{10}{70}
\tileg{2}{9}{70}
\tileg{2}{8}{70}
\tileg{2}{7}{70}
\tileg{3}{7}{70}
\tileg{4}{7}{70}
\tileg{5}{7}{70}
\tileg{6}{7}{70}
\tileg{7}{7}{70}
\tileg{8}{7}{70}
\tileg{9}{7}{70}
\tileg{10}{7}{70}
\tileg{11}{7}{70}
\tileg{11}{6}{70}
\tileg{11}{5}{70}
\tileg{10}{5}{70}
\tileg{9}{5}{70}
\tileg{8}{5}{70}
\tileg{7}{5}{70}
\tileg{6}{5}{70}



\path [dotted, draw, thin] (0,0) grid[step=0.22cm] (26,26);

\draw[dashed] (6,0) -| (6,5.5);
\fill (5.5,5.5) circle (0.16);
\node (D) at (4,5.5) {$i$};
\node (D) at (4.8,1.5) {$l^i$};

\draw[dashed] (14.5,23.5) -| (14.5,26);
\fill (14.5,23.5) circle (0.16);
\node (D) at (11.5,23.5) {\tiny $|P|-1$};
\node (D) at (16,25) {$l^P$};

\node (D) at (14.5,27) {$c$};

\draw[dashed] (21,0) -| (21,19.5);
\fill (20.5,19.5) circle (0.16);
\fill[color=Olive] (21,19.5) circle (0.16);
\node (D) at (20.5,18.2) {$k$};
\node (D) at (20,4.5) {$l^k$};

\node (D) at (21,-1.5) {\tiny $e^P-0.5$};
\node (D) at (5.5,-1.5) {$c_i$};

\node (D) at (24,3.7) {$\mathcal{C}$};
\end{tikzpicture}

\center (a) 
\end{minipage}
\begin{minipage}{0.47\linewidth}
\begin{tikzpicture}[x=0.22cm,y=0.22cm]

\fill[fill=blue!30!white, draw opacity=0.8] (6,5.5) -| (11.5,7.5) -| (2.5,15.5) -| (17.5,17.5) -| (8.5,21.5) -| (19.5,19.5) -| (21.5,23.5) -| (14.5,26) -| (26,0) -|  (6,5.5);

\draw[very thick] (5.5,5.5) -| (11.5,7.5) -| (2.5,15.5) -| (17.5,17.5) -| (8.5,21.5) -| (19.5,19.5) -| (21.5,23.5) -| (14.5,23.5);



\tilef{5}{5}{85}

\tileb{14}{23}{48}
\tileb{15}{23}{48}
\tileb{16}{23}{48}
\tileb{17}{23}{48}
\tileb{18}{23}{48}
\tileb{19}{23}{48}
\tileb{20}{23}{48}
\tileb{21}{23}{48}
\tileb{21}{22}{48}
\tileb{21}{21}{48}
\tileb{21}{20}{48}
\tileb{21}{19}{48}
\tiley{20}{19}{70}
\tiley{19}{19}{70}
\tiley{19}{20}{70}
\tiley{19}{21}{70}
\tiley{18}{21}{70}
\tiley{17}{21}{70}
\tiley{16}{21}{70}
\tiley{15}{21}{70}
\tiley{14}{21}{70}
\tiley{13}{21}{70}
\tiley{12}{21}{70}
\tiley{11}{21}{70}
\tiley{10}{21}{70}
\tiley{9}{21}{70}
\tiley{8}{21}{70}
\tiley{8}{20}{70}
\tiley{8}{19}{70}
\tiley{8}{18}{70}
\tiley{8}{17}{70}
\tiley{9}{17}{70}
\tiley{10}{17}{70}
\tiley{11}{17}{70}
\tiley{12}{17}{70}
\tiley{13}{17}{70}
\tiley{14}{17}{70}
\tiley{15}{17}{70}
\tiley{16}{17}{70}
\tiley{17}{17}{70}
\tiley{17}{16}{70}
\tiley{17}{15}{70}
\tileg{16}{15}{70}
\tileg{15}{15}{70}
\tileg{14}{15}{70}
\tileg{13}{15}{70}
\tileg{12}{15}{70}
\tileg{11}{15}{70}
\tileg{10}{15}{70}
\tileg{9}{15}{70}
\tileg{8}{15}{70}
\tileg{7}{15}{70}
\tileg{6}{15}{70}
\tileg{5}{15}{70}
\tileg{4}{15}{70}
\tileg{3}{15}{70}
\tileg{2}{15}{70}
\tileg{2}{14}{70}
\tileg{2}{13}{70}
\tileg{2}{12}{70}
\tileg{2}{11}{70}
\tileg{2}{10}{70}
\tileg{2}{9}{70}
\tileg{2}{8}{70}
\tileg{2}{7}{70}
\tileg{3}{7}{70}
\tileg{4}{7}{70}
\tileg{5}{7}{70}
\tileg{6}{7}{70}
\tileg{7}{7}{70}
\tileg{8}{7}{70}
\tileg{9}{7}{70}
\tileg{10}{7}{70}
\tileg{11}{7}{70}
\tileg{11}{6}{70}
\tileg{11}{5}{70}
\tileg{10}{5}{70}
\tileg{9}{5}{70}
\tileg{8}{5}{70}
\tileg{7}{5}{70}
\tileg{6}{5}{70}



\path [dotted, draw, thin] (0,0) grid[step=0.22cm] (26,26);

\draw[dashed] (6,0) -| (6,5.5);
\fill (5.5,5.5) circle (0.16);
\node (D) at (4,5.5) {$i$};
\node (D) at (4.8,1.5) {$l^i$};

\draw[dashed] (14.5,23.5) -| (14.5,26);
\fill (14.5,23.5) circle (0.16);
\node (D) at (11.5,23.5) {\tiny $|P|-1$};
\node (D) at (16,25) {$l^P$};

\node (D) at (14.5,27) {$c$};

\draw[dashed] (21,0) -| (21,19.5);
\fill (20.5,19.5) circle (0.16);
\fill[color=Olive] (21,19.5) circle (0.16);
\node (D) at (20.5,18.2) {$k$};
\node (D) at (20,4.5) {$l^k$};

\draw[dashed] (17,0) -| (17,15.5);
\fill (16.5,15.5) circle (0.16);
\fill[color=Olive] (17,15.5) circle (0.16);
\node (D) at (16.5,14.2) {$j$};
\node (D) at (18,7.5) {$l^j$};

\node (D) at (21,-1.5) {\tiny $e^P-0.5$};
\node (D) at (17,-1.5) {$c_1$};
\node (D) at (5.5,-1.5) {$c_i$};

\node (D) at (24,3.7) {$\mathcal{C}$};

\draw[->,thick] (21.5,19.5) -- (17.5,15.5);

\end{tikzpicture}

\center (b) 
\end{minipage}
\caption{Proof of Lemma \ref{Uturn:LinearBound}, Part(1/3) (a) A producible path $P$ with its last tile visible from the north on column $c$, $\glueP{i}{i+1}$ visible from the south on glue column $c_i$ and $\glueP{k}{k+1}$ visible from the south on glue column $e_P-0.5$. The proof will use the area $\mathcal{C}$ in blue. Only $P_{i,\ldots, |P|-1}$ is drawn, $P_{i,\ldots,k}$ is in green while $P_{k+1,\ldots,|P|-1}$ is in light green. Note that after $\glueP{k}{k+1}$, the path $P$ does a U-turn to go back to column $c$ (the length of this U-turn is 6.5 tiles in this example). (b) The $\glueP{j}{j+1}$ is the westernmost glue visible from the south in the blue area with the same type as $\glueP{k}{k+1}$. This is the setting of Lemma \ref{Uturn:main} for $\glueP{j}{j+1}$ and $\glueP{k}{k+1}$: $P_{i+1,\ldots,j}$ is green, $P_{j+1,\ldots,k}$ is yellow and $P_{k+1,\ldots,|P|-1}$ is in blue.}
\label{fig:Uturn:Linear}
\end{figure}
\begin{figure}
\center

\begin{minipage}{0.47\linewidth}
\begin{tikzpicture}[x=0.22cm,y=0.22cm]

\fill[fill=yellow!90!black, draw opacity=0.8] (17.5,15.5) -| (17.5,17.5) -| (8.5,21.5) -| (19.5,19.5) -| (21.5,23.5) -| (14.5,26) -| (26,0) -|  (17,15.5);

\draw[very thick] (5.5,5.5) -| (11.5,7.5) -| (2.5,15.5) -| (17.5,17.5) -| (8.5,21.5) -| (19.5,19.5) -| (21.5,23.5) -| (14.5,23.5);
\draw[very thick] (10.5,19.5) -| (17.5,15.5);



\tilef{5}{5}{85}

\tileb{14}{23}{48}
\tileb{15}{23}{48}
\tileb{16}{23}{48}
\tileb{17}{23}{48}
\tileb{18}{23}{48}
\tileb{19}{23}{48}
\tileb{20}{23}{48}
\tileb{21}{23}{48}
\tileb{21}{22}{48}
\tileb{21}{21}{48}
\tileb{21}{20}{48}
\tileb{21}{19}{48}
\tiley{20}{19}{70}
\tiley{19}{19}{70}
\tiley{19}{20}{70}
\tiley{19}{21}{70}
\tiley{18}{21}{70}
\tiley{17}{21}{70}
\tiley{16}{21}{70}
\tiley{15}{21}{70}
\tiley{14}{21}{70}
\tiley{13}{21}{70}
\tiley{12}{21}{70}
\tiley{11}{21}{70}
\tiley{10}{21}{70}
\tiley{9}{21}{70}
\tiley{8}{21}{70}
\tiley{8}{20}{70}
\tiley{8}{19}{70}
\tiley{8}{18}{70}
\tiley{8}{17}{70}
\tiley{9}{17}{70}
\tiley{10}{17}{70}
\tiley{11}{17}{70}
\tiley{12}{17}{70}
\tiley{13}{17}{70}
\tiley{14}{17}{70}
\tiley{15}{17}{70}
\tiley{16}{17}{70}
\tiley{17}{17}{70}
\tiley{17}{16}{70}
\tiley{17}{15}{70}
\tileg{16}{15}{70}
\tileg{15}{15}{70}
\tileg{14}{15}{70}
\tileg{13}{15}{70}
\tileg{12}{15}{70}
\tileg{11}{15}{70}
\tileg{10}{15}{70}
\tileg{9}{15}{70}
\tileg{8}{15}{70}
\tileg{7}{15}{70}
\tileg{6}{15}{70}
\tileg{5}{15}{70}
\tileg{4}{15}{70}
\tileg{3}{15}{70}
\tileg{2}{15}{70}
\tileg{2}{14}{70}
\tileg{2}{13}{70}
\tileg{2}{12}{70}
\tileg{2}{11}{70}
\tileg{2}{10}{70}
\tileg{2}{9}{70}
\tileg{2}{8}{70}
\tileg{2}{7}{70}
\tileg{3}{7}{70}
\tileg{4}{7}{70}
\tileg{5}{7}{70}
\tileg{6}{7}{70}
\tileg{7}{7}{70}
\tileg{8}{7}{70}
\tileg{9}{7}{70}
\tileg{10}{7}{70}
\tileg{11}{7}{70}
\tileg{11}{6}{70}
\tileg{11}{5}{70}
\tileg{10}{5}{70}
\tileg{9}{5}{70}
\tileg{8}{5}{70}
\tileg{7}{5}{70}
\tileg{6}{5}{70}

\tiler{10}{19}{48}
\tiler{11}{19}{48}
\tiler{12}{19}{48}
\tiler{13}{19}{48}
\tiler{14}{19}{48}
\tiler{15}{19}{48}
\tiler{16}{19}{48}
\tiler{17}{19}{48}
\tiler{17}{18}{48}
\tileor{17}{17}{48}
\tileor{17}{16}{48}
\tileor{17}{15}{48}


\path [dotted, draw, thin] (0,0) grid[step=0.22cm] (26,26);

\draw[dashed] (6,0) -| (6,5.5);
\fill (5.5,5.5) circle (0.16);
\node (D) at (4,5.5) {$i$};
\node (D) at (4.8,1.5) {$l^i$};

\draw[dashed] (14.5,23.5) -| (14.5,26);
\fill (14.5,23.5) circle (0.16);
\node (D) at (11.5,23.5) {\tiny $|P|-1$};
\node (D) at (16,25) {$l^P$};

\node (D) at (14.5,27) {$c$};

\draw[dashed] (17,0) -| (17,15.5);
\fill (16.5,15.5) circle (0.16);
\fill[color=Olive] (17,15.5) circle (0.16);
\node (D) at (16.5,14.2) {$j$};
\node (D) at (18,7.5) {$l^j$};
\node (D) at (17,-1.5) {$c_1$};

\draw[dashed] (21,0) -| (21,19.5);
\fill (20.5,19.5) circle (0.16);
\fill[color=Olive] (21,19.5) circle (0.16);
\node (D) at (20.5,18.2) {$k$};
\node (D) at (20,4.5) {$l^k$};

\node (D) at (21,-1.5) {\tiny $e^P-0.5$};
\node (D) at (5.5,-1.5) {$c_i$};

\draw[->,thick] (14.5,23.5) -- (10.5,19.5);

\node (D) at (24,3.7) {$\mathcal{C}_1$};
\end{tikzpicture}

\center (a) 
\end{minipage}
\begin{minipage}{0.47\linewidth}
\begin{tikzpicture}[x=0.22cm,y=0.22cm]

\fill[fill=blue!30!white, draw opacity=0.8] (6,5.5) -| (11.5,7.5) -| (2.5,15.5) -| (17.5,17.5) -| (8.5,21.5) -| (19.5,19.5) -| (21.5,23.5) -| (14.5,26) -| (26,0) -|  (6,5.5);

\draw[very thick] (5.5,5.5) -| (11.5,7.5) -| (2.5,15.5) -| (17.5,17.5) -| (8.5,21.5) -| (19.5,19.5) -| (21.5,23.5) -| (14.5,23.5);



\tilef{5}{5}{85}

\tileo{14}{23}{48}
\tileo{15}{23}{48}
\tileo{16}{23}{48}
\tileo{17}{23}{48}
\tileo{18}{23}{48}
\tileo{19}{23}{48}
\tileo{20}{23}{48}
\tileo{21}{23}{48}
\tileo{21}{22}{48}
\tileo{21}{21}{48}
\tileo{21}{20}{48}
\tileo{21}{19}{48}
\tileo{20}{19}{70}
\tileo{19}{19}{70}
\tileo{19}{20}{70}
\tileo{19}{21}{70}
\tileo{18}{21}{70}
\tileo{17}{21}{70}
\tileo{16}{21}{70}
\tileo{15}{21}{70}
\tileo{14}{21}{70}
\tileo{13}{21}{70}
\tileo{12}{21}{70}
\tileo{11}{21}{70}
\tileo{10}{21}{70}
\tileo{9}{21}{70}
\tileo{8}{21}{70}
\tileo{8}{20}{70}
\tileo{8}{19}{70}
\tileo{8}{18}{70}
\tileo{8}{17}{70}
\tileo{9}{17}{70}
\tileo{10}{17}{70}
\tileo{11}{17}{70}
\tileo{12}{17}{70}
\tileo{13}{17}{70}
\tileo{14}{17}{70}
\tileo{15}{17}{70}
\tileo{16}{17}{70}
\tileo{17}{17}{70}
\tileo{17}{16}{70}
\tileo{17}{15}{70}
\tileg{16}{15}{70}
\tileg{15}{15}{70}
\tileg{14}{15}{70}
\tileg{13}{15}{70}
\tileg{12}{15}{70}
\tileg{11}{15}{70}
\tileg{10}{15}{70}
\tileg{9}{15}{70}
\tileg{8}{15}{70}
\tileg{7}{15}{70}
\tileg{6}{15}{70}
\tileg{5}{15}{70}
\tileg{4}{15}{70}
\tileg{3}{15}{70}
\tileg{2}{15}{70}
\tileg{2}{14}{70}
\tileg{2}{13}{70}
\tileg{2}{12}{70}
\tileg{2}{11}{70}
\tileg{2}{10}{70}
\tileg{2}{9}{70}
\tileg{2}{8}{70}
\tileg{2}{7}{70}
\tileg{3}{7}{70}
\tileg{4}{7}{70}
\tileg{5}{7}{70}
\tileg{6}{7}{70}
\tileg{7}{7}{70}
\tileg{8}{7}{70}
\tileg{9}{7}{70}
\tileg{10}{7}{70}
\tileg{11}{7}{70}
\tileg{11}{6}{70}
\tileg{11}{5}{70}
\tileg{10}{5}{70}
\tileg{9}{5}{70}
\tileg{8}{5}{70}
\tileg{7}{5}{70}
\tileg{6}{5}{70}



\path [dotted, draw, thin] (0,0) grid[step=0.22cm] (26,26);

\draw[dashed] (6,0) -| (6,5.5);
\fill (5.5,5.5) circle (0.16);
\node (D) at (4,5.5) {$i$};
\node (D) at (4.8,1.5) {$l^i$};

\draw[dashed] (10,0) -| (10,5.5);
\fill (9.5,5.5) circle (0.16);
\node (D) at (9,3.8) {$j'$};
\fill[color=red] (10,5.5) circle (0.16);
\node (D) at (11.3,1.5) {$l^{j'}$};

\draw[dashed] (14.5,23.5) -| (14.5,26);
\fill (14.5,23.5) circle (0.16);
\node (D) at (11.5,23.5) {\tiny $|P|-1$};
\node (D) at (16,25) {$l^P$};

\node (D) at (14.5,27) {$c$};

\fill[color=Olive] (21,19.5) circle (0.16);

\fill (16.5,15.5) circle (0.16);
\fill[color=Olive] (17,15.5) circle (0.16);
\node (D) at (16.5,16.5) {$j$};

\draw[dashed] (16,0) -| (16,15.5);
\fill (15.5,15.5) circle (0.16);
\fill[color=red] (16,15.5) circle (0.16);
\node (D) at (15,14.2) {$k'$};
\node (D) at (15,7.5) {$l^{k'}$};

\node (D) at (16,-1.5) {$c_1-1$};

\node (D) at (10,-1.5) {$c_2$};
\node (D) at (5.5,-1.5) {$c_i$};

\node (D) at (24,3.7) {$\mathcal{C}$};

\draw[->,thick] (16.5,15.5) -- (10.5,5.5);

\end{tikzpicture}

\center (b) 
\end{minipage}

\caption{Proof of Lemma \ref{Uturn:LinearBound}, Part(2/3) (a) By Lemma \ref{Uturn:main}, $P_{k+1,\ldots,|P|-1}-\vect{P_jP_k}$ is a producible path and is inside the area $\mathcal{C}_1$ in yellow. (b) This result implies that, after $\glueP{j}{j+1}$ the path $P$ (this suffix of $P$ is in light green) must do a U-turn of at least the same size than the one of after $\glueP{k}{k+1}$ (6.5 tiles in this example) to allow $P_{k+1,\ldots,|P|-1}-\vect{P_jP_k}$ to assemble. Then, we keep going on with the same reasoning, we consider $\glueP{k'}{k'+1}$ which is visible from the south on glue column $c_1-1$ (and the length of the U-turn after $\glueP{k'}{k'+1}$ is at least the length of the U-turn after $\glueP{j}{j+1}$ minus~$1$). Let $\glueP{j'}{j'+1}$ the westernmost glue visible from the south on glue column~$c_2$, inside the blue area and with the same type as  $\glueP{k'}{k'+1}$. }
\label{fig:Uturn:Linear2}
\end{figure}

Now, either $k=i$ and $c_1=c_i$ or we can keep going on with the same reasoning. Indeed if $c_1>c$, consider $0\leq k' \leq |P|-1$ such that $\glueP{k'}{k'+1}$ is visible from the south in $P$ on glue column $c_1-1$. Note that the type of this glue is different from the type of $\glueP{j}{j+1}$. Indeed, by definition of $j$, if a glue is visible from the south on a glue column $c_i \leq c'\leq c_1-1$ then its type cannot be the type of $\glueP{j}{j+1}$. Let $0 \leq j' \leq |P|-1$ such that $\glueP{j'}{j'+1}$ has the same type as $\glueP{k'}{k'+1}$ and is the westernmost glue visible from the south in $P$ on glue column $c_2$ with $c_i \leq c_2 \leq c_1-1$, see Figures \ref{fig:Uturn:Linear2}(b) and \ref{fig:Uturn:Linear3}. By Lemmas \ref{lem:glue:prop2} and \ref{lem:glue:prop4}, we have $i \leq j' \leq k' < j$ and $\glueP{j'}{j'+1}$ and $\glueP{k'}{k'+1}$ are both pointing east. Then Lemma \ref{Uturn:main} can be applied again on indices $j'$ and $k'$  to obtain:

\begin{align*}
 w_{P_{j'+1,\ldots,|P|-1}} & \leq w_{P_{k'+1,\ldots,|P|-1}}+(c_2-(c_1-1)) \\
& \leq w_{P_{j+1,\ldots,|P|-1}}+(c_2-(c_1-1)) \text{ (since $k'<j'$)} \\
& \leq c+(c_1-e_P+0.5)+(c_2-(c_1-1)) \text{ (by equation (2))}\\
& \leq c+(c_2-e_P+1.5). 
\end{align*}

\input{./tikz/Uturn/ProofLinearPart3}

This reasoning can be done at most $|T|$ time (since each tile type is considered at most one time) until it is applied on the type of the $\glueP{i}{i+1}$, then we obtain: $$w_{P_{i+1,\ldots,|P|-1}}\leq  c+(c_i-e_P+0.5+|T|).$$
\end{proof}

Finally, solving Lock \ref{lock:uturnTwo} requires only to add that $w_P$ is bounded by $w_{\sigma}-2|T|-1$ by Lemma \ref{Uturn:boundWest}.

\begin{corollary}
\label{fact:Bthree}
Consider a producible path $P$, a glue column $\Bone \leq c \leq \Btwo$. Suppose that there exists a $\glueP{i}{i+1}$ which is visible from the south (resp. north) in $P$ and pointing east on glue column $\Bthree(c)+1.5$. If the last tile of $P$ is visible from the north (resp. south) on column $\Bthree(c)+1$ then $e_P< \Bfour(c)$.
\end{corollary}

\begin{proof}
Consider a producible path $P$, a glue column $\Bone \leq c \leq \Btwo$. Suppose that there exists a $\glueP{i}{i+1}$ which is visible from the south in $P$ (the other case is symmetric) and pointing east on glue column $\Bthree(c)+1.5$. If the last tile of $P$ is visible from the north on column $\Bthree(c)+1$ then, by Lemma \ref{Uturn:LinearBound}, we have $$e_P - (\Bthree(c)+1) \leq  (\Bthree(c)+1.5- w_{P_{i+1,\ldots,|P|-1}})+|T|+0.5$$ Thus,
$$w_{P_{i+1,\ldots,|P|-1}} \leq 2\Bthree(c)-e_P +|T|+3$$
 By Lemma \ref{Uturn:boundWest}, we have  $w_P\geq w_{\sigma}-2|T|-1$ which leads to: $$w_{\sigma}-2|T|-1 \leq 2\Bthree(c)-e_P+|T|+3$$ And then, 
 \begin{align*}
 e_P & \leq  2\Bthree(c)-w_{\sigma}+3|T|+4\\
 & \leq 4c-3w_{\sigma}+9|T|+8\\
 & \leq c+3(c-w_{\sigma})+9|T|+8\\
 & \leq c+3(|\sigma|+5|T|+1.5)+9|T|+8\\
 & \leq c+3|\sigma|+24|T|+12.5\\
 & <\Bfour(c)\\
 \end{align*}
 Note that the inequality $c-w_{\sigma}\leq 5|T|+1.5+|\sigma|$ is due the fact that $c\leq \Btwo$, $\Btwo=e_\sigma+5|T|+1.5$ and $e_\sigma-w_\sigma \leq |\sigma|$
\end{proof}

\subsection{U-turn for a canonical path}
\label{sec:U-turn:linearCano}

Consider a glue column $\Bone \leq c \leq \Btwo$ and a path $P$ which is canonical for glue column $c$, the span $(s,n)$ of $P$ on $c$ and the index of the last glue $\lastc$ of $P$ on $c$. The objective of this subsection is to bound $e_{P_{0,\ldots,\lastc}}$ by $\Bthree(c)$ (see Lemma \ref{Uturn:LinearCano}) to solve Lock \ref{lock:uturnOne}. Using Lemma \ref{Uturn:LinearBound} directly and doing some calculus may allow to bound only  $e_{P_{0,\ldots,n+1}}$ which is not sufficient for our objective. To prove this result, we will follow the same steps as in Subsection \ref{sec:U-turn:linear} but we will adapt the proofs to canonical paths. 

The first problem is that the last tile of $P_{0,\ldots,\lastc}$ may not be the easternmost one of $P_{0,\ldots,\lastc}$. The second one is that a glue which is visible in $P_{0,\ldots,\lastc}$ may not be visible in $P$, see Figure \ref{fig:Uturn:useful}. To solve this problem, we will mainly work on the \emph{useful prefix} $Q$ of the canonical path $P$ which is is defined as the shortest prefix of $P$ to reach column $e_{P_{0,\ldots,\lastc}}+1$, see Figure \ref{fig:Uturn:useful}(b). Note that this notion is used only in this Subsection. Now, we prove some useful results on the visible glues of $Q$ to obtain equivalences of Lemmas \ref{lem:glue:prop2} and \ref{lem:glue:prop4}.

\input{./tikz/Uturn/Useful}

\begin{lemma}
\label{lem:glue:prop:cano}
Consider a glue column $\Bone \leq c \leq \Btwo$ and a path $P$ which is canonical for $c$. Let $Q$ be the useful prefix of $P$ and let $(u_i)_{0\leq i \leq t}$ be the decomposition of $P$ into pseudo-visible glue on glue column $c$. Then, there exists $u_0\leq k \leq u_t$ such that $\glueQ{k}{k+1}$ is visible from either the south or north on glue column $e_Q-1.5$. If $\glueQ{k}{k+1}$ is visible from the south (resp. north), then consider $c\leq c_i \leq e_{Q}-1.5$ and $0\leq i \leq |Q|-1$ such that $\glueQ{i}{i+1}$ is visible from the south (resp. north) in $Q$ on glue column $c_i$, then $\glueQ{i}{i+1}$ points east and there exists $0\leq i' < t$ such that $u_{i'}\leq i <u_{i'+1}$ and $(u_{i'},u_{i'+1})$ is an upward (resp. downward) pseudo-span of $Q$ on glue column $c$. Moreover, if $\glueQ{j}{j+1}$ is visible from the south (resp. north) in $Q$ on a glue column $c_i\leq c_j \leq e_{Q}-1.5$, then we have $i\leq j$.
\end{lemma}

\begin{proof}
Let $R$ be the shortest prefix of $Q$ to reach column $e_{P_{0,\ldots,u_t+1}}$ (in orange in Figure \ref{fig:Uturn:useful}). Since $(u_0,u_1)$ is the span of $Q$ on column $c$ we have $e_{P_{0,\ldots,u_0-1}}<e_{P_{u_0,\ldots,u_1}}$. Then, $u_0 \leq |R|-1$ and $\glueR{u_0}{u_0+1}$ is a visible glue of $R$. Without loss of generality, we suppose that $\glueR{u_0}{u_0+1}$ is a visible glue from the north (\emph{i.e.} $(u_0,u_1)$ is a downward span) and $y_{Q_{|R|}}=y_{Q_{|R|-1}}+1$ (the other cases are symmetric) and we call $l^R$ the ray starting at $\pos{R_{|R|-1}}$ and going east. Let $\mathcal{C}$ be the north-east side of the cut delimited by $l^{u_0}$, $R_{u_0,\ldots,|R|-1}$ and $l^R$, see Figure \ref{fig:Uturn:useful}(b).

Remark that the tile $Q_{|R|}$ is inside $\mathcal{C}$ since $y_{Q_{|R|}}=y_{Q_{|R|-1}}+1$. Moreover, $Q_{|R|,\ldots,|Q|-2}$ cannot cross $l^R$ (it cannot reach column $e_R+1$),  $R_{u_0,\ldots,|R|-1}$ (since $P$ is simple) and  $l^{u_0}$ (since  $\glueR{u_0}{u_0+1}$ is visible in $P$). Finally since $\glueQ{|Q|-2}{|Q|-1}$ is on column $e_R+0.5$ then $Q_{|R|,\ldots,|Q|-2}$ is inside $\mathcal{C}$. Thus, consider glue column $c \leq c_i \leq e_Q-1.5$ and $0\leq i \leq |R|-1$ such that $\glueR{i}{i+1}$ is visible from the south in $R$,  since $l^i$ is not inside $\mathcal{C}$ (apart from its starting point) then  $\glueQ{i}{i+1}$ is visible from the south in $Q$. Since the last tile of $Q$ is the easternmost one then Lemmas \ref{lem:glue:prop2} and \ref{lem:glue:prop4} show that $\glueP{i}{i+1}$ point east and that $u_0 \leq i$. Since $i \leq |R|-1 \leq u_t$, there exists $0\leq i' \leq t$ such that $u_{i'} \leq i <u_{i'+1}$. If $(u_{i'},u_{i'+1})$ is an upward pseudo-span then the lemma is true. 
Otherwise, consider the east side of the cut delimited by $l^{u_{i'}}$, $Q_{u_{i'},\ldots,i}$ and $l^i$, note that $l^{u_{i'+1}}$ is not inside this area and since $\glueQ{i}{i+1}$ points east, we obtain a contradiction, see Figure \ref{fig:Uturn:usefulProp}. 
Finally, consider a glue column $c_i \leq c_j \leq e_Q-1.5$ and $0 \leq j \leq |Q|-1$ such that $\glueQ{j}{j+1}$ is visible from the south in $Q$ on glue column $c_j$. Then, $i\leq j$ by Lemma \ref{lem:glue:prop4}. 

\begin{figure}
\center
\begin{tikzpicture}[x=0.2cm,y=0.2cm]

\fill[fill=yellow!30!white, draw opacity=0.8] (16,21.5) -| (17.5,18.5) -| (8.5,9.5) -| (23,9.5) |- (26,0) |- (16,27);

\draw[very thick] (15.5,21.5) -| (17.5,18.5) -| (8.5,9.5) -| (23.5,9.5);

\tileor{15}{21}{48}
\tileor{16}{21}{48}
\tileor{17}{21}{48}
\tileor{17}{20}{48}
\tileor{17}{19}{48}
\tileor{17}{18}{48}
\tileor{16}{18}{48}
\tileor{15}{18}{48}
\tileor{14}{18}{48}
\tileor{13}{18}{48}
\tileor{12}{18}{48}
\tileor{11}{18}{48}
\tileor{10}{18}{48}
\tileor{9}{18}{48}
\tileor{8}{18}{48}
\tileor{8}{17}{48}
\tileor{8}{16}{48}
\tileor{8}{15}{48}
\tileor{8}{14}{48}
\tileor{8}{13}{48}
\tileor{8}{12}{48}
\tileor{8}{11}{48}
\tileor{8}{10}{48}
\tileor{8}{9}{48}
\tileor{9}{9}{48}
\tileor{10}{9}{48}
\tileor{11}{9}{48}
\tileor{12}{9}{48}
\tileor{13}{9}{48}
\tileor{14}{9}{48}
\tileor{15}{9}{48}
\tileor{16}{9}{48}
\tileor{17}{9}{48}
\tileor{18}{9}{48}
\tileor{19}{9}{48}
\tileor{20}{9}{48}
\tileor{21}{9}{48}
\tileor{22}{9}{48}
\tileor{23}{9}{48}

\path [dotted, draw, thin] (0,0) grid[step=0.2cm] (26,27);

\draw[dashed] (23,0) -| (23,9.5);
\fill (22.5,9.5) circle (0.16);
\node (D) at (22.5,11) {$i$};
\node (D) at (24.2,2.5) {$l^{i}$};


\draw[dashed] (16,0) -| (16,8.5);
\fill (15.5,9.5) circle (0.16);
\node (D) at (15.5,10.7) {$s$};
\node (D) at (17.2,2.5) {$l^{s}$};

\fill (15.5,21.5) circle (0.16);
\node (D) at (15.5,20.2) {$u_2$};
\draw[dashed] (16,21.5) -| (16,27);
\node (D) at (17.4,26) {$l^{u_2}$};

\end{tikzpicture} 

\caption{Proof Lemma \ref{lem:glue:prop:cano}. Consider the area in yellow delimited by $l^{u_2}$, $P_{u_2,\ldots,i+1}$ and $l^i$ from Figure \ref{fig:Uturn:useful}. Since $\glueP{i}{i+1}$ points east and is visible in $P_{0,\ldots,\lastc}$ then $P_{i+1,\ldots,u_{3}}$ is inside the yellow area. Consider $u_2 \leq s \leq i$ such that $\glueP{s}{s+1}$ is visible from the south in $P_{u_2,\ldots,i+1}$ on glue column~$c$. If $(u_2,u_3)$ is a downward span then $P_{i+1,\ldots,u_3}$ must reach $l^s$ but it cannot leave the yellow area. Then $(u_2,u_3)$ is an upward span.}
\label{fig:Uturn:usefulProp}
\end{figure}

\end{proof}

\begin{lemma}
\label{Uturn:mainCano}
Consider a glue column $\Bone \leq c \leq \Btwo$ and a path $P$ which is canonical for $c$. Let $Q$ be the useful prefix of $P$ and let $(u_i)_{0\leq i \leq t}$ be the decomposition of $P$ into pseudo-visible glue on column $c$. If there exist $u_{i'} \leq i \leq u_{i'}$ and $u_{j'} \leq j \leq u_{j'}$ such that $\glueQ{i}{i+1}$ and $\glueQ{j}{j+1}$ are visible from the south on glue column $c'$ and $c''$ respectively with $c \leq c' \leq c'' \leq e_Q-1.5$ and both glues have the same type then we have:
 $$w_{Q_{i+1,\ldots,u_{i'+1}}} \leq w_{Q_{j+1,\ldots, u_{j'+1}}}+(c'-c'').$$
\end{lemma}

\begin{proof}
First of all, by Lemma \ref{lem:glue:prop:cano}, we have $i\leq j$ and thus $i'\leq j'$. Consider $l^Q$ the ray starting at $\pos{Q_{|Q|-1}}$ and going south and let $\mathcal{C}$ be the area delimited by $l^i$, $Q_{i,\ldots,|Q|-1}$ and $l^Q$, see Figure \ref{fig:Uturn:LinearCano}(b). Note that $w_{Q_{i,\ldots,|Q|-1}}=w_{Q_{i,\ldots,u_{i'+1}}}$ and if $Q_{j+1,\ldots,u_{j'+1}}-\vect{Q_iQ_j}$ is inside $\mathcal{C}$ then $$w_{Q_{i+1,\ldots,u_{i'+1}}} \leq w_{Q_{j+1,\ldots, u_{j'+1}}}+(c'-c'')$$ and the lemma is true \footnote{One could argue that the Figure \ref{fig:Uturn:LinearCano}(b) shows a case where there is a contradiction. Indeed, $Q_{j+1,\ldots,u_4}-\vect{Q_iQ_j}$ turns right of $Q_{i+1,\ldots,u_{2}}$ and then intersects with $Q$. This a contradiction of the priority of the pseudo-span $(u_1,u_2)$. Nevertheless, we will not use this argument in the proof. Indeed, it is possible that $Q_{j+1,\ldots,u_{4}}-\vect{Q_iQ_j}$ turns right of $Q_{i+1,\ldots, u_3}$ and that $Q_{i+1,\ldots,u_{2}}$ is a prefix of $Q_{j+1,\ldots,u_{4}}-\vect{Q_iQ_j}$. In this case, we cannot use the priority of $(u_2,u_3)$ since it is in the wrong orientation. This remark explains the complexity of this proof.}.

\input{./tikz/Uturn/ProofLinearCano}

For the sake of contradiction assume that $Q_{j+1,\ldots,u_{j'+1}}-\vect{Q_iQ_j}$ is not inside~$\mathcal{C}$. Then $Q_{j+1,\ldots,u_{j'+1}}-\vect{P_iP_j}$ cannot intersect $l^i$ (since $Q_{j+1,\ldots,u_{j'+1}}$ does not intersect $l^j$) and $l^Q$ (since $e_{Q_{j+1,\ldots,u_{j'+1}}-\vect{Q_iQ_j}} < e_Q$). Thus $Q_{j+1,\ldots,u_{j'+1}}-\vect{Q_iQ_j}$ must cross through $Q_{i+1,\ldots, |Q|-1}$. We consider now two sub-cases.

In the first case, $P_{|Q|-1,\ldots,|P|-1}$ does not cross $l^j$, see Figure \ref{fig:Uturn:LinearCanoPart2}. Firstly, remark that $P_{|Q|-1,\ldots,|P|-1}$ does not cross $l^i$. Indeed, since $\glueP{j}{j+1}$ is visible from the south in $P$ on column $c''$ then by Lemma \ref{lem:glue:prop4} and since $c''<c'$, the glue visible from the south on column $c'$ in $P$ cannot belong to $P_{|Q|-1,\ldots,|P|-1}$. Thus $\glueP{i}{i+1}$ is also visible from the south in $P$ and we can defined the area $\mathcal{D}$ delimited by $l^i$, $P_{i,\ldots,|P|-1}$ and $l^P$ (the ray starting at $\pos{P_{|P|-1}}$ and going south).
Secondly, consider $S$ the largest prefix of $Q_{j,\ldots,u_{j'}}-\vect{Q_iQ_j}$ which is entirely inside $\mathcal{D}$ (if $S=Q_{j,\ldots,u_{j'}}-\vect{Q_iQ_j}$ then we are in the previous case and the lemma is true). Let $a$ such that $S_{0,\ldots,a}$ is the largest prefix between $S$ and $P_{i+1,\ldots,|P|-1}$. Let $i+1\leq a' \leq |P|-1$ such that $P_{a'}=S_a$. If $S$ turns left of $P_{i,\ldots,|P|-1}$ then let $b=a$ and $b'=b$. In this case, $R=S_{0,\ldots,b}P_{b'+1,\ldots,|P|-1}=P_{i+1,\ldots,|P|-1}$ is a path and turns right of $S$. Otherwise, let $a< b \leq |S|-1$ such that: $$b = \min \{a< k:S_k \text{ is a tile of } P\}$$ and let $i+a < b < |P|-1$ such that $P_{b'}=S_b$. By definition of $b$, $R=S_{0,\ldots,b}P_{b'+1,\ldots,|P|-1}$ is a path. Remark that $S_{a,\ldots,b}$ and $P_{a',\ldots,b'}$ delimit a finite area, see Figure \ref{fig:Uturn:LinearCanoPart2}(a). This area cannot contains the last tile of $R$ which is the easternmost one. Since $P_{b'+1,\ldots,|P|-1}$ cannot intersect $S_{a,\ldots,b}$(by definition of $b$) and $P_{a',\ldots,b'}$ then $R$ is not inside this area and turns right of $Q_{j+1,\ldots,u_{j'+1}}-\vect{Q_iQ_j}$. 
Thirdly, since $P_{|Q|-1,\ldots,|P|-1}$ does not cross $l^j$, we can define the area $\mathcal{E}$ delimited by $l^j$, $P_{j,\ldots,|P|-1}$ and $l^P$ (the ray starting at $\pos{P_{|P|-1}}$ and going south), see Figure \ref{fig:Uturn:LinearCanoPart2}(b). Remark that $x_{P_{|P|-1}}+(c''-c')>e_P$ and thus $P_{|P|-1}+\vect{P_iP_j}$ is not inside $\mathcal{E}$. Thus $R+\vect{Q_iQ_j}$ must leave $\mathcal{E}$ after turning right of $P_{j,\ldots, u_{j'+1}}$ but it cannot cross $l^j$ (since $P_{b',\ldots,|P|-1}$ does not cross $l^i$) and it cannot intersect with $P$ (by priority of $(u_{j'}, u_{j'+1})$. This is a contradiction.


\input{./tikz/Uturn/ProofLinearCanoPart2}

In the second case, $P_{|Q|-1,\ldots,|P|-1}$ crosses trough cross $l^j$, see Figure \ref{fig:Uturn:LinearCanoPart3}. Firstly, Consider $j<g \leq |P|-1$ the smallest index such that $\glueP{g}{g+1}$ is on $l^j$. Note that $g>|Q|$ since $\glueQ{j}{j+1}$ is visible in $Q$. Consider $|Q|-1\leq d <g$ such that $P_{j,\ldots,d}$ is the shortest prefix of $P_{j,\ldots,|P|-1}$ to reach column $e_{P_{j,\ldots,g}}$. Consider the area $\mathcal{D}$ delimited by $l^i$, $Q_{i,\ldots,|Q|-1|}$ and $l^Q$ (the ray starting at $\pos{Q_{|Q|-1}}$ and going south). Secondly, consider $S$ the largest prefix of $Q_{j,\ldots,u_{j'}}-\vect{Q_iQ_j}$ which is entirely inside $\mathcal{D}$ (if $S=Q_{j,\ldots,u_{j'}}-\vect{Q_iQ_j}$ then we are in the previous case and the lemma is true). Let $a$ such that $S_{0,\ldots,a}$ is the largest prefix between $S$ and $P_{i+1,\ldots,|P|-1}$. Let $i+1\leq a' \leq |P|-1$ such that $P_{a'}=S_a$. If $S$ turns left of $P_{i,\ldots,|P|-1}$ then let $b=a$ and $b'=b$. In this case, $R=S_{0,\ldots,b}P_{b'+1,\ldots,d}=P_{i+1,\ldots,d}$ is a path and turns right of $S$. Otherwise, let $a< b \leq |S|-1$ such that: $$b = \min \{a< k:S_k \text{ is a tile of } P_{i,\ldots,d}\}$$ and let $i+a < b < d$ such that $P_{b'}=S_b$. By definition of $b$, $R=S_{0,\ldots,b}P_{b'+1,\ldots,d}$ is a path. Remark that $S_{a,\ldots,b}$ and $P_{a',\ldots,b'}$ delimit a finite area, see Figure \ref{fig:Uturn:LinearCanoPart2}(a). This area cannot contains the last tile of $R$ which is the easternmost one. Since $P_{b'+1,\ldots,d}$ cannot intersect $S_{a,\ldots,b}$(by definition of $b$) and $P_{a',\ldots,b'}$ then $R$ is not inside this area and turns right of $Q_{j+1,\ldots,u_{j'+1}}-\vect{Q_iQ_j}$. Thirdly, let $A=P_{j,\ldots,g+1}$. The path $A$ and the ray $l^j$ delimit a finite area of the plane called $\inter{A}$ (we develop this argument in the following Section \ref{sec:decompo}), see Figure \ref{fig:Uturn:LinearCanoPart3}(b). Remark that $x_{P_{d}}+(c''-c')>e_{P_{j,\ldots,g}}$ and thus $P_{|P|-1}+\vect{P_iP_j}$ is not inside $\inter{A}$. Thus $R+\vect{Q_iQ_j}$ must leave $\inter{A}$ after turning right of $P_{j,\ldots, u_{j'+1}}$ but it cannot cross $l^j$ (since $d<g$ then $P_{b',\ldots,d}$ does not cross $l^i$), cannot intersect with $P$ (by priority of $(u_{j'}, u_{j'+1})$ or crosses trough $l^P$ (since $P$ is extremal). This is a contradiction.

\input{./tikz/Uturn/ProofLinearCanoPart3}

\end{proof}


\begin{lemma}
\label{Uturn:LinearCano}
Consider a column $\Bone \leq c \leq \Btwo$, a canonical path $P$ for column $c$ and its last glue $\lastc$ on column $c$. Then, $$e_{P_{0,\ldots,\lastc}} \leq \Bthree(c).$$
\end{lemma}

\begin{proof}
Consider the useful prefix $Q$ of $P$ and the decomposition $(u_i)_{0\leq i \leq t}$ into pseudo-visible glue on glue column $c$. We assume that $t\geq 2$: the case $t=0$ is trivial. The case $t=1$ can be dealt with Lemma \ref{Uturn:LinearBound} and the calculus are identical the ones done in this proof. Then by Lemma \ref{lem:glue:prop:cano}, we can assume that there exist $u_{k'}\leq k <{u_{k'+1}}$ such that $\glueQ{k}{k+1}$ is visible from the south on glue column $e_{Q}-1.5$. Let $0 \leq s \leq |Q|-1$ such that $\glueQ{s}{s+1}$ is visible from the south on glue column $c$ in $Q$, we have either $s=u_0$ or $s=u_1$. If $s=u_0$, let $n=u_1$ otherwise let $n=u_2.$ By the same reasoning as in the proof of Lemma \ref{Uturn:LinearBound}, let $u_{j'}\leq j < u_{j'+1}$ such that $\glueP{j}{j+1}$ has the same type as $\glueP{k}{k+1}$ and is the westernmost glue visible from the south in $Q$ on glue column $c_1$ with $c \leq c_1 \leq e_Q-1.5$. Then by Lemma \ref{Uturn:mainCano} we have 
 \begin{align*}
w_{Q_{j+1,\ldots,u_{j'+1}}} & \leq w_{Q_{k+1,\ldots, u_{k'+1}}}+(c_1-(e_Q-1.5)) \\
 \end{align*}
By iterating this reasoning at most $|T|$ times, we obtain: 
 \begin{align*}
w_{Q_{s,\ldots,n}} & \leq w_{Q_{k+1,\ldots, u_{k'+1}}}+(c-(e_Q-1.5)+|T| \\
 \end{align*}
Since $\glueQ{u_{k'+1}}{u_{k'+1}+1}$ is on glue column $c$ then $Q_{u_{k'+1}}$ is on column $c-0.5$. Thus, we have $w_{Q_{k+1,\ldots, u_{k'+1}}}\leq c-0.5$:
 \begin{align*}
w_{Q_{s,\ldots,n}} & \leq 2c-e_Q+|T|+1 \\
 \end{align*}
  By Lemma \ref{Uturn:boundWest}, we have  $w_P\geq w_{\sigma}-2|T|-1$ which leads to: $$w_{\sigma}-2|T|-1\leq 2c-e_Q+|T|+1$$ And then, 
 \begin{align*}
 e_Q & \leq 2c-w_{\sigma}+3|T|+2\\
 \end{align*}
 \end{proof}


\subsection{Conclusion}

Now, we can combine Corollary \ref{fact:Bthree} and Lemma \ref{Uturn:LinearCano} together to obtain the main result of this section, Lemma \ref{Uturn:Conc}, which solves Lock \ref{lock:uturn} and prove the efficiency of the shield. 

Note that this result is stronger than the one described in the roadmap (Section \ref{sec:roadmap}). Indeed, Figure \ref{fig:ShieldFinal} shows a path $Q$ growing on a tile of $P_{0,\ldots,\lastc}$ and turning around the shield. The following lemma claims that if the path $Q$ grows far away from $P_{0,\ldots,\lastc}$ then it cannot come near it anymore. Thus, we are able to give this result even if a shield is not formally defined yet.

\begin{lemma}
\label{Uturn:Conc}
Consider a path $P$ which is canonical for glue column $\Bone \leq c \leq \Btwo$. Consider the index $\lastc$ of the last glue of $P$ on column $c$. Consider a path $Q=Q'Q''$ such that:
\begin{itemize}
\item  $Q_0$ binds with a tile $P_i$ of $P_{0,\ldots,\lastc}$;
\item  the last tile of $Q'$ is the easternmost one and $e_{Q'}\geq \Bfour(c)$;
\item  $P_{0,\ldots,i}Q'$ is a producible path. 
\end{itemize}
Then $P_{0,\ldots,i}Q$ is a producible path.
\end{lemma}

\begin{proof}
By Lemma \ref{Uturn:LinearCano}, we have $e_{P_{0,\ldots,\lastc}} \leq \Bthree(c)$. Then $w_{Q'}\leq \Bthree(c)+1$ and thus $Q'$ crosses the glue column $\Bthree(c)+1.5$, see Figure \ref{fig:Uturn:Conc}. Remark that since $e_{P_{0,\ldots,\lastc}} \leq \Bthree(c)$, the glues of $Q'$ which are visible on glue column $\Bthree(c)+1.5$ are also visible in $R=P_{0,\ldots,i}Q'$. Thus, we can define $(s,n)$ the span of $R$ on glue column $\Bthree(c)+1.5$. Without loss of generality, we suppose that this span is an upward span and let $\mathcal{C}$ be the workspace of the span $(s,n)$ of $R$ on glue column $\Bthree(c)+1.5$. Since $e_{P_{0,\ldots,\lastc}} \leq \Bthree(c)$ and the last tile of $Q$ is the easternmost one then $\glueR{s+1}{s+1}$ and $\glueR{n+1}{n+1}$ both point east and no tile of $P_{0,\ldots,\lastc}$ is inside $\mathcal{C}$. Now if $Q''$ is inside $\mathcal{C}$ then $P_{0,\ldots,i}Q$ is producible and the lemma is true. 

Otherwise let $S$ be the shortest prefix of $Q''$ which is not entirely inside $\mathcal{C}$. Since $Q'Q''$ is a path, then $S$ must cross $l^n$ or $l^s$. Both cases are symmetric thus assume that $S$ crosses $l^n$. Then $x_{S_{|S|-1}}=\Bthree(c)+1$ and since $e_{P_{0,\ldots,\lastc}} \leq \Bthree(c)$ then $S$ does not intersect with $P_{0,\ldots,\lastc}$ and $RS$ is a producible path. By hypothesis, we have $e_{RS}\geq \Bfour(c)$. Moreover, the tile $S_{|S|-1}$ is visible from the south on column $\Bthree(c)+1$ and $\glueR{s}{s+1}$ is visible from the south in $RS$ on glue column $\Bthree(c)+1.5$. All these hypotheses contradict Corollary~\ref{fact:Bthree}.

\begin{figure}
\center
\begin{tikzpicture}[x=0.22cm,y=0.22cm]

\fill[fill=blue!30!white, draw opacity=0.8] (21,0) |- (16.5,6.5) |- (22.5,13.5) |- (21,16.5) |- (38,24)  |- (21,0);

\draw[very thick] (1.5,3.5) -| (13.5,6.5) -| (2.5,21.5) -| (13.5,18.5) -| (5.5,9.5)-| (13.5,12.5) -| (8.5,15.5) -| (10.5,15.5);
\draw[very thick] (12.5,18.5) |- (22.5,16.5) |- (16.5,13.5) |- (35.5,6.5) |- (19.5,9.5) |- (25.5,11.5) |- (20.5,19.5);

\draws{1}{3}

\tileg{2}{3}{48}
\tileg{3}{3}{48}
\tileg{4}{3}{48}
\tileg{5}{3}{48}
\tileg{6}{3}{48}
\tileg{7}{3}{48}
\tileg{8}{3}{48}
\tileg{9}{3}{48}
\tileg{10}{3}{48}

\tileg{11}{3}{48}
\tileg{12}{3}{48}
\tileg{13}{3}{48}
\tileg{13}{4}{48}
\tileg{13}{5}{48}
\tileg{13}{6}{48}
\tileg{12}{6}{48}
\tileg{11}{6}{48}
\tileg{10}{6}{48}
\tileg{9}{6}{48}
\tileg{8}{6}{48}
\tileg{7}{6}{48}
\tileg{6}{6}{48}
\tileg{5}{6}{48}
\tileg{4}{6}{48}
\tileg{3}{6}{48}
\tileg{2}{6}{48}
\tileg{2}{7}{48}
\tileg{2}{8}{48}
\tileg{2}{9}{48}
\tileg{2}{10}{48}
\tileg{2}{11}{48}
\tileg{2}{12}{48}
\tileg{2}{13}{48}
\tileg{2}{14}{48}
\tileg{2}{15}{48}
\tileg{2}{16}{48}
\tileg{2}{17}{48}
\tileg{2}{18}{48}
\tileg{2}{19}{48}
\tileg{2}{20}{48}
\tileg{2}{21}{48}
\tileg{3}{21}{48}
\tileg{4}{21}{48}
\tileg{5}{21}{48}
\tileg{6}{21}{48}
\tileg{7}{21}{48}
\tileg{8}{21}{48}
\tileg{9}{21}{48}
\tileg{10}{21}{48}

\tileg{11}{21}{48}
\tileg{12}{21}{48}
\tileg{13}{21}{48}
\tileg{13}{20}{48}
\tileg{13}{19}{48}
\tileg{13}{18}{48}
\tileg{12}{18}{48}
\tileg{11}{18}{48}
\tileg{10}{18}{48}
\tileg{9}{18}{48}
\tileg{8}{18}{48}
\tileg{7}{18}{48}
\tileg{6}{18}{48}
\tileg{5}{18}{48}
\tileg{5}{17}{48}
\tileg{5}{16}{48}
\tileg{5}{15}{48}
\tileg{5}{14}{48}
\tileg{5}{13}{48}
\tileg{5}{12}{48}
\tileg{5}{11}{48}
\tileg{5}{10}{48}
\tileg{5}{9}{48}
\tileg{6}{9}{48}
\tileg{7}{9}{48}
\tileg{8}{9}{48}
\tileg{9}{9}{48}
\tileg{10}{9}{48}

\tileg{11}{9}{48}
\tileg{12}{9}{48}
\tileg{13}{9}{48}
\tileg{13}{10}{48}
\tileg{13}{11}{48}
\tileg{13}{12}{48}
\tileg{12}{12}{48}
\tileg{11}{12}{48}
\tileg{10}{12}{48}
\tileg{9}{12}{48}
\tileg{8}{12}{48}
\tileg{8}{13}{48}
\tileg{8}{14}{48}
\tileg{8}{15}{48}
\tileg{9}{15}{48}
\tileg{10}{15}{48}

\tiley{12}{17}{85}
\tiley{12}{16}{85}
\tiley{13}{16}{85}
\tiley{14}{16}{85}
\tiley{15}{16}{85}
\tiley{16}{16}{85}
\tiley{17}{16}{85}
\tiley{18}{16}{85}
\tiley{19}{16}{85}
\tiley{20}{16}{85}
\tiley{21}{16}{85}
\tiley{22}{16}{85}
\tiley{22}{15}{85}
\tiley{22}{14}{85}
\tiley{22}{13}{85}
\tiley{21}{13}{85}
\tiley{20}{13}{85}
\tiley{19}{13}{85}
\tiley{18}{13}{85}
\tiley{17}{13}{85}
\tiley{16}{13}{85}
\tiley{16}{12}{85}
\tiley{16}{11}{85}
\tiley{16}{10}{85}
\tiley{16}{9}{85}
\tiley{16}{8}{85}
\tiley{16}{7}{85}
\tiley{16}{6}{85}
\tiley{17}{6}{85}
\tiley{18}{6}{85}
\tiley{19}{6}{85}
\tiley{20}{6}{85}
\tiley{21}{6}{85}
\tiley{22}{6}{85}
\tiley{23}{6}{85}
\tiley{24}{6}{85}
\tiley{25}{6}{85}
\tiley{26}{6}{85}
\tiley{27}{6}{85}
\tiley{28}{6}{85}
\tiley{29}{6}{85}
\tiley{30}{6}{85}

\tileor{31}{6}{85}
\tileor{32}{6}{85}
\tileor{33}{6}{85}
\tileor{34}{6}{85}
\tileor{35}{6}{85}
\tileor{35}{7}{85}
\tileor{35}{8}{85}
\tileor{35}{9}{85}
\tileor{34}{9}{85}
\tileor{33}{9}{85}
\tileor{32}{9}{85}
\tileor{31}{9}{85}
\tileor{30}{9}{85}
\tileor{29}{9}{85}
\tileor{28}{9}{85}
\tileor{27}{9}{85}
\tileor{26}{9}{85}
\tileor{25}{9}{85}
\tileor{24}{9}{85}
\tileor{23}{9}{85}
\tileor{22}{9}{85}
\tileor{21}{9}{85}
\tileor{20}{9}{85}
\tileor{19}{9}{85}
\tileor{19}{10}{85}
\tileor{19}{11}{85}
\tileor{20}{11}{85}
\tileor{21}{11}{85}
\tileor{22}{11}{85}
\tileor{23}{11}{85}
\tileor{24}{11}{85}
\tileor{25}{11}{85}
\tileor{25}{12}{85}
\tileor{25}{13}{85}
\tileor{25}{14}{85}
\tileor{25}{15}{85}
\tileor{25}{16}{85}
\tileor{25}{17}{85}
\tileor{25}{18}{85}
\tileor{25}{19}{85}
\tileor{24}{19}{85}
\tileor{23}{19}{85}
\tileor{22}{19}{85}
\tileor{21}{19}{85}
\tileor{20}{19}{85}

\path [dotted, draw, thin] (0,0) grid[step=0.22cm] (38,24);

\draw[dashed] (30,0) -| (30,24);

\draw[dashed] (21,0) -| (21,6.5);
\fill (20.5,6.5) circle (0.16);
\node (D) at (20.5,7.7) {$s$};
\node (D) at (22.2,5) {$l^{s}$};

\node (D) at (32.7,1) {$\Bfour(c)$};
\node (D) at (16.4,1) {$\Bthree(c)+1.5$};

\draw[dashed] (21,24) -| (21,16.5);
\fill (20.5,16.5) circle (0.16);
\node (D) at (20.5,17.7) {$n$};
\node (D) at (22.2,21) {$l^{n}$};


\fill (12.5,18.5) circle (0.16);
\node (D) at (12.5,20.2) {$i$};

\fill (31.5,6.5) circle (0.16);
\node (D) at (31.5,7.9) {$S_0$};

\fill (10.5,15.5) circle (0.16);
\node (D) at (10.5,16.9) {$\ell$};
\end{tikzpicture}

\caption{Proof of Lemma \ref{Uturn:Conc}. Consider a path $P$ and the index $\lastc$ of its last glue on column $c$. The prefix $P_{0,\ldots,\lastc}$ is drawn in green and is west of column $\Bthree(c)$. A path $Q'$ is in yellow, its first tile binds with $P_i$ and its last tile is on column $\Bfour(c)+0.5$. The path $R=P_{0,\ldots,\lastc}Q'$ is producible by hypothesis and  $(n,s)$ is the downward span of $R$ on glue column $\Bthree(c)+1.5$. The workspace $\mathcal{C}$ of this span of $R$ is in blue. The path $S$ is orange and by hypothesis $Q'S$ is a path. To leave $\mathcal{C}$, $S$ crosses through $l^n$ but then $RS$ is producible which is a contradiction by Corollary \ref{fact:Bthree}. }
\label{fig:Uturn:Conc}
\end{figure}
\end{proof}

\section{Decomposition into dominant arcs}
\label{sec:decompo}

Consider an extremal path $P$ which is canonical for a glue column $\Bone \leq c \leq \Btwo$. A key point of the roadmap was to use a shield in order to protect a glue of $P$ on glue column $c$, see Subsection \ref{road:pseudo-visibility}. In the previous example of Figure \ref{fig:PseudoVisible}, we assumed that this protected glue was pseudo-visible, see Lock \ref{lock:pseudo}. Unfortunately, this is not always the case, see Figure \ref{fig:Protected:NotPseudoVisible}. We need to focus on the other glues of $P$ on glue column $c$. To do so, we consider subpaths of $P$ whose first and last glues are two consecutive glues on glue column $c$. These subpaths will be called \emph{arcs} of $P$.

The notion of arcs is properly defined in Subsection \ref{Sec:hole:def} alongside the more general concept of $\emph{holes}$. Then, we focus on the subpaths of a path which are arcs, see Subsection \ref{Sec:arc:paths}, and we analyze their properties, see Subsection \ref{Sec:arc:dominant}. These results lead to the definition of the decomposition of an extremal path into dominant arcs in Subsection \ref{Sec:arc:deompo}. This decomposition will be the one used in the final Section \ref{sec:analysis}. 
Indeed, the previous decomposition into pseudo-visible glues has been useful to define canonical paths, see Section \ref{sec:canon}. Nevertheless, canonical paths have only three properties which will be relevant in the final Section \ref{sec:analysis}: the first one is Lemma \ref{Uturn:LinearCano} (the bound on $P_{0,\ldots,\lastc}$), the second one is Lemma \ref{Uturn:Conc} (the final result of Section~\ref{sec:U-turn}) and the third one is achieved here in Corollary \ref{cor:decompo:cano}. This result is obtained by studying how the priority of the pseudo-spans affects the decomposition of a canonical path into dominant arcs.

\input{./tikz/Hole/NotPseudoVisible}

\subsection{Holes and arcs}
\label{Sec:hole:def}

We start with some basic definitions. Consider a path $\Hole$ which is a subassembly of the terminal assembly $\uniterm$\footnote{This hypothesis is important and we will always consider paths that can be assembled by the tile assembly system.}. Suppose that $\glueN{0}{1}$ and $\glueN{|\Hole|-2}{|\Hole|-1}$ are both on the same glue column $c$, we define $\cork{\Hole}$ as the finite line between $\glueN{0}{1}$ and $\glueN{|\Hole|-2}{|\Hole|-1}$. We say that $\Hole$ is a \emph{hole} on glue column $c$ if and only if the only intersection between $\cork{\Hole}$ and $\Hole$ are $\glueN{0}{1}$ and $\glueN{|\Hole|-2}{|\Hole|-1}$, see Figure \ref{fig:Hole:def}. The finite line $\cork{\Hole}$ is called the \emph{frontier} of the hole. We say that the hole $\Hole$ is an \emph{upward} (resp. \emph{downward}) hole if and only if $y_{\Hole_0}<y_{\Hole_{|\Hole|-1}}$ (resp. $y_{\Hole_0}>y_{\Hole_{|\Hole|-1}}$).

\begin{figure}
\center
\begin{tikzpicture}[x=0.22cm,y=0.22cm]

\fill[fill=green!30!white, draw opacity=0.8] (15,9.5) -| (16.5,7.5) -| (13.5,5.5) -| (16.5,3.5) -| (14.5,1.5) -| (23.5,7.5) -| (21.5,16.5) -| (13.5,18.5) -| (18.5,20.5) -| (11.5,14.5) -| (17.5,12.5) -| (15,12.5);
\fill[fill=yellow!30!white, draw opacity=0.8] (40,9.5) -| (36.5,5.5) -| (32.5,8.5) -| (29.5,16.5) -| (28.5,19.5) -| (35.5,15.5) -| (33.5,13.5) -| (37.5,15.5) -| (40,15.5);
\fill[fill=orange!30!white, draw opacity=0.8] (46,9.5) -| (47.5,7.5) -| (49.5,5.5)  -| (52.5,13.5) -| (48.5,12.5) -| (46,12.5);

\draw[very thick] (4.5,6.5) -| (7.5,11.5) -| (2.5,15.5) -| (5.5,15.5);
\draw[very thick] (14.5,9.5) -| (16.5,7.5) -| (13.5,5.5) -| (16.5,3.5) -| (14.5,1.5) -| (23.5,7.5) -| (21.5,16.5) -| (13.5,18.5) -| (18.5,20.5) -| (11.5,14.5) -| (17.5,12.5) -| (14.5,12.5);
\draw[very thick] (40.5,9.5) -| (36.5,5.5) -| (32.5,8.5) -| (29.5,16.5) -| (28.5,19.5) -| (35.5,15.5) -| (33.5,13.5) -| (37.5,15.5) -| (40.5,15.5);
\draw[very thick] (45.5,9.5) -| (47.5,7.5) -| (49.5,5.5)  -| (52.5,13.5) -| (48.5,12.5) -| (45.5,12.5);



\tiler{4}{6}{85}
\tiler{5}{6}{85}
\tiler{6}{6}{85}
\tiler{7}{6}{85}
\tiler{7}{7}{85}
\tiler{7}{8}{85}
\tiler{7}{9}{85}
\tiler{7}{10}{85}
\tiler{7}{11}{85}
\tiler{6}{11}{85}
\tiler{5}{11}{85}
\tiler{4}{11}{85}
\tiler{3}{11}{85}
\tiler{2}{11}{85}
\tiler{2}{12}{85}
\tiler{2}{13}{85}
\tiler{2}{14}{85}
\tiler{2}{15}{85}
\tiler{3}{15}{85}
\tiler{4}{15}{85}
\tiler{5}{15}{85}

\tileg{14}{9}{85}
\tileg{15}{9}{85}
\tileg{16}{9}{85}
\tileg{16}{8}{85}
\tileg{16}{7}{85}
\tileg{15}{7}{85}
\tileg{14}{7}{85}
\tileg{13}{7}{85}
\tileg{13}{6}{85}
\tileg{13}{5}{85}
\tileg{14}{5}{85}
\tileg{15}{5}{85}
\tileg{16}{5}{85}
\tileg{16}{4}{85}
\tileg{16}{3}{85}
\tileg{15}{3}{85}
\tileg{14}{3}{85}
\tileg{14}{2}{85}
\tileg{14}{1}{85}
\tileg{15}{1}{85}
\tileg{16}{1}{85}
\tileg{17}{1}{85}
\tileg{18}{1}{85}
\tileg{19}{1}{85}
\tileg{20}{1}{85}
\tileg{21}{1}{85}
\tileg{22}{1}{85}
\tileg{23}{1}{85}
\tileg{23}{2}{85}
\tileg{23}{3}{85}
\tileg{23}{4}{85}
\tileg{23}{5}{85}
\tileg{23}{6}{85}
\tileg{23}{7}{85}
\tileg{22}{7}{85}
\tileg{21}{7}{85}
\tileg{21}{8}{85}
\tileg{21}{9}{85}
\tileg{21}{10}{85}
\tileg{21}{11}{85}
\tileg{21}{12}{85}
\tileg{21}{13}{85}
\tileg{21}{14}{85}
\tileg{21}{15}{85}
\tileg{21}{16}{85}
\tileg{20}{16}{85}
\tileg{19}{16}{85}
\tileg{18}{16}{85}
\tileg{17}{16}{85}
\tileg{16}{16}{85}
\tileg{15}{16}{85}
\tileg{14}{16}{85}
\tileg{13}{16}{85}
\tileg{13}{17}{85}
\tileg{13}{18}{85}
\tileg{14}{18}{85}
\tileg{15}{18}{85}
\tileg{16}{18}{85}
\tileg{17}{18}{85}
\tileg{18}{18}{85}
\tileg{18}{19}{85}
\tileg{18}{20}{85}
\tileg{17}{20}{85}
\tileg{16}{20}{85}
\tileg{15}{20}{85}
\tileg{14}{20}{85}
\tileg{13}{20}{85}
\tileg{12}{20}{85}
\tileg{11}{20}{85}
\tileg{11}{19}{85}
\tileg{11}{18}{85}
\tileg{11}{17}{85}
\tileg{11}{16}{85}
\tileg{11}{15}{85}
\tileg{11}{14}{85}
\tileg{12}{14}{85}
\tileg{13}{14}{85}
\tileg{14}{14}{85}
\tileg{15}{14}{85}
\tileg{16}{14}{85}
\tileg{17}{14}{85}
\tileg{17}{13}{85}
\tileg{17}{12}{85}
\tileg{16}{12}{85}
\tileg{15}{12}{85}
\tileg{14}{12}{85}

\tiley{40}{9}{85}
\tiley{39}{9}{85}
\tiley{38}{9}{85}
\tiley{37}{9}{85}
\tiley{36}{9}{85}
\tiley{36}{8}{85}
\tiley{36}{7}{85}
\tiley{36}{6}{85}
\tiley{36}{5}{85}
\tiley{35}{5}{85}
\tiley{34}{5}{85}
\tiley{33}{5}{85}
\tiley{32}{5}{85}
\tiley{32}{6}{85}
\tiley{32}{7}{85}
\tiley{32}{8}{85}
\tiley{31}{8}{85}
\tiley{30}{8}{85}
\tiley{29}{8}{85}
\tiley{29}{9}{85}
\tiley{29}{10}{85}
\tiley{29}{11}{85}
\tiley{29}{12}{85}
\tiley{29}{13}{85}
\tiley{29}{14}{85}
\tiley{29}{15}{85}
\tiley{29}{16}{85}
\tiley{28}{16}{85}
\tiley{28}{17}{85}
\tiley{28}{18}{85}
\tiley{28}{19}{85}
\tiley{29}{19}{85}
\tiley{30}{19}{85}
\tiley{31}{19}{85}
\tiley{32}{19}{85}
\tiley{33}{19}{85}
\tiley{34}{19}{85}
\tiley{35}{19}{85}
\tiley{35}{18}{85}
\tiley{35}{17}{85}
\tiley{35}{16}{85}
\tiley{35}{15}{85}
\tiley{34}{15}{85}
\tiley{33}{15}{85}
\tiley{33}{14}{85}
\tiley{33}{13}{85}
\tiley{34}{13}{85}
\tiley{35}{13}{85}
\tiley{36}{13}{85}
\tiley{37}{13}{85}
\tiley{37}{14}{85}
\tiley{37}{15}{85}
\tiley{38}{15}{85}
\tiley{39}{15}{85}
\tiley{40}{15}{85}

\tileor{45}{12}{85}
\tileor{46}{12}{85}
\tileor{47}{12}{85}
\tileor{48}{12}{85}
\tileor{48}{13}{85}
\tileor{49}{13}{85}
\tileor{50}{13}{85}
\tileor{51}{13}{85}
\tileor{52}{13}{85}
\tileor{52}{12}{85}
\tileor{52}{11}{85}
\tileor{52}{10}{85}
\tileor{52}{9}{85}
\tileor{52}{8}{85}
\tileor{52}{7}{85}
\tileor{52}{6}{85}
\tileor{52}{5}{85}
\tileor{51}{5}{85}
\tileor{50}{5}{85}
\tileor{49}{5}{85}
\tileor{49}{6}{85}
\tileor{49}{7}{85}
\tileor{48}{7}{85}
\tileor{47}{7}{85}
\tileor{47}{8}{85}
\tileor{47}{9}{85}
\tileor{46}{9}{85}
\tileor{45}{9}{85}

\path [dotted, draw, thin] (0,0) grid[step=0.22cm] (55,22);

\draw [dashed, color=red] (5,0) -| (5,22);
\draw [thick, color=red] (5,6.5) -| (5,15.5);
\draw [dashed, color=green!70!black] (15,0) -| (15,22);
\draw [thick, color=green!70!black] (15,9.5) -| (15,12.5);
\draw [dashed, color=yellow!70!black] (40,0) -| (40,22);
\draw [thick, color=yellow!70!black] (40,9.5) -| (40,15.5);
\draw [dashed, color=orange] (46,0) -| (46,22);
\draw [thick, color=orange!70!black] (46,9.5) -| (46,12.5);

\node (D) at (20,5) {$\inter{H}$};

\fill (14.5,9.5) circle (0.16);
\node (D) at (12.7,9.5) {$H_{0}$};
\fill (14.5,12.5) circle (0.16);

\fill (4.5,6.5) circle (0.16);
\node (D) at (2.7,6.5) {$P_0$};
\fill (5.5,15.5) circle (0.16);

\fill (40.5,15.5) circle (0.16);
\node (D) at (42.2,15.5) {$A_0$};
\node (D) at (33,11) {$\inter{A}$};

\fill (45.5,9.5) circle (0.16);
\node (D) at (43.7,9.5) {$B_0$};
\node (D) at (50.5,10) {$\inter{B}$};

\end{tikzpicture}
\caption{Four examples of path whose first and last glues are on the same glue column: the red path $P$ intersects with its frontier (at a position which is not its extremities) then $P$ is not a hole. The green path $H$ is an upward hole (its last tile is north of its first tile), its interior $\inter{H}$ is in green. The yellow path $A$ is a negative (apart from its first and last tile, $A$ is in the west side of the yellow glue column) downward arc, its interior $\inter{A}$ is in yellow. The orange path $B$ is a positive upward arc, its interior $\inter{B}$ is in orange.}
\label{fig:Hole:def}
\end{figure}

Note that, $\Hole$ and $\cork{\Hole}$ delimit a finite area called the \emph{interior} $\inter{\Hole}$ of the hole. We call $\prodHole{\inter{\Hole}}$ the set of holes which are starting by $\Hole_0\Hole_1$, ending by  $\Hole_{|\Hole|-2}\Hole_{|\Hole|-1}$ and which are inside $\inter{\Hole}$ (apart from $\Hole_0$ and $\Hole_{|\Hole|-1}$). Note that $\prodHole{\inter{\Hole}}$ contains at least $\Hole$ and if $\prodHole{\inter{\Hole}}=\{\Hole\}$, we say that $\Hole$ is \emph{minimum}, see Figure \ref{fig:MinArc:def}. Otherwise, if $\inter{S}$ is the left side (resp. right side) of $\Hole$ we designate by $\min(\inter{\Hole})$ the left-priority (resp. right-priority) path of this set. Note that, $\min(\inter{\Hole})$ is a minimum hole whose interior is included into $\inter{\Hole}$. Thus, we also have:

\begin{figure}
\center
\begin{tikzpicture}[x=0.22cm,y=0.22cm]

\fill[fill=green!30!white, draw opacity=0.8] (15,9.5) -| (16.5,7.5) -| (13.5,5.5) -| (16.5,3.5) -| (14.5,1.5) -| (23.5,7.5) -| (21.5,16.5) -| (13.5,18.5) -| (18.5,20.5) -| (11.5,14.5) -| (17.5,12.5) -| (15,12.5);
\fill[fill=orange!30!white, draw opacity=0.8] (15,9.5) -| (16.5,7.5) -| (18.5,5.5)  -| (21.5,13.5) -| (17.5,12.5) -| (15,12.5);

\draw[very thick] (14.5,9.5) -| (16.5,7.5) -| (13.5,5.5) -| (16.5,3.5) -| (14.5,1.5) -| (23.5,7.5) -| (21.5,16.5) -| (13.5,18.5) -| (18.5,20.5) -| (11.5,14.5) -| (17.5,12.5) -| (14.5,12.5);
\draw[very thick] (14.5,9.5) -| (16.5,7.5) -| (18.5,5.5)  -| (21.5,13.5) -| (17.5,12.5) -| (14.5,12.5);


\tileg{14}{9}{85}
\tileg{15}{9}{85}
\tileg{16}{9}{85}
\tileg{16}{8}{85}
\tileg{16}{7}{85}
\tileg{15}{7}{85}
\tileg{14}{7}{85}
\tileg{13}{7}{85}
\tileg{13}{6}{85}
\tileg{13}{5}{85}
\tileg{14}{5}{85}
\tileg{15}{5}{85}
\tileg{16}{5}{85}
\tileg{16}{4}{85}
\tileg{16}{3}{85}
\tileg{15}{3}{85}
\tileg{14}{3}{85}
\tileg{14}{2}{85}
\tileg{14}{1}{85}
\tileg{15}{1}{85}
\tileg{16}{1}{85}
\tileg{17}{1}{85}
\tileg{18}{1}{85}
\tileg{19}{1}{85}
\tileg{20}{1}{85}
\tileg{21}{1}{85}
\tileg{22}{1}{85}
\tileg{23}{1}{85}
\tileg{23}{2}{85}
\tileg{23}{3}{85}
\tileg{23}{4}{85}
\tileg{23}{5}{85}
\tileg{23}{6}{85}
\tileg{23}{7}{85}
\tileg{22}{7}{85}
\tileg{21}{7}{85}
\tileg{21}{8}{85}
\tileg{21}{9}{85}
\tileg{21}{10}{85}
\tileg{21}{11}{85}
\tileg{21}{12}{85}
\tileg{21}{13}{85}
\tileg{21}{14}{85}
\tileg{21}{15}{85}
\tileg{21}{16}{85}
\tileg{20}{16}{85}
\tileg{19}{16}{85}
\tileg{18}{16}{85}
\tileg{17}{16}{85}
\tileg{16}{16}{85}
\tileg{15}{16}{85}
\tileg{14}{16}{85}
\tileg{13}{16}{85}
\tileg{13}{17}{85}
\tileg{13}{18}{85}
\tileg{14}{18}{85}
\tileg{15}{18}{85}
\tileg{16}{18}{85}
\tileg{17}{18}{85}
\tileg{18}{18}{85}
\tileg{18}{19}{85}
\tileg{18}{20}{85}
\tileg{17}{20}{85}
\tileg{16}{20}{85}
\tileg{15}{20}{85}
\tileg{14}{20}{85}
\tileg{13}{20}{85}
\tileg{12}{20}{85}
\tileg{11}{20}{85}
\tileg{11}{19}{85}
\tileg{11}{18}{85}
\tileg{11}{17}{85}
\tileg{11}{16}{85}
\tileg{11}{15}{85}
\tileg{11}{14}{85}
\tileg{12}{14}{85}
\tileg{13}{14}{85}
\tileg{14}{14}{85}
\tileg{15}{14}{85}
\tileg{16}{14}{85}
\tileg{17}{14}{85}
\tileg{17}{13}{85}
\tileg{17}{12}{85}
\tileg{16}{12}{85}
\tileg{15}{12}{85}
\tileg{14}{12}{85}

\dotor{14}{12}{85}
\dotor{15}{12}{85}
\dotor{16}{12}{85}
\dotor{17}{12}{85}
\dotor{17}{13}{85}
\dotor{18}{13}{85}
\dotor{19}{13}{85}
\dotor{20}{13}{85}
\dotor{21}{13}{85}
\dotor{21}{12}{85}
\dotor{21}{11}{85}
\dotor{21}{10}{85}
\dotor{21}{9}{85}
\dotor{21}{8}{85}
\dotor{21}{7}{85}
\dotor{21}{6}{85}
\dotor{21}{5}{85}
\dotor{20}{5}{85}
\dotor{19}{5}{85}
\dotor{18}{5}{85}
\dotor{18}{6}{85}
\dotor{18}{7}{85}
\dotor{17}{7}{85}
\dotor{16}{7}{85}
\dotor{16}{8}{85}
\dotor{16}{9}{85}
\dotor{15}{9}{85}
\dotor{14}{9}{85}

\path [dotted, draw, thin] (6,0) grid[step=0.22cm] (28,22);

\draw [dashed, color=green!70!black] (15,0) -| (15,22);
\draw [thick, color=green!70!black] (15,9.5) -| (15,12.5);

\fill (14.5,9.5) circle (0.16);
\node (D) at (10.2,9.5) {$H_{0}=B_0$};
\fill (14.5,12.5) circle (0.16);

\end{tikzpicture}
\caption{Consider the hole $H$ (in green) and the arc $B$ (the orange dots) of Figure \ref{fig:Hole:def} and suppose that $H_0=B_0$. Then, the interior of $B$ is included into the interior of $H$. In this case, $H$ is not minimum. The arc $B$ is minimum if and if only it is the only path from $B_0$ to $B_{|B|-1}$ which is inside $\inter{B}$ and which belongs to the terminal assembly $\uniterm$. In this case, it is also the leftmost priority path from $B_0$ to $B_{|B|-1}$.}
\label{fig:MinArc:def}
\end{figure}

\begin{fact}
\label{fact:arc:int}
For a hole $\Hole$, we have $w_{\Hole}\leq w_{\min(\inter{\Hole})} < e_{\min(\inter{\Hole})} \leq e_{\Hole}$.
\end{fact}

If $\glueN{0}{1}$ and $\glueN{|\Hole|-2}{|\Hole|-1}$ are the only two glues of $\Hole$ which are on glue column $c$ then $\Hole$ is called an \emph{arc}. In this case, if $x_{\Hole_{1,\ldots,|\Hole|-2}}>c$, we say that $\Hole$ is a \emph{positive} arc and $\glueN{0}{1}$ points east while $\glueN{|\Hole|-2}{|\Hole|-1}$ points west. Otherwise, we have $x_{\Hole_{1,\ldots,|\Hole|-2}}<c$ and $\Hole$ is a \emph{negative} arc and $\glueN{0}{1}$ points west while $\glueN{|\Hole|-2}{|\Hole|-1}$ points east.  Finally, here is a technical lemma which will be useful later.

\begin{lemma} \label{combi:half:neq:Three}
Consider a hole $H$ on glue column $c$ and its interior $\inter{H}$. If $A$ is a positive arc of glue column $c$ which is inside $\inter{H}$ (apart from its first and last tiles) then there exist $0\leq i \leq j \leq |H|-1$ such that $H_{i,\ldots,j}$ is a positive arc on glue column $c$ with $$\min\{y_{H_i}, y_{H_j}\} \leq \min\{y_{A_0}, y_{A_{|A|-1}}\} \leq  \max\{y_{A_0}, y_{A_{|A|-1}}\} \leq \max\{y_{H_i}, y_{H_j}\}.$$
\end{lemma}

\begin{proof}
Consider $0\leq k \leq |A|-1$ such that $x_{A_k}=e_A$ and consider the ray $l^k$ starting at position $\pos{A_k}$ and going east, see Figure \ref{fig:Arc:ArcInHole}. If this ray encounter no tile of $H$ then $A_k$ (and thus $A$) cannot be inside $\inter{H}$ since this area is finite. Now suppose that there exists $0\leq k' \leq |H|-1$ such that $H_{k'}$ is on $l^k$. Since $H$ is hole there exists $i\leq k' \leq j$ such that $H_{i,\ldots, j}$ is a positive arc of glue column $c$. If we have $\min\{y_{H_i}, y_{H_j}\} \leq \min\{y_{A_0}, y_{A_{|A|-1}}\} \leq  \max\{y_{A_0}, y_{A_{|A|-1}}\} \leq \max\{y_{H_i}, y_{H_j}\}$ then the lemma is true\footnote{Note that we are looking for a necessary condition not a sufficient one.}. Otherwise, if $\max\{y_{H_i}, y_{H_j}\} \leq \min\{y_{A_0}, y_{A_{|A|-1}}\}$, it is possible to modify the ray $l^k$ such that it turns around this obstacle by going north without crossing glue column $c$, see Figure \ref{fig:Arc:ArcInHole}. A similar reasoning can be done if  $\min\{y_{H_i}, y_{H_j}\} \geq \max\{y_{A_0}, y_{A_{|A|-1}}\}$. Finally, if $\min\{y_{A_0}, y_{A_{|A|-1}}\} < y_{H_i}  < \max\{y_{A_0}, y_{A_{|A|-1}}\}$ then there are two possibilities. The first one is that $H_{i,\ldots,j}$ crosses through $A$. In this case at least one glue of $A$ is not inside $\inter{H}$ and $A$ cannot be inside $\inter{H}$. The other case is that $H_k$ is a tile of $A$ and the start-point of $l^k$, see Figure \ref{fig:Arc:ArcInHole}, and we also have $\min\{y_{A_0}, y_{A_{|A|-1}}\} < y_{H_j} < \max\{y_{A_0}, y_{A_{|A|-1}}\}$. Nevertheless, the arc $H_{i,\ldots,j}$ is not sufficient to show that $A$ is inside $\inter{H}$. 

\begin{figure}
\center
\begin{tikzpicture}[x=0.22cm,y=0.22cm]

\fill[fill=lightblue, draw opacity=0.8]  (40,30.5) -|  (36.5,24.5) -| (50.5,20.5)-| (44.5,17.5)-| (35.5,11.5) -| (54.5,30.5)-| (59.5,6.5) -| (35.5,2.5)-| (64.5,38.5)-| (40,38.5);

\draw[very thick]  (40.5,30.5) -|  (36.5,24.5) -| (50.5,20.5)-| (44.5,17.5)-| (35.5,11.5) -| (54.5,30.5)-| (59.5,6.5) -| (35.5,2.5)-| (64.5,38.5)-| (39.5,38.5);
\draw[thick] (39.5,27.5) -| (50.5,14.5)-| (39.5,14.5);



\tile{40}{30}{85}
\tile{39}{30}{85}
\tile{38}{30}{85}
\tile{37}{30}{85}
\tile{36}{30}{85}
\tile{36}{29}{85}
\tile{36}{28}{85}
\tile{36}{27}{85}
\tile{36}{26}{85}
\tile{36}{25}{85}
\tile{36}{24}{85}
\tile{37}{24}{85}
\tile{38}{24}{85}

\tilelb{39}{24}{85}
\tilelb{40}{24}{85}
\tilelb{41}{24}{85}
\tilelb{42}{24}{85}
\tilelb{43}{24}{85}
\tilelb{44}{24}{85}
\tilelb{45}{24}{85}
\tilelb{46}{24}{85}
\tilelb{47}{24}{85}
\tilelb{48}{24}{85}
\tilelb{49}{24}{85}
\tilelb{50}{24}{85}
\tilelb{50}{23}{85}
\tilelb{50}{22}{85}
\tilelb{50}{21}{85}
\tilelb{50}{20}{85}
\tilelb{49}{20}{85}
\tilelb{48}{20}{85}
\tilelb{47}{20}{85}
\tilelb{46}{20}{85}
\tilelb{45}{20}{85}
\tilelb{44}{20}{85}
\tilelb{44}{19}{85}
\tilelb{44}{18}{85}
\tilelb{44}{17}{85}
\tilelb{43}{17}{85}
\tilelb{42}{17}{85}
\tilelb{41}{17}{85}
\tilelb{40}{17}{85}
\tilelb{39}{17}{85}
\tile{38}{17}{85}
\tile{37}{17}{85}
\tile{36}{17}{85}
\tile{35}{17}{85}
\tile{35}{16}{85}
\tile{35}{15}{85}
\tile{35}{14}{85}
\tile{35}{13}{85}
\tile{35}{12}{85}
\tile{35}{11}{85}
\tile{36}{11}{85}
\tile{37}{11}{85}
\tile{38}{11}{85}
\tilemb{39}{11}{85}
\tilemb{40}{11}{85}
\tilemb{41}{11}{85}
\tilemb{42}{11}{85}
\tilemb{43}{11}{85}
\tilemb{44}{11}{85}
\tilemb{45}{11}{85}
\tilemb{46}{11}{85}
\tilemb{47}{11}{85}
\tilemb{48}{11}{85}
\tilemb{49}{11}{85}
\tilemb{50}{11}{85}
\tilemb{51}{11}{85}
\tilemb{52}{11}{85}
\tilemb{53}{11}{85}
\tilemb{54}{11}{85}
\tilemb{54}{12}{85}
\tilemb{54}{13}{85}
\tilemb{54}{14}{85}
\tilemb{54}{15}{85}
\tilemb{54}{16}{85}
\tilemb{54}{17}{85}
\tilemb{54}{18}{85}
\tilemb{54}{19}{85}
\tilemb{54}{20}{85}
\tilemb{54}{21}{85}
\tilemb{54}{22}{85}
\tilemb{54}{23}{85}
\tilemb{54}{24}{85}
\tilemb{54}{25}{85}
\tilemb{54}{26}{85}
\tilemb{54}{27}{85}
\tilemb{54}{28}{85}
\tilemb{54}{29}{85}
\tilemb{54}{30}{85}
\tilemb{55}{30}{85}
\tilemb{56}{30}{85}
\tilemb{57}{30}{85}
\tilemb{58}{30}{85}
\tilemb{59}{30}{85}
\tilemb{59}{29}{85}
\tilemb{59}{28}{85}
\tilemb{59}{27}{85}
\tilemb{59}{26}{85}
\tilemb{59}{25}{85}
\tilemb{59}{24}{85}
\tilemb{59}{23}{85}
\tilemb{59}{22}{85}
\tilemb{59}{21}{85}
\tilemb{59}{20}{85}
\tilemb{59}{19}{85}
\tilemb{59}{18}{85}
\tilemb{59}{17}{85}
\tilemb{59}{16}{85}
\tilemb{59}{15}{85}
\tilemb{59}{14}{85}
\tilemb{59}{13}{85}
\tilemb{59}{12}{85}
\tilemb{59}{11}{85}
\tilemb{59}{10}{85}
\tilemb{59}{9}{85}
\tilemb{59}{8}{85}
\tilemb{59}{7}{85}
\tilemb{59}{6}{85}
\tilemb{58}{6}{85}
\tilemb{57}{6}{85}
\tilemb{56}{6}{85}
\tilemb{55}{6}{85}
\tilemb{54}{6}{85}
\tilemb{53}{6}{85}
\tilemb{52}{6}{85}
\tilemb{51}{6}{85}
\tilemb{50}{6}{85}
\tilemb{49}{6}{85}
\tilemb{48}{6}{85}
\tilemb{47}{6}{85}
\tilemb{46}{6}{85}
\tilemb{45}{6}{85}
\tilemb{44}{6}{85}
\tilemb{43}{6}{85}
\tilemb{42}{6}{85}
\tilemb{41}{6}{85}
\tilemb{40}{6}{85}
\tilemb{39}{6}{85}
\tile{38}{6}{85}
\tile{37}{6}{85}
\tile{36}{6}{85}
\tile{35}{6}{85}
\tile{35}{5}{85}
\tile{35}{4}{85}
\tile{35}{3}{85}
\tile{35}{2}{85}
\tile{36}{2}{85}
\tile{37}{2}{85}
\tile{38}{2}{85}
\tileb{39}{2}{85}
\tileb{40}{2}{85}
\tileb{41}{2}{85}
\tileb{42}{2}{85}
\tileb{43}{2}{85}
\tileb{44}{2}{85}
\tileb{45}{2}{85}
\tileb{46}{2}{85}
\tileb{47}{2}{85}
\tileb{48}{2}{85}
\tileb{49}{2}{85}
\tileb{49}{2}{85}
\tileb{50}{2}{85}
\tileb{51}{2}{85}
\tileb{52}{2}{85}
\tileb{53}{2}{85}
\tileb{54}{2}{85}
\tileb{55}{2}{85}
\tileb{56}{2}{85}
\tileb{57}{2}{85}
\tileb{58}{2}{85}
\tileb{59}{2}{85}
\tileb{60}{2}{85}
\tileb{61}{2}{85}
\tileb{62}{2}{85}
\tileb{63}{2}{85}
\tileb{64}{2}{85}
\tileb{64}{3}{85}
\tileb{64}{4}{85}
\tileb{64}{5}{85}
\tileb{64}{6}{85}
\tileb{64}{7}{85}
\tileb{64}{8}{85}
\tileb{64}{9}{85}
\tileb{64}{10}{85}
\tileb{64}{11}{85}
\tileb{64}{12}{85}
\tileb{64}{13}{85}
\tileb{64}{14}{85}
\tileb{64}{15}{85}
\tileb{64}{16}{85}
\tileb{64}{17}{85}
\tileb{64}{18}{85}
\tileb{64}{19}{85}
\tileb{64}{20}{85}
\tileb{64}{21}{85}
\tileb{64}{22}{85}
\tileb{64}{23}{85}
\tileb{64}{24}{85}
\tileb{64}{25}{85}
\tileb{64}{26}{85}
\tileb{64}{27}{85}
\tileb{64}{28}{85}
\tileb{64}{29}{85}
\tileb{64}{30}{85}
\tileb{64}{31}{85}
\tileb{64}{32}{85}
\tileb{64}{33}{85}
\tileb{64}{34}{85}
\tileb{64}{35}{85}
\tileb{64}{36}{85}
\tileb{64}{37}{85}
\tileb{64}{38}{85}
\tileb{63}{38}{85}
\tileb{62}{38}{85}
\tileb{61}{38}{85}
\tileb{60}{38}{85}
\tileb{59}{38}{85}
\tileb{58}{38}{85}
\tileb{57}{38}{85}
\tileb{56}{38}{85}
\tileb{55}{38}{85}
\tileb{54}{38}{85}
\tileb{53}{38}{85}
\tileb{52}{38}{85}
\tileb{51}{38}{85}
\tileb{50}{38}{85}
\tileb{49}{38}{85}
\tileb{48}{38}{85}
\tileb{47}{38}{85}
\tileb{46}{38}{85}
\tileb{45}{38}{85}
\tileb{44}{38}{85}
\tileb{43}{38}{85}
\tileb{42}{38}{85}
\tileb{41}{38}{85}
\tileb{40}{38}{85}
\tileb{39}{38}{85}

\doty{39}{27}{85}
\doty{40}{27}{85}
\doty{41}{27}{85}
\doty{42}{27}{85}
\doty{43}{27}{85}
\doty{44}{27}{85}
\doty{45}{27}{85}
\doty{46}{27}{85}
\doty{47}{27}{85}
\doty{48}{27}{85}
\doty{49}{27}{85}
\doty{50}{27}{85}
\doty{50}{26}{85}
\doty{50}{25}{85}
\doty{50}{24}{85}
\doty{50}{23}{85}
\doty{50}{22}{85}
\doty{50}{21}{85}
\doty{50}{20}{85}
\doty{50}{19}{85}
\doty{50}{18}{85}
\doty{50}{17}{85}
\doty{50}{16}{85}
\doty{50}{15}{85}
\doty{50}{14}{85}
\doty{49}{14}{85}
\doty{48}{14}{85}
\doty{47}{14}{85}
\doty{46}{14}{85}
\doty{45}{14}{85}
\doty{44}{14}{85}
\doty{43}{14}{85}
\doty{42}{14}{85}
\doty{41}{14}{85}
\doty{40}{14}{85}
\doty{39}{14}{85}

\path [dotted, draw, thin] (33,0) grid[step=0.22cm] (70,41);

\draw [dashed] (40,0) -| (40,41);
\draw [thick, color=blue] (40,38.5) -| (40,30.5);

\draw [thick, color=red] (50.5,20.5) -| (54.2,30.8) -| (59.8,20.5) -- (65,20.5);
\draw [thick, dashed, color=red,->] (64,20.5) -- (70,20.5);

\fill (40.5,30.5) circle (0.16);
\node (D) at (42.2,30.5) {$\Hole_0$};

\fill (39.5,24.5) circle (0.16);
\node (D) at (39.5,23) {$\Hole_i$};

\fill (39.5,17.5) circle (0.16);
\node (D) at (39.5,19) {$\Hole_{j}$};

\fill (39.5,11.5) circle (0.16);
\node (D) at (39.5,10) {$\Hole_{i'}$};

\fill (39.5,6.5) circle (0.16);
\node (D) at (39.5,7.7) {$\Hole_{j'}$};

\fill (39.5,2.5) circle (0.16);
\node (D) at (39.5,0.8) {$\Hole_{i''}$};

\fill (50.5,20.5) circle (0.16);
\node (D) at (52.2,19.2) {$A_k$};

\fill (39.5,38.5) circle (0.16);
\node (D) at (36.2,38.5) {$\Hole_{|\Hole|-1}$};

\end{tikzpicture}
\caption{Proof of Lemma \ref{lem:next:order}: consider a hole $H$ (the rectangle tiles) of glue column $c$ and an arc $A$ (the yellow dots) of glue column $c$. In this example, $A$ is inside $\inter{H}$ (the area colored in light blue). The arc $H_{i,\ldots,j}$ (in light blue) intersects with $A$ but it is not sufficient by itself to show that $A$ is inside $\inter{H}$. Consider the red ray starting at $A_k$ and on going west. This ray can turn around the arc $H_{i'\ldots,j'}$ (in blue) but it cannot turn around $H_{i'',\ldots,|H|-1}$ (in dark blue) without crossing glue column $c$.}
\label{fig:Arc:ArcInHole}
\end{figure}

\end{proof}

\subsection{Arcs of a path}
\label{Sec:arc:paths}

Consider a producible path $P$ whose last tile if the easternmost one, a glue column $\Bone\leq c \leq \Btwo$ 
and two indices $0\leq i < j \leq |P|-1$. We say that $P_{i,\ldots,j}$ is an \emph{arc} of $P$ on glue column $c$ if and only if $P_{i,\ldots,j}$ is a positive arc on glue column $c$. Indeed, with the chosen conventions, we focus mainly on the east side of glue column $c$ and thus on the positive arc. Note that since $P_{i,\ldots,j}$ is positive then $\glueP{i}{i+1}$ must point east and $\glueP{j}{j+1}$ must point west.

Now we introduce some notions to compare and to order the different arcs of $P$, see Figure \ref{fig:Arc:Path} for illustrations.
If $P_{i,\ldots,j}$ is an upward (resp. downward) arc, the \emph{next} glue of $P$ according to arc $P_{i,\ldots,j}$ is defined as the southernmost (resp. northernmost) glue of $P$ among the glues which are on glue column $c$ and strictly north (resp. south) of $\glueP{j-1}{j}$\footnote{Note that the next glue is always correctly defined since $\glueP{j-1}{j}$ points west and thus cannot be one of the two glues of $P$ which are visible from the north or south since they both point east by Lemma \ref{lem:glue:prop2} and Lemma \ref{Uturn:glueWest}.},\emph{i.e.} the next glue is on glue column $c$ and its $y$-coordinate is:

\begin{figure}
\center
\begin{tikzpicture}[x=0.22cm,y=0.22cm]

\fill[fill=orange!30!white, draw opacity=0.8] (8,4.5) -| (13.5,13.5) -| (8,13.5);
\fill[fill=yellow!30!white, draw opacity=0.8] (8,7.5) -| (10.5,10.5) -| (8,10.5);
\fill[fill=orange!30!white, draw opacity=0.8] (25,4.5) -| (29.5,9.5) -| (25,9.5);
\fill[fill=orange!30!white, draw opacity=0.8] (25,11.5) -| (29.5,16.5) -| (25,16.5);
\fill[fill=orange!30!white, draw opacity=0.8] (41,4.5) -| (45.5,9.5) -| (41,9.5);
\fill[fill=blue!30!white, draw opacity=0.8] (41,13.5) -| (45.5,18.5) -| (41,13.5);

\draw[very thick] (7.5,4.5) -| (13.5,13.5) -| (5.5,10.5) -| (10.5,7.5) -| (2.5,16.5) -| (15.5,16.5);
\draw[very thick] (24.5,4.5) -| (29.5,9.5) -| (22.5,11.5) -| (29.5,16.5) -| (22.5,18.5) -| (31.5,18.5);
\draw[very thick] (40.5,4.5) -| (45.5,9.5)-| (36.5,18.5) -| (45.5,13.5)-| (38.5,11.5) -| (48.5,11.5);



\tileor{7}{4}{85}
\tileor{8}{4}{85}
\tileor{9}{4}{85}
\tileor{10}{4}{85}
\tileor{11}{4}{85}
\tileor{12}{4}{85}
\tileor{13}{4}{85}
\tileor{13}{5}{85}
\tileor{13}{6}{85}
\tileor{13}{7}{85}
\tileor{13}{8}{85}
\tileor{13}{9}{85}
\tileor{13}{10}{85}
\tileor{13}{11}{85}
\tileor{13}{12}{85}
\tileor{13}{13}{85}
\tileor{12}{13}{85}
\tileor{11}{13}{85}
\tileor{10}{13}{85}
\tileor{9}{13}{85}
\tileor{8}{13}{85}
\tileor{7}{13}{85}
\tile{6}{13}{85}
\tile{5}{13}{85}
\tile{5}{12}{85}
\tile{5}{11}{85}
\tile{5}{10}{85}
\tile{6}{10}{85}
\tiley{7}{10}{40}
\tiley{8}{10}{85}
\tiley{9}{10}{85}
\tiley{10}{10}{85}
\tiley{10}{9}{85}
\tiley{10}{8}{85}
\tiley{10}{7}{85}
\tiley{9}{7}{85}
\tiley{8}{7}{85}
\tiley{7}{7}{85}
\tile{6}{7}{85}
\tile{5}{7}{85}
\tile{4}{7}{85}
\tile{3}{7}{85}
\tile{2}{7}{85}
\tile{2}{8}{85}
\tile{2}{9}{85}
\tile{2}{10}{85}
\tile{2}{11}{85}
\tile{2}{12}{85}
\tile{2}{13}{85}
\tile{2}{14}{85}
\tile{2}{15}{85}
\tile{2}{16}{85}
\tile{3}{16}{85}
\tile{4}{16}{85}
\tile{5}{16}{85}
\tile{6}{16}{85}
\tile{7}{16}{85}
\tile{8}{16}{85}
\tile{9}{16}{85}
\tile{10}{16}{85}
\tile{11}{16}{85}
\tile{12}{16}{85}
\tile{13}{16}{85}
\tile{14}{16}{85}
\tile{15}{16}{85}

\tiley{24}{4}{85}
\tiley{25}{4}{85}
\tiley{26}{4}{85}
\tiley{27}{4}{85}
\tiley{28}{4}{85}
\tiley{29}{4}{85}
\tiley{29}{5}{85}
\tiley{29}{6}{85}
\tiley{29}{7}{85}
\tiley{29}{8}{85}
\tiley{29}{9}{85}
\tiley{28}{9}{85}
\tiley{27}{9}{85}
\tiley{26}{9}{85}
\tiley{25}{9}{85}
\tiley{24}{9}{85}
\tile{23}{9}{85}
\tile{22}{9}{85}
\tile{22}{9}{85}
\tile{22}{10}{85}
\tile{22}{11}{85}
\tile{23}{11}{85}
\tileor{24}{11}{85}
\tileor{25}{11}{85}
\tileor{26}{11}{85}
\tileor{27}{11}{85}
\tileor{28}{11}{85}
\tileor{29}{11}{85}
\tileor{29}{12}{85}
\tileor{29}{13}{85}
\tileor{29}{14}{85}
\tileor{29}{15}{85}
\tileor{29}{16}{85}
\tileor{28}{16}{85}
\tileor{27}{16}{85}
\tileor{26}{16}{85}
\tileor{25}{16}{85}
\tileor{24}{16}{85}
\tile{23}{16}{85}
\tile{22}{16}{85}
\tile{22}{17}{85}
\tile{22}{18}{85}
\tile{23}{18}{85}
\tile{24}{18}{85}
\tile{25}{18}{85}
\tile{26}{18}{85}
\tile{27}{18}{85}
\tile{28}{18}{85}
\tile{29}{18}{85}
\tile{30}{18}{85}
\tile{31}{18}{85}

\tileor{40}{4}{85}
\tileor{41}{4}{85}
\tileor{42}{4}{85}
\tileor{43}{4}{85}
\tileor{44}{4}{85}
\tileor{45}{4}{85}
\tileor{45}{5}{85}
\tileor{45}{6}{85}
\tileor{45}{7}{85}
\tileor{45}{8}{85}
\tileor{45}{9}{85}
\tileor{44}{9}{85}
\tileor{43}{9}{85}
\tileor{42}{9}{85}
\tileor{41}{9}{85}
\tileor{40}{9}{85}
\tile{39}{9}{85}
\tile{38}{9}{85}
\tile{37}{9}{85}
\tile{36}{9}{85}
\tile{36}{10}{85}
\tile{36}{11}{85}
\tile{36}{12}{85}
\tile{36}{13}{85}
\tile{36}{14}{85}
\tile{36}{15}{85}
\tile{36}{16}{85}
\tile{36}{17}{85}
\tile{36}{18}{85}
\tile{37}{18}{85}
\tile{38}{18}{85}
\tile{39}{18}{85}
\tileb{40}{18}{85}
\tileb{41}{18}{85}
\tileb{42}{18}{85}
\tileb{43}{18}{85}
\tileb{44}{18}{85}
\tileb{45}{18}{85}
\tileb{45}{17}{85}
\tileb{45}{16}{85}
\tileb{45}{15}{85}
\tileb{45}{14}{85}
\tileb{45}{13}{85}
\tileb{44}{13}{85}
\tileb{43}{13}{85}
\tileb{42}{13}{85}
\tileb{41}{13}{85}
\tileb{40}{13}{85}
\tile{39}{13}{85}
\tile{38}{13}{85}
\tile{38}{12}{85}
\tile{38}{11}{85}
\tile{39}{11}{85}
\tile{40}{11}{85}
\tile{41}{11}{85}
\tile{42}{11}{85}
\tile{43}{11}{85}
\tile{44}{11}{85}
\tile{45}{11}{85}
\tile{46}{11}{85}
\tile{47}{11}{85}
\tile{48}{11}{85}

\path [dotted, draw, thin] (0,0) grid[step=0.22cm] (55,22);

\draw [dashed] (8,0) -| (8,22);
\draw [thick, color=orange] (8,4.5) -| (8,13.5);
\draw [thick, color=yellow] (8,7.5) -| (8,10.5);

\draw [dashed] (25,0) -| (25,22);
\draw [thick, color=orange] (25,11.5) -| (25,16.5);
\draw [thick, color=orange] (25,4.5) -| (25,9.5);

\draw [dashed] (41,0) -| (41,22);
\draw [thick, color=orange] (41,4.5) -| (41,9.5);
\draw [thick, color=blue] (41,18.5) -| (41,13.5);

\fill (7.5,4.5) circle (0.16);
\node (D) at (5.7,4.5) {$P_0$};

\fill (7.5,16.5) circle (0.16);
\node (D) at (6.7,18.2) {$P_\lastc$};

\fill (24.5,4.5) circle (0.16);
\node (D) at (22.7,4.5) {$Q_0$};

\fill (24.5,18.5) circle (0.16);
\node (D) at (23.7,20.2) {$Q_\lastc$};

\fill (40.5,4.5) circle (0.16);
\node (D) at (38.7,4.5) {$R_0$};

\fill (40.5,11.5) circle (0.16);
\node (D) at (40,12.7) {$R_\lastc$};

\end{tikzpicture}
\caption{Three examples of path and their arcs. For each path we denote by $\lastc$ the index of the last glue on the relevant glue column. In the first one, the path $P$ has two arcs, the upward orange arc is dominating the downward yellow arc. The next glue according to the orange arc is $\glueP{\lastc}{\lastc+1}$.
In the second example, the path $Q$ has two dominant arcs, the upward arc in yellow is south and anterior of the upward arc in orange. The next glue according to the yellow arc is the first glue of the orange arc and the next glue according the orange arc is $Q_{\lastc}$. In the third example, the path $R$ has two dominant arcs, the upward arc in orange is south and anterior of the downward arc in blue. The next glue of both these arcs is $\glueR{\lastc}{\lastc+1}$.}
\label{fig:Arc:Path}
\end{figure}

\begin{align*}
\min\{y:\text{there } & \text{exists $0\leq k\leq |P|-1$ such that $\glueP{k}{k+1}$ is} \\ & \text{on glue column $c$, $y_{P_k}=y$ and $y>y_{P_j}$ (resp. $y<y_{P_j}$)}\}.
\end{align*}
Now, consider another arc $P_{i',\ldots,j'}$ on glue column $c$. We say that $P_{i',\ldots,j'}$ is \emph{north} of $P_{i,\ldots,j}$ (and $P_{i,\ldots,j}$ is \emph{south }of $P_{i',\ldots,j'}$) if and only if $\max\{y_{P_{i}},y_{P_{j}}\}<\min\{y_{P_{i'}},y_{P_{j'}}\}.$ Now, if $$\min\{y_{P_{i}},y_{P_{j}}\}<\min\{y_{P_{i'}},y_{P_{j'}}\}<\max\{y_{P_{i'}},y_{P_{j'}}\}<\max\{y_{P_{i}},y_{P_{j}}\},$$  we say that the arc $P_{i,\ldots,j}$ \emph{dominates} $P_{i',\ldots,j'}$. If an arc $P$ is not dominated by another arc of $P$, we say that this arc is \emph{dominant}. Note that the case $$\min\{y_{P_{i}},y_{P_{j}}\}<\min\{y_{P_{i'}},y_{P_{j'}}\}<\max\{y_{P_{i}},y_{P_{j}}\}<\max\{y_{P_{i'}},y_{P_{j'}}\}$$ is not possible otherwise the arcs $P_{i,\ldots,j}$ and $P_{i',\ldots,j'}$ would intersect, contradicting the fact that $P$ is a path (which is  self-avoiding by definition). Thus, when considering two dominant arcs, one has to be north of the other one.  Finally, if $j< i'$, we say that the arc $P_{i,\ldots,j}$ is \emph{anterior} to the arc $P_{i',\ldots,j'}$ otherwise $P_{i,\ldots,j}$ is \emph{posterior} to the arc $P_{i',\ldots,j'}$. 

To summarize, the notion of anteriority/posteriority induces a total order on the arcs of $P$ (according to their appearances in $P$). The relation of domination induces a partial order on the arcs of $P$ and the notion of south/north induces a total order on the dominant arcs of $P$. The main result of the following section is to show that the order induced by the anteriority/posteriority restricted to the upward dominant arcs (resp. downward dominant arcs) is equivalent to the order induced by the notion of south/north.

\subsection{Properties of dominant arcs}
\label{Sec:arc:dominant}

The next subsection focus on analyzing the relationship between the three orders introduced in Subsection \ref{Sec:arc:paths}. To do so, we consider an arc $P_{i,\ldots,j}$ of $P$ and we focus on its next $\glueP{i'}{i'+1}$, the first step is to show that $i<j<i'$.

\begin{lemma}
\label{lem:next:order}
Consider a producible path $P$ whose last tile is the easternmost one, a glue column $c$ with $\Bone \leq c \leq \Btwo$, two indices $0<i<j<|P|-1$ such that $P_{i,\ldots,j}$ is a dominant arc of $P$ on glue column $c$ and the index $i'$ of the next glue of $P$ according to this arc. Then, we have $i<j<i'$.
\end{lemma}

\begin{proof}
For the sake of contradiction suppose that $i' \leq j$. Then $i'<i$ since $\glueP{i'}{i'+1}$ is neither $\glueP{i}{i+1}$ nor $\glueP{j}{j+1}$ and no other glue of the arc $P_{i,\ldots,j}$ is on glue column $c$. Consider the path $\Hole=P_{i',\ldots,j}$. By definition of $i'$ as the index of the next glue of $P$ according to arc $P_{i,\ldots,j}$ then $\Hole$ is a hole, see Figure \ref{fig:Arc:Order}. Let $\cork{\Hole}$ be the frontier of this hole and $\inter{\Hole}$ be its interior. We remind that $\inter{\Hole}$ is the finite area delimited by the cut of the $2D$ plane done by $\cork{\Hole}$ and  $\Hole$. We consider that the curve generated by $\cork{\Hole}$ and  $\Hole$ has the same orientation as $\Hole$. Since $P_{i,\ldots,j}$ is dominant then if $P_{i,\ldots,j}$ is an upward (resp. downward) arc then $\inter{\Hole}$ is the left (resp. right) side of the cut done by $\cork{\Hole}$ and  $P_{i,\ldots,j}$. By definition of an arc of $P$, $P_{i,\ldots,j}$ is a positive arc for glue column $c$. Thus, $\glueP{j-1}{j}$ points west implying that $P_{i,\ldots,j}$ turns left (resp. right) of the curve generated by $P_{i,\ldots,j}$ and $\cork{\Hole}$ at the position of $\glueP{j-1}{j}$. Thus, the tile $P_{j}$ is inside $\inter{\Hole}$. By hypothesis, the last tile of $P$ is the easternmost one implying that $P_{|P|-1}$ cannot be inside the finite area $\inter{\Hole}$ then $P_{j,\ldots,|P|-1}$ must intersect with either $P_{k,\ldots,j-1}$ contradicting the fact that $P$ is simple or it must cross $\cork{\Hole}$ contradicting the definition of $i'$. Then $j<i'$.
\begin{figure}
\center
\begin{tikzpicture}[x=0.22cm,y=0.22cm]

\fill[fill=lightblue, draw opacity=0.8] (41,12.5) -| (48.5,4.5)-| (33.5,20.5)-| (41,20.5);

\draw[very thick] (40,12.5) -| (48.5,4.5)-| (33.5,20.5)-| (42,20.5);



\tile{39}{4}{85}
\tile{38}{4}{85}
\tile{37}{4}{85}
\tile{36}{4}{85}
\tile{35}{4}{85}
\tile{34}{4}{85}
\tile{33}{4}{85}
\tile{33}{5}{85}
\tile{33}{6}{85}
\tile{33}{7}{85}
\tile{33}{8}{85}
\tile{33}{9}{85}
\tile{33}{10}{85}
\tile{33}{11}{85}
\tile{33}{12}{85}
\tile{33}{13}{85}
\tile{33}{14}{85}
\tile{33}{15}{85}
\tile{33}{16}{85}
\tile{33}{17}{85}
\tile{33}{18}{85}
\tile{33}{19}{85}
\tile{33}{20}{85}
\tile{34}{20}{85}
\tile{35}{20}{85}
\tile{36}{20}{85}
\tile{37}{20}{85}
\tile{38}{20}{85}
\tile{39}{20}{85}

\tiley{40}{4}{85}
\tiley{41}{4}{85}
\tiley{42}{4}{85}
\tiley{43}{4}{85}
\tiley{44}{4}{85}
\tiley{45}{4}{85}
\tiley{46}{4}{85}
\tiley{47}{4}{85}
\tiley{48}{4}{85}
\tiley{48}{5}{85}
\tiley{48}{6}{85}
\tiley{48}{7}{85}
\tiley{48}{8}{85}
\tiley{48}{9}{85}
\tiley{48}{10}{85}
\tiley{48}{11}{85}
\tiley{48}{12}{85}
\tiley{47}{12}{85}
\tiley{46}{12}{85}
\tiley{45}{12}{85}
\tiley{44}{12}{85}
\tiley{43}{12}{85}
\tiley{42}{12}{85}
\tiley{41}{12}{85}
\tiley{40}{12}{85}

\tileor{41}{20}{85}
\tileor{40}{20}{85}

\path [dotted, draw, thin] (27,0) grid[step=0.22cm] (55,25);

\draw [dashed] (41,0) -| (41,25);
\draw [thick, color=blue] (41,12.5) -| (41,20.5);

\draw [thick, color=blue, ->] (48.5,9) -- (45,9);
\draw [thick, color=blue, ->] (40.2,12.5) -- (38,12.5);

\fill (40.5,4.5) circle (0.16);
\node (D) at (39.7,2.7) {$P_i$};

\fill (40.5,12.5) circle (0.16);
\node (D) at (39.7,14.2) {$P_j$};

\node (D) at (37.3,8.7) {$\inter{\Hole}$};

\fill (41.5,20.5) circle (0.16);
\node (D) at (45.7,20.5) {$\Hole_0=P_{i'}$};

\end{tikzpicture}
\caption{Proof of Lemma \ref{lem:next:order}: $P_{i,\ldots,j}$ is an arc (in yellow) of a path $P$. The $\glueP{i'}{i'+1}$ (in orange) is the next glue of $P$ according to the arc $P_{i,\ldots,j}$. For the sake of contradiction, we assume $i'<i<j$ and then $\Hole=P_{i',\ldots,j}$ is a hole. Its interior $\inter{\Hole}$ is in light blue. In this example, $P_{i,\ldots,j}$ is an upward arc and then $\inter{H}$ is the left side of the curve generated $P_{i,\ldots,j}$ and $\cork{\Hole}$. Since $P_{i,\ldots,j}$ is a positive arc for $c$ then $\glueP{j-1}{j}$ points west. Then $P_j$ is inside $\inter{\Hole}$. This is a contradiction since the last tile of $P$ is the easternmost one and cannot be inside $\inter{\Hole}$.}
\label{fig:Arc:Order}
\end{figure}
\end{proof}

Now, we distinguished three cases. The first one is when $i'$ is the index of the last glue $\lastc$ of $P$ on glue column $c$. This case is possible, see for example the orange arc of path $Q$ of Figure \ref{fig:Arc:Path}. Note that, $\glueP{i'}{i'+1}$ points east in this case. Otherwise $i'<\lastc$ and if $P_{i,\ldots,j}$ is an upward (resp. downward) arc of $P$, the second case is when $\glueP{i'}{i'+1}$ is pointing east. In this case there exists $j'>i'$ such that $P_{i',\ldots,j'}$ is a dominant arc of $P$ which is north (resp. south) and posterior of $P_{i,\ldots,j}$. This case is possible, see for example the yellow arc of path $Q$ of Figure \ref{fig:Arc:Path} 

\begin{lemma}
\label{decompo:next}
Consider a producible path $P$ whose last tile is the easternmost one, a glue column $\Bone \leq c \leq \Btwo$ and let $0\leq \lastc < |P|-1$ be the index of the last glue of $P$ on glue column $c$. Consider two indices $0<i<j<|P|-1$ such that $P_{i,\ldots,j}$ is an upward (resp. downward) dominant arc on glue column $c$ and the index of the next glue $i'$ of $P$ according to this arc. Then, if  $i'<\lastc$ and $\glueP{i'}{i'+1}$ points east then there exists $i'<j'<\lastc$ such that $P_{i',\ldots,j'}$ is a dominant arc of $P$ which is north (resp. south) and posterior of $P_{i,\ldots,j}$.
\end{lemma}

\begin{proof}
We consider that $P_{i,\ldots,j}$ is an upward arc, the other case is symmetric, see Figure \ref{fig:Arc:ProofDom}. Consider $j'=\min\{k>i'+1: \glueP{k-1}{k} \text{ is on glue column }c\}$. By definition of $j'$, $P_{i',\ldots,j'}$ is an arc of glue column $c$. By hypothesis, $\glueP{i'}{i'+1}$ points east implying that $P_{i',\ldots,j'}$ is a positive arc and thus an arc of $P$. By Lemma \ref{lem:next:order}, $P_{i',\ldots,j'}$ is posterior of $P_{i,\ldots,j}$.
\begin{figure}
\center
\begin{tikzpicture}[x=0.22cm,y=0.22cm]

\fill[fill=yellow!30!white, draw opacity=0.8] (41,4.5) -| (45.5,9.5) -| (41,9.5);
\fill[fill=orange!30!white, draw opacity=0.8] (41,14.5) -| (45.5,19.5) -| (41,14.5);

\draw[very thick] (40.5,4.5) -| (45.5,9.5)-| (36.5,14.5) -| (45.5,19.5) -| (40.5,19.5);
\draw[very thick] (40.5,23.5) -| (51.5,12.5);



\tiley{40}{4}{85}
\tiley{41}{4}{85}
\tiley{42}{4}{85}
\tiley{43}{4}{85}
\tiley{44}{4}{85}
\tiley{45}{4}{85}
\tiley{45}{5}{85}
\tiley{45}{6}{85}
\tiley{45}{7}{85}
\tiley{45}{8}{85}
\tiley{45}{9}{85}
\tiley{44}{9}{85}
\tiley{43}{9}{85}
\tiley{42}{9}{85}
\tiley{41}{9}{85}
\tiley{40}{9}{85}
\tile{39}{9}{85}
\tile{38}{9}{85}
\tile{37}{9}{85}
\tile{36}{9}{85}
\tile{36}{10}{85}
\tile{36}{11}{85}
\tile{36}{12}{85}
\tile{36}{13}{85}
\tile{36}{14}{85}
\tile{37}{14}{85}
\tile{38}{14}{85}
\tile{39}{14}{85}

\tileor{40}{19}{85}
\tileor{41}{19}{85}
\tileor{42}{19}{85}
\tileor{43}{19}{85}
\tileor{44}{19}{85}
\tileor{45}{19}{85}
\tileor{45}{18}{85}
\tileor{45}{17}{85}
\tileor{45}{16}{85}
\tileor{45}{15}{85}
\tileor{45}{14}{85}
\tileor{44}{14}{85}
\tileor{43}{14}{85}
\tileor{42}{14}{85}
\tileor{41}{14}{85}
\tileor{40}{14}{85}

\tile{39}{14}{85}

\tiler{40}{23}{85}
\tiler{41}{23}{85}
\tiler{42}{23}{85}
\tiler{43}{23}{85}
\tiler{44}{23}{85}
\tiler{45}{23}{85}
\tiler{46}{23}{85}
\tiler{47}{23}{85}
\tiler{48}{23}{85}
\tiler{49}{23}{85}
\tiler{50}{23}{85}
\tiler{51}{23}{85}
\tiler{51}{22}{85}
\tiler{51}{21}{85}
\tiler{51}{20}{85}
\tiler{51}{19}{85}
\tiler{51}{18}{85}
\tiler{51}{17}{85}
\tiler{51}{16}{85}
\tiler{51}{15}{85}
\tiler{51}{14}{85}
\tiler{51}{13}{85}
\tiler{51}{12}{85}

\path [dotted, draw, thin] (33,-1) grid[step=0.22cm] (55,25);

\draw [dashed, thick, color=red] (41,19.5) -| (41,25);

\draw [dashed, thick, color=green!80!black] (41,9.5) -| (41,14.5);

\draw [dashed, thick, color=blue] (41,-1) -| (41,4.5);

\draw [thick, color=yellow] (41,4.5) -| (41,9.5);
\draw [thick, color=orange] (41,19.5) -| (41,14.5);

\fill (40.5,4.5) circle (0.16);
\node (D) at (38.7,4.5) {$P_i$};

\fill (40.5,9.5) circle (0.16);
\node (D) at (39.7,11.2) {$P_j$};

\fill (40.5,14.5) circle (0.16);
\node (D) at (39.7,16.2) {$P_{i'}$};

\fill (40.5,19.5) circle (0.16);
\node (D) at (38.7,19.5) {$P_{j'}$};

\fill (40.5,23.5) circle (0.16);
\node (D) at (38.7,23.5) {$A_0$};

\node (D) at (51.5,10) {\textcolor{red}{\Huge$?$}};

\end{tikzpicture}
\caption{Proof of Lemma \ref{decompo:next}. Consider an upward dominant arc $P_{i,\ldots,j}$ of a path $P$ on glue column $c$ and the next $\glueP{i'}{i'+1}$ of $P$ according to this arc. If $i'$ is not the index of the last glue of $P$ on glue column $c$ then there exists $j'>i'$ such that $P_{i',\ldots,j'}$ is an arc. This arc is posterior to $P_{i,\ldots,j}$ since $i<j<i'$ by Lemma \ref{lem:next:order}. Also, $P_{i',\ldots,j'}$ must be an upward arc otherwise it would dominate $P_{i,\ldots,j}$. Now if an arc $A$ of $P$ dominates $P_{i',\ldots,j'}$ then one of its extremities is on the red part of glue column $c$. The other extremity must either be on the green line (which correspond to the line $b$ in the proof) contradicting the definition of $i'$ or the yellow line (contradicting that $P$ is simple) or the blue ray (contradicting that $P_{i,\ldots,j}$ is dominant).}
\label{fig:Arc:ProofDom}
\end{figure}

Let $b$ be the vertical line between $\glueP{j-1}{j}$ and $\glueP{i'}{i'+1}$. By definition of $i'$ the only intersections between $P$ and $b$ are $\glueP{j-1}{j}$ and $\glueP{i'}{i'+1}$. Since $P_{i,\ldots,j}$ is an upward arc, we have $y_{P_{i'}}>y_{P_{j}}>y_{P_{i}}$. 
Suppose that an arc $A$ of $P$ dominates $P_{i',\ldots,j'}$ then we have $\min\{y_{A_0},y_{A_{|A|-1}}\}<y_{P_{i'}}<\max\{y_{A_0},y_{A_{|A|-1}}\}$. Since $A$ cannot intersect $b$ and $P_{i,\ldots,j}$ then we have $\max\{y_{A_0},y_{A_{|A|-1}}\}>y_{P_{i'}}>y_{P_{j}}>y_{P_{i}}>\min\{y_{A_0},y_{A_{|A|-1}}\}$. In this case, $A$ would also dominates the arc $P_{i,\ldots,j}$ of $P$ which is a contradiction. Then $P_{i',\ldots,j'}$ is a dominant arc of $P$ on glue column $c$. 
Finally, if $P_{i',\ldots,j'}$ is a downward arc then since $P_{i'+1,\ldots,j'}$ cannot intersect with $b$ and $P_{i,\ldots,j}$ and we would have $y_{P_{i'}}>y_{P_{j}}>y_{P_{i}}>y_{P_{j'}}$. In other words, $P_{i',\ldots,j'}$ would dominates $P_{i,\ldots,j}$ which is a contradiction. Then, $P_{i',\ldots,j'}$ is an upward arc of $P$ and we have $y_{P_{j'}}>y_{P_{i'}}>y_{P_{j}}>y_{P_{i}}$, \emph{i.e. } $P_{i',\ldots,j'}$ is north of $P_{i,\ldots,j}$.
\end{proof}

Finally the last case is when $\glueP{i'}{i'+1}$ points west. This case leads to a contradiction and is impossible.


\begin{lemma}
\label{decompo:direction}
Consider a producible path $P$ whose last tile is the easternmost one, a glue column $c$ with  $\Bone \leq c \leq \Btwo$, two indices $0<i<j<|P|-1$ such that $P_{i,\ldots,j}$ is a dominant arc of $P$ on glue column $c$ and the index $i'$ of the next glue of $P$ according to this arc. Then, $\glueP{i'}{i'+1}$ points east.
\end{lemma}

\begin{proof}
We consider that $P_{i,\ldots,j}$ is an upward arc, the other case is symmetric, see Figure \ref{fig:Arc:Direction}. For the sake of contradiction suppose that $\glueP{i'}{i'+1}$ points west and consider $j'=\max\{k<i': \glueP{k}{k+1} \text{ is on glue column }c\}$. By definition of $j'$, $P_{j',\ldots,i'+1}$ is an arc of $P$ on glue column $c$. By a reasoning similar to the one of the proof of Lemma \ref{decompo:next}, $P_{j',\ldots,i'+1}$ is a dominant downward arc. By Lemma \ref{lem:next:order}, we have $i<j<i'$ and thus $i<j<j'<i'$ since the two arcs share no tile in common. Consider the path $\Hole=P_{j-1,\ldots,i'+1}$. By definition of $i'$ as the index of the next glue of $P$ according to arc $P_{i,\ldots,j}$ then $\Hole$ is a hole. By a reasoning similar to the one of the proof of Lemma \ref{lem:next:order}, the tile $P_j$ is inside the interior $\inter{H}$ of the hole $\Hole$. This is a contradiction since the last tile of $P$ is the easternmost one.
\begin{figure}
\center
\begin{tikzpicture}[x=0.22cm,y=0.22cm]

\fill[fill=yellow!30!white, draw opacity=0.8] (41,4.5) -| (45.5,9.5) -| (41,9.5);
\fill[fill=lightblue, draw opacity=0.8] (41,14.5) -| (45.5,19.5) -| (36.5,9.5) -| (41,14.5);

\draw[very thick] (40.5,4.5) -| (45.5,9.5)-| (36.5,19.5) -| (45.5,14.5) -| (40,14.5);
\draw[very thick] (40.5,23.5) -| (51.5,12.5);



\tiley{40}{4}{85}
\tiley{41}{4}{85}
\tiley{42}{4}{85}
\tiley{43}{4}{85}
\tiley{44}{4}{85}
\tiley{45}{4}{85}
\tiley{45}{5}{85}
\tiley{45}{6}{85}
\tiley{45}{7}{85}
\tiley{45}{8}{85}
\tiley{45}{9}{85}
\tiley{44}{9}{85}
\tiley{43}{9}{85}
\tiley{42}{9}{85}
\tiley{41}{9}{85}
\tiley{40}{9}{85}
\tile{39}{9}{85}
\tile{38}{9}{85}
\tile{37}{9}{85}
\tile{36}{9}{85}
\tile{36}{10}{85}
\tile{36}{11}{85}
\tile{36}{12}{85}
\tile{36}{13}{85}
\tile{36}{14}{85}
\tile{36}{15}{85}
\tile{36}{16}{85}
\tile{36}{17}{85}
\tile{36}{18}{85}
\tile{36}{19}{85}
\tile{37}{19}{85}
\tile{38}{19}{85}
\tile{39}{19}{85}

\tileor{40}{19}{85}
\tileor{41}{19}{85}
\tileor{42}{19}{85}
\tileor{43}{19}{85}
\tileor{44}{19}{85}
\tileor{45}{19}{85}
\tileor{45}{18}{85}
\tileor{45}{17}{85}
\tileor{45}{16}{85}
\tileor{45}{15}{85}
\tileor{45}{14}{85}
\tileor{44}{14}{85}
\tileor{43}{14}{85}
\tileor{42}{14}{85}
\tileor{41}{14}{85}
\tileor{40}{14}{85}

\tiler{40}{23}{85}
\tiler{41}{23}{85}
\tiler{42}{23}{85}
\tiler{43}{23}{85}
\tiler{44}{23}{85}
\tiler{45}{23}{85}
\tiler{46}{23}{85}
\tiler{47}{23}{85}
\tiler{48}{23}{85}
\tiler{49}{23}{85}
\tiler{50}{23}{85}
\tiler{51}{23}{85}
\tiler{51}{22}{85}
\tiler{51}{21}{85}
\tiler{51}{20}{85}
\tiler{51}{19}{85}
\tiler{51}{18}{85}
\tiler{51}{17}{85}
\tiler{51}{16}{85}
\tiler{51}{15}{85}
\tiler{51}{14}{85}
\tiler{51}{13}{85}
\tiler{51}{12}{85}

\path [dotted, draw, thin] (33,-1) grid[step=0.22cm] (55,25);

\draw [dashed, thick, color=red] (41,19.5) -| (41,25);

\draw [thick, color=blue] (41,9.5) -| (41,14.5);
\draw [thick, color=blue,->] (40,14.5) -- (38.5,14.5);
\draw [thick, color=blue,->] (45.5,16.5) -- (44,16.5);

\draw [dashed, thick, color=green!80!black] (41,-1) -| (41,4.5);

\draw [thick, color=yellow] (41,4.5) -| (41,9.5);

\fill (40.5,4.5) circle (0.16);
\node (D) at (38.7,4.5) {$P_i$};

\fill (40.5,9.5) circle (0.16);
\node (D) at (40.5,7.7) {$P_j=\Hole_0$};

\node (D) at (43,17) {$\inter{\Hole}$};

\fill (41.5,14.5) circle (0.16);
\node (D) at (42,13) {$P_{i'}$};

\fill (40.5,19.5) circle (0.16);
\node (D) at (39.7,21) {$P_{j'}$};

\fill (40.5,23.5) circle (0.16);
\node (D) at (38.7,23.5) {$A_0$};

\node (D) at (51.5,10) {\textcolor{red}{\Huge$?$}};

\end{tikzpicture}
\caption{Proof of Lemma \ref{decompo:direction}. Consider an upward dominant arc $P_{i,\ldots,j}$ of a path $P$ on glue column $c$ and the next $\glueP{i'}{i'+1}$ of $P$ according to this arc. If $\glueP{i'}{i'+1}$ point west then there exists $j'<i'$ such that $P_{j',\ldots,i'+1}$ is an arc of $P$. Also, $P_{j',\ldots,i'+1}$ must be a downward dominant arc otherwise $P_{i,\ldots,j}$ would not be dominant (see Figure \ref{fig:Arc:ProofDom} for more details). Since $i<j<i'$ by Lemma \ref{lem:next:order}, then
 $\Hole=P_{j,\ldots,i'+1}$ is a hole. Its interior $\inter{\Hole}$ is in light blue. Remark that $\inter{H}$ is the right side of the curve generated by $P_{j',\ldots,i'+1}$ and $\cork{\Hole}$. Since $\glueP{i'}{i'+1}$ points west, then $P_{i'+1}$ is inside $\inter{\Hole}$. This is a contradiction since the last tile of $P$ is the easternmost one and cannot be inside $\inter{\Hole}$. }
\label{fig:Arc:Direction}
\end{figure}
\end{proof}

See Figure \ref{fig:Decompo:dominant} for the implications of these lemmas. To summarize, the dominant arcs which are below (resp. over) the last glue of $P$ on glue column $c$ must be upward (resp. downward) arcs. Moreover, for upward (resp. downward) arcs the orders implied by the notions of south/north and  anteriority/posteriority are equivalent. Finally, the last remaining lemma of this subsection shows that the smallest element of this order is the arc starting by $P_s$ (resp. $P_n$).

\input{./tikz/Hole/ExempleDecompo}

\begin{lemma}
\label{decompo:init}
Consider a producible path $P$ whose last tile is the easternmost one, a glue column $c$ with  $\Bone \leq c \leq \Btwo$ and let $0\leq \lastc < |P|-1$ be the index of the last glue of $P$ on glue column $c$. Consider $0\leq s \leq |P|-1$ such that $\glueP{s}{s+1}$ is visible from the south (resp. north) in $P$ on glue column $c$. Then, we have either  $s=\lastc$ or there exists $s<j<\lastc$ such that $P_{s,\ldots,j}$ is an upward (resp. downward) dominant arc of $P$ such that no arc of $P$ is south (resp. north) of $P_{s,\ldots,j}$.
\end{lemma}

\begin{proof}
Up to symmetries, we can consider that $\glueP{s}{s+1}$ is visible from the south and let $l^s$ be its glue ray. If $s=\lastc$ then the lemma is true, otherwise we can define $j=\min\{k>s: \glueP{k}{k+1} \text{ is on glue column }c\}$. Since $\glueP{j}{j+1}$ cannot be on $l^s$ then $y_{P_j}>y_{P_s}$ and $P_{s,\ldots,j+1}$ is an upward arc. Moreover, by Lemma \ref{Uturn:glueWest} and Lemma \ref{lem:glue:prop2}, $\glueP{s}{s+1}$ is pointing east then $P_{s,\ldots,j+1}$ is a positive arc and then an arc of $P$.  Since the only intersection between $P$ and $l^s$ is $\glueP{s}{s+1}$, then no arc of $P$ can dominate $P_{s,\ldots,j+1}$ and thus $P_{s,\ldots,j+1}$ is a dominant arc of $P$. Moreover, no other arc of $P$ can intersect $l^s$ and be south of $\glueP{s}{s+1}$. 
\end{proof}

\subsection{Decomposition into dominant arcs}
\label{Sec:arc:deompo}

Consider a producible path $P$ whose last tile is the easternmost one, a glue column   $\Bone \leq c \leq \Btwo$ and let $0\leq \lastc < |P|-1$ be the index of the last glue of $P$ on glue column $c$. Consider $0<s<|P|-1$ and $0<n<|P|-1$ such that $\glueP{s}{s+1}$ is visible from the south on glue column $c$ in $P$ and $0<n<|P|-1$ such that $\glueP{n}{n+1}$ is visible from the north on glue column $c$ in $P$. 

We introduce the \emph{upward (resp. downward) decomposition of} $P$ \emph{into dominant arcs} on glue column $c$ which is made of two sequences of indices $(\main{i})_{0\leq i \leq t}$ and $(\second{i})_{0\leq i < t}$ such that $\main{0}=s$ (resp. $\main{0}=n$), $\main{t}=\lastc$ and for all $0\leq i < t$, $P_{\main{i}, \ldots,\second{i}}$ is an upward (resp. downward) dominant arc of $P$ and if $i<t$, $\main{i+1}$ is the index of the next glue of $P$ according to $P_{\main{i}, \ldots,\second{i}}$. Lemmas \ref{decompo:next}, \ref{decompo:direction} and \ref{decompo:init} imply that this decomposition is correctly defined and that if $i<t-1$ then the arc $P_{\main{i+1}, \ldots,\second{i+1}}$ is north (resp. south) and posterior of arc $P_{\main{i}, \ldots,\second{i}}$.

\begin{fact}
\label{decompo:alldominant}
Consider a producible path $P$ whose last tile is the easternmost one and a glue column   $\Bone \leq c \leq \Btwo$. Consider $0\leq j < k \leq |P|-1$ such that $P_{j,\ldots,k}$ is an upward (resp. downward) dominant arc of $P$ then consider $(\main{i})_{0\leq i \leq t}$ and $(\second{i})_{0\leq i < t}$ the upward (resp. downward) decomposition of $P$ into dominants arcs. Then, there exists $0\leq a <t$ such $P_{j,\ldots,k}=P_{\main{a},\ldots,\second{a}}$. 
\end{fact}

The indices $(\main{i})_{0\leq i \leq t}$ are called the indices of the \emph{southern main} (resp. \emph{northern main}) glues of $P$ on glue column $c$ and these glues all point east while the glues whose indices are $(\second{i})_{0\leq i < t}$ all point west and are called the \emph{southern backup}  (resp. \emph{northern backup}) glues of $P$. They will allow a better characterization of the workspace of the span of $P$ on glue column $c$, see Figures \ref{fig:Example:Area}, \ref{fig:Decompo:dominant:south} and \ref{fig:Decompo:dominant:north} for an illustration of the following notions. We define $\bord{0}$ as the glue ray of $\glueP{s}{s+1}$ (resp.  $\glueP{n}{n+1}$) to the south (resp. north). We define $\bord{t+1}$ as the glue ray of  $\glueP{|P|-2}{|P|-1}$ going south (resp. north). For each $0< i \leq t$, we define $\bord{i}$ as the vertical segment between $\glueP{\second{i-1}-1}{\second{i-1}}$ and $\glueP{\main{i}}{\main{i}+1}$. These segments are called the \emph{frontiers} of the decomposition. Remark that, by the definition of the next glue, for all $0< i \leq t$, $P$ intersects $\bord{i}$ only at $\glueP{\second{i-1}-1}{\second{i-1}}$ and $\glueP{\main{i}}{\main{i}+1}$. Thus, $P_{\second{i-1}-1, \ldots, \main{i}+1}$ is a hole of $P$ on glue column $c$ whose interior is delimited by $P_{\second{i-1}-1, \ldots, \main{i}+1}$ and its frontier $\bord{i}$. To simplify the notations, we designate by $\inter{i}$ the area $\inter{P_{\second{i-1}-1, \ldots, \main{i}+1}}$ and call $(\inter{i})_{0< i \leq t}$ the \emph{interiors} of the decomposition.
Let $\mathcal{C}$ be the workspace of the span of $P$ on glue column $c$ (this span is $(s,n)$ if $s\leq n$ and $(n,s)$ otherwise). This workspace can be cut in three parts: the area $\mathcal{C}_3$ which the intersection of $\mathcal{C}$ with the half-plane defined by $x>e_\uniterm$ (no tile of $\uniterm$ is inside this area); the area $\mathcal{C}_2$ is the left side of the area delimited by $l^n$, $P_{n,\ldots,|P|-1}$ and the ray starting at $\pos{P_{|P|-1}}$ and going north; the area $\mathcal{C}_1$ is the right side of the area delimited by $l^s$, $P_{s,\ldots,|P|-1}$ and the ray starting at $\pos{P_{|P|-1}}$ and going south. Note that, for all $0< i \leq t$, $\inter{i}$ is included into $\mathcal{C}_1$ (resp. $\mathcal{C}_2$). Now, we definite $\es$ as the  area delimited by: $$ \bord{0}, P_{s,\ldots,\second{0}}, \bord{1}, P_{\main{1},\ldots,\second{1}}, \bord{2} \ldots \bord{t}, P_{\lastc,\ldots, |P|-1}, \bord{t+1}.$$ This area is also included into $\mathcal{C}_1$ (resp. $\mathcal{C}_2$) since it is the intersection of this area with the half-plane delimited by $x \geq c$. Then $\mathcal{C}_1$ (resp. $\mathcal{C}_2$)  is the union of $\es$ and all the interiors $\inter{i}$ for all $0< i \leq t$. Moreover, $\es$ is connected to all the interiors whereas an interior $\inter{i}$ is connected only to $\es$ via $\bord{i}$. Indeed, consider $0 < i <j \leq t$ and a path $Q$ such that $Q_0$ is in $\inter{i}$ and $Q_{|Q|-1}$ is in $\inter{j}$ then a prefix a $Q$ intersects with $\bord{i}$ to exit $\inter{i}$ and enter inside $\es$ whereas a suffix of $Q$  intersects with $\bord{j}$ to exit $\es$ and enter inside $\inter{j}$.  

\input{./tikz/Hole/ExempleArea}

\input{./tikz/Hole/ExempleDecompoSouth}

\input{./tikz/Hole/ExempleDecompoNorth}

\subsection{Links between the two decompositions}

Indices of the decomposition of $P$ into dominant arcs are not linked with the indices $(u_i)_{0\leq i \leq t'}$ of the decomposition of $P$ into pseudo-visible glues  on glue column $c$. Indeed, a northern or southern main glues may not be pseudo-visible (see Figure \ref{fig:Protected:NotPseudoVisible}) and a pseudo-visible $\glueP{u_i}{u_i+1}$ is not necessary a main glue of $P$, see Figure \ref{fig:Decompo:withSpan}. Nevertheless, some correspondencies can still be done. First, we introduce the notion of weakly dominance which is a restriction of the notion of dominance to a pseudo-span of $P$, see Figure \ref{fig:Decompo:withSpan}.

\input{./tikz/Hole/ExempleWithSpan}

\begin{definition}
\label{def:weak:dom}
Consider a producible path $P$ whose last tile is the easternmost one, a glue column  $\Bone \leq c \leq \Btwo$ and the decomposition $(u_i)_{0\leq i \leq t}$ of $P$ into pseudo-visible glues on glue column $c$. Consider $0\leq i <t$ and $u_i\leq j <k<u_{i+1}$ such that $P_{i,\ldots,j}$ is an arc of $P$, we say that this arc is \emph{weakly dominant} if and only there no $u_i\leq a < b <u_{i+1}$ such that $P_{a,\ldots, b}$ is an arc of $P$ which dominates $P_{j,\ldots,k}$.
\end{definition}

Then, we show that a dominant arc of $P$ on glue column $c$ is weakly dominant (Lemma \ref{lem:link:dom} and Corollary \ref{cor:link:dom}) for some pseudo-span of $P$ on glue column $c$. Afterwards, we show that if a pseudo-span $(u_i,u_{i+1})$ is an upward (resp. downward) pseudo-span then any weakly dominant arc of this pseudo-span is also an upward (resp. downward) arc (see Lemma \ref{lem:link:same}). These two lemmas imply the main result of this section in Corollary \ref{cor:decompo:cano}: the orientation of the spans matches the orientation of the dominant arcs, thus assembling a path which turns right of such an arc and which intersects with the end of $P$ would contradict the definition of a canonical path. 

\begin{lemma}
\label{lem:link:dom}
Consider a producible path $P$ whose last tile is the easternmost one, a glue column  $\Bone \leq c \leq \Btwo$ and the decomposition $(u_i)_{0\leq i \leq t}$ of $P$ into pseudo-visible glues on glue column $c$. Consider $0\leq j < k < |P|-1$ such that $P_{j,\ldots,k}$ is a dominant arc of $P$ on glue column $c$. Then, there exists $0\leq i <t$ such that $u_i\leq j <k < u_{i+1}$.
\end{lemma}

\begin{proof}
Up to symmetries, we can consider that $P_{j,\ldots,k}$ is an upward dominant arc of $P$ on glue column $c$. We remind that $u_t$ is, by definition, the index of the last glue of $P$ on glue column $c$ and since $\glueP{k}{k+1}$ is on glue column $c$ then $j< k \leq u_t$. By Fact \ref{decompo:alldominant}, $\glueP{j}{j+1}$ is also a main southern glue. 
Let $0\leq s \leq |P|-1$ such that $\glueP{s}{s+1}$ is visible from the south in $P$ on glue column $c$. Remind that $\glueP{s}{s+1}$ is either $u_0$ or $u_1$ depending if the span of $P$ is upward or downward and that $\glueP{s}{s+1}$ is also the first main southern glue of $P$. Then, we have $u_0\leq s \leq j$. Thus, there exists $0\leq i < t$ such that $u_i\leq j < u_{i+1}$. Note that, since $k=\min\{a>j:\glueP{a}{a+1} \text{ is on glue column $c$}\}$ and since $\glueP{u_{i+1}}{u_{i+1}+1}$ is also on glue column $c$ then $u_i\leq j < k \leq  u_{i+1}$. Moreover, since $P_{j\ldots,k}$ is a positive arc then $\glueP{k}{k+1}$ points west while $\glueP{u_{i+1}}{u_{i+1}+1}$ points east (as all glues of the decomposition of $P$ into pseudo-visible glues) then $k<u_{i+1}$.  
\end{proof}

\begin{corollary}
\label{cor:link:dom}
Consider a producible path $P$ whose last tile is the easternmost one, a glue column $\Bone \leq c \leq \Btwo$ and the decomposition $(u_i)_{0\leq i \leq t}$ of $P$ into pseudo-visible glues on glue column $c$. Consider $0\leq j < k < |P|-1$ such that $P_{j,\ldots,k}$ is a dominant arc of $P$ on glue column $c$ and $0\leq i \leq t$ such that $u_i\leq j <k < u_{i+1}$. Then, $P_{j,\ldots,k}$ is weakly dominant for the pseudo-span $(u_i,u_{i+1})$ of $P$.
\end{corollary}

\begin{lemma}
\label{lem:link:same}
Consider a producible path $P$ whose last tile is the easternmost one, a glue column $\Bone \leq c \leq \Btwo$ and the decomposition $(u_i)_{0\leq i \leq t}$ of $P$ into pseudo-visible glues on glue column $c$. Consider $0\leq j < k < |P|-1$ such that $P_{j,\ldots,k}$ is an upward (resp. downward) dominant arc of $P$ on glue column $c$ and $0\leq i \leq t$ such that $u_i\leq j <k < u_{i+1}$. Then, $(u_i,u_{i+1})$ is an upward (resp. downward) pseudo-span of $P$.
\end{lemma}

\begin{proof}
In fact, this proof follows the steps and reasonings which were done in Subsection \ref{Sec:arc:dominant}, see Figure \ref{fig:Decompo:ExemplePseudoDom}.
Consider a producible path $P$ whose last tile is the easternmost one, a glue column $\Bone \leq c \leq \Btwo$ and the decomposition $(u_i)_{0\leq i \leq t}$ of $P$ into pseudo-visible glues on glue column $c$. Consider $0\leq i < t$ such that $(u_i,u_{i+1})$ is a pseudo-span of $P$ on glue column $c$. Without loss of generality, we suppose that $(u_i,u_{i+1})$ is an upward pseudo-span. Let $s=u_i$ and $n=u_{i+1}$ such that $\glueP{s}{s+1}$ is pseudo-visible from the south in $P$ while $\glueP{n}{n+1}$ is visible from the north in $P_{s,\ldots,|P|-1}$. Let $j=\min\{k>s: \glueP{k}{k+1} \text{ is on glue column }c\}$, thus $P_{s,\ldots,j+1}$ is an arc of $P$ ($\glueP{s}{s+1}$ points east since it is a glue of the decomposition). By definition of $j$, we have $j<n$ since $\glueP{n}{n+1}$ points east and is also on glue column $c$. Since $P_{s,\ldots,|P|-1}$ cannot cross $l^s$ (the glue ray of $\glueP{s}{s+1}$ going south) then $P_{s,\ldots,j+1}$ is an upward arc and no arc of $P_{s,\ldots,|P|-1}$ can dominate or be south of $P_{s,\ldots,j+1}$. Then $P_{s,\ldots,j+1}$ is the southernmost weakly dominant arc of the span $(s,n)$ of $P$. Now let $a$ be the index of the southernmost glue of $P_{s,\ldots,n+1}$ which is strictly to the north of $\glueP{j}{j+1}$. If $a=n$ then $P_{s,\ldots,j+1}$ is the only weakly dominant arc of the span $(s,n)$ of $P$ (since $P_{s,\ldots,n+1}$ can intersect the glue ray of $\glueP{n}{n+1}$ going north only at its start-point). Thus the lemma is true in this case. Then suppose that $a<n$.

\begin{figure}
\center
\begin{tikzpicture}[x=0.2cm,y=0.2cm]


\draw[very thick] (4.5,4.5) -| (30.5,18.5) -| (12.5,16.5) -| (28.5,6.5)  -| (2.5,49.5) -| (19.5,52.5) -| (9.5,60.5) -| (15.5,60.5);


\draws{4}{4}

\tile{5}{4}{85}
\tile{6}{4}{85}
\tile{7}{4}{85}
\tile{8}{4}{85}
\tile{9}{4}{85}
\tile{10}{4}{85}
\tile{11}{4}{85}
\tile{12}{4}{85}
\tile{13}{4}{85}
\tile{14}{4}{85}
\tiler{14}{4}{85}
\tiler{15}{4}{85}
\tiler{16}{4}{85}
\tiler{17}{4}{85}
\tiler{18}{4}{85}
\tiler{19}{4}{85}
\tiler{20}{4}{85}
\tiler{21}{4}{85}
\tiler{22}{4}{85}
\tiler{23}{4}{85}
\tiler{24}{4}{85}
\tiler{25}{4}{85}
\tiler{26}{4}{85}
\tiler{27}{4}{85}
\tiler{28}{4}{85}
\tiler{29}{4}{85}
\tiler{30}{4}{85}
\tiler{30}{5}{85}
\tiler{30}{6}{85}
\tiler{30}{7}{85}
\tiler{30}{8}{85}
\tiler{30}{9}{85}
\tiler{30}{10}{85}
\tiler{30}{11}{85}
\tiler{30}{12}{85}
\tiler{30}{13}{85}
\tiler{30}{14}{85}
\tiler{30}{15}{85}
\tiler{30}{16}{85}
\tiler{30}{17}{85}
\tiler{30}{18}{85}
\tiler{29}{18}{85}
\tiler{28}{18}{85}
\tiler{27}{18}{85}
\tiler{26}{18}{85}
\tiler{25}{18}{85}
\tiler{24}{18}{85}
\tiler{23}{18}{85}
\tiler{22}{18}{85}
\tiler{21}{18}{85}
\tiler{20}{18}{85}
\tiler{19}{18}{85}
\tiler{18}{18}{85}
\tiler{17}{18}{85}
\tiler{16}{18}{85}
\tiler{15}{18}{85}
\tiler{14}{18}{85}
\tile{13}{18}{85}
\tile{12}{18}{85}
\tile{12}{17}{85}
\tile{12}{16}{85}
\tile{13}{16}{85}
\tiley{14}{16}{85}
\tiley{15}{16}{85}
\tiley{16}{16}{85}
\tiley{17}{16}{85}
\tiley{18}{16}{85}
\tiley{19}{16}{85}
\tiley{20}{16}{85}
\tiley{21}{16}{85}
\tiley{22}{16}{85}
\tiley{23}{16}{85}
\tiley{24}{16}{85}
\tiley{25}{16}{85}
\tiley{26}{16}{85}
\tiley{27}{16}{85}
\tiley{28}{16}{85}
\tiley{28}{15}{85}
\tiley{28}{14}{85}
\tiley{28}{13}{85}
\tiley{28}{12}{85}
\tiley{28}{11}{85}
\tiley{28}{10}{85}
\tiley{28}{9}{85}
\tiley{28}{8}{85}
\tiley{28}{7}{85}
\tiley{28}{6}{85}
\tiley{27}{6}{85}
\tiley{26}{6}{85}
\tiley{25}{6}{85}
\tiley{24}{6}{85}
\tiley{23}{6}{85}
\tiley{22}{6}{85}
\tiley{21}{6}{85}
\tiley{20}{6}{85}
\tiley{19}{6}{85}
\tiley{18}{6}{85}
\tiley{17}{6}{85}
\tiley{16}{6}{85}
\tiley{15}{6}{85}
\tiley{14}{6}{85}
\tile{13}{6}{85}
\tile{12}{6}{85}
\tile{11}{6}{85}
\tile{10}{6}{85}
\tile{9}{6}{85}
\tile{8}{6}{85}
\tile{7}{6}{85}
\tile{6}{6}{85}
\tile{5}{6}{85}
\tile{4}{6}{85}
\tile{3}{6}{85}
\tile{2}{6}{85}
\tile{2}{7}{85}
\tile{2}{8}{85}
\tile{2}{9}{85}
\tile{2}{10}{85}
\tile{2}{11}{85}
\tile{2}{12}{85}
\tile{2}{13}{85}
\tile{2}{14}{85}
\tile{2}{15}{85}
\tile{2}{16}{85}
\tile{2}{17}{85}
\tile{2}{18}{85}
\tile{2}{19}{85}
\tile{2}{20}{85}
\tile{2}{21}{85}
\tile{2}{22}{85}
\tile{2}{23}{85}
\tile{2}{24}{85}
\tile{2}{25}{85}
\tile{2}{26}{85}
\tile{2}{27}{85}
\tile{2}{28}{85}
\tile{2}{29}{85}
\tile{2}{30}{85}
\tile{2}{31}{85}
\tile{2}{32}{85}
\tile{2}{33}{85}
\tile{2}{34}{85}
\tile{2}{35}{85}
\tile{2}{36}{85}
\tile{2}{37}{85}
\tile{2}{38}{85}
\tile{2}{39}{85}
\tile{2}{40}{85}
\tile{2}{41}{85}
\tile{2}{42}{85}
\tile{2}{43}{85}
\tile{2}{44}{85}
\tile{2}{45}{85}
\tile{2}{46}{85}
\tile{2}{47}{85}
\tile{2}{48}{85}
\tile{2}{49}{85}
\tile{3}{49}{85}
\tile{4}{49}{85}
\tile{5}{49}{85}
\tile{6}{49}{85}
\tile{7}{49}{85}
\tile{8}{49}{85}
\tile{9}{49}{85}
\tile{10}{49}{85}
\tile{11}{49}{85}
\tile{12}{49}{85}
\tile{13}{49}{85}
\tileor{14}{49}{85}
\tileor{15}{49}{85}
\tileor{16}{49}{85}
\tileor{17}{49}{85}
\tileor{18}{49}{85}
\tileor{19}{49}{85}
\tileor{19}{50}{85}
\tileor{19}{51}{85}
\tileor{19}{52}{85}
\tileor{18}{52}{85}
\tileor{17}{52}{85}
\tileor{16}{52}{85}
\tileor{15}{52}{85}
\tileor{14}{52}{85}
\tile{13}{52}{85}
\tile{12}{52}{85}
\tile{11}{52}{85}
\tile{10}{52}{85}
\tile{9}{52}{85}
\tile{9}{53}{85}
\tile{9}{54}{85}
\tile{9}{55}{85}
\tile{9}{56}{85}
\tile{9}{57}{85}
\tile{9}{58}{85}
\tile{9}{59}{85}
\tile{9}{60}{85}
\tile{10}{60}{85}
\tile{11}{60}{85}
\tile{12}{60}{85}
\tile{13}{60}{85}
\tile{14}{60}{85}
\tile{15}{60}{85}

\path [dotted, draw, thin] (0,0) grid[step=0.2cm] (60,64);

\draw [dashed] (15,0) -| (15,64);

\fill (14.5,4.5) circle (0.16);
\node (D) at (14.5,3) {$P_{s}=P_{u_0}$};

\fill (14.5,60.5) circle (0.16);
\node (D) at (14.5,62.2) {$P_{n}=P_{u_1}$};


\end{tikzpicture}
\caption{The span $(s,n)$ of the path $P$ from Figure \ref{fig:Protected:NotPseudoVisible}. There are three arcs in $P_{0,\ldots,n}$. The red one is dominant (and thus weakly dominant). The orange one is weakly dominant but not dominant (another arc of $P_{n,\ldots,|P|-1}$ dominates this arc, see Figure \ref{fig:Decompo:dominant}). The last one in yellow is dominated by the red arc. All weakly dominant arcs are upwards arcs since $(s,n)$ is an upward span.}
\label{fig:Decompo:ExemplePseudoDom}
\end{figure}

Firstly, for the sake of contradiction, suppose that $a<s<j$ \footnote{One could argue that $a$ belongs to $P_{s,\ldots,n}$ but we will not used this argument here since we present a sketch of the proof of Lemma \ref{lem:next:order} which can be used in the general case.}. Similarly to the proof of Lemma \ref{lem:next:order}, $P_{a,\ldots,j+1}$ is a hole and $P_{j+1}$ is in the interior of this hole (since $\glueP{j}{j+1}$ points west). Since $\glueP{n}{n+1}$ is the northernmost glue of $P_{s,\ldots,n}$, it cannot be inside the interior of this hole. Thus $P_{j+1,\ldots,n}$ must intersect with either $P_{a,\ldots,j}$ (which is not possible since $P$ is simple) or the frontier of the hole (contradicting the definition of $a$). Thus, we have $s<j<a$.

Secondly for the sake of contradiction suppose that $\glueP{a}{a+1}$ points west. Similarly as in the proof of Lemma \ref{decompo:direction}, $P_{j,\ldots,a}$ is a hole and $P_{a+1}$ is inside the interior of this hole (since $\glueP{a}{a+1}$ points west). Since $\glueP{n}{n+1}$ is the northernmost glue of $P_{s,\ldots,n}$, it cannot be inside the interior of this hole. Thus, $P_{a+1,\ldots,n}$ must intersect with either $P_{j,\ldots,a}$ (which is not possible since $P$ is simple) or the frontier of the hole (contradicting the definition of $a$). Thus, $\glueP{a}{a+1}$ points east.

Then let $b=\min\{k>a: \glueP{k}{k+1} \text{ is on glue column }c\}$ and thus $P_{a,\ldots,b+1}$ is an arc of $P$. By definition of $j$, we have $j<n$ since $\glueP{n}{n+1}$ points east and is also on glue column $c$. Also, since $P_{a+1,\ldots,b+1}$ cannot cross $l^s$, $P_{s,\ldots,j+1}$ and the vertical segment between $\glueP{j}{j+1}$ and $\glueP{a}{a+1}$, then $P_{a,\ldots,b+1}$ is an upward arc of $P$. Similarly, no arc of $P_{s,\ldots,n}$ can dominate $P_{a,\ldots,b+1}$ since it would need to cross either $l^s$, $P_{s,\ldots,j+1}$ or the vertical segment between $\glueP{j}{j+1}$ and $\glueP{a}{a+1}$. Then $P_{a,\ldots,b+1}$ is a weakly dominant upward arc of the span $(s,n)$ of $P$ which is north of the arc $P_{s,\ldots,j+1}$ of the same span. By iterating this reasoning, we obtain that any weakly dominant arc of the span $(s,n)$ of $P$ is an upward arc of $P$ and the lemma is true.

\end{proof}

When the path $P$ is a canonically path for glue column $c$, it means that an upward (resp. downward) dominant arc of $P$ on glue column $c$ would be a subpath of an upward (resp. downward) pseudo-span of $P$. The definition of a canonical path leads to the following corollary.

\begin{corollary}
\label{cor:decompo:cano}
Consider a glue column $\Bone \leq c \leq \Btwo$ and let $P$ be a canonical path of glue column $c$. Consider $0\leq i <j <|P|-1$ such that $P_{i,\ldots,j}$ is an upward (resp. downward) dominant arc of $P$ on glue column $c$ then let $0\leq s \leq |P|-1$ such that $\glueP{s}{s+1}$ is visible from the south (resp. north) in $P$. Consider a path $B$ such that $P_{0,\ldots,k}B$ is a path for some $i<k<j$, $B_0 \neq P_{k+1}$ and $B$ is inside the area delimited by $l^s$, $P_{s,\ldots,|P|-1}$ and the ray starting at $\pos{P_{|P|-1}}$ and going south (resp. north). Then $B$ does not intersect with~$P$.
\end{corollary}

\section{Shield and conclusion}
\label{sec:analysis}

This section aims at showing that shields can be found and that whenever two identical protected glues are found on two glue columns on different canonical paths, either shields can be used to translate paths to go beyond the eastern limit of the assembly or a new protected glue can be found. Eventually, we obtain a contradiction and this concludes the proof.

We start by defining a shield in Subsection \ref{sec:def:shield}. Then, we present the three kinds of protected glues and their properties in Subsection~\ref{sec:def:protect}. Afterwards, we explain how to combine two protected main glues of the same type in Subsection \ref{sec:combining:full} in order to find a new protected glue. Similarly, we explain how to patch the previous results to combine two protected backup glues of the same type in Subsection \ref{sec:combining:half} and finally we conclude with the main result in Theorem \ref{the:end}.

\subsection{Shield}\label{sec:def:shield}

%
%
%

Before introducing the formal definition of a shield, we need to define the index of the \emph{first glue} of $P$ on glue column $c$ as $$\min\{i: \glueP{i}{i+1} \text{ is on glue column $c$}\}.$$ The definition of a shield is illustrated in Figure \ref{fig:Final:Mainexemple}.

\begin{definition}
\label{def:shield}
Suppose that $e_\uniterm>\Bfinal$ and consider a path $P$ which is canonical for some glue column $c$ with $\Bone\leq c \leq \Btwo$. Let $0\leq \last < |P|-1$ be the index of the first glue of $P$ on glue column $\Bfour(c)$ and $0\leq \lastc < |P|-1$ the index of the last glue of $P$ on glue column $c$. Consider $0<s\leq |P|-1$ such that $\glueP{s}{s+1}$ is visible from the south (resp. north) in $P_{0,\ldots,\last+1}$ on glue column $c$ and let $\mathcal{C}$ be the right (resp. left)  side of the area delimited by $l^s$, $P_{s,\ldots,\last+1}$ and the ray starting at $\glueP{\last}{\last+1}$ and going south (resp. north). A path $S$ is a southern (resp. northern) \emph{shield} of $P$ on glue column $c$ at position $s \leq \posS \leq \last$ iff:
\begin{itemize}
\item $S$ is a sub-assembly of $\uniterm$
\item $S$ does not intersect with $P_{s,\ldots,\last+1}$ and $S_{1,\ldots,|S|-1}$ is in $\mathcal{C}$;
\item $A=SP_{\posS,\ldots,\last+1}$ is an upward (resp. downward) negative arc of glue column $\Bfour(c)$;
\item $SP_\posS$ is a prefix of $\min(\inter{A})$.
\item $w_{A}\leq c$;
\end{itemize}
\end{definition}

\input{./tikz/Final/MainExemple}

Here are some remarks about this definition. Firstly, we will mainly work on $P_{0,\ldots,\last+1}$ in this section. By lemma \ref{Uturn:LinearCano}, we have $s\leq \lastc<\last$ implying that the indices of all glues of $P$ on glue column $c$ are smaller than $\last$. For example, the glue visible from the south in $P$ and $P_{0,\ldots,\last+1}$ is the same one. Thus, not taking the end of the path $P_{\last+2,\ldots,|P|-1}$ into account is not problematic for our study. This subpath $P_{\last+2,\ldots,|P|-1}$ is only relevant, in some cases, when it is copied and pasted to the east to obtain a contradiction. Secondly, we remind that $A$ is self-avoiding by definition. Thus there is no intersection between $S$ and $P_{a,\ldots,\last+1}$ (but intersections between $S$ and $P_{\last+2,\ldots,|P|-1}$ may exist). Thirdly, since $A$ is a negative arc of glue column $\Bfour(c)$ then $\glueS{0}{1}$ is on glue column $\Bfour(c)$ and points west. Moreover, the tile $S_{|S|-1}$ must bind with $P_a$. We also have $e_{A_{1\ldots,|A|-2}}<\Bfour(c)$ and more importantly, see Figure \ref{fig:Final:Mainexemple}:

\begin{fact}
\label{fact:shield:inter}
Consider a path $P$ which is canonical for some glue column $\Bone\leq c \leq \Btwo$, with $0\leq s \leq |P|-1$ as the index of a glue visible on glue column $c$ in $P$ and $\last$ as the index of its first glue on glue column $\Bfour(c)$. Consider a shield $S$ of $P$ on glue column $c$ at position $s \leq \posS \leq \last$ and let $A=SP_{\posS,\ldots,\last+1}$. Then, the interior $\inter{A}$ of arc $A$ is included into $\mathcal{C}$.
\end{fact}

This fact is implicitly used every time we assemble a path: we always work into areas where there is no tile of the seed $\sigma$ and the beginning of the path. Finally, the third condition of definition \ref{def:shield} implies that there is no ``shortcut'' in $\inter{A}$ between a tile of $S$ and another tile of $A$, \emph{i.e.} consider $0\leq i \leq j \leq |A| -1$ such that $i \leq |S|-1$ then for any path $Q$ which is sub-assembly of $\uniterm$ and such that $Q$ is in $\inter{A}$, $Q_0=A_i$ and $Q_{|Q|-1}=A_j$ then $S_{i,\ldots,|S|-1}P_\posS$ is a prefix of $Q$ if $j\geq |S|$ and $Q=S_{i,\ldots,j}$ otherwise. Finally, in order to deal with the final Lock \ref{lock:shieldone}, we distinguished two types of shield depending on the fact that $S$ crosses glue column $c$ or not.

\begin{definition}
\label{def:shield:type}
Consider a path $P$ which is canonical for some glue column $\Bone\leq c \leq \Btwo$, with $s$ as the index of a glue visible on glue column $c$ in $P$ and $\last$ as the index of its first glue on glue column $\Bfour(c)$. Consider a shield $S$ of $P$ on glue column $c$ at position $s \leq \posS \leq \last$. Shield $S$ is a \emph{full} shield if and only if $w_{SP_a}\leq c$ otherwise $S$ is called an \emph{half} shield.
\end{definition}

\subsection{Protected glues}\label{sec:def:protect}

Suppose that $e_\uniterm>\Bfinal$ and consider a path $P$ which is canonical for some glue column $c$ with $\Bone\leq c \leq \Btwo$. Also consider $(\main{i})_{0\leq i \leq t}$ and $(\second{i})_{0\leq i < t}$ the upward (resp. downward) decomposition of $P$ into dominants arcs on glue column $c$. 
Now, the aim of this section is to \emph{protect} the glues of this decomposition. A notion which can be seen as a generalization of pseudo-visibility.

To do so, we will have to consider three kinds of \emph{protected} glues. The first one is when the glue is visible in $P$ on glue column $c$, see Subsection \ref{sub:vis:protect}. The second one is when a glue is protected by a full-shield, see Subsection \ref{sub:vis:full}. The last one is when a glue is protected by an half-shield, see Subsection \ref{sub:vis:half}. For each kind of protected glue, we will define a \emph{category} (which is either ``main'' or ``backup''), a \emph{type} (belonging to the set of tile types $T$), a \emph{protected area} (which is included into the area $\mathcal{C}$ of Definition \ref{def:shield}) and a \emph{practical form} (which is a path included into the protected area and which is a subassembly of $\uniterm$) and a \emph{restrained area} (whose definition relies on the practical form and which is included into the protected area). In order to achieve this goal, we introduce here the \emph{practical form} of $P$ which can be defined conjointly for the three types of protected glue:

\begin{definition}
\label{def:rightform}
Suppose that $e_\uniterm>\Bfinal$, consider a path $P$ which is canonical for some glue column $c$ with $\Bone\leq c \leq \Btwo$, and let $\last$ be the index of its first glue on glue column $\Bfour(c)$. Also consider $(\main{i})_{0\leq i \leq t}$ and $(\second{i})_{0\leq i < t}$ to be the upward and downward decomposition of $P$ into dominant arcs on glue column $c$ and take an index $\proG$ such that $\proG=\main{i}$ or $\proG=\second{i}-1$. Suppose that an area called the protected area (defined later and differently in the two sub-cases) of $\proG$ is defined. We define the set of paths $\prodpaths{\proG}$ such that it contains a path $Q$ if and only if: 
\begin{itemize}
\item $Q_0=P_{\proG}$, $Q_1=P_{\proG+1}$, $Q_{|Q|-2}=P_{\last}$ and $Q_{|Q|-1}=P_{\last+1}$;
\item $Q_{1,\ldots,|Q|-2}$ is in the protected area of $\proG$; 
\item $P_{0,\ldots,\proG-1}Q$ is a producible path whose last tile is the easternmost one. 
\item $Q_{|Q|-1}$ is the only tile of $Q$ on column $\Bfour(c)+0.5$. 
\end{itemize}
Let $\rightF$ be the path such that $P_{\proG}R$ is the rightmost (resp. leftmost) priority path of this set, $\rightF$ is called the \emph{practical form} of $P$ at index $\proG$.
\end{definition}

We will give examples and illustrations for these notions for the three different kinds of protected glues. 

\subsubsection{The visible glue is protected}\label{sub:vis:protect}

\begin{definition}[Notations for a visible glue]
\label{notation:visible}
Suppose that $e_\uniterm>\Bfinal$ and consider a path $P$ which is canonical for some glue column $c$ with $\Bone\leq c \leq \Btwo$ and let $\last$ be the index of its first glue on glue column $\Bfour(c)$. Also consider $(\main{i})_{0\leq i \leq t}$ and $(\second{i})_{0\leq i < t}$ the upward (resp. downward) decomposition of $P$ into dominants arcs on glue column $c$. 
\end{definition}

We remind that $\glueP{\main{0}}{\main{0}+1}$ is the glue visible from the south (resp. north) in $P$ and we consider it as protected by default, see Figure \ref{fig:Final:VisIsPro}. It is categorized as a ``main'' glue, its type is the type of $\glueP{\main{0}}{\main{0}+1}$, its protected area $\area{\main{0}}$ is $\mathcal{C}$ the right side (resp. left side) of the area delimited by $l^{\main{0}}$, $P_{\main{0},\ldots,\last+1}$ and the ray starting at $\glueP{\last}{\last+1}$ and going south (resp. north). Let $\rightF$ be the practical form of $P$ at index $\main{0}$. In this case, $\rightF_{0,\ldots,|\rightF|-2}$ is in $\area{\main{0}}$ and we define the restrained area $\areaRes{\main{0}}$ as the right side (resp. left side) of the area delimited by $l^{\main{0}}$, $P_{\main{0}}R$ and the ray starting at $\glueP{\last}{\last+1}$ and going south (resp. north). This restrained area is inside~$\mathcal{C}$, see Figure \ref{fig:Final:VisIsPro}.

\input{./tikz/Final/VisibleProtected}

Variants of the two following lemmas are required for each kind of protected glue. We remind that the final goal, explained in Subsections \ref{sub:vis:full} and \ref{sub:vis:half}, is to combine two protected glues to find a new one. To achieve this goal, Lemma \ref{lem:full:turnRestrainedVis} will allow us to copy and pasted the practical form of one protected glue into the restrained area of the other one (as in Subsection \ref{road:newIdea} of the roadmap). Later in the proof, a third protected glue may be required to conclude (if the two protected glues are southern glues then the third one will be a northern one and vice versa). Then, Lemma \ref{lem:vis:LastTechnical} allow us to use efficiently this third protected glue.

\begin{lemma}
\label{lem:full:turnRestrainedVis}
Consider a path $P$ and the notations \ref{notation:visible}. Let $\area{\main{0}}$ be the protected area of $\glueP{\main{0}}{\main{0}+1}$ and let $\rightF$ the practical form of $P$ at index $\main{0}$. Consider a path $Q$ such that $Q_0=\rightF_0$. If $P_{\main{0}}Q$ is more right-priority (resp. left-priority) than $P_{\main{0}}\rightF$ and if there exist $0\leq u\leq v \leq |Q|-1$ such that $\rightF_{0,\ldots,u}$ is the largest common prefix between $\rightF$ and $Q$ and such that $Q_{0,\ldots,v}$ intersects with $\rightF_{u+1,\ldots,|\rightF|-1}$, then $Q_{0,\ldots,v}$ intersects with the glue ray of $\glueP{\main{0}}{\main{0}+1}$ or $e_{Q_{0,\ldots,v}}> \Bfour(c)$. 
\end{lemma}

\begin{proof}
Consider a path $Q$ such that $Q_0=\rightF_0$ and let $0\leq u \leq |Q|-1$ such that $Q_{0,\ldots,u}$ is the largest common prefix between $\rightF$ and $Q$, see Figure \ref{fig:Final:VisIsPro:LemOne}. If $Q_{u+1,\ldots,|Q|-1}$ intersects with $\rightF$ then we can define $v'=\max\{i: \pos{\rightF_i} \text{ is occupied by a tile of } Q\}$, let $u\leq v\leq |Q|-1$ such that $Q_{v}=\rightF_{v'}$. By definition of $v$ and $v'$, $Q_{0,\ldots,v}\rightF_{v'+1,\ldots, |R|-1}$ is a path. Now, if $Q_{0,\ldots,v}$ does not intersect the glue ray of $\glueP{\main{0}}{\main{0}+1}$, if $x_{Q_{0,\ldots,v}}<\Bfour(c)$ and if $P_{\main{0}}Q$ is more right-priority (resp. left-priority) than $P_{\main{0}}\rightF$ then $Q_{0,\ldots,v}$ is in $\area{\main{0}}$. Thus $P_{0,\ldots,\main{0}}Q_{0,\ldots,v}$ is producible and we have $Q_{v}=\rightF_{v'}$ since the tile assembly system is directed. Then $P_{0,\ldots, \main{0}}Q_{0,\ldots,v}\rightF_{v'+1,\ldots,|\rightF|-1}$ is a producible path and if $u<v$ then this path is more right-priority (resp. left-priority) than $P_{\main{\goB}}\rightF$ and satisfies all properties of \ref{def:rightform}. This is a contradiction of the definition of $\rightF$. 

\input{./tikz/Final/VisibleProtectedLem1}
\end{proof}

\begin{lemma}
\label{lem:vis:LastTechnical}
Consider a path $P$, the notations \ref{notation:visible} and the index $0\leq \lastc \leq |P|-1$ of the last glue of $P$ on glue column $c$. If there exists a path $B$ such that $B$ is in the restrained area $\areaRes{\main{0}}$ of $\glueP{\main{0}}{\main{0}+1}$, $B$ is a subassembly of $\uniterm$, $w_B>c$, $B_0$ is a tile of $P_{\lastc+1,\ldots,\last}$ then $B$ does not intersect with $P_{\main{0}+1,\ldots,\lastc}$. 
\end{lemma}

\begin{proof}
If $B$ intersects with $P_{\main{0}+1,\ldots,\lastc}$ then let $B_{0,\ldots,v}$ be the shortest prefix of $B$ which intersects with $P_{\main{0}+1,\ldots,\lastc}$, see Figure \ref{fig:Final:VisIsPro:LemTwo}. Let $B_{u,\ldots,v}$ be the shortest suffix of  $B_{0,\ldots,v}$ which intersects with $P_{\lastc+1,\ldots,\last}$. Since, $w_{B}> c$ then there exist $i\leq j < k$ such that $P_{i,\ldots,k}$ is a dominant arc of $P$ on glue column $c$ and $P_j=B_{v}$ (since $B$ is a subassembly of $\uniterm$). Since $B_{u,\ldots,v}$ is a subassembly of $\uniterm$ and in $\areaRes{\main{0}}$ then there is a contradiction with Corollary \ref{cor:decompo:cano} since $P$ is canonical for glue column $c$. 

\input{./tikz/Final/VisibleProtectedLem2}
\end{proof}

Note that, for the special case where the width of $P$ on glue column $c$ is $0$, then $\main{0}=\lastc$ and $\glueP{\lastc}{\lastc+1}$ is the only ``main'' glue of the upward decomposition of $P$ into dominant arcs and the only ``main'' glue of the downward decomposition of $P$ into dominant arcs.

\subsubsection{Glues protected by a full-shield}\label{sub:vis:full}

Along this subsection, we will use the following notations:

\begin{definition}[Notations for a glue protected by a full-shield]
\label{notation:full}
Suppose that $e_\uniterm>\Bfinal$ and consider a path $P$ which is canonical for some glue column $c$ with $\Bone\leq c \leq \Btwo$ and let $\last$ be the index of its first glue on glue column $\Bfour(c)$. Let $(\main{i})_{0\leq i \leq t}$ and $(\second{i})_{0\leq i < t}$ the upward (resp. downward) decomposition of $P$ into dominants arcs on glue column $c$ and let $(\bord{i})_{0\leq i \leq t+1}$ be the frontiers of this decomposition. Let $S$ be a southern (resp. northern) full-shield of $P$ on column $c$ at index $\main{0} <\posS \leq \last$. We denote by $\mathcal{C}$ the area delimited by $l^{\main{0}}$, $P_{\main{0},\ldots,\last+1}$ and the ray starting at $\glueP{\last}{{\last+1}}$ and going south (resp. north).
\end{definition}

\begin{lemma}
\label{lem:full:exists}
Consider the notations \ref{notation:full}. If $S$ is a full-shield then, there exist $0 \leq \enSh \leq |S|-1$ and $0< \goB \leq t$ such that the $\glueS{\enSh}{\enSh+1}$\footnote{If $\enSh=|S|-1$,  we consider that $S_{\enSh+1}=P_\posS$ by abuse of notation.} intersects the frontier $\bord{\goB}$ and $x_{S_{0,\ldots,\enSh}}>c$.
\end{lemma}

\begin{proof}
By definition of a full-shield, we have $x_{SP_a}\leq c$, see Figure \ref{fig:Final:Mainexemple:short}. Thus, we can define $\enSh=\min\{i:\glueS{i}{i+1} \text{ is on glue column $c$}\}$. Then, by definition of $e$, we have $x_{S_e}=c+0.5$ and $x_{S_{e+1}}=c-0.5$. Now, the workspace $\mathcal{C}$ can be decomposed into its east side $\es$ and the interiors $\inter{i}$ for all $0< i \leq t$. Since $x_{S_e}>c$ then $\pos{S_e}$ is in $\es$. However, $x_{S_{e+1}}\leq c$ and $\pos{S_{e+1}}\in \mathcal{C}$ thus $\pos{S_{e+1}}$ belongs to $\inter{\goB}$ for some $0< \goB \leq t$. Then, $\glueS{\enSh}{\enSh+1}$ must intersect with $\bord{\goB}$.

\input{./tikz/Final/FullShort}
\end{proof}

Thanks to Lemma \ref{lem:full:exists}, we now assume that there exists $0 \leq \enSh \leq |S|-1$ and $0< \goB \leq t$ such that $\glueS{\enSh}{\enSh+1}$ intersects with the frontier $\bord{\goB}$ and $x_{S_{0,\ldots,\enSh}}>c$. In this case, we say that $\glueP{\main{\goB}}{\main{\goB}+1}$ is \emph{protected} by the beginning of the full-shield $S_{0,\ldots,e}$, see Figure \ref{fig:Final:Mainexemple:protect}. The category of this protected glue is ``main'' and its type is the type of $\glueP{\main{\goB}}{\main{\goB}+1}$. Its \emph{protected area} $\area{\main{\goB}}$ is defined as the finite area delimited by:
\begin{itemize}
\item $P_{\main{\goB},\ldots,\last+1}$, 
\item the vertical segment between $\glueP{\main{\goB}}{\main{\goB}+1}$ and $\glueS{\enSh}{\enSh+1}$,
\item $S_{0,\ldots,\enSh+1}$;
\item the vertical segment between $\glueP{\last}{\last+1}$ and $\glueS{0}{1}$. 
\end{itemize}
Note that $\area{\main{\goB}}$ may be included into the interior of the arc $SP_{\posS,\ldots,\last+1}$ but this is not always the case. 
Remark that $P_{0,\ldots,\main{\goB}}$ is not inside $\area{\main{\goB}}$. Indeed, $\glueP{\main{\goB}}{\main{\goB}+1}$ points east and thus $P_{\main{\goB}}$ is not inside $\area{\main{\goB}}$. Moreover, $P_{0,\ldots,\main{\goB}}$ cannot intersect with $P_{\main{\goB}+1,\ldots,\last+1}$ (since $P$ is self-avoiding), the vertical segment between $\glueP{\main{\goB}}{\main{\goB}+1}$ and $\glueS{\enSh}{\enSh+1}$ (since this segment is part of the frontier $\bord{\goB}$), $S_{0,\ldots,\enSh}$ (since, by definition, $S$ does not intersect $P_{0,\ldots,\last+1}$ and is in $\mathcal{C}$) and the vertical segment between $\glueP{\last}{\last+1}$ and $\glueS{0}{1}$ (since $P_{0,\ldots,\main{\goB}}$ has no tile on column $\Bfour(c)+0.5$ by Lemma \ref{Uturn:LinearCano}). This remark implies the following fact:

\input{./tikz/Final/FullProtected}

\begin{fact}
\label{fact:full:area}
Consider the notations \ref{notation:full} and $0< \goB \leq t$ such that $\glueP{\main{\goB}}{\main{\goB}+1}$ is protected by the full-shield $S$. Let $\area{\main{\goB}}$ be  its protected area.
Consider a path $Q$ such that $Q$ is inside $\area{\main{\goB}}$ and such that $Q_0=P_{\main{\goB}+1}$ then $P_{0,\ldots,\main{\goB}}Q$ is a producible path.
\end{fact}

Moreover, the following lemma shows that the beginning $S_{0,\ldots,e}$ of the full-shield $S$ must be avoided by any path $Q$ growing inside $\area{\main{\goB}}$.

\begin{lemma}
\label{lem:full:visible}
Consider the notations \ref{notation:full}, $0< \goB \leq t$ and $0\leq  e \leq |S|-1$ such that $\glueP{\main{\goB}}{\main{\goB}+1}$ is protected by the beginning $S_{0,\ldots,e}$ of the full-shield $S$. Let $\area{\main{\goB}}$ be  its protected area.
There exists no path $Q$ such that $Q$ is inside the area $\area{\main{\goB}}$, $Q_0=P_{\main{\goB}+1}$ and $Q$ intersects with $S_{0,\ldots,\enSh}$.
\end{lemma}

\begin{proof}
For the sake of contradiction suppose that such a path $Q$ exists (see Figure \ref{fig:Final:Mainexemple:ShieldLemma:PartOne}), then by Fact \ref{fact:full:area}, $P_{0,\ldots,\main{\goB}}Q$ is a producible path. Let $A$ be the arc $SP_{\posS,\ldots,\last+1}$. Firstly, consider the case where $\posS \leq \main{\goB}$, see Figure \ref{fig:Final:Mainexemple:ShieldLemma:PartTwo}. In this case, $P_{\main{\goB}+1}$ is inside the interior $\inter{A}$ of arc $A$ and let $0\leq j \leq |Q|-1$ be such that $Q_{0,\ldots,j}$ in the shortest prefix of $Q$ which intersects with $S$ (by hypothesis, this index is correctly defined) and let $0\leq i < j$ be such that $Q_{i,\ldots,j}$ is the shortest suffix of $Q_{0,\ldots,j}$ which intersects with $P_{\main{\goB}+1,\ldots, \last+1}$ (by hypothesis, this index is correctly defined since $Q_0=P_{{\main{\goB}}+1}$). Let $\main{\goB}+1\leq i' \leq \last$ be such that $P_{i'}=Q_i$ and $0\leq j' \leq |S|-1$ be such that $S_{j'}=Q_j$ then $A'=S_{0,\ldots,j'-1}\rev{Q_{i,\ldots,j}}P_{i'+1,\ldots,\last+1}$ is a simple path\footnote{$\rev{Q}$ designates the path $Q$ in reverse: the first tile becomes the first one and vice versa.}. Since $Q$ is in $\area{\main{\goB}}$ then $e_Q< \Bfour(c)$ and thus $A'_{1,\ldots,|A'|-1}$ is inside $\inter{A}$ and is a negative arc of glue column $\Bfour(c)$. Moreover, $A'$ does not pass by $P_{\main{\goB}}$ since this tile is not inside $\area{\main{\goB}}$. The existence of this arc contradicts the definition of a shield, since $S$ should be the prefix of  $\min(\inter{A})$. 

\input{./tikz/Final/FullLemmaShieldPart1}
\input{./tikz/Final/FullLemmaShieldPart2}

Secondly, consider the case where $\posS> \main{\goB}$, see Figure \ref{fig:Final:Mainexemple:ShieldLemma:PartThree}. In this case, $P_{\posS}$ is inside $\area{\main{\goB}}$ but $S_{\enSh+1}$ is not inside this area. Then, let $\enSh<h\leq |S|-1$ be such that $S_{h,\ldots,|S|-1}$ is the shortest suffix of $S$ such that $S_{h,\ldots,|S|-1}P_\posS$ is not inside $\area{\main{\goB}}$. Without loss of generality, we suppose that $h<|S|-1$ to ease the notation. This definition implies that $\glueS{h}{h+1}$ intersects the vertical segment between $\glueP{\main{\goB}}{\main{\goB}+1}$ and $\glueS{e}{e+1}$. 
Since $\main{\goB}<\posS$, the path $Q$ must intersect with $S_{h+1,\ldots,|S|-1}P_a$ in order to reach $S_{0,\ldots,\enSh}$. Let $0\leq i \leq |Q|-1$ be such that $Q_{i,\ldots,|Q|-1}$ in the shortest suffix of $Q$ which intersects with $S_{h+1,\ldots,|S|-1}P_{a,\ldots,\last+1}$ and let $ i < j \leq |S|-1$ be such that $Q_{i,\ldots,j}$ is the shortest prefix of $Q_{i,\ldots,|Q|-1}$ which intersects with $S_{0,\ldots,h-1}$ (this index is correctly defined since $Q$ intersects with $S_{0,\ldots,\enSh}$). Let $h< i' \leq |S|-1$ be such that $S_{i'}=Q_i$ and $0 \leq j' < h$ be such that $S_{j'}=Q_j$ then $A'=S_{0,\ldots,j'-1}\rev{Q_{i,\ldots,j}}S_{i'+1,\ldots,|S|-1}P_{\posS,\ldots,\last+1}$ is a simple path. Since $Q$ is in $\area{\main{\goB}}$ then $e_Q < \Bfour(c)$ and thus $A'_{1,\ldots,|A'|-2}$ is inside $\inter{A}$ and $A'$ is a negative arc of glue column $\Bfour(c)$. Moreover, this arc does not pass by $S_{h}$ since this tile is not inside $\area{\main{\goB}}$. Such an arc contradicts the definition of a shield, since $S$ should be the prefix of $\min(\inter{A})$.

\input{./tikz/Final/FullLemmaShieldPart3}
\end{proof}

Let $\rightF$ be the practical form of $P$ at index $\main{\goB}$, see Figure \ref{fig:Final:Mainexemple:fullrestrained}. Lemma \ref{lem:full:visible} implies that $\rightF$ does not intersect $S_{0,\ldots,\enSh}$. Then, we can define $\areaRes{\main{\goB}}$ as the \emph{restrained protected area} delimited by:
\begin{itemize}
\item $P_{\main{\goB}}\rightF$, 
\item the vertical segment between $\glueP{\main{\goB}}{\main{\goB}+1}$ and $\glueS{\enSh}{\enSh+1}$,
\item $S_{0,\ldots,\enSh+1}$,
\item the vertical segment between $\glueP{\last}{\last+1}$ and $\glueS{0}{1}$. 
\end{itemize}
This area is included into $\area{\main{\goB}}$.

\input{./tikz/Final/FullRestrained}

Variants of the two following lemmas are required for each kind of protected glue, see Subsection \ref{sub:vis:protect} dedicated to visible glues.

\begin{lemma}
\label{lem:full:turnRestrained}
Consider the notations \ref{notation:full}, $0< \goB \leq t$ and $0\leq  e \leq |S|-1$ such that $\glueP{\main{\goB}}{\main{\goB}+1}$ is protected by the beginning $S_{0,\ldots,e}$ of the full-shield $S$. Let $\rightF$ be  the practical form of $P$ at index $\main{\goB}$ and let $\areaRes{\main{\goB}}$ be its restrained protected area. Consider a path $Q$ such that $Q_0=\rightF_0$. If $P_{\main{\goB}}Q$ is more right-priority (resp. left-priority) than $P_{\main{\goB}}\rightF$ either $Q$ is in $\areaRes{\main{\goB}}$ or $e_Q>\Bfour(c)$ or $Q$ intersects with the vertical segment between $\glueP{\main{\goB}}{\main{\goB}+1}$ and $\glueS{\enSh}{\enSh+1}$. 
\end{lemma}

\begin{proof}
Let $0\leq v \leq |Q|-1$ such that $Q_{0,\ldots,v}$ is the largest prefix of $Q$ which is inside $\areaRes{\main{\goB}}$, see Figure \ref{fig:Final:Mainexemple:fullrestrained:Lem1}. Let $0\leq u \leq |Q|-1$ such that $Q_{0,\ldots,u}$ is the largest common prefix between $Q$ and $\rightF$. Since $P_{\main{\goB}}Q$ is more right-priority (resp. left-priority) than $P_{\main{\goB}}\rightF$ then $u \leq v$. By Lemma \ref{lem:full:visible}, $Q_{0,\dots, v}$ cannot intersect $S_{0,\ldots,e}$. If $Q_{0,\dots, v}$ intersects with $\rightF_{u+1,\ldots,|\rightF|-1}$ then let $z=\max\{i: \pos{R_i} \text{ is occupied by a tile of } Q_{0,\ldots,v}\}$ and $u<z'\leq v$ such that $P_z=Q_{z'}$. Then, $Q_{0,\dots,z'}R_{z+1,\ldots,|R|-1}$ is in $\areaResUn{\main{\goB}}$ and $P_{0,\ldots,\main{\goB}}Q_{0,\dots,z'}R_{z+1,\ldots,|R|-1}$ is producible and $P_{\main{\goB}}Q_{0,\dots,z'}R_{z+1,\ldots,|R|-1}$ is more right-priority (resp. left-priority) than $P_{\main{\goB}}\rightF$ and satisfies all properties of definition \ref{def:rightform}. This is a contradiction of the definition of $\rightF$. Thus, if $v<|Q|-1$, then either $Q_{v+1}$ is on column $\Bfour(c)+0.5$ or $\glueQ{v}{v+1}$ is on the vertical segment between $\glueP{\main{\goB}}{\main{\goB}+1}$ and $\glueS{\enSh}{\enSh+1}$.  

\input{./tikz/Final/FullRestrainedLem1}
\end{proof}

\begin{lemma}
\label{lem:full:LastTechnical}
Consider the notations \ref{notation:full}, $0< \goB \leq t$ and $0\leq  e \leq |S|-1$ such that $\glueP{\main{\goB}}{\main{\goB}+1}$ is protected by the beginning $S_{0,\ldots,e}$ of the full-shield $S$. Let $\area{\main{\goB}}$ be the protected area of $\glueP{\main{\goB}}{\main{\goB}+1}$. Let $0\leq \lastc \leq |P|-1$ be the  index of the last glue of $P$ on glue column $c$. If there exists a path $B$ such that $B$ is in $\area{\main{\goB}}$, $B$ is a subassembly of the terminal assembly $\uniterm$, $w_B>c$, $B_0$ is a tile of $P_{\lastc+1,\ldots,\last}$ then $B$ does not intersect with $P_{\main{\goB}+1,\ldots,\lastc}$ or $S_{0,\ldots,e}$. 
\end{lemma}

\begin{proof}
If $B$ intersects with $P_{\main{\goB}+1,\ldots,\lastc}$ then let $B_{0,\ldots,v}$ be the shortest prefix of $B$ which intersects with $P_{\main{\goB}+1,\ldots,\lastc}$, see Figure \ref{fig:Final:Mainexemple:fullrestrained:Lem2}. Let $B_{u,\ldots,v}$ be the shortest suffix of  $B_{0,\ldots,v}$ which intersects with $P_{\lastc+1,\ldots,\last}$. Since, $w_{B}> c$ then there exists $i< j < k$ such that $P_{i,\ldots,k}$ is a dominant arc of $P$ on glue column $c$ and $P_j=B_{v}$ (since $B$ is a subassembly of $\uniterm$). Since $B_{u,\ldots,v}$ is a subassembly of the terminal assembly of $\uniterm$ and in $\area{\main{\goB}}$ then there is a contradiction with Corollary \ref{cor:decompo:cano} since $P$ is canonical for glue column $c$. 

If $B$ intersects with $S_{0,\ldots,e}$ then let $B_{0,\ldots,v}$ be the shortest prefix of $B$ which intersects with $S_{0,\ldots,e}$. Let $B_{u,\ldots,v}$ be the shortest suffix of  $B_{0,\ldots,v}$ which intersects with $P_{\lastc+1,\ldots,\last}$. There exists $\lastc+1\leq u' \leq \last$ such that $P_{u'}=B_u$. By the previous paragraph,  $B_{u,\ldots,v}$ does not intersect with $P_{\main{\goB}+1,\ldots,\lastc}$. Thus $P_{\main{\goB}+1,\ldots,u'-1}B_{u,\ldots,v}$ is path which is in inside $\area{\main{\goB}}$. This path contradicts Lemma \ref{lem:full:visible}. 

\input{./tikz/Final/FullRestrainedLem2}
\end{proof}

\subsubsection{Glues protected by an half-shield}\label{sub:vis:half}

Along this subsection, we will use the following notations:

\begin{definition}[Notations for a glue protected by an half-shield]
\label{notation:half:def}
Suppose that $e_\uniterm>\Bfinal$ and consider a path $P$ which is canonical for some glue column $c$ with $\Bone\leq c \leq \Btwo$ and let $\last$ be the index of its first glue on glue column $\Bfour(c)$. Let $(\main{i})_{0\leq i \leq t}$ and $(\second{i})_{0\leq i < t}$ the upward (resp. downward) decomposition of $P$ into dominants arcs on glue column $c$ and let $(\bord{i})_{0\leq i \leq t+1}$ be the frontiers of this decomposition. Let $S$ be an southern (resp. northen) half-shield of $P$ on glue column $c$ at index $\main{0} <\posS \leq \last$. Let $A=SP_{a,\ldots,\last+1}$.
\end{definition}

\begin{lemma}
\label{lem:half:exists}
Consider the notations \ref{notation:half:def}, there exists $0\leq \goB < t$ such that $\main{\goB}<\posS<\second{\goB}$.
\end{lemma}

\begin{proof}
In this case, consider $\mathcal{C}$ the area defined by the right (resp. left) side of the curve made of $l^{\main{0}}$, $P_{\main{0}, \ldots,\last+1}$ and the ray starting at $\glueP{\last}{\last+1}$ and going south (resp. north), see Figure \ref{fig:Final:Mainexemple:half}. Let $\es$ be its  east side (the intersection of $\mathcal{C}$ with the half-plane of equation $x \geq c$). By definition of an half-shield, $w_{SP_a}>c$ then $P_a$ is in $\es$. Thus, either $\main{t}<\posS\leq \last$ or there exists $0\leq \goB < t$ such that $\main{\goB}<\posS<\second{\goB}$. If $\goB<t$ then the lemma is true. Otherwise, we have $\main{t}<\posS\leq \last$. In this case, remind that $\main{t}$ is the index of the last glue of $P$ on glue column $c$. Thus, $w_{P_{\posS,\ldots,\last+1}}>c$. Then, we would have $w_A>c$ which contradicts the definition of a shield.
\end{proof}

Thanks to Lemma \ref{lem:half:exists}, we can now assume that there exists $0\leq \goB < t$ such that $\main{\goB}<\posS<\second{\goB}$, see Figure \ref{fig:Final:Mainexemple:half}. In this case, we say that the $\glueP{\second{\goB}-1}{\second{\goB}}$ is \emph{protected} by the half-shield $S$. The category of this protected glue is a ``backup''. The protected area $\area{\second{\goB}}$ of $\second{\goB}$ is $\inter{A}$, the interior of the hole $A$. 

\input{./tikz/Final/HalfRestrained}

\begin{lemma}
\label{lem:half:utile}
Consider the notations \ref{notation:half:def}, let $\rightF$ be the practical form of $P$ at index $\second{\goB}-1$. Then $SP_{a,\ldots,\second{\goB}-1}\rightF$ is a negative upward (resp. downward) arc of glue column $\Bfour(c)$ included into $\area{\second{\goB}}$ (apart from its first and last tiles).
\end{lemma}

\begin{proof}
By definition of a practical form, $P_{0,\ldots,\second{\goB}-1}\rightF$ is a producible path and then $P_{a,\ldots,\second{\goB}-1}$ does not intersect with $\rightF$. Moreover $\rightF$ is inside $\area{\second{\goB}}$. If $\rightF$ intersects with $S$ then there is a contradiction with the hypothesis that $SP_{a}$ is prefix the minimal arc of $\inter{A}$ (by a reasoning similar to the first case of Lemma \ref{lem:full:visible} for example). Thus $\rightF$ and $S$ do not intersect and $SP_{a,\ldots,\second{\goB}-1}\rightF$ is a negative arc upward (resp. downward) arc of glue column $\Bfour(c)$ included into $\area{\second{\goB}}$ (apart from its first and last tiles).  
\end{proof}

Lemma \ref{lem:half:utile} allows us to defined the restrained area $\areaRes{\second{\goB}}$ of $\second{\goB}$ as the interior of the arc $SP_{a,\ldots,\second{\goB}-1}\rightF$. This area is included into $\area{\second{\goB}}$, see Figure~\ref{fig:Final:Mainexemple:half}. 

The following lemma is the equivalent of Lemma \ref{lem:vis:LastTechnical} and Lemma \ref{lem:full:LastTechnical} for a backup glue protected by an half-shield. As explained in Subsection \ref{sub:vis:protect}, this lemma will be useful in Subsection \ref{sec:combining:full}. Indeed, combining two main southern protected glues may require a third northern protected glue. This northern glue may be protected by an half-shield. 

\begin{lemma}
\label{lem:half:LastTechnical}
Consider the notations \ref{notation:half:def}, let $0\leq \lastc \leq |P|-1$ be the  index of the last glue of $P$ on glue column $c$. If there exists a path $B$ such that $B$ is in $\area{\second{\goB}}$, $B$ is a subassembly of  $\uniterm$, $w_B>c$, $B_0$ is a tile of $P_{\lastc+1,\ldots,\last}$ then $B$ does not intersect with $P_{\posS,\ldots,\lastc}$ or $S$. 
\end{lemma}

\begin{proof}
If $B$ intersects with $P_{a,\ldots,\lastc}$ then let $B_{0,\ldots,v}$ be the shortest prefix of $B$ which intersects with $P_{a,\ldots,\lastc}$. Let $B_{u,\ldots,v}$ be the shortest suffix of  $B_{0,\ldots,v}$ which intersects with $P_{\lastc+1,\ldots,\last}$. Since, $w_{B}> c$ then there exists $i\leq j < k$ such that $P_{i,\ldots,k}$ is a dominant arc of $P$ on glue column $c$ and $P_j=B_{v}$ (since $B$ is a subassembly of $\uniterm$). Since $B_{u,\ldots,v}$ is a subassembly of $\uniterm$ and in $\area{\second{\goB}}$ then there is a contradiction with Corollary \ref{cor:decompo:cano} since $P$ is canonical for glue column $c$. 

If $B$ intersects with $S$ then let $B_{0,\ldots,v}$ be the shortest prefix of $B$ which intersects with $S$. Let $B_{u,\ldots,v}$ be the shortest suffix of  $B_{0,\ldots,v}$ which intersects with $P_{\lastc+1,\ldots,\last}$. There exists $\lastc+1\leq u' \leq \last$ such that $P_{u'}=B_u$. By the previous paragraph,  $B_{u,\ldots,v}$ does not intersect with $P_{a,\ldots,\lastc}$. Thus $P_{a,\ldots,u'-1}B_{u,\ldots,v}$ is path which is in inside $\area{\second{\goB}}$. We obtain a contradiction by a reasoning similar to the first case of Lemma \ref{lem:half:utile}. 
\end{proof}

With this result, we may proceed to explain how to combine two ``main'' protected glues to find a new protected one. Concerning the equivalence of Lemma \ref{lem:full:turnRestrainedVis} and Lemma \ref{lem:full:turnRestrained}, the whole Subsection \ref{sec:combining:half} is dedicated to this result. Indeed, achieving this goal is far more complex for a glue protected by an half-shield.

\subsection{Combining two identical ``main'' protected glues}\label{sec:combining:full}

Along this subsection, we will use the following notations:

\begin{definition}[Notations for combining two ``main'' glues]
\label{notation:full:combi}
Suppose that $e_\uniterm>\Bfinal$ and consider two paths $P$ and $Q$ which are canonical for some glue columns $\colUn$ and $\colDeux$ respectively with $\Bone\leq \colUn < \colDeux \leq \Btwo$. 
Let $0\leq \lastUn < |P|-1$ be the index of the first glue of $P$ on glue column $\Bfour(\colUn)$. 
Let $0\leq \lastDeux < |Q|-1$ be the index of the first glue of $Q$ on glue column $\Bfour(\colDeux)$. 
Let $(\main{i})_{0\leq i \leq t}$ and $(\second{i})_{0\leq i < t}$ the upward (resp. downward) decomposition of $P$ into dominants arcs on glue column $\colUn$.
Let $0\leq  \goB \leq t$ such that $\glueP{\main{\goB}}{\main{\goB}+1}$ is a protected ``main'' glue of $P$. Let $\rightF$ be the practical form of $P$ at index $\main{\goB}$ and let $\areaResUn{\main{\goB}}$ be its restrained protected area.
Consider the index $0\leq \proSLast{\colDeux} \leq |Q|-1$ of a protected ``main'' glue of the upward (resp. downward) decomposition of $Q$ into dominants arcs on glue column $\colDeux$. 
Let $\rightFDeux$ be the practical form of $Q$ at index $\proSLast{\colDeux}$ and let $\areaResDeux{\proSLast{\colDeux}}$ be its restrained protected area.
Suppose that $\glueP{\main{\goB}}{\main{\goB+1}}$ and $\glueQ{\proSLast{\colDeux}}{\proSLast{\colDeux}+1}$ have the same type and let $\vu=\vect{\pos{P_{\main{\goB}}}\pos{Q_{\proSLast{\colDeux}}}}$.
\end{definition}

As an example for the path $P$, we will first use the path $P$ of Figure \ref{fig:Final:Combining:Pexample} and for the path $Q$, we will use the example from Figure \ref{fig:Final:Mainexemple:fullrestrained}. The following lemma corresponds to the reasoning of the roadmap where we explain that only the first case of Fact \ref{fact:turn} is possible.

\begin{figure}
\center
\begin{tikzpicture}[x=0.2cm,y=0.2cm]

\fill[fill=orange!30!white, draw opacity=0.8] (15,20.5) -| (23.5,24.5) -| (12.5,30.5) -| (18.5,38.5) -| (51,15.5) -| (22.5,16.5) -| (15,20.5);
\fill[fill=yellow!30!white, draw opacity=0.8] (15,20.5) -| (23.5,24.5) -| (12.5,30.5) -| (18.5,27.5) -| (23.5,29.5) -| (28.5,33.5) -| (47.5,38.5) -| (51,15.5) -| (22.5,16.5) -| (15,20.5);

\draw[thick] (14.5,20.5) -| (23.5,24.5) -| (12.5,30.5) -| (18.5,27.5) -| (23.5,29.5) -| (28.5,33.5) -| (47.5,38.5) -| (51.5,38.5);
\draw[very thick] (14.5,20.5) -| (23.5,24.5) -| (12.5,30.5) -| (18.5,38.5) -| (51.5,38.5);
\draw[thick] (15,16.5) -| (22.5,15.5) -| (51.5,15.5);



\tile{14}{20}

\tileor{15}{20}
\tileor{16}{20}
\tileor{17}{20}
\tileor{18}{20}
\tileor{19}{20}
\tileor{20}{20}
\tileor{21}{20}
\tileor{22}{20}
\tileor{23}{20}
\tileor{23}{21}
\tileor{23}{22}
\tileor{23}{23}
\tileor{23}{24}
\tileor{22}{24}
\tileor{21}{24}
\tileor{20}{24}
\tileor{19}{24}
\tileor{18}{24}
\tileor{17}{24}
\tileor{16}{24}
\tileor{15}{24}
\tileor{14}{24}
\tileor{13}{24}
\tileor{12}{24}
\tileor{12}{25}
\tileor{12}{26}
\tileor{12}{27}
\tileor{12}{28}
\tileor{12}{29}
\tileor{12}{30}
\tileor{13}{30}
\tileor{14}{30}
\tileor{15}{30}
\tileor{16}{30}
\tileor{17}{30}
\tileor{18}{30}
\tileor{18}{31}
\tileor{18}{32}
\tileor{18}{33}
\tileor{18}{34}
\tileor{18}{35}
\tileor{18}{36}
\tileor{18}{37}
\tileor{18}{38}
\tileor{19}{38}
\tileor{20}{38}
\tileor{21}{38}
\tileor{22}{38}
\tileor{23}{38}
\tileor{24}{38}
\tileor{25}{38}
\tileor{26}{38}
\tileor{27}{38}
\tileor{28}{38}
\tileor{29}{38}
\tileor{30}{38}
\tileor{31}{38}
\tileor{32}{38}
\tileor{33}{38}
\tileor{34}{38}
\tileor{35}{38}
\tileor{36}{38}
\tileor{37}{38}
\tileor{38}{38}
\tileor{39}{38}
\tileor{40}{38}
\tileor{41}{38}
\tileor{42}{38}
\tileor{43}{38}
\tileor{44}{38}
\tileor{45}{38}
\tileor{46}{38}
\tileor{47}{38}
\tileor{48}{38}
\tileor{49}{38}
\tileor{50}{38}
\tileor{51}{38}

\doty{15}{20}
\doty{16}{20}
\doty{17}{20}
\doty{18}{20}
\doty{19}{20}
\doty{20}{20}
\doty{21}{20}
\doty{22}{20}
\doty{23}{20}
\doty{23}{21}
\doty{23}{22}
\doty{23}{23}
\doty{23}{24}
\doty{22}{24}
\doty{21}{24}
\doty{20}{24}
\doty{19}{24}
\doty{18}{24}
\doty{17}{24}
\doty{16}{24}
\doty{15}{24}
\doty{14}{24}
\doty{13}{24}
\doty{12}{24}
\doty{12}{25}
\doty{12}{26}
\doty{12}{27}
\doty{12}{28}
\doty{12}{29}
\doty{12}{30}
\doty{13}{30}
\doty{14}{30}
\doty{15}{30}
\doty{16}{30}
\doty{17}{30}
\doty{18}{30}
\doty{18}{29}
\doty{18}{28}
\doty{18}{27}
\doty{19}{27}
\doty{20}{27}
\doty{21}{27}
\doty{22}{27}
\doty{23}{27}
\doty{23}{28}
\doty{23}{29}
\doty{24}{29}
\doty{25}{29}
\doty{26}{29}
\doty{27}{29}
\doty{28}{29}
\doty{28}{30}
\doty{28}{31}
\doty{28}{32}
\doty{28}{33}
\doty{29}{33}
\doty{30}{33}
\doty{31}{33}
\doty{32}{33}
\doty{33}{33}
\doty{34}{33}
\doty{35}{33}
\doty{36}{33}
\doty{37}{33}
\doty{38}{33}
\doty{39}{33}
\doty{40}{33}
\doty{41}{33}
\doty{42}{33}
\doty{43}{33}
\doty{44}{33}
\doty{45}{33}
\doty{46}{33}
\doty{47}{33}
\doty{47}{34}
\doty{47}{35}
\doty{47}{36}
\doty{47}{37}
\doty{47}{38}
\doty{48}{38}
\doty{49}{38}
\doty{50}{38}
\doty{51}{38}

\dotmb{15}{16}
\dotmb{16}{16}
\dotmb{17}{16}
\dotmb{18}{16}
\dotmb{19}{16}
\dotmb{20}{16}
\dotmb{21}{16}
\dotmb{22}{16}
\dotmb{22}{15}
\dotmb{23}{15}
\dotmb{24}{15}
\dotmb{25}{15}
\dotmb{26}{15}
\dotmb{27}{15}
\dotmb{28}{15}
\dotmb{29}{15}
\dotmb{30}{15}
\dotmb{31}{15}
\dotmb{32}{15}
\dotmb{33}{15}
\dotmb{34}{15}
\dotmb{35}{15}
\dotmb{36}{15}
\dotmb{37}{15}
\dotmb{38}{15}
\dotmb{39}{15}
\dotmb{40}{15}
\dotmb{41}{15}
\dotmb{42}{15}
\dotmb{43}{15}
\dotmb{44}{15}
\dotmb{45}{15}
\dotmb{46}{15}
\dotmb{47}{15}
\dotmb{48}{15}
\dotmb{49}{15}
\dotmb{50}{15}
\dotmb{51}{15}

\path [dotted, draw, thin] (0,13) grid[step=0.2cm] (60,43);

\draw [dashed] (15,13) -| (15,43);
\draw [dashed] (51,13) -| (51,43);

\fill (14.5,30.5) circle (0.16);
\node (D) at (14,32.2) {$P_{\lastc}$};

\fill (51.5,15.5) circle (0.16);
\node (D) at (53.2,15.5) {$S_{0}$};

\fill (50.5,38.5) circle (0.16);
\node (D) at (50,40.2) {$P_{\last}$};

\fill (15.5,16.5) circle (0.16);
\node (D) at (15.5,15) {$S_{e}$};

\fill (14.5,20.5) circle (0.16);
\node (D) at (13.5,18.8) {$P_{\main{\goB}}$};

\end{tikzpicture}
\caption{A path $P$ which is canonical for glue column $c$. Only $P_{\main{\goB},\ldots,\last+1}$ is represented (the orange tiles). In this example, the main southern $\glueP{\main{\goB}}{\main{\goB}+1}$ is protected by the beginning $S_{0,\ldots,e}$ of the shield $S$. Only $S_{0,\ldots,e}$ is represented by the blue dots. The colored area (orange and yellow) represents the protected area  $\area{\main{\goB}}$. The practical form $\rightF$ of $P$ at index $\main{\goB}$ corresponds to the yellow dots. The restrained area  $\areaResUn{\main{\goB}}$ of $P$ is in yellow.}
\label{fig:Final:Combining:Pexample}
\end{figure}

\begin{lemma}
\label{lem:full:assWrongArc}
Consider the notations \ref{notation:full:combi}, $\rightFUn_{0,\ldots,|R|-2}+\vu$ is not inside $\areaResDeux{\proSLast{\colDeux}}$. 
\end{lemma}

\begin{proof}
Without loss of generality, we suppose that $\glueP{\main{\goB}}{\main{\goB}+1}$ and $\glueQ{\proSLast{\colDeux}}{\proSLast{\colDeux}+1}$ are southern main glues.
For the sake of contradiction, suppose that $\rightFUn_{0,\ldots,|\rightFUn|-2}+\vu$ is inside $\areaResDeux{\proSLast{\colDeux}}$, see Figure \ref{fig:Final:Combining:FirstPartOne}. We remind that by definition of $\rightFUn_0$, we have $\rightFUn_0=P_{\main{\goB}+1}$. Also, by definition of $\vu$, $\pos{\rightFUn_0}+\vu=\pos{Q_\proSLast{\colDeux}+1}$. By hypothesis, $\glueP{\main{\goB}}{\main{\goB+1}}$ and $\glueQ{\proSLast{\colDeux}}{\proSLast{\colDeux}+1}$ have the same type thus $\rightFUn_0+\vu$ and $Q_{\proSLast{\colDeux}}$ binds together. Then by Fact \ref{fact:full:area}, $Q_{0,\ldots,\proSLast{\colDeux}}(\rightFUn_{0,\ldots,|R|-2}+\vu)$ is a producible path. 

\input{./tikz/Final/CombiningMain/FirstPartOne}

Firstly, suppose that there is no intersection between $\rightFUn_{0,\ldots,|\rightFUn|-2}$ and $P_{\last+1,\ldots,|P|-1}$. By definition of $\rightFUn$, we have $\rightFUn_{|\rightFUn|-2}=P_{\last}$ then $\rightFUn_{|\rightFUn|-2}$ and $P_{\last+1}$ binds together (which also true for $\rightFUn_{|\rightFUn|-2}+\vu$ and $P_{\last+1}+\vu$). Let $0 \leq \lastcDeux \leq |Q|-1$ be the index of the last glue of $Q$ on glue column $\colDeux$. Since $\glueQ{\proSLast{\colDeux}}{\proSLast{\colDeux}+1}$ is on glue column $\colDeux$, we have $\proSLast{\colDeux} \leq \lastcDeux$. Then, it is possible to apply Lemma \ref{Uturn:Conc} on $Q_{0,\ldots,\proSLast{\colDeux}}(\rightFUn+\vu)$ and $P_{\last+2,\ldots,|P|-1}+\vu$ (the bound is satisfied since $\Bfour(\colDeux)=\Bfour(\colUn)+(c'-c)$, see Fact \ref{fact:true:Bthree}). We obtain that $Q_{0,\ldots,\proSLast{\colDeux}}(\rightFUn+\vu)(P_{\last+2,\ldots,|P|-1}+\vu)$ is a subassembly of $\uniterm$. Nevertheless, since $P$ is canonical for glue column $c$ then by definition $x_{P_{|P|-1}}=e_\uniterm$ and then $x_{P_{|P|-1}+\vu}>e_\uniterm$. This is a contradiction. 

Secondly, suppose that there are some intersections between $\rightFUn_{0,\ldots,|\rightFUn|-2}$ and $P_{\last+1,\ldots,|P|-1}$, see Figure \ref{fig:Final:Combining:FirstPartTwo}. Since, $P$ and $P_{0,\ldots,\mainUn{\goBUn}}\rightFUn$ are producible path, then all intersections between $\rightFUn_{0,\ldots,|\rightFUn|-2}$ and $P_{\last+1,\ldots,|P|-1}$ are agreements. Let $$u=\max\{i: \pos{P_i} \text{ is occupied by a tile of } \rightFUn \}$$ and let $u'$ such that $\rightFUn_{u'}=P_u$. We are in case where $u \geq \last+1$ and by definition of $u$, $\rightFUn_{0,\ldots,u'}P_{u+1,\ldots,|P|-1}$ is a path. For the sake of contradiction suppose that $\rightFUn_{0,\ldots,u'}P_{u+1,\ldots,|P|-1}$ is more left-priority than $\rightFUn$ then $P_{u+1}$ is inside a finite area delimited by $\rightFUn$ and $P_{\mainUn{\goBUn},\ldots,\last+1}$, see Figure \ref{fig:Final:Combining:FirstPartTwo}. Nevertheless, since $x_{P_{|P|-1}}=e_{\uniterm}>\Bfour(\colUn)+0.5$ and $\Bfour(\colUn)+0.5=\max\{e_\rightFUn,e_{P_{\mainUn{\goBUn},\ldots,\last+1}}\}$ then $P_{u+1,\ldots,|P|-1}$ must intersect with either $P_{\mainUn{\goBUn},\ldots,\last+1}$ (which cannot occur since $P$ is self-avoiding) or $\rightFUn$ (which cannot occur by definition of $u$). This is a contradiction and then $\rightFUn_{0,\ldots,u'}P_{u+1,\ldots,|P|-1}$ is more right-priority than $\rightFUn$. Since $(\rightFUn+\vu)$ is inside $\areaResDeux{\proSLast{\colDeux}}$ then $Q_{\proSLast{\colDeux}}(\rightFUn_{0,\ldots,u'}P_{u+1,\ldots,|P|-1})+\vu$ is more right-priority than $Q_{\proSLast{\colDeux}}\rightFDeux$, see Figure \ref{fig:Final:Combining:FirstPartThree}. Note that $u\geq \last$ and then $w_{P_{u,\ldots,|P|-1}}>\colUn$ (implying that $w_{P_{u,\ldots,|P|-1}}+\vu>c'$). Then, by Lemma \ref{lem:full:turnRestrainedVis} (if $\glueQ{\proSLast{\colDeux}}{\proSLast{\colDeux}+1}$ is visible from the south) or Lemma \ref{lem:full:turnRestrained} (if $\glueQ{\proSLast{\colDeux}}{\proSLast{\colDeux}+1}$ is not visible from the south), $(\rightFUn_{0,\ldots,u'}P_{u+1,\ldots,|P|-1})+\vu$ can leave $\areaResDeux{\proSLast{\colDeux}}$ only by intersecting with glue column $\Bfour(\colDeux)$. Let $u\leq v \leq |P|-1$ such that $\glueP{v}{v+1}$ is the first glue of $(\rightFUn_{0,\ldots,u'}P_{u+1,\ldots,|P|-1})+\vu$ on glue column $\Bfour(\colDeux)$. then $(\rightFUn_{0,\ldots,u'}P_{u+1,\ldots,v})+\vu$ is inside $\areaResDeux{\proSLast{\colDeux}}$ and $Q_{0,\ldots,\proSLast{\colDeux}}(\rightFUn_{0,\ldots,u'}P_{u+1,\ldots,v})+\vu$ is producible by Fact \ref{fact:full:area}. We can now apply Lemma \ref{Uturn:Conc} on  $Q_{0,\ldots,\proSLast{\colDeux}}(\rightFUn_{0,\ldots,u'}P_{u+1,\ldots,v+1})+\vu$ and $(P_{v+2,\ldots,|P|-1})+\vu$  and  we obtain that $Q_{0,\ldots,\proSLast{\colDeux}}(\rightFUn_{0,\ldots,u'}P_{u+1,\ldots,|P|-1})+\vu$ is producible. Nevertheless, since $P$ is canonical for glue column $\colUn$ then by definition $x_{P_{|P|-1}}=e_\uniterm$ and then $x_{P_{|P|-1}+\vu}>e_\uniterm$. This is a contradiction. 

\begin{figure}
\center
\begin{tikzpicture}[x=0.2cm,y=0.2cm]

\fill[fill=orange!30!white, draw opacity=0.8] (15,20.5) -| (23.5,24.5) -| (12.5,30.5) -| (18.5,38.5) -| (51,15.5) -| (22.5,16.5) -| (15,20.5);
\fill[fill=yellow!30!white, draw opacity=0.8] (15,20.5) -| (23.5,24.5) -| (12.5,30.5) -| (18.5,27.5) -| (23.5,29.5) -| (28.5,33.5) -| (47.5,38.5) -| (51,15.5) -| (22.5,16.5) -| (15,20.5);

\draw[thick] (14.5,20.5) -| (23.5,24.5) -| (12.5,30.5) -| (18.5,27.5) -| (23.5,29.5) -| (28.5,33.5) -| (47.5,38.5) -| (51.5,38.5); 
\draw[very thick] (14.5,20.5) -| (23.5,24.5) -| (12.5,30.5) -| (18.5,38.5) -| (53.5,35.5) -| (41.5,30.5) -| (56.5,30.5);
\draw[thick] (15,16.5) -| (22.5,15.5) -| (51.5,15.5);



\tile{14}{20}

\tileor{15}{20}
\tileor{16}{20}
\tileor{17}{20}
\tileor{18}{20}
\tileor{19}{20}
\tileor{20}{20}
\tileor{21}{20}
\tileor{22}{20}
\tileor{23}{20}
\tileor{23}{21}
\tileor{23}{22}
\tileor{23}{23}
\tileor{23}{24}
\tileor{22}{24}
\tileor{21}{24}
\tileor{20}{24}
\tileor{19}{24}
\tileor{18}{24}
\tileor{17}{24}
\tileor{16}{24}
\tileor{15}{24}
\tileor{14}{24}
\tileor{13}{24}
\tileor{12}{24}
\tileor{12}{25}
\tileor{12}{26}
\tileor{12}{27}
\tileor{12}{28}
\tileor{12}{29}
\tileor{12}{30}
\tileor{13}{30}
\tileor{14}{30}
\tileor{15}{30}
\tileor{16}{30}
\tileor{17}{30}
\tileor{18}{30}
\tileor{18}{31}
\tileor{18}{32}
\tileor{18}{33}
\tileor{18}{34}
\tileor{18}{35}
\tileor{18}{36}
\tileor{18}{37}
\tileor{18}{38}
\tileor{19}{38}
\tileor{20}{38}
\tileor{21}{38}
\tileor{22}{38}
\tileor{23}{38}
\tileor{24}{38}
\tileor{25}{38}
\tileor{26}{38}
\tileor{27}{38}
\tileor{28}{38}
\tileor{29}{38}
\tileor{30}{38}
\tileor{31}{38}
\tileor{32}{38}
\tileor{33}{38}
\tileor{34}{38}
\tileor{35}{38}
\tileor{36}{38}
\tileor{37}{38}
\tileor{38}{38}
\tileor{39}{38}
\tileor{40}{38}
\tileor{41}{38}
\tileor{42}{38}
\tileor{43}{38}
\tileor{44}{38}
\tileor{45}{38}
\tileor{46}{38}
\tileor{47}{38}
\tileor{48}{38}
\tileor{49}{38}
\tileor{50}{38}
\tileor{51}{38}

\doty{15}{20}
\doty{16}{20}
\doty{17}{20}
\doty{18}{20}
\doty{19}{20}
\doty{20}{20}
\doty{21}{20}
\doty{22}{20}
\doty{23}{20}
\doty{23}{21}
\doty{23}{22}
\doty{23}{23}
\doty{23}{24}
\doty{22}{24}
\doty{21}{24}
\doty{20}{24}
\doty{19}{24}
\doty{18}{24}
\doty{17}{24}
\doty{16}{24}
\doty{15}{24}
\doty{14}{24}
\doty{13}{24}
\doty{12}{24}
\doty{12}{25}
\doty{12}{26}
\doty{12}{27}
\doty{12}{28}
\doty{12}{29}
\doty{12}{30}
\doty{13}{30}
\doty{14}{30}
\doty{15}{30}
\doty{16}{30}
\doty{17}{30}
\doty{18}{30}
\doty{18}{29}
\doty{18}{28}
\doty{18}{27}
\doty{19}{27}
\doty{20}{27}
\doty{21}{27}
\doty{22}{27}
\doty{23}{27}
\doty{23}{28}
\doty{23}{29}
\doty{24}{29}
\doty{25}{29}
\doty{26}{29}
\doty{27}{29}
\doty{28}{29}
\doty{28}{30}
\doty{28}{31}
\doty{28}{32}
\doty{28}{33}
\doty{29}{33}
\doty{30}{33}
\doty{31}{33}
\doty{32}{33}
\doty{33}{33}
\doty{34}{33}
\doty{35}{33}
\doty{36}{33}
\doty{37}{33}
\doty{38}{33}
\doty{39}{33}
\doty{40}{33}
\doty{41}{33}
\doty{42}{33}
\doty{43}{33}
\doty{44}{33}
\doty{45}{33}
\doty{46}{33}
\doty{47}{33}
\doty{47}{34}
\doty{47}{35}
\doty{47}{36}
\doty{47}{37}
\doty{47}{38}
\doty{48}{38}
\doty{49}{38}
\doty{50}{38}
\doty{51}{38}

\dotmb{15}{16}
\dotmb{16}{16}
\dotmb{17}{16}
\dotmb{18}{16}
\dotmb{19}{16}
\dotmb{20}{16}
\dotmb{21}{16}
\dotmb{22}{16}
\dotmb{22}{15}
\dotmb{23}{15}
\dotmb{24}{15}
\dotmb{25}{15}
\dotmb{26}{15}
\dotmb{27}{15}
\dotmb{28}{15}
\dotmb{29}{15}
\dotmb{30}{15}
\dotmb{31}{15}
\dotmb{32}{15}
\dotmb{33}{15}
\dotmb{34}{15}
\dotmb{35}{15}
\dotmb{36}{15}
\dotmb{37}{15}
\dotmb{38}{15}
\dotmb{39}{15}
\dotmb{40}{15}
\dotmb{41}{15}
\dotmb{42}{15}
\dotmb{43}{15}
\dotmb{44}{15}
\dotmb{45}{15}
\dotmb{46}{15}
\dotmb{47}{15}
\dotmb{48}{15}
\dotmb{49}{15}
\dotmb{50}{15}
\dotmb{51}{15}

\tiler{52}{38}
\tiler{53}{38}
\tiler{53}{37}
\tiler{53}{36}
\tiler{53}{35}
\tiler{52}{35}
\tiler{51}{35}
\tiler{50}{35}
\tiler{49}{35}
\tiler{48}{35}
\tiler{47}{35}
\tiler{46}{35}
\tiler{45}{35}
\tiler{44}{35}
\tiler{43}{35}
\tiler{42}{35}
\tiler{41}{35}
\tiler{41}{34}
\tiler{41}{33}
\tiler{41}{32}
\tiler{41}{31}
\tiler{41}{30}
\tiler{42}{30}
\tiler{43}{30}
\tiler{44}{30}
\tiler{45}{30}
\tiler{46}{30}
\tiler{47}{30}
\tiler{48}{30}
\tiler{49}{30}
\tiler{50}{30}
\tiler{51}{30}
\tiler{52}{30}
\tiler{53}{30}
\tiler{54}{30}
\tiler{55}{30}
\tiler{56}{30}

\path [dotted, draw, thin] (0,13) grid[step=0.2cm] (60,43);

\draw [dashed] (15,13) -| (15,43);
\draw [dashed] (51,13) -| (51,43);
\draw [dashed] (56.5,13) -| (56.5,43);

\fill (14.5,30.5) circle (0.16);
\node (D) at (14,32.2) {$P_{\lastc}$};

\fill (56.5,30.5) circle (0.16);
\node (D) at (56.5,32.2) {$P_{|P|-1}$};

\fill (51.5,15.5) circle (0.16);
\node (D) at (53.2,15.5) {$S_{0}$};

\fill (50.5,38.5) circle (0.16);
\node (D) at (50,40.2) {$P_{\last}$};

\fill (41.5,33.5) circle (0.16);
\node (D) at (45.7,32.2) {$P_{u}=\rightF_{u'}$};

\fill (15.5,16.5) circle (0.16);
\node (D) at (15.5,15) {$S_{e}$};

\fill (14.5,20.5) circle (0.16);
\node (D) at (13.5,18.8) {$P_{\main{\goB}}$};

\end{tikzpicture}
\caption{Proof of Lemma \ref{lem:full:assWrongArc}, part (2/3): The red tiles represent $P_{\last+2,\ldots,|P|-1}$, the end of path $P$ from Figure \ref{fig:Final:Combining:Pexample}. They intersect with $\rightF$ (the yellow dots). We fuse the yellow dots and the red tiles at $P_{u}=\rightF_{u'}$. This new path turns right of $\rightF$. Otherwise, $P_{u,\ldots,|P|-1}$ would be inside the orange area but this area cannot contain $P_{|P|-1}$ which is on column~$e_\uniterm$. }
\label{fig:Final:Combining:FirstPartTwo}
\end{figure}
\input{./tikz/Final/CombiningMain/FirstPartThree}
\end{proof}

The following lemma corresponds to the reasoning of the roadmap where we explore the consequence of the first case of Fact \ref{fact:turn}. Here, this is the first step toward assembling a new shield for path $P$.

\begin{lemma}
\label{lem:full:goodArc}
Consider the notations \ref{notation:full:combi}, then $P_{0,\ldots,\mainUn{\goBUn}}(\rightFDeux-\vu)$ is producible. Moreover, there exists $0\leq a \leq |\rightFUn|-1$ such that $(\rev{\rightFDeux_{a,\ldots,|\rightFDeux|-1}}-\vu)\rightFUn_{a+1,|\rightFUn|-1}$ is a negative upward (resp. downward) arc of glue column $\Bfour(\colUn)$ and its interior is included into $\areaRes{\mainUn{\goBUn}}$. 
\end{lemma}

\begin{proof}
Without loss of generality, we suppose that $\glueP{\main{\goB}}{\main{\goB}+1}$ and $\glueQ{\proSLast{\colDeux}}{\proSLast{\colDeux}+1}$ are southern main glues. 
We remind that by definition of $\rightFUn_0$, we have $\rightFUn_0=P_{\main{\goB}+1}$. Also, by definition of $\vu$, $\pos{\rightFUn_0}+\vu=\pos{Q_\proSLast{\colDeux}+1}$. By hypothesis, $\glueP{\main{\goB}}{\main{\goB+1}}$ and $\glueQ{\proSLast{\colDeux}}{\proSLast{\colDeux}+1}$ have the same type thus $\rightFUn_0+\vu$ and $Q_{\proSLast{\colDeux}}$ binds together. Since the tile assembly system is directed, we have $\rightFUn_0+\vu=\rightFDeux_0$. By Lemma \ref{lem:full:assWrongArc}, $(\rightFUn_{0,\ldots,|\rightFUn|-2}+\vu)$ is not inside $\areaResDeux{\proSLast{\colDeux}}$. Thus $(\rightFUn_{0,\ldots,|\rightFUn|-2}+\vu)$ cannot be a prefix of $\rightFDeux$. 

For the sake of contradiction, suppose that $Q_{\proSLast{\colDeux}}(\rightFUn+\vu)$ is more right-priority than $Q_{\proSLast{\colDeux}}\rightFDeux$, see Figure \ref{fig:Final:Combining:NewPathP} for a similar reasoning to show that $\rightFDeux-\vu$ does not cross the segment between $\glueP{\main{\goB}}{\main{\goB}+1}$ and $\glueS{e}{e+1}$. Consider $v$ such that $(\rightFUn_{0,\ldots,v}+\vu)$ is the largest prefix of $(\rightFUn+\vu)$ which is inside $\areaResDeux{\proSLast{\colDeux}}$. Since $(\rightFUn_{0,\ldots,|\rightFUn|-2}+\vu)$ is not inside $\areaResDeux{\proSLast{\colDeux}}$ then $v<|\rightFUn|-2$. Moreover, since $(\rightFUn_{0,\ldots,|\rightFUn|-2}+\vu)< \Bfour(\colDeux)$, then by Lemma \ref{lem:full:turnRestrained} or Lemma \ref{lem:full:turnRestrainedVis}, $\glueR{v}{v+1}$ is on the vertical ray starting at $\glueQ{\proSLast{\colDeux}}{\proSLast{\colDeux}+1}$ and going south. Note that this glue should point west to leave $\areaResDeux{\proSLast{\colDeux}}$. Then, consider $0 \leq u\leq v$ such that $(\rightFUn_{u,\ldots,v+1}+\vu)$ is a positive arc on column $\colDeux$. If $\glueQ{\proSLast{\colDeux}}{\proSLast{\colDeux}+1}$ is visible from the south then $\glueR{v}{v+1}$ is on the glue ray of $\glueQ{\proSLast{\colDeux}}{\proSLast{\colDeux}+1}$. The $\glueR{u}{u+1}$ cannot be on this glue ray by minimality of $v$ and $(\rightFUn_{u,\ldots,v+1}+\vu)$ must be a downward arc. If $\glueQ{\proSLast{\colDeux}}{\proSLast{\colDeux}+1}$ is not visible from the south then it is protected by the beginning $S'_{0,\ldots,e'}$ of a shield $S'$. Then, $(\rightFUn_{u,\ldots,v+1}+\vu)$ must be a downward arc to avoid $S'_{0,\ldots,e'}$. In all cases,  $(\rightFUn_{u,\ldots,v+1}+\vu)$ is a downward arc.   Then $\rightFUn_{u,\ldots,v+1}$ is a positive downward arc on column $\colUn$ such that $y_{\rightFUn_v}<y_{P_{\main{\goB}}}\leq y_{\rightFUn_u}$.
If $\goB=0$ then $\rightFUn_{u,\ldots,v+1}$ would intersect the glue ray of $\glueP{\main{0}}{\main{0}+1}$ contradicting the definition of $R$. If $\goB>0$ then there exists a shield $S$ and $0\leq e \leq |S|-1$ such that $\glueP{\main{\goB}}{\main{\goB}+1}$ is protected by the end $S_{0,\ldots,e}$ of shield $S$. In this case, either $\rightFUn_{u,\ldots,v+1}$ intersects $S_{0,\ldots,e}$ (which is a contradiction by Lemma \ref{lem:full:visible}) or $\rightFUn_{u,\ldots,v+1}$ leaves $\areaResUn{\main{\goB}}$ by crossing the vertical segment between $\glueP{\main{\goB}}{\main{\goB}+1}$ and $\glueS{e}{e+1}$ (contradicting the definition of $R$).

The last possible case is that $\rightFDeux$ is more right-priority than $(\rightFUn+\vu)$, see Figure \ref{fig:Final:Combining:NewPathP}. This means that $\rightFDeux-\vu$ is more right-priority than $\rightFUn$. Remark that the paths $\rightFUn$ and $\rightFDeux$ and the protected areas $\areaResUn{\main{\goBUn}}$ and $\areaResDeux{\proSLast{\colDeux}}$ have the same properties. Then with the same kind of reasoning, it is possible to prove that $\rightFDeux_{0,\ldots,|\rightFDeux|-2}-\vu$ is inside $\areaRes{\main{\goBUn}}$. By Fact \ref{fact:full:area}, we obtain that $P_{0,\ldots,\mainUn{\goBUn}}(\rightFDeux_{0,\ldots,|\rightFDeux|-2}-\vu)$ is producible. Since $x_{\rightFDeux_{|\rightFDeux|-1}-\vu}=\Bfour(\colUn)+0.5$ and $e_{P_{0,\ldots,\mainUn{\goBUn}}(\rightFDeux_{0,\ldots,|\rightFDeux|-2}-\vu)} < \Bfour(\colUn)$ then $P_{0,\ldots,\mainUn{\goBUn}}(\rightFDeux-\vu)$ is also producible. Consider $0\leq a \leq |\rightFUn|-1$ such that $\rightFUn_{0,\ldots,a}$ is the largest common prefix between $\rightFUn$ and $\rightFDeux-\vu$. By Lemma \ref{lem:full:turnRestrainedVis} (if $\goB=0$) or Lemma \ref{lem:full:turnRestrained} (if $\goB>0$), $\rightFUn_{a+1,\ldots,|\rightFUn|-2}$ and $\rightFDeux_{a+1,\ldots,|\rightFDeux|-2}-\vu$ do not intersect. Then $\pos{\rightFUn_{|\rightFUn|-2}}\neq \pos{\rightFDeux_{|\rightFDeux|-2}}-\vu$ and $y_{\rightFUn_{|\rightFUn|-2}}> y_{\rightFDeux_{|\rightFDeux|-2}}-\vu$ (since $x_{\rightFUn_{|\rightFUn|-2}}=x_{\rightFDeux_{|\rightFDeux|-2}-\vu}=\Bfour(c)-0.5$), then we have $\pos{\rightFUn_{|\rightFUn|-1}}\neq \pos{\rightFDeux_{|\rightFDeux|-1}}-\vu$. Then, $\rightFUn_{a+1,\ldots,|\rightFUn|-1}$ and $\rightFDeux_{a+1,\ldots,|\rightFDeux|-1}-\vu$ do not intersect. To conclude, the path $\rev{\rightFDeux_{a,\ldots,|\rightFDeux|-1}-\vu}\rightFUn_{a+1,|\rightFUn|-1}$ is a negative upward arc of column $\Bfour(\colUn)$ and its interior is included into $\areaRes{\mainUn{\goBUn}}$.

\input{./tikz/Final/CombiningMain/PathPExampleTwo}
\end{proof}
 
In Lemma \ref{lem:full:goodArc}, a negative upward arc on column $\Bfour(c)$ which is inside $\areaRes{\mainUn{\goBUn}}$ was found. Now, we focus on the minimum arc in the interior of this arc.
 
 \begin{lemma}
\label{lem:full:MinimalArc}
Consider the notations \ref{notation:full:combi}, there exists an upward (resp. downward) minimum arc $\newA$ on glue column $\Bfour(\colUn)$ such that its interior $\inter{\newA}$ is included into $\areaRes{\mainUn{\goBUn}}$, such that $\newA$ is a subassembly of $\uniterm$, such that $\newA_0=Q_{\lastDeux+1}-\vu$ and $\newA_{|\newA|-1}=P_{\lastUn+1}$.
\end{lemma}

\begin{proof}
By  Lemma \ref{lem:full:goodArc}, there exists a negative upward (resp. downward) arc $B$ of glue column $\Bfour(\colUn)$ such $B_0=\rightFDeux_{|\rightFDeux|-1}-\vu=Q_{\lastDeux+1}-\vu$ and $B_{|B|-1}=\rightFUn_{|\rightFUn|-1}=P_{\lastUn+1}$. Moreover, $B$ is a subassembly of  $\uniterm$ and its interior is included into $\areaRes{\mainUn{\goBUn}}$. The path $\newA$ defined as $\min(\inter{B})$ satisfies the condition of the lemma.
\end{proof}

Under some hypothesis a new shield and a new protected glue can be extracted from the arc $A'$.

\begin{lemma}
\label{lem:full:GoodCase}
Consider the notations \ref{notation:full:combi} and the arc $\newA$ of Lemma \ref{lem:full:MinimalArc}. If $w_{\newA}\leq \colUn$, then either there exists $\goBUn \leq i < \decoUn$ such that $\second{i}$ is protected or there exists $\goBUn < i \leq \decoUn$ such that $\main{i}$ is protected.
 \end{lemma}

\begin{proof}
By definition, we have $\newA_{|\newA|-2}=P_{\last}$, then $\newA$ intersects with $P_{\mainUn{\goBUn},\ldots,\last+1}$. Let $\newS$ be the largest prefix of $\newA$ such that $\newS$ does not intersect $P_{\mainUn{\goBUn},\ldots,\last+1}$. 
Note that, $\newS$ does not intersect with $P_{0,\ldots,\mainUn{\goBUn}-1}$ since the interior of $\newA$ is in $\areaResUn{\main{\goB}}$.
Also, $\newS$ contains at least $\newA_0$ since $x_{\newA_0}=\Bfour(\colUn)+0.5$, $\pos{\newA_0}\neq\pos{P_{\last+1}}$ and $e_{P_{\mainUn{\goBUn},\ldots,\last}}<\Bfour(\colUn)$. 
Since $\newS_{1,\ldots,|\newS|-2}$ is in $\area{\mainUn{\goBUn}}$ and by definition of $\newS$ there exists $\mainUn{\goBUn}<\newProG\leq \last$ such that $\newS P_{\newProG}$ is a prefix of $\newA$.  Then, $\newS P_{\newProG,\ldots,\last+1}$ is also a negative upward arc of glue column $\Bfour(\colUn)$. Moreover, since $\newA$ is minimum then $\newA$ is inside the interior of the arc $\newS P_{\newProG,\ldots,\last+1}$. Then, we have $w_{\newS P_{\newProG,\ldots,\last+1}}\leq w_{\newA} \leq \colUn$. This implies that $\newS$ is a southern shield of $P$ on glue column $\colUn$.  

We remind that $(\main{i})_{0\leq i \leq t}$ and $(\second{i})_{0\leq i < t}$ are the upward decomposition of $P$ into dominants arcs on glue column $c$. Let $(\bord{i})_{0\leq i \leq t+1}$ be the frontiers of this decomposition. We denote by $\mathcal{C}$ the area delimited by $l^{\main{0}}$, $P_{\main{0},\ldots,\last+1}$ and the ray starting at $\glueP{\last}{{\last+1}}$ and going south. Now, either $\newS$ is a full-shield but since $\newS$ is in $\area{\main{\goBUn}}$ then it cannot intersect $\bordUn{i}$ for $i\leq \goBUn$, see Figure \ref{fig:Final:Combining:NewArcFull}. Thus, by Lemma \ref{lem:full:exists} there exists $i>\goBUn$ such that $\main{i}$ is protected and the lemma is true. Now, if $\newS$ is an half-shield then $\newS P_{\newProG}$ is in the east side $\esUn$ of $\mathcal{C}$ and we have $\mainUn{i}<\newProG<\secondUn{i}$ for some  $0\leq i < \decoUn$, see Figure~\ref{fig:Final:Combining:NewArcHalf}. Also, since $\newS_{1,\ldots,|\newS|-1}$ is in $\area{\main{\goBUn}}$, then $i\geq \goBUn$ and the  $\glueP{\secondUn{i}}{\secondUn{i}+1}$ is protected with $\secondUn{i}>\main{\goBUn}$.

\begin{figure}
\center
\begin{tikzpicture}[x=0.2cm,y=0.2cm]


\fill[fill=lightblue, draw opacity=0.8] (51,66.5) -| (9.5,63.5) -| (24.5,60.5) -| (15,39.5) -| (20.5,37.5) -| (30.5,36.5) -| (41.5,41.5) -| (51,41.5);

\draw[very thick] (14.5,20.5) -| (23.5,24.5) -| (6.5,10.5)-| (19.5,6.5) -| (2.5,60.5) -| (24.5,63.5) -| (9.5,66.5) -| (51.5,66.5);
\draw[thick] (51.5,14.5) -| (30.5,12.5) -| (15,12.5);
\draw[thick] (22.5,60.5) -| (22.5,53.5) -|  (10.5,51.5) -| (6.5,39.5) -| (20.5,37.5) -| (30.5,36.5) -| (41.5,41.5) -| (51.5,41.5);




\tile{14}{20}{85}
\tileor{15}{20}{85}
\tileor{16}{20}{85}
\tileor{17}{20}{85}
\tileor{18}{20}{85}
\tileor{19}{20}{85}
\tileor{20}{20}{85}
\tileor{21}{20}{85}
\tileor{22}{20}{85}
\tileor{23}{20}{85}
\tileor{23}{21}{85}
\tileor{23}{22}{85}
\tileor{23}{23}{85}
\tileor{23}{24}{85}
\tileor{22}{24}{85}
\tileor{21}{24}{85}
\tileor{20}{24}{85}
\tileor{19}{24}{85}
\tileor{18}{24}{85}
\tileor{17}{24}{85}
\tileor{16}{24}{85}
\tileor{15}{24}{85}
\tileor{14}{24}{85}
\tileor{13}{24}{85}
\tileor{12}{24}{85}
\tileor{11}{24}{85}
\tileor{10}{24}{85}
\tileor{9}{24}{85}
\tileor{8}{24}{85}
\tileor{7}{24}{85}
\tileor{6}{24}{85}
\tileor{6}{23}{85}
\tileor{6}{22}{85}
\tileor{6}{21}{85}
\tileor{6}{20}{85}
\tileor{6}{19}{85}
\tileor{6}{18}{85}
\tileor{6}{17}{85}
\tileor{6}{16}{85}
\tileor{6}{15}{85}
\tileor{6}{14}{85}
\tileor{6}{13}{85}
\tileor{6}{12}{85}
\tileor{6}{11}{85}
\tileor{6}{10}{85}
\tileor{7}{10}{85}
\tileor{8}{10}{85}
\tileor{9}{10}{85}
\tileor{10}{10}{85}
\tileor{11}{10}{85}
\tileor{12}{10}{85}
\tileor{13}{10}{85}
\tileor{14}{10}{85}
\tileor{15}{10}{85}
\tileor{16}{10}{85}
\tileor{17}{10}{85}
\tileor{18}{10}{85}
\tileor{19}{10}{85}
\tileor{19}{9}{85}
\tileor{19}{8}{85}
\tileor{19}{7}{85}
\tileor{19}{6}{85}
\tileor{18}{6}{85}
\tileor{17}{6}{85}
\tileor{16}{6}{85}
\tileor{15}{6}{85}
\tileor{14}{6}{85}
\tileor{13}{6}{85}
\tileor{12}{6}{85}
\tileor{11}{6}{85}
\tileor{10}{6}{85}
\tileor{9}{6}{85}
\tileor{8}{6}{85}
\tileor{7}{6}{85}
\tileor{6}{6}{85}
\tileor{5}{6}{85}
\tileor{4}{6}{85}
\tileor{3}{6}{85}
\tileor{2}{6}{85}
\tileor{2}{7}{85}
\tileor{2}{8}{85}
\tileor{2}{9}{85}
\tileor{2}{10}{85}
\tileor{2}{11}{85}
\tileor{2}{12}{85}
\tileor{2}{13}{85}
\tileor{2}{14}{85}
\tileor{2}{15}{85}
\tileor{2}{16}{85}
\tileor{2}{17}{85}
\tileor{2}{18}{85}
\tileor{2}{19}{85}
\tileor{2}{20}{85}
\tileor{2}{21}{85}
\tileor{2}{22}{85}
\tileor{2}{23}{85}
\tileor{2}{24}{85}
\tileor{2}{25}{85}
\tileor{2}{26}{85}
\tileor{2}{27}{85}
\tileor{2}{28}{85}
\tileor{2}{29}{85}
\tileor{2}{30}{85}
\tileor{2}{31}{85}
\tileor{2}{32}{85}
\tileor{2}{33}{85}
\tileor{2}{34}{85}
\tileor{2}{35}{85}
\tileor{2}{36}{85}
\tileor{2}{37}{85}
\tileor{2}{38}{85}
\tileor{2}{39}{85}
\tileor{2}{40}{85}
\tileor{2}{41}{85}
\tileor{2}{42}{85}
\tileor{2}{43}{85}
\tileor{2}{44}{85}
\tileor{2}{45}{85}
\tileor{2}{46}{85}
\tileor{2}{47}{85}
\tileor{2}{48}{85}
\tileor{2}{49}{85}
\tileor{2}{50}{85}
\tileor{2}{51}{85}
\tileor{2}{52}{85}
\tileor{2}{53}{85}
\tileor{2}{54}{85}
\tileor{2}{55}{85}
\tileor{2}{56}{85}
\tileor{2}{57}{85}
\tileor{2}{58}{85}
\tileor{2}{59}{85}
\tileor{2}{60}{85}
\tileor{3}{60}{85}
\tileor{4}{60}{85}
\tileor{5}{60}{85}
\tileor{6}{60}{85}
\tileor{7}{60}{85}
\tileor{8}{60}{85}
\tileor{9}{60}{85}
\tileor{10}{60}{85}
\tileor{11}{60}{85}
\tileor{12}{60}{85}
\tileor{13}{60}{85}
\tileor{14}{60}{85}
\tileor{15}{60}{85}
\tileor{16}{60}{85}
\tileor{17}{60}{85}
\tileor{18}{60}{85}
\tileor{19}{60}{85}
\tileor{20}{60}{85}
\tileor{21}{60}{85}
\tileor{22}{60}{85}
\tileor{23}{60}{85}
\tileor{24}{60}{85}
\tileor{24}{61}{85}
\tileor{24}{62}{85}
\tileor{24}{63}{85}
\tileor{23}{63}{85}
\tileor{22}{63}{85}
\tileor{21}{63}{85}
\tileor{20}{63}{85}
\tileor{19}{63}{85}
\tileor{18}{63}{85}
\tileor{17}{63}{85}
\tileor{16}{63}{85}
\tileor{15}{63}{85}
\tileor{14}{63}{85}
\tileor{13}{63}{85}
\tileor{12}{63}{85}
\tileor{11}{63}{85}
\tileor{10}{63}{85}
\tileor{9}{63}{85}
\tileor{9}{64}{85}
\tileor{9}{65}{85}
\tileor{9}{66}{85}
\tileor{10}{66}{85}
\tileor{11}{66}{85}
\tileor{12}{66}{85}
\tileor{13}{66}{85}
\tileor{14}{66}{85}
\tileor{15}{66}{85}
\tileor{16}{66}{85}
\tileor{17}{66}{85}
\tileor{18}{66}{85}
\tileor{19}{66}{85}
\tileor{20}{66}{85}
\tileor{21}{66}{85}
\tileor{22}{66}{85}
\tileor{23}{66}{85}
\tileor{24}{66}{85}
\tileor{25}{66}{85}
\tileor{26}{66}{85}
\tileor{27}{66}{85}
\tileor{28}{66}{85}
\tileor{29}{66}{85}
\tileor{30}{66}{85}
\tileor{31}{66}{85}
\tileor{32}{66}{85}
\tileor{33}{66}{85}
\tileor{34}{66}{85}
\tileor{35}{66}{85}
\tileor{36}{66}{85}
\tileor{37}{66}{85}
\tileor{38}{66}{85}
\tileor{39}{66}{85}
\tileor{40}{66}{85}
\tileor{41}{66}{85}
\tileor{42}{66}{85}
\tileor{43}{66}{85}
\tileor{44}{66}{85}
\tileor{45}{66}{85}
\tileor{46}{66}{85}
\tileor{47}{66}{85}
\tileor{48}{66}{85}
\tileor{49}{66}{85}
\tileor{50}{66}{85}
\tileor{51}{66}{85}

\dotlb{22}{60}{85}
\dotlb{22}{59}{85}
\dotlb{22}{58}{85}
\dotlb{22}{57}{85}
\dotlb{22}{56}{85}
\dotlb{22}{55}{85}
\dotlb{22}{54}{85}
\dotlb{22}{53}{85}
\dotlb{21}{53}{85}
\dotlb{20}{53}{85}
\dotlb{19}{53}{85}
\dotlb{18}{53}{85}
\dotlb{17}{53}{85}
\dotlb{16}{53}{85}
\dotlb{15}{53}{85}
\dotlb{14}{53}{85}
\dotlb{13}{53}{85}
\dotlb{12}{53}{85}
\dotlb{11}{53}{85}
\dotlb{10}{53}{85}
\dotlb{10}{52}{85}
\dotlb{10}{51}{85}
\dotlb{41}{41}{85}
\dotlb{42}{41}{85}
\dotlb{43}{41}{85}
\dotlb{44}{41}{85}
\dotlb{45}{41}{85}
\dotlb{46}{41}{85}
\dotlb{47}{41}{85}
\dotlb{48}{41}{85}
\dotlb{49}{41}{85}
\dotlb{50}{41}{85}
\dotlb{51}{41}{85}

\dotmb{51}{14}
\dotmb{50}{14}
\dotmb{49}{14}
\dotmb{48}{14}
\dotmb{47}{14}
\dotmb{46}{14}
\dotmb{45}{14}
\dotmb{44}{14}
\dotmb{43}{14}
\dotmb{42}{14}
\dotmb{41}{14}
\dotmb{40}{14}
\dotmb{39}{14}
\dotmb{38}{14}
\dotmb{37}{14}
\dotmb{36}{14}
\dotmb{35}{14}
\dotmb{34}{14}
\dotmb{33}{14}
\dotmb{32}{14}
\dotmb{31}{14}
\dotmb{30}{14}
\dotmb{30}{13}
\dotmb{30}{12}
\dotmb{29}{12}
\dotmb{28}{12}
\dotmb{27}{12}
\dotmb{26}{12}
\dotmb{25}{12}
\dotmb{24}{12}
\dotmb{23}{12}
\dotmb{22}{12}
\dotmb{21}{12}
\dotmb{20}{12}
\dotmb{19}{12}
\dotmb{18}{12}
\dotmb{17}{12}
\dotmb{16}{12}
\dotmb{15}{12}

\dotlb{9}{51}
\dotlb{8}{51}
\dotlb{7}{51}
\dotlb{6}{51}
\dotlb{6}{50}
\dotlb{6}{49}
\dotlb{6}{48}
\dotlb{6}{47}
\dotlb{6}{46}
\dotlb{6}{45}
\dotlb{6}{44}
\dotlb{6}{43}
\dotlb{6}{42}
\dotlb{6}{41}
\dotlb{6}{40}
\dotlb{6}{39}
\dotlb{6}{39}
\dotlb{7}{39}
\dotlb{8}{39}
\dotlb{9}{39}
\dotlb{10}{39}
\dotlb{11}{39}
\dotlb{12}{39}
\dotlb{13}{39}
\dotlb{14}{39}
\dotlb{15}{39}
\dotlb{16}{39}
\dotlb{17}{39}
\dotlb{18}{39}
\dotlb{19}{39}
\dotlb{20}{39}
\dotlb{20}{38}
\dotlb{20}{37}
\dotlb{21}{37}
\dotlb{22}{37}
\dotlb{23}{37}
\dotlb{24}{37}
\dotlb{25}{37}
\dotlb{26}{37}
\dotlb{27}{37}
\dotlb{28}{37}
\dotlb{29}{37}
\dotlb{30}{37}
\dotlb{30}{36}
\dotlb{31}{36}
\dotlb{32}{36}
\dotlb{33}{36}
\dotlb{34}{36}
\dotlb{35}{36}
\dotlb{36}{36}
\dotlb{37}{36}
\dotlb{38}{36}
\dotlb{39}{36}
\dotlb{40}{36}
\dotlb{41}{36}
\dotlb{41}{37}
\dotlb{41}{38}
\dotlb{41}{39}
\dotlb{41}{40}

\path [dotted, draw, thin] (0,5) grid[step=0.2cm] (60,69);

\draw [dashed] (15,5) -| (15,69);
\draw [dashed] (51,5) -| (51,69);
\draw [thick, color=medblue] (15,39.5) -| (15,60.5);

\fill (14.5,66.5) circle (0.16);
\node (D) at (14,68.2) {$P_{\lastc}$};

\fill (50.5,66.5) circle (0.16);
\node (D) at (50.5,68.2) {$P_{\last}$};

\fill (51.5,41.5) circle (0.16);
\node (D) at (53.2,41.5) {$A'_0$};

\fill (51.5,14.5) circle (0.16);
\node (D) at (53.2,14.5) {$S_{0}$};

\fill (15.5,12.5) circle (0.16);
\node (D) at (13.5,12.5) {$S_{e}$};

\fill (14.5,60.5) circle (0.16);
\node (D) at (13.5,58.8) {$P_{\main{\goB+1}}$};

\fill (22.5,60.5) circle (0.16);
\node (D) at (22.5,62) {$P_{a'}$};

\fill (14.5,20.5) circle (0.16);
\node (D) at (13.5,18.8) {$P_{\main{\goB}}$};

\end{tikzpicture}
\caption{Proof of Lemma \ref{lem:full:GoodCase}, part (1/2): the light blue dots represent the beginning of $A'=\min(\inter{A})$ until it intersects with $P_{\main{\goB},\ldots,\last+1}$ and where $A$ is the arc of Figure \ref{fig:Final:Combining:NewPathP}. In this case, the blue dots are a full shield. This full shield protects $\glueP{\main{\goB+1}}{\main{\goB+1}+1}$ and the protected area $\area{\main{\goB+1}}$ is in light blue.}
\label{fig:Final:Combining:NewArcFull}
\end{figure}
\begin{figure}
\center
\begin{tikzpicture}[x=0.2cm,y=0.2cm]


\fill[fill=lightblue, draw opacity=0.8] (51,66.5) -| (9.5,63.5) -| (24.5,60.5) -| (22.5,37.5) -| (30.5,36.5) -| (41.5,41.5) -| (51,41.5);

\draw[very thick] (14.5,20.5) -| (23.5,24.5) -| (6.5,10.5)-| (19.5,6.5) -| (2.5,60.5) -| (24.5,63.5) -| (9.5,66.5) -| (51.5,66.5);
\draw[thick] (51.5,14.5) -| (30.5,12.5) -| (15,12.5);
\draw[thick] (22.5,60.5) -| (22.5,37.5) -| (30.5,36.5) -| (41.5,41.5) -| (51.5,41.5);




\tile{14}{20}{85}
\tileor{15}{20}{85}
\tileor{16}{20}{85}
\tileor{17}{20}{85}
\tileor{18}{20}{85}
\tileor{19}{20}{85}
\tileor{20}{20}{85}
\tileor{21}{20}{85}
\tileor{22}{20}{85}
\tileor{23}{20}{85}
\tileor{23}{21}{85}
\tileor{23}{22}{85}
\tileor{23}{23}{85}
\tileor{23}{24}{85}
\tileor{22}{24}{85}
\tileor{21}{24}{85}
\tileor{20}{24}{85}
\tileor{19}{24}{85}
\tileor{18}{24}{85}
\tileor{17}{24}{85}
\tileor{16}{24}{85}
\tileor{15}{24}{85}
\tileor{14}{24}{85}
\tileor{13}{24}{85}
\tileor{12}{24}{85}
\tileor{11}{24}{85}
\tileor{10}{24}{85}
\tileor{9}{24}{85}
\tileor{8}{24}{85}
\tileor{7}{24}{85}
\tileor{6}{24}{85}
\tileor{6}{23}{85}
\tileor{6}{22}{85}
\tileor{6}{21}{85}
\tileor{6}{20}{85}
\tileor{6}{19}{85}
\tileor{6}{18}{85}
\tileor{6}{17}{85}
\tileor{6}{16}{85}
\tileor{6}{15}{85}
\tileor{6}{14}{85}
\tileor{6}{13}{85}
\tileor{6}{12}{85}
\tileor{6}{11}{85}
\tileor{6}{10}{85}
\tileor{7}{10}{85}
\tileor{8}{10}{85}
\tileor{9}{10}{85}
\tileor{10}{10}{85}
\tileor{11}{10}{85}
\tileor{12}{10}{85}
\tileor{13}{10}{85}
\tileor{14}{10}{85}
\tileor{15}{10}{85}
\tileor{16}{10}{85}
\tileor{17}{10}{85}
\tileor{18}{10}{85}
\tileor{19}{10}{85}
\tileor{19}{9}{85}
\tileor{19}{8}{85}
\tileor{19}{7}{85}
\tileor{19}{6}{85}
\tileor{18}{6}{85}
\tileor{17}{6}{85}
\tileor{16}{6}{85}
\tileor{15}{6}{85}
\tileor{14}{6}{85}
\tileor{13}{6}{85}
\tileor{12}{6}{85}
\tileor{11}{6}{85}
\tileor{10}{6}{85}
\tileor{9}{6}{85}
\tileor{8}{6}{85}
\tileor{7}{6}{85}
\tileor{6}{6}{85}
\tileor{5}{6}{85}
\tileor{4}{6}{85}
\tileor{3}{6}{85}
\tileor{2}{6}{85}
\tileor{2}{7}{85}
\tileor{2}{8}{85}
\tileor{2}{9}{85}
\tileor{2}{10}{85}
\tileor{2}{11}{85}
\tileor{2}{12}{85}
\tileor{2}{13}{85}
\tileor{2}{14}{85}
\tileor{2}{15}{85}
\tileor{2}{16}{85}
\tileor{2}{17}{85}
\tileor{2}{18}{85}
\tileor{2}{19}{85}
\tileor{2}{20}{85}
\tileor{2}{21}{85}
\tileor{2}{22}{85}
\tileor{2}{23}{85}
\tileor{2}{24}{85}
\tileor{2}{25}{85}
\tileor{2}{26}{85}
\tileor{2}{27}{85}
\tileor{2}{28}{85}
\tileor{2}{29}{85}
\tileor{2}{30}{85}
\tileor{2}{31}{85}
\tileor{2}{32}{85}
\tileor{2}{33}{85}
\tileor{2}{34}{85}
\tileor{2}{35}{85}
\tileor{2}{36}{85}
\tileor{2}{37}{85}
\tileor{2}{38}{85}
\tileor{2}{39}{85}
\tileor{2}{40}{85}
\tileor{2}{41}{85}
\tileor{2}{42}{85}
\tileor{2}{43}{85}
\tileor{2}{44}{85}
\tileor{2}{45}{85}
\tileor{2}{46}{85}
\tileor{2}{47}{85}
\tileor{2}{48}{85}
\tileor{2}{49}{85}
\tileor{2}{50}{85}
\tileor{2}{51}{85}
\tileor{2}{52}{85}
\tileor{2}{53}{85}
\tileor{2}{54}{85}
\tileor{2}{55}{85}
\tileor{2}{56}{85}
\tileor{2}{57}{85}
\tileor{2}{58}{85}
\tileor{2}{59}{85}
\tileor{2}{60}{85}
\tileor{3}{60}{85}
\tileor{4}{60}{85}
\tileor{5}{60}{85}
\tileor{6}{60}{85}
\tileor{7}{60}{85}
\tileor{8}{60}{85}
\tileor{9}{60}{85}
\tileor{10}{60}{85}
\tileor{11}{60}{85}
\tileor{12}{60}{85}
\tileor{13}{60}{85}
\tileor{14}{60}{85}
\tileor{15}{60}{85}
\tileor{16}{60}{85}
\tileor{17}{60}{85}
\tileor{18}{60}{85}
\tileor{19}{60}{85}
\tileor{20}{60}{85}
\tileor{21}{60}{85}
\tileor{22}{60}{85}
\tileor{23}{60}{85}
\tileor{24}{60}{85}
\tileor{24}{61}{85}
\tileor{24}{62}{85}
\tileor{24}{63}{85}
\tileor{23}{63}{85}
\tileor{22}{63}{85}
\tileor{21}{63}{85}
\tileor{20}{63}{85}
\tileor{19}{63}{85}
\tileor{18}{63}{85}
\tileor{17}{63}{85}
\tileor{16}{63}{85}
\tileor{15}{63}{85}
\tileor{14}{63}{85}
\tileor{13}{63}{85}
\tileor{12}{63}{85}
\tileor{11}{63}{85}
\tileor{10}{63}{85}
\tileor{9}{63}{85}
\tileor{9}{64}{85}
\tileor{9}{65}{85}
\tileor{9}{66}{85}
\tileor{10}{66}{85}
\tileor{11}{66}{85}
\tileor{12}{66}{85}
\tileor{13}{66}{85}
\tileor{14}{66}{85}
\tileor{15}{66}{85}
\tileor{16}{66}{85}
\tileor{17}{66}{85}
\tileor{18}{66}{85}
\tileor{19}{66}{85}
\tileor{20}{66}{85}
\tileor{21}{66}{85}
\tileor{22}{66}{85}
\tileor{23}{66}{85}
\tileor{24}{66}{85}
\tileor{25}{66}{85}
\tileor{26}{66}{85}
\tileor{27}{66}{85}
\tileor{28}{66}{85}
\tileor{29}{66}{85}
\tileor{30}{66}{85}
\tileor{31}{66}{85}
\tileor{32}{66}{85}
\tileor{33}{66}{85}
\tileor{34}{66}{85}
\tileor{35}{66}{85}
\tileor{36}{66}{85}
\tileor{37}{66}{85}
\tileor{38}{66}{85}
\tileor{39}{66}{85}
\tileor{40}{66}{85}
\tileor{41}{66}{85}
\tileor{42}{66}{85}
\tileor{43}{66}{85}
\tileor{44}{66}{85}
\tileor{45}{66}{85}
\tileor{46}{66}{85}
\tileor{47}{66}{85}
\tileor{48}{66}{85}
\tileor{49}{66}{85}
\tileor{50}{66}{85}
\tileor{51}{66}{85}

\dotmb{51}{14}
\dotmb{50}{14}
\dotmb{49}{14}
\dotmb{48}{14}
\dotmb{47}{14}
\dotmb{46}{14}
\dotmb{45}{14}
\dotmb{44}{14}
\dotmb{43}{14}
\dotmb{42}{14}
\dotmb{41}{14}
\dotmb{40}{14}
\dotmb{39}{14}
\dotmb{38}{14}
\dotmb{37}{14}
\dotmb{36}{14}
\dotmb{35}{14}
\dotmb{34}{14}
\dotmb{33}{14}
\dotmb{32}{14}
\dotmb{31}{14}
\dotmb{30}{14}
\dotmb{30}{13}
\dotmb{30}{12}
\dotmb{29}{12}
\dotmb{28}{12}
\dotmb{27}{12}
\dotmb{26}{12}
\dotmb{25}{12}
\dotmb{24}{12}
\dotmb{23}{12}
\dotmb{22}{12}
\dotmb{21}{12}
\dotmb{20}{12}
\dotmb{19}{12}
\dotmb{18}{12}
\dotmb{17}{12}
\dotmb{16}{12}
\dotmb{15}{12}

\dotlb{22}{60}{85}
\dotlb{22}{59}{85}
\dotlb{22}{58}{85}
\dotlb{22}{57}{85}
\dotlb{22}{56}{85}
\dotlb{22}{55}{85}
\dotlb{22}{54}{85}
\dotlb{22}{53}{85}
\dotlb{22}{52}
\dotlb{22}{51}
\dotlb{22}{50}
\dotlb{22}{49}
\dotlb{22}{48}
\dotlb{22}{47}
\dotlb{22}{46}
\dotlb{22}{45}
\dotlb{22}{44}
\dotlb{22}{43}
\dotlb{22}{42}
\dotlb{22}{41}
\dotlb{22}{40}
\dotlb{22}{39}
\dotlb{22}{38}
\dotlb{22}{37}
\dotlb{23}{37}
\dotlb{24}{37}
\dotlb{25}{37}
\dotlb{26}{37}
\dotlb{27}{37}
\dotlb{28}{37}
\dotlb{29}{37}
\dotlb{30}{37}
\dotlb{30}{36}
\dotlb{31}{36}
\dotlb{32}{36}
\dotlb{33}{36}
\dotlb{34}{36}
\dotlb{35}{36}
\dotlb{36}{36}
\dotlb{37}{36}
\dotlb{38}{36}
\dotlb{39}{36}
\dotlb{40}{36}
\dotlb{41}{36}
\dotlb{41}{37}
\dotlb{41}{38}
\dotlb{41}{39}
\dotlb{41}{40}
\dotlb{41}{41}{85}
\dotlb{42}{41}{85}
\dotlb{43}{41}{85}
\dotlb{44}{41}{85}
\dotlb{45}{41}{85}
\dotlb{46}{41}{85}
\dotlb{47}{41}{85}
\dotlb{48}{41}{85}
\dotlb{49}{41}{85}
\dotlb{50}{41}{85}
\dotlb{51}{41}{85}

\path [dotted, draw, thin] (0,5) grid[step=0.2cm] (60,69);

\draw [dashed] (15,5) -| (15,69);
\draw [dashed] (51,5) -| (51,69);

\fill (14.5,66.5) circle (0.16);
\node (D) at (14,68.2) {$P_{\lastc}$};

\fill (50.5,66.5) circle (0.16);
\node (D) at (50.5,68.2) {$P_{\last}$};

\fill (51.5,41.5) circle (0.16);
\node (D) at (53.2,41.5) {$A'_0$};

\fill (51.5,14.5) circle (0.16);
\node (D) at (53.2,14.5) {$S_{0}$};

\fill (15.5,12.5) circle (0.16);
\node (D) at (13.5,12.5) {$S_{e}$};

\fill (14.5,63.5) circle (0.16);
\node (D) at (13.5,61.8) {$P_{\second{\goB+1}}$};

\fill (22.5,60.5) circle (0.16);
\node (D) at (22.5,62) {$P_{a'}$};

\fill (14.5,20.5) circle (0.16);
\node (D) at (13.5,18.8) {$P_{\main{\goB}}$};

\end{tikzpicture}
\caption{Proof of Lemma \ref{lem:full:GoodCase}, part (2/2): 
the light blue dots represent the beginning of $A'=\min(\inter{A})$ until it intersects with $P_{\main{\goB},\ldots,\last+1}$ and where $A$ is the arc of Figure \ref{fig:Final:Combining:NewPathP}. In this case, the blue dots are an half-shield. This half-shield protects $\glueP{\second{\goB+1}}{\second{\goB+1}+1}$ and the protected area $\area{\second{\goB+1}}$ is in light blue.}
\label{fig:Final:Combining:NewArcHalf}
\end{figure}

\end{proof}

If no new shield and no new protected glue can be extracted from the arc $A'$, then a contradiction can be found by introducing a third protected glue in the hypothesis. 

\begin{lemma}
\label{lem:full:DoNotGo}
Consider the notations \ref{notation:full:combi} and the arc $\newA$ of Lemma \ref{lem:full:MinimalArc}. If there exists an index $\proNDeux \geq \proSLast{\colDeux}$ such that the $\glueQ{\proNDeux}{\proNDeux+1}$ is a northern (resp. southern) protected glue (either a ``main'' or ``backup'' one) of the downward (resp. upward) decomposition of $Q$ into dominant arcs on glue column $\colDeux$ then $w_{\newA} \leq \colUn$.
\end{lemma}

\begin{proof}
Without loss of generality, we suppose that $\glueP{\main{\goB}}{\main{\goB}+1}$ and $\glueQ{\proSLast{\colDeux}}{\proSLast{\colDeux}+1}$ are southern main glues. Thus, $\glueQ{\proNDeux}{\proNDeux+1}$ is a northern protected glue. For the sake of contradiction, assume that $w_{\newA} > \colUn$, see Figure \ref{fig:Final:Combining:Contra},  and we denote by $0\leq \lastcDeux\leq |Q|-1$ the index of the last glue of $Q$ of glue column $\colDeux$. 

\begin{figure}
\center
\begin{tikzpicture}[x=0.2cm,y=0.2cm]



\draw[very thick] (14.5,20.5) -| (23.5,24.5) -| (6.5,10.5)-| (19.5,6.5) -| (2.5,60.5) -| (24.5,63.5) -| (9.5,66.5) -| (51.5,66.5);
\draw[thick] (51.5,14.5) -| (30.5,12.5) -| (15,12.5);
\draw[thick] (36.5,66.5) |- (40.5,48.5) |- (51.5,41.5);




\tile{14}{20}{85}
\tileor{15}{20}{85}
\tileor{16}{20}{85}
\tileor{17}{20}{85}
\tileor{18}{20}{85}
\tileor{19}{20}{85}
\tileor{20}{20}{85}
\tileor{21}{20}{85}
\tileor{22}{20}{85}
\tileor{23}{20}{85}
\tileor{23}{21}{85}
\tileor{23}{22}{85}
\tileor{23}{23}{85}
\tileor{23}{24}{85}
\tileor{22}{24}{85}
\tileor{21}{24}{85}
\tileor{20}{24}{85}
\tileor{19}{24}{85}
\tileor{18}{24}{85}
\tileor{17}{24}{85}
\tileor{16}{24}{85}
\tileor{15}{24}{85}
\tileor{14}{24}{85}
\tileor{13}{24}{85}
\tileor{12}{24}{85}
\tileor{11}{24}{85}
\tileor{10}{24}{85}
\tileor{9}{24}{85}
\tileor{8}{24}{85}
\tileor{7}{24}{85}
\tileor{6}{24}{85}
\tileor{6}{23}{85}
\tileor{6}{22}{85}
\tileor{6}{21}{85}
\tileor{6}{20}{85}
\tileor{6}{19}{85}
\tileor{6}{18}{85}
\tileor{6}{17}{85}
\tileor{6}{16}{85}
\tileor{6}{15}{85}
\tileor{6}{14}{85}
\tileor{6}{13}{85}
\tileor{6}{12}{85}
\tileor{6}{11}{85}
\tileor{6}{10}{85}
\tileor{7}{10}{85}
\tileor{8}{10}{85}
\tileor{9}{10}{85}
\tileor{10}{10}{85}
\tileor{11}{10}{85}
\tileor{12}{10}{85}
\tileor{13}{10}{85}
\tileor{14}{10}{85}
\tileor{15}{10}{85}
\tileor{16}{10}{85}
\tileor{17}{10}{85}
\tileor{18}{10}{85}
\tileor{19}{10}{85}
\tileor{19}{9}{85}
\tileor{19}{8}{85}
\tileor{19}{7}{85}
\tileor{19}{6}{85}
\tileor{18}{6}{85}
\tileor{17}{6}{85}
\tileor{16}{6}{85}
\tileor{15}{6}{85}
\tileor{14}{6}{85}
\tileor{13}{6}{85}
\tileor{12}{6}{85}
\tileor{11}{6}{85}
\tileor{10}{6}{85}
\tileor{9}{6}{85}
\tileor{8}{6}{85}
\tileor{7}{6}{85}
\tileor{6}{6}{85}
\tileor{5}{6}{85}
\tileor{4}{6}{85}
\tileor{3}{6}{85}
\tileor{2}{6}{85}
\tileor{2}{7}{85}
\tileor{2}{8}{85}
\tileor{2}{9}{85}
\tileor{2}{10}{85}
\tileor{2}{11}{85}
\tileor{2}{12}{85}
\tileor{2}{13}{85}
\tileor{2}{14}{85}
\tileor{2}{15}{85}
\tileor{2}{16}{85}
\tileor{2}{17}{85}
\tileor{2}{18}{85}
\tileor{2}{19}{85}
\tileor{2}{20}{85}
\tileor{2}{21}{85}
\tileor{2}{22}{85}
\tileor{2}{23}{85}
\tileor{2}{24}{85}
\tileor{2}{25}{85}
\tileor{2}{26}{85}
\tileor{2}{27}{85}
\tileor{2}{28}{85}
\tileor{2}{29}{85}
\tileor{2}{30}{85}
\tileor{2}{31}{85}
\tileor{2}{32}{85}
\tileor{2}{33}{85}
\tileor{2}{34}{85}
\tileor{2}{35}{85}
\tileor{2}{36}{85}
\tileor{2}{37}{85}
\tileor{2}{38}{85}
\tileor{2}{39}{85}
\tileor{2}{40}{85}
\tileor{2}{41}{85}
\tileor{2}{42}{85}
\tileor{2}{43}{85}
\tileor{2}{44}{85}
\tileor{2}{45}{85}
\tileor{2}{46}{85}
\tileor{2}{47}{85}
\tileor{2}{48}{85}
\tileor{2}{49}{85}
\tileor{2}{50}{85}
\tileor{2}{51}{85}
\tileor{2}{52}{85}
\tileor{2}{53}{85}
\tileor{2}{54}{85}
\tileor{2}{55}{85}
\tileor{2}{56}{85}
\tileor{2}{57}{85}
\tileor{2}{58}{85}
\tileor{2}{59}{85}
\tileor{2}{60}{85}
\tileor{3}{60}{85}
\tileor{4}{60}{85}
\tileor{5}{60}{85}
\tileor{6}{60}{85}
\tileor{7}{60}{85}
\tileor{8}{60}{85}
\tileor{9}{60}{85}
\tileor{10}{60}{85}
\tileor{11}{60}{85}
\tileor{12}{60}{85}
\tileor{13}{60}{85}
\tileor{14}{60}{85}
\tileor{15}{60}{85}
\tileor{16}{60}{85}
\tileor{17}{60}{85}
\tileor{18}{60}{85}
\tileor{19}{60}{85}
\tileor{20}{60}{85}
\tileor{21}{60}{85}
\tileor{22}{60}{85}
\tileor{23}{60}{85}
\tileor{24}{60}{85}
\tileor{24}{61}{85}
\tileor{24}{62}{85}
\tileor{24}{63}{85}
\tileor{23}{63}{85}
\tileor{22}{63}{85}
\tileor{21}{63}{85}
\tileor{20}{63}{85}
\tileor{19}{63}{85}
\tileor{18}{63}{85}
\tileor{17}{63}{85}
\tileor{16}{63}{85}
\tileor{15}{63}{85}
\tileor{14}{63}{85}
\tileor{13}{63}{85}
\tileor{12}{63}{85}
\tileor{11}{63}{85}
\tileor{10}{63}{85}
\tileor{9}{63}{85}
\tileor{9}{64}{85}
\tileor{9}{65}{85}
\tileor{9}{66}{85}
\tileor{10}{66}{85}
\tileor{11}{66}{85}
\tileor{12}{66}{85}
\tileor{13}{66}{85}
\tileor{14}{66}{85}
\tileor{15}{66}{85}
\tileor{16}{66}{85}
\tileor{17}{66}{85}
\tileor{18}{66}{85}
\tileor{19}{66}{85}
\tileor{20}{66}{85}
\tileor{21}{66}{85}
\tileor{22}{66}{85}
\tileor{23}{66}{85}
\tileor{24}{66}{85}
\tileor{25}{66}{85}
\tileor{26}{66}{85}
\tileor{27}{66}{85}
\tileor{28}{66}{85}
\tileor{29}{66}{85}
\tileor{30}{66}{85}
\tileor{31}{66}{85}
\tileor{32}{66}{85}
\tileor{33}{66}{85}
\tileor{34}{66}{85}
\tileor{35}{66}{85}
\tileor{36}{66}{85}
\tileor{37}{66}{85}
\tileor{38}{66}{85}
\tileor{39}{66}{85}
\tileor{40}{66}{85}
\tileor{41}{66}{85}
\tileor{42}{66}{85}
\tileor{43}{66}{85}
\tileor{44}{66}{85}
\tileor{45}{66}{85}
\tileor{46}{66}{85}
\tileor{47}{66}{85}
\tileor{48}{66}{85}
\tileor{49}{66}{85}
\tileor{50}{66}{85}
\tileor{51}{66}{85}

\dotmb{51}{14}
\dotmb{50}{14}
\dotmb{49}{14}
\dotmb{48}{14}
\dotmb{47}{14}
\dotmb{46}{14}
\dotmb{45}{14}
\dotmb{44}{14}
\dotmb{43}{14}
\dotmb{42}{14}
\dotmb{41}{14}
\dotmb{40}{14}
\dotmb{39}{14}
\dotmb{38}{14}
\dotmb{37}{14}
\dotmb{36}{14}
\dotmb{35}{14}
\dotmb{34}{14}
\dotmb{33}{14}
\dotmb{32}{14}
\dotmb{31}{14}
\dotmb{30}{14}
\dotmb{30}{13}
\dotmb{30}{12}
\dotmb{29}{12}
\dotmb{28}{12}
\dotmb{27}{12}
\dotmb{26}{12}
\dotmb{25}{12}
\dotmb{24}{12}
\dotmb{23}{12}
\dotmb{22}{12}
\dotmb{21}{12}
\dotmb{20}{12}
\dotmb{19}{12}
\dotmb{18}{12}
\dotmb{17}{12}
\dotmb{16}{12}
\dotmb{15}{12}

\dotlb{51}{66}
\dotlb{50}{66}
\dotlb{49}{66}
\dotlb{48}{66}
\dotlb{47}{66}
\dotlb{46}{66}
\dotlb{45}{66}
\dotlb{44}{66}
\dotlb{43}{66}
\dotlb{42}{66}
\dotlb{41}{66}
\dotlb{40}{66}
\dotlb{39}{66}
\dotlb{38}{66}
\dotlb{37}{66}
\dotlb{36}{66}
\dotlb{36}{65}
\dotlb{36}{64}
\dotlb{36}{63}
\dotlb{36}{62}
\dotlb{36}{61}
\dotlb{36}{60}
\dotlb{36}{59}
\dotlb{36}{58}
\dotlb{36}{57}
\dotlb{36}{56}
\dotlb{36}{55}
\dotlb{36}{54}
\dotlb{36}{53}
\dotlb{36}{52}
\dotlb{36}{51}
\dotlb{36}{50}
\dotlb{36}{49}
\dotlb{36}{48}
\dotlb{37}{48}
\dotlb{38}{48}
\dotlb{39}{48}
\dotlb{40}{48}
\dotlb{40}{47}
\dotlb{40}{46}
\dotlb{40}{45}
\dotlb{40}{44}
\dotlb{40}{43}
\dotlb{40}{42}
\dotlb{40}{41}
\dotlb{41}{41}{85}
\dotlb{42}{41}{85}
\dotlb{43}{41}{85}
\dotlb{44}{41}{85}
\dotlb{45}{41}{85}
\dotlb{46}{41}{85}
\dotlb{47}{41}{85}
\dotlb{48}{41}{85}
\dotlb{49}{41}{85}
\dotlb{50}{41}{85}
\dotlb{51}{41}{85}

\path [dotted, draw, thin] (0,5) grid[step=0.2cm] (60,69);

\draw [dashed] (15,5) -| (15,69);
\draw [dashed] (51,5) -| (51,69);

\fill (14.5,66.5) circle (0.16);
\node (D) at (14,68.2) {$P_{\lastc}$};

\fill (50.5,66.5) circle (0.16);
\node (D) at (50.5,68.2) {$P_{\last}$};

\fill (51.5,41.5) circle (0.16);
\node (D) at (53.2,41.5) {$A'_0$};

\fill (51.5,14.5) circle (0.16);
\node (D) at (53.2,14.5) {$S_{0}$};

\fill (15.5,12.5) circle (0.16);
\node (D) at (13.5,12.5) {$S_{e}$};

\fill (14.5,63.5) circle (0.16);
\node (D) at (13.5,61.8) {$P_{\second{\goB+1}}$};

\fill (36.5,66.5) circle (0.16);
\node (D) at (36.5,68) {$P_{a'}$};

\fill (14.5,20.5) circle (0.16);
\node (D) at (13.5,18.8) {$P_{\main{\goB}}$};

\end{tikzpicture}
\caption{Proof of Lemma \ref{lem:full:DoNotGo}, part (1/4): 
the light blue dots represent the arc $A'=\min(\inter{A})$ where $A$ is the arc of Figure \ref{fig:Final:Combining:NewPathP}. Here, the light blue dots connect directly with $P_{\lastc+1,\ldots,|P|-1}$ without crossing glue column $c$. Then, we have $w_{A'}>c$ and no tile is of $P$ is protected. In this case, we will reach a contradiction.}
\label{fig:Final:Combining:Contra}
\end{figure}

First of all, consider the protected area $\area{\proNDeux}$ (be careful  $\glueQ{\proNDeux}{\proNDeux+1}$ is a northern main glue). Note that $Q_{\lastDeux+1}=\newA_{0}+\vu$ and $Q_{\lastDeux}=\newA_{1}+\vu$. Then we can defined $0<v \leq |\newA|-1$ such that $\newA_{1,\ldots,v}+\vu$ is the largest prefix of $\newA_{1,\ldots,|\newA|-1}+\vu$ which is inside $\area{\proNDeux}$. Now, since $P_{0,\ldots,\main{\goB}}(\rightFDeux-\vu)$ is producible (by Lemma \ref{lem:full:goodArc}) and since $\newA$ is a subassembly of $\uniterm$ (by Lemma~\ref{lem:full:MinimalArc}) then all intersections between $\newA+\vu$ and $\rightFDeux$ are agreements. Since $\proNDeux \geq \proSLast{\colDeux}$ then the only tiles of $Q_{0,\ldots,\proSLast{\colDeux}}\rightFDeux$ which are inside $\area{\proNDeux}$ are tiles of $\rightFDeux$: there is no intersection between $Q_{0,\ldots,\proSLast{\colDeux}}$ and $\newA_{1,\ldots,v}+\vu$. Then, since $Q_{0,\ldots,\proSLast{\colDeux}}\rightFDeux$ is producible (by definition of $\rightFDeux$), then $\asm{Q_{0,\ldots,\proSLast{\colDeux}}\rightFDeux}\cup \asm{\newA_{1,\ldots,v}+\vu}$ is a producible assembly. Thus, $\newA_{1,\ldots,v}+\vu$ is a subassembly of $\uniterm$ and in particular, all intersections between $\newA_{1,\ldots,v}+\vu$ and $Q_{\proSLast{\colDeux}+1,\ldots,\lastDeux}$ are agreements, see Figure \ref{fig:Final:Combining:FullLastOne}.

\input{./tikz/Final/CombiningMain/PathPShrinkProofOne}

Since $w_{\newA}+\vu> \colDeux$ then $\newA_{1,\ldots,v}+\vu$ cannot intersect with the ray starting in $\glueQ{\proNDeux}{\proNDeux+1}$ and going north. By Lemma \ref{lem:vis:LastTechnical} (if $\glueQ{\proNDeux}{\proNDeux+1}$ is visible from the north), Lemma \ref{lem:full:LastTechnical} (if $\glueQ{\proNDeux}{\proNDeux+1}$ is a ``main'' glue protected by a shield) or Lemma \ref{lem:half:LastTechnical} (if $\glueQ{\proNDeux}{\proNDeux+1}$ is a ``backup'' glue protected by a shield) then $\newA_{1,\ldots,v}+\vu$ cannot intersect with $Q_{\proNDeux+1,\ldots,\lastcDeux}$ and the shield of $\glueQ{\proNDeux}{\proNDeux+1}$  (in the relevant cases). Finally, if there exists $\lastcDeux \leq v' \leq \lastDeux$ such that $Q_{v'}=\newA_v+\vu$ and $\rev{Q_{v',\ldots,\lastDeux+1}}(\newA_{v+1}+\vu)$ turns left of $\rev{Q_{v'-1,\ldots,\lastDeux+1}}$ then there exist $\lastcDeux \leq u' \leq \lastDeux$ and $1 \leq u \leq v$ such that $Q_{u'}=\newA_{u}+\vu$ and $(\newA_{1,\ldots,u}+\vu)Q_{u'+1}$ turns right of $(\newA_{1,\ldots,u+1}+\vu)$. Since $x_{Q_{\lastcDeux}}=\colDeux-0.5$ and $w_{\newA+\vu} >\colDeux$, then $Q_{u'+1,\ldots, \lastcDeux+1}$ intersects with $\newA+\vu$ contradicting the fact that $\newA+\vu$ (and $\newA$) is a minimum arc. Then, $v\geq |\newA|-2$ and $\newA_{1,\ldots,|\newA|-2}+\vu$ is in $\area{\proNDeux}$ and is a subassembly of  $\uniterm$. Moreover, $\newA_{|\newA|-1}+\vu$ is also a tile of $\uniterm$ since it binds with $\newA_{|\newA|-2}+\vu$ and is on column $\Bfour(c')+0.5$.

Since $\newA_1+\vu=Q_{\lastDeux}=\rightFDeux_{|\rightFDeux|-2}$ then $\newA_{1,\ldots,|\newA|-1}+\vu$ and $\rightFDeux$ intersect. Let $0\leq b \leq \rightFDeux$ such that $\rightFDeux_{0,\ldots,b}$ is the shortest prefix of $\rightFDeux$ which intersects with $\newA+\vu$. Let $0\leq b' \leq |\newA|-1$ such that $\newA_{b'}+\vu=\rightFDeux_b$ and then $T=\rightFDeux_{0,\ldots,b}(\newA_{b'+1,\ldots,|\newA|-1}+\vu)$ is a path and $Q_{0,\ldots,\proSLast{\colDeux}}T$ is a producible path. Note that $P_{\last+1}+\vu=\newA_{|\newA|-1}$ and $P_{\last+1}+\vu$ is on column $\Bfour(\colDeux)+0.5$. By Lemma~\ref{Uturn:Conc}, if $P_{\last+2,\ldots,|P|-1}+\vu$ does not intersect with $T$,  then $Q_{0,\ldots,\proSLast{\colDeux}}T(P_{\last+2,\ldots,|P|-1}+\vu)$ is a producible path. Nevertheless, $x_{P_{|P|-1}+\vu}>e_\uniterm$ and we have a contradiction.

Finally, suppose that there is an intersection between $T$ and $(P_{\last+2,\ldots,|P|-1}+\vu)$, see Figures \ref{fig:Final:Combining:FullLastTwo} and \ref{fig:Final:Combining:FullLastThree}. By definition $P_{\last+2,\ldots,|P|-1}$ is a subassembly of $\uniterm$, $\newA$ is a subassembly of $\uniterm$ by Lemma \ref{lem:full:MinimalArc} and $\rightFDeux-\vu$ is a subassembly of $\uniterm$ by Lemma \ref{lem:full:goodArc}. Thus all intersections between these paths are agreements. Let $d=\max\{i: \pos{P_i}+\vu \text{ is occupied by a tile of } T\}$, we are in a case where $d>\lastUn+1$. Let $0\leq d' \leq |T|-1$ such that $T_{d'}=P_d+\vu$. Let $d\leq f'' \leq |P|-1$ be the index of the first tile of $P_{d,\ldots,|P|-1}+\vu$ which is on column $\Bfour(\colDeux)+0.5$. If $T_{0,\ldots,d'}(P_{d+1,\dots,f''}+\vu)$ is more left-priority than $T$, then $P_{d+1,\dots,f''}$ is inside a finite area delimited by $T-\vu$ and $P_{\main{\goB},\ldots,\lastUn+1}$. Nevertheless, $P_{d+1,\dots,f''}$ must leave this area but it cannot intersect with $P_{\main{\goB},\ldots,\lastUn+1}$ (since $P$ is simple) or $T-\vu$ (by definition of $d$). 
Then $T_{0,\ldots,d'}(P_{d+1,\dots,f''}+\vu)$ is more right-priority than $T$. If $P_{d+1,\dots,f''}+\vu$ is inside $\area{\proNDeux}\cup \areaResDeux{\proSLast{\colDeux}}$ then $$Q_{0,\ldots,\proSLast{\colDeux}}T_{0,\ldots,d'}(P_{d+1,\dots,f''}+\vu)$$ is producible and by Lemma \ref{Uturn:Conc},  $Q_{0,\ldots,\proSLast{\colDeux}}T_{0,\ldots,d'}(P_{d+1,\dots,|P|-1}+\vu)$ is producible. Again, $x_{P_{|P|-1}+\vu}>e_\uniterm$ and we have a contradiction. 

\input{./tikz/Final/CombiningMain/PathPShrinkProofTwo}
\input{./tikz/Final/CombiningMain/PathPShrinkProofThree}

Suppose that  $P_{d+1,\dots,f''}+\vu$ leaves $\area{\proNDeux}\cup \areaResDeux{\proSLast{\colDeux}}$ by leaving $\area{\proNDeux}$ first. Then since $w_T\leq \colDeux < \Bfour(\colDeux) \leq e_T$ and $\colDeux< w_{P_{d+1,\dots,f''}+\vu} <  e_{P_{d+1,\dots,f''-1}+\vu} \leq  \Bfour(\colDeux)$,  then $P_{d+1,\dots,f''}+\vu$ must intersect with $T$. This is a contradiction by the definition of $d$. If $P_{d+1,\dots,f''}+\vu$ leaves $\area{\proNDeux}\cup \areaResDeux{\proSLast{\colDeux}}$ by leaving $\areaResDeux{\proSLast{\colDeux}}$. Since $w_{P_{d+1,\dots,f''}+\vu}>\colDeux$ then $P_{d+1,\dots,f''}+\vu$ should intersect with the shield of $\glueQ{\proSLast{\colDeux}}{\proSLast{\colDeux}+1}$. In this case, it is possible to find a path in $\asm{T_{0,\ldots,d'}(P_{d+1,\dots,f''}+\vu)}\cup \asm{Q_{\lastcDeux,\ldots,\lastDeux}}$ which intersects with the shield and is in $\areaResDeux{\proSLast{\colDeux}}$ which is a contradiction by Lemma \ref{lem:full:visible}. All cases lead to contradiction and the lemma is true.

\end{proof}

Finally, here is the main result of this section.

\begin{lemma}
\label{combi:full:end}
Consider the notations \ref{notation:full:combi}, if there exists an index $\proNDeux \geq \proSLast{\colDeux}$ such that the $\glueQ{\proNDeux}{\proNDeux+1}$ is a northern (resp. southern) protected glue (either a ``main'' or ``backup'' one) of the downward (resp. upward) decomposition of $Q$ into dominant arcs on glue column $\colDeux$ then either there exists $\goBUn \leq i < \decoUn$ such that $\second{i}$ is protected or there exists $\goBUn < i \leq \decoUn$ such that $\main{i}$ is protected.
\end{lemma}

\begin{proof}
By Lemma \ref{lem:full:MinimalArc}, there exists an upward (resp. downward) minimum arc $\newA$ on glue column $\Bfour(\colUn)$ such that its interior $\inter{\newA}$ is included into $\area{\mainUn{\goBUn}}$, such that $\newA$ is a subassembly of $\uniterm$, such that $\newA_0=Q_{\lastDeux+1}-\vu$ and $\newA_{|\newA|-1}=P_{\lastUn+1}$. Moreover, by Lemma \ref{lem:full:DoNotGo}, we have $w_{\newA} \leq \colUn$. Then by Lemma \ref{lem:full:GoodCase}, the lemma is true.
\end{proof}

\subsection{Combining two identical ``backup'' protected glues}\label{sec:combining:half}

Along this subsection, we will use the following notations:

\begin{definition}[Notations for combining two ``half'' glues]
\label{notation:half}
Suppose that $e_\uniterm>\Bfinal$ and consider two paths $P$ and $Q$ which are canonical for some glue columns $\colUn$ and $\colDeux$ with $\Bone\leq \colUn < \colDeux \leq \Btwo$. 
\begin{itemize}
\item Let $0\leq \lastUn < |P|-1$ be the index of the first glue of $P$ on glue column $\Bfour(\colUn)$. 
\item Let $0\leq \lastDeux < |Q|-1$ be the index of the first glue of $Q$ on glue column $\Bfour(\colDeux)$. 
\item Let $(\main{i})_{0\leq i \leq t}$ and $(\second{i})_{0\leq i < t}$ the upward (resp. downward) decomposition of $P$ into dominants arcs on glue column $\colUn$.
 Let $0\leq \goB < t$, a path $\shieldUn$ and an index $\main{\goB} < \posS < \second{\goB}$ such that $\glueP{\second{\goB}-1}{\second{\goB}}$ is a ``backup'' glue of $P$ protected by the half-shield $S$ at position $\posS$. Let $\rightF$ be the practical form of $P$ at index  $\second{\goB}-1$ and let $\areaResUn{\second{\goB}}$ be its restrained protected area.
\item Let $0\leq \proSLast{\colDeux} < |Q|-1$ such that $\glueQ{\proS-1}{\proS}$ is a ``backup'' glue of the upward (resp. downward) decomposition of $Q$ into dominants arcs on glue column $\colDeux$. This $\glueQ{\proS-1}{\proS}$ is protected by the half-shield $\shieldDeux$ at position $\posSDeux$. Let $\rightFDeux$ be the practical form of $Q$ at index  $\second{\goB}-1$ and let $\areaResDeux{\proS}$ its restrained protected area.
\end{itemize}
\end{definition}

Remark that in the Subsection \ref{sub:vis:half}, we did not define the type of a protected backup glue. At first glance, it seems logical to consider that $\glueP{\second{\goB}}{\second{\goB}+1}$ and $\glueQ{\proSLast{\colDeux}}{\proSLast{\colDeux}+1}$ should have the same type. Then, we define $\vu=\vect{\pos{P_{\second{\goB}}}\pos{Q_{\proSLast{\colDeux}}}}$ to try to assemble $\rightFDeux-\vu$ at the end of $P_{0,\ldots,\second{\goB}-1}$. Unfortunately, Figure \ref{fig:Final:Combining:half:bug} shows a problematic case. Indeed, $\rightFDeux-\vu$ may turn right of $\rightFUn$, stay inside $\areaResUn{\second{\goB}}$ until intersecting with $P_{\posS,\ldots,\second{\goB}-1}$. In this case, a tile of  $\rightFDeux-\vu$ shares the same position as a tile of $P_{\posS,\ldots,\second{\goB}-1}$ but the two tiles may not share the same type: assembling $\rightFDeux-\vu$ requires to assemble $P_{0,\ldots,\second{\goB}-1}$ first. Thus, the growth of $\rightFDeux-\vu$ can be stopped and we cannot apply the same reasoning as in Subsection \ref{sec:combining:full}.

\input{./tikz/Final/CombiningBackup/HalfBug}

Thus, in order to solve the last remaining Lock \ref{lock:shieldone}, we need to introduce a new notion. The type of a backup glue while not be the type of the protected glue but the type of a \emph{hidden} glue, see Figure \ref{fig:Final:Combining:half:hid:def}. The index $\hid$ of the \emph{hidden} glue associated to $\glueP{\second{\goB}-1}{\second{\goB}}$ is defined as the northernmost (resp. southernmost) $\glueP{\hid}{\hid+1}$ among the glues of $P_{0,\ldots,\main{\goB}+1}$ which are on glue column $c$ and strictly to the south (resp. north) of $\glueP{\second{\goB}-1}{\second{\goB}}$, \emph{i.e.} $0\leq \hid \leq \main{\goB}$ and $\glueP{\hid}{\hid+1}$ is on glue column $c$ and its $y$-coordinate is 
\begin{align*}
\max\{y:\text{there } & \text{exists $0\leq i\leq \main{\goB}$ such that $\glueP{i}{i+1}$ is} \\ & \text{on glue column $c$, $y_{P_i}=y$ and $y<y_{P_{\second{\goB}}}$}\}.
\end{align*}
\begin{align*}
\text{(resp. } \min\{y:\text{there } & \text{exists $0\leq i\leq \main{\goB}$ such that $\glueP{i}{i+1}$ is} \\ & \text{on glue column $c$, $y_{P_i}=y$ and $y>y_{P_{\second{\goB}}}$}\}.)
\end{align*}
Note that $\hid$ is correctly defined since $\glueP{\main{\goB}}{\main{\goB}+1}$ is on glue column $c$ and strictly to the south (resp. north) of the protected $\glueP{\second{\goB}-1}{\second{\goB}}$\footnote{In some case, we can have $\hid=\main{\goB}$  or in other cases, it is possible that $\hid<\main{0}$. Nevertheless, these different cases are not relevant for the rest of the proof.}. The \emph{type} of a protected backup glue is the type of $\glueP{\hid}{\hid+1}$. The \emph{width} of the backup glue is $|y_{P_{\second{\goB}}}-y_{P_{\hid}}|$. From this definition, it follows that $H=P_{\hid,\ldots,\second{\goB}}$ is a hole of glue column $c$ and the only tiles of $P_{0,\ldots,\second{\goB}}$ inside its interior $\inter{H}$ are $P_{\hid+1,\ldots,\second{\goB}-1}$ (by definition of $\Bone$, the seed cannot be inside $\inter{H}$). Thus, for any path $B$ such that $B_0=P_{\hid+1}$ and $B$ is in $\inter{H}$, the path $P_{0,\ldots,\hid}B$ is producible. In particular, consider $\leftF=\min(\inter{H})$ then $P_{0,\ldots,\hid}\leftF_{1,\ldots,|\leftF|-1}$ is producible. 

\input{./tikz/Final/CombiningBackup/HalfHiddenDef}

\begin{definition}[Notations for the hidden glues]
\label{notation:hidden}

We add the following notations to the ones of Definition \ref{notation:half}:
\begin{itemize}
\item Let $0\leq \hid \leq \main{\goB}$ be the index of the hidden glue of $P$ associated to $\glueP{\second{\goB}-1}{\second{\goB}}$ and let $w$ be its width. Let $H$ be the hole $P_{\hid,\ldots,\second{\goB}}$ and let $\leftF$ be $\min(\inter{H})$. 
\item Let $0\leq \hidDeux < \proSLast{\colDeux}$ be the index of the hidden glue of $Q$ associated to $\glueQ{\proSLast{\colDeux}-1}{\proSLast{\colDeux}}$ and let $w'$ be its width. Let $H'$ be the hole $Q_{\hidDeux,\ldots,\proSLast{\colDeux}}$ and let $\leftFDeux$ be $\min(\inter{H'})$. 
\end{itemize}
Instead of using equality between visible glues, we will switch in this case to looking for equality between the hidden glue associated with the visible glues, \emph{i.e.} we suppose that $\glueP{\hid}{\hid+1}$ and $\glueQ{\hidDeux}{\hidDeux+1}$ have the same type and we denote by $\vu=\vect{P_{\hid+1}Q_{\hidDeux+1}}$.
\end{definition}

Now, we distinguish two cases, either $w=w'$ or not. For each case, we give a patch to avoid the problem shown in Figure \ref{fig:Final:Combining:half:bug}. Then, the reasoning of Subsection \ref{sec:combining:full} can be followed to conclude in Subsection \ref{sec:combining:half:conc}.

\subsubsection{Case where $w=w'$}\label{sec:combining:half:equality}

Consider the notations \ref{notation:hidden}, when $w=w'$, the key argument is to remark that $\leftF_{1,\ldots,|\leftF|-1}=\leftFDeux_{1,\ldots,|\leftFDeux|-1}-\vu$.

\begin{lemma}
\label{combi:half:equal:One}
Consider the notations \ref{notation:hidden}, If $w=w'$ then $\leftF_{1,\ldots,|\leftF|-1}=\leftFDeux_{1,\ldots,|\leftFDeux|-1}-\vu$. 
\end{lemma}

\begin{proof}
W. l. o. g., we suppose that $\leftF$ is an upward hole\footnote{This means that $\glueP{\second{\goB-1}}{\second{\goB}}$ and $\glueQ{\second{\goBDeux-1}}{\second{\goBDeux}}$ are southern protected glues.}, see Figure \ref{fig:Final:Combining:half:hid:equality}. In this case, the interior $\inter{\leftF}$ corresponds to the left side of $\leftF$. Now, let $1\leq u \leq |\leftF|-1$ such that $\leftF_{1,\ldots,u}$ is the largest common prefix between $\leftF_{1,\ldots,|\leftF|-1}$ and $\leftFDeux_{1,\ldots,|\leftFDeux|-1}-\vu$. If $u= |\leftF|-1=|\leftFDeux|-1$ then the lemma is true. 

For the sake of contradiction suppose that $u< |\leftF|-1$ and w.l.o.g., we suppose that $\leftFDeux_{1,\ldots,|\leftFDeux|-1}-\vu$ turns left of $\leftF_{1,\ldots,|\leftF|-1}$ (the other case is symmetric). Since by hypothesis, $\glueP{\hid}{\hid+1}=\glueQ{\hidDeux}{\hidDeux+1}-\vu$ and $w=w'$ then $\pos{\leftFDeux_{|\leftFDeux|-1}}-\vu=\pos{\leftF_{|\leftF|-1}}$.
This remark implies that there is at  least one intersection between $\leftF_{u+1,\ldots,|\leftF|-1}$ and $\leftFDeux_{u+1,\ldots,|\leftFDeux|-1}-\vu$. Then, we can define $u<v'\leq |\leftFDeux|-1$ such that $\leftFDeux_{u+1,\ldots,v'}-\vu$ is the shortest prefix of $\leftFDeux_{u+1,\ldots,|\leftFDeux|-1}-\vu$ to intersect $\leftF_{u+1,\ldots,|\leftF|-1}$. Let $u<v\leq |\leftF|-1$ such that $\pos{\leftFDeux_{v'}}-\vu=\pos{\leftF_v}$. Remark that the path $\leftFDeux_{1,\ldots,v'}-\vu$ cannot leave $\inter{\leftF}$ by crossing the frontier of the hole $\leftF$ (otherwise $\leftFDeux_{1,\ldots,v'}$ would also cross the frontier of the hole $\leftFDeux$ which contradicts the definition of a hole) or by crossing through $\leftF$ (by definition of $u$ and $v'$). Then $\leftFDeux_{1,\ldots,v'}-\vu$ is inside $\inter{\leftF}$. Thus $P_{0,\ldots,\hid}(\leftFDeux_{1,\ldots,v'}-\vu)$ is producible since no tile of $P_{0,\ldots,\hid}$ or the seed $\sigma$ is inside $\inter{\leftF}$. In particular, $P_{0,\ldots,\hid}(\leftFDeux_{1,\ldots,v'}-\vu)\leftF_{v+1,\ldots,|\leftF|-1}$ is producible, is inside $\inter{\leftF}$ and turns left of $P_{0,\ldots,\hid}\leftF_{1,\ldots,|\leftF|-1}$. This is a contradiction of the definition of $\leftF$ as~$\min(\inter{\leftF})$.
\end{proof}

In particular, this lemma implies: $$P_{\second{\goB}}=\leftF_{|\leftF|-1}=\leftFDeux_{|\leftFDeux|-1}-\vu=Q_{\proS}-\vu.$$ Then we can try to grow $\rightFDeux-\vu$ from $P_{0,\ldots,\hid}\leftF_{1,\ldots,|\leftF|-2}$ and the problem shown in Figure \ref{fig:Final:Combining:half:bug} cannot occurs anymore. Indeed, if $\rightFDeux$ and $\leftFDeux$ does not conflict then $\rightFDeux-\vu$ and $\leftF$ does not conflict since $\leftF=\leftFDeux-\vu$.

\begin{lemma}
\label{combi:half:equal:Two}
Consider the notations \ref{notation:hidden}, a path $B$ such that $Q_{0,\ldots,\second{\goBDeux}-1}B$ is producible and such that $T-\vu$ does not intersect with $P_{0,\ldots,\hid}$ and the seed $\sigma$. If $w=w'$ then $B-\vu$ is a sub-assembly of $\uniterm$. 
\end{lemma}

\begin{proof}
If $Q_{0,\ldots,\second{\goBDeux}-1}B$ is producible, then $B$ does not conflict with $\leftFDeux$ \footnote{Note that there may be several intersections between $T$ and $\leftFDeux$.} (since $Q_{0,\ldots,\hidDeux}\leftFDeux$ is producible). By Lemma \ref{combi:half:equal:One}, we have $\leftF_{1,\ldots,|\leftF|-1}=\leftFDeux_{1,\ldots,|\leftFDeux|-1}-\vu$. Then there is no conflict between $B-\vu$ and $\leftF_{1,\ldots,|\leftF|-1}$ (which implies that $B_0-\vu=\leftF_{|\leftF|-1}=P_{\second{\goB}})$. Thus, $\asm{B-\vu}$ can grow from $\asm{P_{0,\ldots,\hid}\leftF_{1,\ldots,|\leftF|-1}}$ since by hypothesis $B-\vu$ does not intersect with $P_{0,\ldots,\hid}$ and the seed $\sigma$.
\end{proof}

Of course there exists a version of this lemma obtained by switching the role of the two paths.

\begin{lemma}
\label{combi:half:equal:TwoBis}
Consider the notations \ref{notation:hidden}, a path $B$ such that $P_{0,\ldots,\second{\goB}-1}B$ is producible and such that $B+\vu$ does not intersect with $Q_{0,\ldots,\hidDeux}$ and the seed $\sigma$. If $w=w'$ then $B+\vu$ is a sub-assembly of $\uniterm$. 
\end{lemma}

In the setting of notations \ref{notation:hidden}, the yellow diamonds of Figure \ref{fig:Final:Combining:half:bug} cannot be stopped by a tile of $P_{a,\ldots,\second{\goB}-1}$. Indeed, the tile marked by a red diamond does not belong to $\leftF$ according to Figure \ref{fig:Final:Combining:half:hid:def}. Then, we do not need to assemble this part of $P_{a,\ldots,\second{\goB}-1}$ if $\leftF$ is assembled first. Thus, the red diamond of Figure \ref{fig:Final:Combining:half:bug} is in fact an agreement with the tile of $P_{a,\ldots,\second{\goB}-1}$ and there is no conflict. Moreover, such an agreement would create a contradiction with the priority of $P_{\main{\goB},\ldots,\second{\goB}}$ according to Corollary \ref{cor:decompo:cano}. This last remark allow us to conclude this section by stating a lemma which is an equivalent to Lemma \ref{lem:full:turnRestrainedVis} and Lemma \ref{lem:full:turnRestrained}.

\begin{lemma}
\label{combi:half:equal:Three}
Consider the notations \ref{notation:hidden}, a path $B$ such that $e_B<\Bfour(c')$ and such that $Q_{0,\ldots,\proS-1}B$ is producible. If $P_{\second{\goB}-1}(B-\vu)$ is more right-priority (resp. left-priority) than $P_{\second{\goB}-1}\rightF$ then $B-\vu$ is inside $\areaRes{\second{\goB}}$ and does not intersect with $S$ and $P_{a,\ldots,\second{\goB}-1}$.
\end{lemma}

\begin{proof}
Consider that $P_{\second{\goB}}(T-\vu)$ is more right-priority than $P_{\second{\goB}}\rightF$ and that $T-\vu$ intersects with $S$ or $P_{a,\ldots,\second{\goB}}$. Then let $0\leq u \leq |T|-1$ such that $T_{0,\ldots,u}-\vu$ is the largest common prefix between $T-\vu$ and $P_{\second{\goB}+1,\ldots,\last}$ and let $$v=\min\{u< i \leq |T|: T_v\text{ occupies a position of a tile of $S$ or $P_{a,\ldots,\last}$}\}.$$  Since $e_T<\Bfour(c')$ then $e_{T-\vu}<\Bfour(c)$. Also, since $P_{\second{\goB}}(T-\vu)$ is more right-priority (resp. left-priority) than $P_{\second{\goB}}\rightF$ then $T_{0,\ldots,v}$ is inside $\areaRes{\second{\goB}}$. 

\begin{figure}
\center
\begin{tikzpicture}[x=0.22cm,y=0.22cm]

\fill[fill=yellow!40!white, draw opacity=0.8] (15,9.5) -| (16.5,7.5) -| (13.5,5.5) -| (16.5,3.5) -| (14.5,1.5) -| (23.5,7.5) -| (21.5,16.5) -| (13.5,18.5) -| (18.5,20.5) -| (11.5,14.5) -| (17.5,12.5) -| (15,12.5);
\fill[fill=orange!30!white, draw opacity=0.8] (15,9.5) -| (16.5,7.5) -| (13.5,5.5) -| (16.5,3.5) -| (14.5,1.5) -| (21.5,16.5) -| (13.5,18.5) -| (18.5,20.5) -| (11.5,14.5) -| (17.5,12.5) -| (15,12.5);

\draw[very thick] (14.5,9.5) -| (16.5,7.5) -| (13.5,5.5) -| (16.5,3.5) -| (14.5,1.5) -| (23.5,7.5) -| (21.5,16.5) -| (13.5,18.5) -| (18.5,20.5) -| (11.5,14.5) -| (17.5,12.5) -| (14.5,12.5);
\draw[thick] (14.5,9.5) -| (16.5,7.5) -| (13.5,5.5) -| (16.5,3.5) -| (14.5,1.5) -| (21.5,13.5) -| (17.5,12.5)-| (14.5,12.5);


\tiley{14}{9}{85}
\tiley{15}{9}{85}
\tiley{16}{9}{85}
\tiley{16}{8}{85}
\tiley{16}{7}{85}
\tiley{15}{7}{85}
\tiley{14}{7}{85}
\tiley{13}{7}{85}
\tiley{13}{6}{85}
\tiley{13}{5}{85}
\tiley{14}{5}{85}
\tiley{15}{5}{85}
\tiley{16}{5}{85}
\tiley{16}{4}{85}
\tiley{16}{3}{85}
\tiley{15}{3}{85}
\tiley{14}{3}{85}
\tiley{14}{2}{85}
\tiley{14}{1}{85}
\tiley{15}{1}{85}
\tiley{16}{1}{85}
\tiley{17}{1}{85}
\tiley{18}{1}{85}
\tiley{19}{1}{85}
\tiley{20}{1}{85}
\tiley{21}{1}{85}
\tiley{22}{1}{85}
\tiley{23}{1}{85}
\tiley{23}{2}{85}
\tiley{23}{3}{85}
\tiley{23}{4}{85}
\tiley{23}{5}{85}
\tiley{23}{6}{85}
\tiley{23}{7}{85}
\tiley{22}{7}{85}
\tiley{21}{7}{85}
\tiley{21}{8}{85}
\tiley{21}{9}{85}
\tiley{21}{10}{85}
\tiley{21}{11}{85}
\tiley{21}{12}{85}
\tiley{21}{13}{85}
\tiley{21}{14}{85}
\tiley{21}{15}{85}
\tiley{21}{16}{85}
\tiley{20}{16}{85}
\tiley{19}{16}{85}
\tiley{18}{16}{85}
\tiley{17}{16}{85}
\tiley{16}{16}{85}
\tiley{15}{16}{85}
\tiley{14}{16}{85}
\tiley{13}{16}{85}
\tiley{13}{17}{85}
\tiley{13}{18}{85}
\tiley{14}{18}{85}
\tiley{15}{18}{85}
\tiley{16}{18}{85}
\tiley{17}{18}{85}
\tiley{18}{18}{85}
\tiley{18}{19}{85}
\tiley{18}{20}{85}
\tiley{17}{20}{85}
\tiley{16}{20}{85}
\tiley{15}{20}{85}
\tiley{14}{20}{85}
\tiley{13}{20}{85}
\tiley{12}{20}{85}
\tiley{11}{20}{85}
\tiley{11}{19}{85}
\tiley{11}{18}{85}
\tiley{11}{17}{85}
\tiley{11}{16}{85}
\tiley{11}{15}{85}
\tiley{11}{14}{85}
\tiley{12}{14}{85}
\tiley{13}{14}{85}
\tiley{14}{14}{85}
\tiley{15}{14}{85}
\tiley{16}{14}{85}
\tiley{17}{14}{85}
\tiley{17}{13}{85}
\tiley{17}{12}{85}
\tiley{16}{12}{85}
\tiley{15}{12}{85}
\tiley{14}{12}{85}

\diaor{14}{12}{85}
\diaor{15}{12}{85}
\diaor{16}{12}{85}
\diaor{17}{12}{85}
\diaor{17}{13}{85}
\diaor{18}{13}{85}
\diaor{19}{13}{85}
\diaor{20}{13}{85}
\diaor{21}{13}{85}
\diaor{21}{12}{85}
\diaor{21}{11}{85}
\diaor{21}{10}{85}
\diaor{21}{9}{85}
\diaor{21}{8}{85}
\diaor{21}{7}{85}
\diaor{21}{6}{85}
\diaor{21}{5}{85}
\diaor{21}{4}{85}
\diaor{21}{3}{85}
\diaor{21}{2}{85}
\diaor{21}{1}{85}
\diaor{20}{1}{85}
\diaor{19}{1}{85}
\diaor{18}{1}{85}
\diaor{17}{1}{85}
\diaor{16}{1}{85}
\diaor{15}{1}{85}
\diaor{14}{1}{85}
\diaor{14}{2}{85}
\diaor{14}{3}{85}
\diaor{15}{3}{85}
\diaor{16}{3}{85}
\diaor{16}{4}{85}
\diaor{16}{5}{85}
\diaor{15}{5}{85}
\diaor{14}{5}{85}
\diaor{13}{5}{85}
\diaor{13}{6}{85}
\diaor{13}{7}{85}
\diaor{14}{7}{85}
\diaor{15}{7}{85}
\diaor{16}{7}{85}
\diaor{16}{8}{85}
\diaor{16}{9}{85}
\diaor{15}{9}{85}


\path [dotted, draw, thin] (6,0) grid[step=0.22cm] (28,22);

\draw [dashed, color=yellow!70!black] (15,0) -| (15,22);

\fill (21.5,1.5) circle (0.16);
\node (D) at (21.5,0.2) {$L'_{u}-\vu$};

\fill (21.5,7.5) circle (0.16);
\node (D) at (25.2,9) {$L'_{v}-\vu$};

\draw [<->, thick] (15,9.5) -- (15,12.5);
\node (D) at (18.2,11.5) {$w=w'$};


\fill (14.5,9.5) circle (0.16);
\node (D) at (8.2,9.5) {$L_{0}=P_\hid$};
\node (D) at (3,12.5) {$L_{|L|-1}=L'_{|L'|-1}-\vu=P_{\second{\goB}}$};
\fill (14.5,12.5) circle (0.16);

\end{tikzpicture}
\caption{Proof of Lemma \ref{combi:half:equal:One}: consider the upward hole $L$ (the yellow tiles) of the path $P$ on glue column $c$, its interior is the colored area (orange and yellow). The path $L'_{1,\ldots,|L'|-1}-\vu$ (the orange diamonds) starts in $L_1=P_{\hid+1}$ (by definition of $\vu$) and ends in $L_{|L|-1}$ (since $w=w'$). Here $L_0(L'_{1,\ldots,v'}-\vu)L_{v+1,\ldots,|L|-1}$ turns left of $L$, is inside $\inter{L}$ and is a hole of glue column $c$ (its interior is in orange). This contradicts the minimality of $L$.}
\label{fig:Final:Combining:half:hid:equality}
\end{figure}

If $\pos{T_v}-\vu=\pos{\rightF_{v'}}$ for some $v'>\second{\goB}+u$, then $P_{0,\ldots,\second{\goB}}(T_{0,\ldots,v}-\vu)\rightF_{v'+1,\ldots,\last}$ is producible. Nevertheless, $(T_{0,\ldots,v}-\vu)\rightF_{v'+1,\ldots,|\rightF|-1}$ is more right-priority (resp. left-priority) than $\rightF$ which is a contradiction of the definition of $\rightF$. 

Now if $\pos{T_v}-\vu=\pos{S_{v'}}$ for some $0\leq v' \leq |S|-1$, then $P_{0,\ldots,\second{\goB}}(T_{0,\ldots,v}-\vu)$ is producible implying that $T_v-\vu=S_{v'}$. Nevertheless, $S_{0,\ldots,v'-1}\rev{T_{u,\ldots,v}-\vu}\rightF_{u+1,\ldots,|R|-1}$ is more right-priority (resp. left-priority) than $SP_{a,\ldots,\last}$ which is a contradiction of the definition of shield $S$. 

Finally, if $\pos{T_v}-\vu=\pos{P_{v'}}$ for some $a \leq v' \leq \second{\goB}$, then $T_v-\vu=P_{v'}$ by Lemma \ref{combi:half:equal:Two}. Nevertheless, the existence of $T_{u,\ldots,v}-\vu$ is a contradiction by Corollary \ref{cor:decompo:cano}. 
\end{proof}

Of course there exists a version of this lemma obtained by switching the role of the two paths.

\begin{lemma}
\label{combi:half:equal:ThreeBis}
Consider the notations \ref{notation:hidden}, a path $B$ such that $e_B<\Bfour(c)$ and such that $P_{0,\ldots,\second{\goB}-1}B$ is producible. If $Q_{\proS-1}(B+\vu)$ is more right-priority (resp. left-priority) than $Q_{\proS}\rightFDeux$ then $B+\vu$ is inside $\areaRes{\second{\goBDeux}}$ and does not intersect with $S'$ and $Q_{a',\ldots,\proS-1}$.
\end{lemma}

Lemma \ref{combi:half:equal:Three} and Lemma \ref{combi:half:equal:ThreeBis} are the patches required to follow the steps of Subsection \ref{sec:combining:full}. Indeed, now we can try to grow $\rightFDeux-\vu$ from $P_{0,\ldots,\second{\goB}-1}$ without fearing conflict from $S$ or $P_{a,\ldots,\second{\goB}-1}$. We can also try to grow $\rightF+\vu$ from $Q_{0,\ldots,\proS-1}$ without fearing conflict from $S'$ or $Q_{a',\ldots,\proS-1}$. These results are the first half of the final Lemma \ref{combi:half:end} of this section.

\subsubsection{Case where $w \neq w'$}\label{sec:combining:half:not}

First of all, to simplify the notation we suppose that $w < w'$ and that $\glueP{\second{\goB}-1}{\second{\goB}}$ is a southern glue (all other cases are symmetric to this one), see Figure \ref{fig:Final:Combining:back:running}. Now, the result here is harder to prove than the one in the case $w=w'$ but it is also stronger. We will conclude with Lemma \ref{combi:half:neq:Six} which allows to copy a path assembled at the end of $Q_{0,\ldots,\proS-1}$ and translates it by $-\vu$ directly into $\areaRes{\second{\goB}}$. To prove this result, we first consider the path $SP_{a,\ldots,\second{\goB}}$ and show that if a path $B$ turns right of this path then it is inside $\areaRes{\second{\goB}}$ and cannot leave this area.

\input{./tikz/Final/CombiningBackup/RunningExample}

\begin{lemma} \label{combi:half:neq:One}
Consider the notations \ref{notation:hidden}. Consider the path $F=SP_{a,\ldots,\second{\goB}}$ and a path $B$ such that:
\begin{itemize}
\item all intersections between $B$ and $F$ are agreements;
\item there exists $0 \leq u < |F|-1$ such that $B_0=F_u$;
\item $F_{0,\ldots,u}B_1$ turns right of $F$;
\item $e_B\leq \Bfour(c)$.
\end{itemize}
Then $B$ is inside $\areaRes{\second{\goB}}$ and does not intersect with $A=F\rightF_{1,\ldots,|\rightF|-1}$ apart from $B_0$.
\end{lemma}

\begin{proof}
We remind that, with the chosen conventions, $A$ is a negative upward arc of column $\Bfour(c)$ and its interior is $\areaRes{\second{\goB}}$, see Figure \ref{fig:Final:Combining:back:running:LemOne}. Suppose that $B_{1,\ldots,|B|-1}$ does not intersect with $A$. In this case, since $F_{0,\ldots,u}B_1$ turns right of $F$ and since $e_B\leq \Bfour(c)$, then $B$ is inside $\areaRes{\second{\goB}}$. Since $B_0=F_u$, the assembly $\asm{B}$ can grow from $\asm{P_{0,\ldots,\second{\goB}}} \cup \asm{F}$. The lemma is true in this case. 

\begin{figure}
\center
\begin{tikzpicture}[x=0.2cm,y=0.2cm]

\fill[fill=blue!60!white, draw opacity=0.8] (51,15.5) -| (30.5,27.5) -| (11.5,35.5) -| (51,33.5);

\draw [dashed] (15,0) -| (15,49);
\draw [dashed] (51,0) -| (51,49);

\draw[very thick] (4.5,44.5) -| (24.5,37.5) -| (8.5,21.5) -| (21.5,13.5) -| (11.5,15.5) -| (19.5,19.5)-| (6.5,39.5) -| (22.5,42.5) -| (1.5,4.5) -| (30.5,27.5) -| (11.5,35.5) -| (51.5,35.5);
\draw[thick] (51.5,15.5) -| (30.5,15.5);
\draw[thick] (39.5,15.5) -| (39.5,35.5);
\draw[thick] (27.5,27.5) -| (27.5,35.5);

\draws{4}{44}
\tile{5}{44}
\tile{6}{44}
\tile{7}{44}
\tile{8}{44}
\tile{9}{44}
\tile{10}{44}
\tile{11}{44}
\tile{12}{44}
\tile{13}{44}
\tile{14}{44}
\tile{15}{44}
\tile{16}{44}
\tile{17}{44}
\tile{18}{44}
\tile{19}{44}
\tile{20}{44}
\tile{21}{44}
\tile{22}{44}
\tile{23}{44}
\tile{24}{44}
\tile{24}{43}
\tile{24}{42}
\tile{24}{41}
\tile{24}{40}
\tile{24}{39}
\tile{24}{38}
\tile{24}{37}
\tile{23}{37}
\tile{22}{37}
\tile{21}{37}
\tile{20}{37}
\tile{19}{37}
\tile{18}{37}
\tile{17}{37}
\tile{16}{37}
\tile{15}{37}
\tile{14}{37}
\tile{13}{37}
\tile{12}{37}
\tile{11}{37}
\tile{10}{37}
\tile{9}{37}
\tile{8}{37}
\tile{8}{36}
\tile{8}{35}
\tile{8}{34}
\tile{8}{33}
\tile{8}{32}
\tile{8}{31}
\tile{8}{30}
\tile{8}{29}
\tile{8}{28}
\tile{8}{27}
\tile{8}{26}
\tile{8}{25}
\tile{8}{24}
\tile{8}{23}
\tile{8}{22}
\tile{8}{21}
\tile{9}{21}
\tile{10}{21}
\tile{11}{21}
\tile{12}{21}
\tile{13}{21}
\tile{14}{21}
\tile{15}{21}
\tile{16}{21}
\tile{17}{21}
\tile{18}{21}
\tile{19}{21}
\tile{20}{21}
\tile{21}{21}
\tile{21}{20}
\tile{21}{19}
\tile{21}{18}
\tile{21}{17}
\tile{21}{16}
\tile{21}{15}
\tile{21}{14}
\tile{21}{13}
\tile{20}{13}
\tile{19}{13}
\tile{18}{13}
\tile{17}{13}
\tile{16}{13}
\tile{15}{13}
\tile{14}{13}
\tile{13}{13}
\tile{12}{13}
\tile{11}{13}
\tile{11}{14}
\tile{11}{15}
\tile{12}{15}
\tile{13}{15}
\tile{14}{15}
\tile{15}{15}
\tile{16}{15}
\tile{17}{15}
\tile{18}{15}
\tile{19}{15}
\tile{19}{16}
\tile{19}{17}
\tile{19}{18}
\tile{19}{19}
\tile{18}{19}
\tile{17}{19}
\tile{16}{19}
\tile{15}{19}
\tile{14}{19}
\tile{13}{19}
\tile{12}{19}
\tile{11}{19}
\tile{10}{19}
\tile{9}{19}
\tile{8}{19}
\tile{7}{19}
\tile{6}{19}
\tile{6}{20}
\tile{6}{21}
\tile{6}{22}
\tile{6}{23}
\tile{6}{24}
\tile{6}{25}
\tile{6}{26}
\tile{6}{27}
\tile{6}{28}
\tile{6}{29}
\tile{6}{30}
\tile{6}{31}
\tile{6}{32}
\tile{6}{33}
\tile{6}{34}
\tile{6}{35}
\tile{6}{36}
\tile{6}{37}
\tile{6}{38}
\tile{6}{39}
\tile{7}{39}
\tile{8}{39}
\tile{9}{39}
\tile{10}{39}
\tile{11}{39}
\tile{12}{39}
\tile{13}{39}
\tile{14}{39}
\tile{15}{39}
\tile{16}{39}
\tile{17}{39}
\tile{18}{39}
\tile{19}{39}
\tile{20}{39}
\tile{21}{39}
\tile{22}{39}
\tile{22}{40}
\tile{22}{41}
\tile{22}{42}
\tile{21}{42}
\tile{20}{42}
\tile{19}{42}
\tile{18}{42}
\tile{17}{42}
\tile{16}{42}
\tile{15}{42}
\tile{14}{42}
\tile{13}{42}
\tile{12}{42}
\tile{11}{42}
\tile{10}{42}
\tile{9}{42}
\tile{8}{42}
\tile{7}{42}
\tile{6}{42}
\tile{5}{42}
\tile{4}{42}
\tile{3}{42}
\tile{2}{42}
\tile{1}{42}
\tile{1}{41}
\tile{1}{40}
\tile{1}{39}
\tile{1}{38}
\tile{1}{37}
\tile{1}{36}
\tile{1}{35}
\tile{1}{34}
\tile{1}{33}
\tile{1}{32}
\tile{1}{31}
\tile{1}{30}
\tile{1}{29}
\tile{1}{28}
\tile{1}{27}
\tile{1}{26}
\tile{1}{25}
\tile{1}{24}
\tile{1}{23}
\tile{1}{22}
\tile{1}{21}
\tile{1}{20}
\tile{1}{19}
\tile{1}{18}
\tile{1}{17}
\tile{1}{16}
\tile{1}{15}
\tile{1}{14}
\tile{1}{13}
\tile{1}{12}
\tile{1}{11}
\tile{1}{10}
\tile{1}{9}
\tile{1}{8}
\tile{1}{7}
\tile{1}{6}
\tile{1}{5}
\tile{1}{4}
\tile{2}{4}
\tile{3}{4}
\tile{4}{4}
\tile{5}{4}
\tile{6}{4}
\tile{7}{4}
\tile{8}{4}
\tile{9}{4}
\tile{10}{4}
\tile{11}{4}
\tile{12}{4}
\tile{13}{4}
\tile{14}{4}
\tile{15}{4}
\tile{16}{4}
\tile{17}{4}
\tile{18}{4}
\tile{19}{4}
\tile{20}{4}
\tile{21}{4}
\tile{22}{4}
\tile{23}{4}
\tile{24}{4}
\tile{25}{4}
\tile{26}{4}
\tile{27}{4}
\tile{28}{4}
\tile{29}{4}
\tile{30}{4}
\tile{30}{5}
\tile{30}{6}
\tile{30}{7}
\tile{30}{8}
\tile{30}{9}
\tile{30}{10}
\tile{30}{11}
\tile{30}{12}
\tile{30}{13}
\tile{30}{14}
\tilelb{30}{15}
\tilelb{30}{16}
\tilelb{30}{17}
\tilelb{30}{18}
\tilelb{30}{19}
\tilelb{30}{20}
\tilelb{30}{21}
\tilelb{30}{22}
\tilelb{30}{23}
\tilelb{30}{24}
\tilelb{30}{25}
\tilelb{30}{26}
\tilelb{30}{27}
\tilelb{29}{27}
\tilelb{28}{27}
\tilelb{27}{27}
\tilelb{26}{27}
\tilelb{25}{27}
\tilelb{24}{27}
\tilelb{23}{27}
\tilelb{22}{27}
\tilelb{21}{27}
\tilelb{20}{27}
\tilelb{19}{27}
\tilelb{18}{27}
\tilelb{17}{27}
\tilelb{16}{27}
\tilelb{15}{27}
\tilelb{14}{27}
\tilemb{13}{27}
\tilemb{12}{27}
\tilemb{11}{27}
\tilemb{11}{28}
\tilemb{11}{29}
\tilemb{11}{30}
\tilemb{11}{31}
\tilemb{11}{32}
\tilemb{11}{33}
\tilemb{11}{34}
\tilemb{11}{35}
\tilemb{12}{35}
\tilemb{13}{35}
\tilemb{14}{35}
\tilemb{15}{35}
\tilemb{16}{35}
\tilemb{17}{35}
\tilemb{18}{35}
\tilemb{19}{35}
\tilemb{20}{35}
\tilemb{21}{35}
\tilemb{22}{35}
\tilemb{23}{35}
\tilemb{24}{35}
\tilemb{25}{35}
\tilemb{26}{35}
\tilemb{27}{35}
\tilemb{28}{35}
\tilemb{29}{35}
\tilemb{30}{35}
\tilemb{31}{35}
\tilemb{32}{35}
\tilemb{33}{35}
\tilemb{34}{35}
\tilemb{35}{35}
\tilemb{36}{35}
\tilemb{37}{35}
\tilemb{38}{35}
\tilemb{39}{35}
\tilemb{40}{35}
\tilemb{41}{35}
\tilemb{42}{35}
\tilemb{43}{35}
\tilemb{44}{35}
\tilemb{45}{35}
\tilemb{46}{35}
\tilemb{47}{35}
\tilemb{48}{35}
\tilemb{49}{35}
\tilemb{50}{35}
\tilemb{51}{35}

\dotlb{51}{15}
\dotlb{50}{15}
\dotlb{49}{15}
\dotlb{48}{15}
\dotlb{47}{15}
\dotlb{46}{15}
\dotlb{45}{15}
\dotlb{44}{15}
\dotlb{43}{15}
\dotlb{42}{15}
\dotlb{41}{15}
\dotlb{40}{15}
\dotlb{39}{15}
\dotlb{38}{15}
\dotlb{37}{15}
\dotlb{36}{15}
\dotlb{35}{15}
\dotlb{34}{15}
\dotlb{33}{15}
\dotlb{32}{15}
\dotlb{31}{15}

\doty{39}{15}
\doty{39}{16}
\doty{39}{17}
\doty{39}{18}
\doty{39}{19}
\doty{39}{20}
\doty{39}{21}
\doty{39}{22}
\doty{39}{23}
\doty{39}{24}
\doty{39}{25}
\doty{39}{26}
\doty{39}{27}
\doty{39}{28}
\doty{39}{29}
\doty{39}{30}
\doty{39}{31}
\doty{39}{32}
\doty{39}{33}
\doty{39}{34}
\doty{39}{35}

\dotor{27}{35}
\dotor{27}{34}
\dotor{27}{33}
\dotor{27}{32}
\dotor{27}{31}
\dotor{27}{30}
\dotor{27}{29}
\dotor{27}{28}
\dotor{27}{27}
\path [dotted, draw, thin] (0,0) grid[step=0.2cm] (57,49);



\fill (14.5,4.5) circle (0.16);
\node (D) at (14,2.7) {$P_{\main{\goB}}$};

\fill (14.5,27.5) circle (0.16);
\node (D) at (17.5,29) {$P_{\second{\goB}}=F_{|F|-1}$};

\fill (50.5,35.5) circle (0.16);
\node (D) at (50.5,37.2) {$P_{\last}=\rightF_{|\rightF|-2}$};

\fill (39.5,35.5) circle (0.16);
\node (D) at (39.5,37.2) {$B_v$};

\fill (27.5,35.5) circle (0.16);
\node (D) at (27.5,37.2) {$B_v$};

\fill (27.5,27.5) circle (0.16);
\node (D) at (27.5,25.7) {$B_0$};

\fill (51.5,15.5) circle (0.16);
\node (D) at (51.5,13.7) {$S_{0}=F_0$};

\fill (39.5,15.5) circle (0.16);
\node (D) at (39.5,13.7) {$B_{0}=F_u$};

\fill (30.5,15.5) circle (0.16);
\node (D) at (28.5,15.5) {$P_{a}$};

\end{tikzpicture}
\caption{Proof of Lemma \ref{combi:half:neq:One}: consider the path $P$ from Figure \ref{fig:Final:Combining:back:running}. The path $F=SP_{a,\ldots,\second{\goB}}$ is represented in light blue (the dots for the half-shield $S$ and the squares for $P_{a,\ldots,\second{\goB}}$). The area $\areaRes{\second{\goB}}$ (in dark blue) is the interior of the arc $A=F\rightF_{1,\ldots,|\rightF|-1}$ (the tiles of $\rightF_{1,\ldots,|\rightF|-1}$ are blue). Any path which turns rights of $F$ stays inside $\areaRes{\second{\goB}}$ and does not intersect with $A$ after turning right. In the first example, the path $B$ (the yellow dots) links a tile of the half-shield $S$ with a tile of $\rightF$. It is a contradiction, since $S$ should be a prefix of $\min(\inter{A})$. In the second example, the path $B$ (the orange dots) links a tile of $P_{a,\ldots, \second{\goB}-1}$ to a tile of $\rightF$. Since $P$ is canonical for glue column $c$, there is a contradiction with the priority of the dominant arc  $P_{\main{\goB},\ldots,\second{\goB}}$ according to Corollary \ref{cor:decompo:cano}.}
\label{fig:Final:Combining:back:running:LemOne}
\end{figure}

For the sake of contradiction suppose that $B_{1,\ldots,|B|-1}$  intersects $A$ then we can define $0\leq v \leq |B|-1$ such that $B_{1,\ldots,v}$ is the shortest prefix of $B_{1,\ldots,|B|-1}$ to intersect with $A$. Again by hypothesis, $F_{0,\ldots,u}B_1$ turns right of $F$ and  $e_B\leq \Bfour(c)$ then $B_{0,\ldots,v}$ is inside $\areaRes{\second{\goB}}$. Moreover, by hypothesis, all intersections between $B$ and $F$ are agreements  then  $B_{0,\ldots,v}$ can grow from $\asm{P_{0,\ldots,\second{\goB}}} \cup \asm{F}$. In this case, if $B_0$ or $B_v$ is a tile of $S$, then we have a contradiction since $S$ must be the prefix of $\min(\inter{A})$ by definition of a shield. Otherwise, $B_0$ is a tile of $P_{a,\ldots,\second{\goB}-1}$ and $B_v$ is a tile of $P_{a,\ldots,\second{\goB}-1}$ or a tile of $\rightF$ (we remind that $\rightF_0=P_{\second{\goB}}$), this is contradiction by Corollary \ref{cor:decompo:cano} due to the priority of $P_{\main{\goB},\ldots,\second{\goB}}$.
\end{proof}

Now, to achieve our goal, we just need to show that the hypothesis of Lemma \ref{combi:half:equal:One} are met for some $Q_u-\vu$ with $0 \leq u < \proS$. To do so we will focus at the end of $\leftF$ and  $\leftFDeux$, especially at the indices $0\leq i \leq |\leftF|-1$ and $0\leq i' \leq |\leftFDeux|-1$ such that $\leftF_{i,\ldots,|\leftF|-1}$ and $\leftFDeux_{i',\ldots,|\leftFDeux|-1}$ are arcs on glue columns $c$ and $c'$ respectively.

\begin{lemma} \label{combi:half:neq:Two}
Consider the notations \ref{notation:hidden}, there exist $0\leq i \leq |\leftF|$ such that $\leftF_{i,\ldots,|\leftF|-1}$ is a positive upward arc of glue column $c$ and we have $y_{\leftF_{|\leftF|-1}}>y_{\leftF_0}\geq y_{\leftF_i}\geq y_{P_{\main{\goB}}}$.
\end{lemma}

\begin{proof}
With the chosen conventions, $P_{\main{\goB},\ldots,\second{\goB}}$ is a positive upward arc. By definition of $\leftF$, its last glue is $\glueP{\second{\goB}-1}{\second{\goB}}$ and this glue points west on glue column $c$. Since the first glue of $\leftF$ is $\glueP{\hid}{\hid+1}$, it is also on glue column $c$. Then, we can define $\leftF_{i,\ldots,|\leftF|-1}$ such that this path is an arc of glue column $c$, see Figure \ref{fig:Final:Combining:back:running}. This arc is positive since its last glue points west. Now, the interior of the hole $H=P_{\hid,\ldots,\second{\goB}}$ corresponds to the left side of $P_{\hid,\ldots,\second{\goB}}$. Thus, if there exists $i< j <|\leftF| -1$ and $\main{\goB}<j'<\second{\goB}$ such that $\rev{P_{\main{\goB},\ldots,j'}\leftF_{j+1}}$ turns left of $\rev{P_{\main{\goB},\ldots,j'+1}}$ then the path $\leftF$ is not inside the interior of the hole $H$. This is a contradiction, since $\leftF$ is defined as $\min(\inter{H})$. Then, $\leftF_{i,\ldots,|\leftF|-1}$ is inside the interior of arc $P_{\main{\goB},\ldots,\second{\goB}}$.  Of course, $\leftF_{i,\ldots,|\leftF|-1}$ cannot cross the frontier of $\leftF$ apart from the extremities. Thus we have $y_{\leftF_{|\leftF|-1}}>y_{\leftF_0}\geq y_{\leftF_i}\geq y_{P_{\main{\goB}}}$ and $\leftF_{i,\ldots,|\leftF|-1}$ is an upward arc.
\end{proof}

The next step is to compare the two holes $\leftF$ and $\leftFDeux-\vu$.

\begin{lemma} \label{combi:half:neq:Four}
Consider the notations \ref{notation:hidden}, consider $0\leq u \leq |\leftF|$ such that $\leftF_{1,\ldots,u}$ is the largest common prefix between $\leftF_{1,\ldots,|\leftF|-1}$ and $\leftFDeux_{1,\ldots,|\leftFDeux|-1}-\vu$ then: 
\begin{itemize}
\item $\leftF_0(\leftFDeux_{1,\ldots,u+1}-\vu)$ turns right of $\leftF$;
\item $\leftFDeux_{u+1,\ldots,|\leftFDeux|-1}-\vu$ and $\leftF_{u+1,\ldots,|\leftF|-1}$ do not intersect;
\item $\leftF_{1,\ldots,|\leftF|-2}$ is inside the interior of $\leftFDeux-\vu$.
\end{itemize}
\end{lemma}

\begin{proof}
Since the tile assembly system is directed and since $P_{0,\ldots,\hid}(\leftFDeux_{1,\ldots,u}-\vu)$ is producible, we have two cases to study: either $\leftF_0(\leftFDeux_{1,\ldots,u+1}-\vu)$ turns left of $\leftF$ or $\leftF_0(\leftFDeux_{1,\ldots,u+1}-\vu)$ turns right of $\leftF$.

If $\leftF_0(\leftFDeux_{1,\ldots,u+1}-\vu)$ turns left of $\leftF$, see Figure \ref{fig:Final:Combining:back:running:LemFour}. Consider $u\leq v \leq |\leftFDeux|-1$ such that $(\leftFDeux_{1,\ldots,v}-\vu)$ is the largest prefix of $(\leftFDeux_{1,\ldots,|\leftFDeux|-1}-\vu)$ to be inside $\inter{\leftF}$. Then, $P_{0,\ldots,\hid}(\leftFDeux_{1,\ldots,v}-\vu)$ is producible. If $\leftFDeux_{u+1,\ldots,v}-\vu$ intersects with $\leftF_{u+1,\ldots,|\leftF|-1}$, then it is possible to contradict the fact $\leftF$ is a minimum hole. If $\leftFDeux_{0,\ldots,v}-\vu$ crosses the frontier of the hole $\leftF$, since $w<w'$ then $\leftFDeux_{0,\ldots,v}$ crosses the frontier of the hole $\leftFDeux$ contradicting the definition of a hole. 
Then $v=|\leftFDeux|-1$. Nevertheless, by Lemma \ref{combi:half:neq:Two} there exists $0\leq i' \leq |\leftFDeux|-1$ such that $\leftFDeux_{i',\ldots,|\leftFDeux|-1}-\vu$ is a positive upward arc of glue column $c$ with $y_{\leftFDeux_{|\leftFDeux|-1}}>y_{\leftFDeux_0}\geq y_{\leftFDeux_{i'}}$. Since this arc is inside the interior of the hole $\leftF$ and thus in the interior of the hole $P_{\hid,\ldots,\second{\goB}}$
then by Lemma \ref{combi:half:neq:Three}, there exists $0 \leq i'' \leq  j'' \leq |\leftF|-1$ such that $P_{i'',\ldots,j''}$ is a positive arc of glue column $c$ such that  $\min\{y_{P_{i''}},y_{P_{j'''}}\} \leq y_{\leftFDeux_{i'}-\vu} \leq y_{\leftFDeux_{|\leftFDeux|-1}-\vu} \leq \max\{y_{P_{i''}},y_{P_{j''}}\}$. Since $\pos{\leftFDeux_{0}-\vu}=\pos{P_{\hid}}$ and $w<w'$, then $\min\{y_{P_{i''}},y_{P_{j''}}\} \leq y_{P_{\hid}} < y_{P_{\second{\goB}}}< y_{\leftFDeux_{|\leftFDeux|-1}-\vu} \leq \max\{y_{P_{i''}},y_{P_{j''}}\}$. 
Then, the arc $P_{i'',\ldots,j''}$ either crosses or dominates the arc $P_{\main{\goB},\ldots,\second{\goB}}$ contradicting the fact that $P$ is simple or the definition of a main glue.

\begin{figure}
\center
\begin{tikzpicture}[x=0.2cm,y=0.2cm]

\fill[fill=lightblue, draw opacity=0.8] (15,21.5) -| (21.5,13.5) -| (11.5,15.5) -| (19.5,19.5)-| (6.5,39.5) -| (18.5,42.5) -| (1.5,7.5) -| (25.5,21.5) -| (30.5,23.5) -| (23.5,27.5) -| (20.5,24.5) -| (17.5,27.5) -| (15,27.5);

\draw [dashed] (15,0) -| (15,49);
\draw [thick, color=lightblue] (15,21.5) -| (15,27.5);

\draw[very thick] (4.5,44.5) -| (24.5,37.5) -| (8.5,21.5) -| (14.5,21.5);
\draw[very thick] (14.5,21.5) -| (21.5,13.5) -| (11.5,15.5) -| (19.5,19.5)-| (6.5,39.5) -| (18.5,42.5) -| (1.5,7.5) -| (25.5,21.5) -| (30.5,23.5) -| (23.5,27.5) -| (20.5,24.5) -| (17.5,27.5) -| (14.5,27.5);
\draw[thick] (6.5,19.5) |- (18,11.5);
\draw[thick] (14.5,32.5) |- (18,32.5);

\draws{4}{44}
\tile{5}{44}
\tile{6}{44}
\tile{7}{44}
\tile{8}{44}
\tile{9}{44}
\tile{10}{44}
\tile{11}{44}
\tile{12}{44}
\tile{13}{44}
\tile{14}{44}
\tile{15}{44}
\tile{16}{44}
\tile{17}{44}
\tile{18}{44}
\tile{19}{44}
\tile{20}{44}
\tile{21}{44}
\tile{22}{44}
\tile{23}{44}
\tile{24}{44}
\tile{24}{43}
\tile{24}{42}
\tile{24}{41}
\tile{24}{40}
\tile{24}{39}
\tile{24}{38}
\tile{24}{37}
\tile{23}{37}
\tile{22}{37}
\tile{21}{37}
\tile{20}{37}
\tile{19}{37}
\tile{18}{37}
\tile{17}{37}
\tile{16}{37}
\tile{15}{37}
\tile{14}{37}
\tile{13}{37}
\tile{12}{37}
\tile{11}{37}
\tile{10}{37}
\tile{9}{37}
\tile{8}{37}
\tile{8}{36}
\tile{8}{35}
\tile{8}{34}
\tile{8}{33}
\tile{8}{32}
\tile{8}{31}
\tile{8}{30}
\tile{8}{29}
\tile{8}{28}
\tile{8}{27}
\tile{8}{26}
\tile{8}{25}
\tile{8}{24}
\tile{8}{23}
\tile{8}{22}
\tile{8}{21}
\tile{9}{21}
\tile{10}{21}
\tile{11}{21}
\tile{12}{21}
\tile{13}{21}

\tilemb{14}{21}
\tilemb{15}{21}
\tilemb{16}{21}
\tilemb{17}{21}
\tilemb{18}{21}
\tilemb{19}{21}
\tilemb{20}{21}
\tilemb{21}{21}
\tilemb{21}{20}
\tilemb{21}{19}
\tilemb{21}{18}
\tilemb{21}{17}
\tilemb{21}{16}
\tilemb{21}{15}
\tilemb{21}{14}
\tilemb{21}{13}
\tilemb{20}{13}
\tilemb{19}{13}
\tilemb{18}{13}
\tilemb{17}{13}
\tilemb{16}{13}
\tilemb{15}{13}
\tilemb{14}{13}
\tilemb{13}{13}
\tilemb{12}{13}
\tilemb{11}{13}
\tilemb{11}{14}
\tilemb{11}{15}
\tilemb{12}{15}
\tilemb{13}{15}
\tilemb{14}{15}
\tilemb{15}{15}
\tilemb{16}{15}
\tilemb{17}{15}
\tilemb{18}{15}
\tilemb{19}{15}
\tilemb{19}{16}
\tilemb{19}{17}
\tilemb{19}{18}
\tilemb{19}{19}
\tilemb{18}{19}
\tilemb{17}{19}
\tilemb{16}{19}
\tilemb{15}{19}
\tilemb{14}{19}
\tilemb{13}{19}
\tilemb{12}{19}
\tilemb{11}{19}
\tilemb{10}{19}
\tilemb{9}{19}
\tilemb{8}{19}
\tilemb{7}{19}
\tilemb{6}{19}
\tilemb{6}{20}
\tilemb{6}{21}
\tilemb{6}{22}
\tilemb{6}{23}
\tilemb{6}{24}
\tilemb{6}{25}
\tilemb{6}{26}
\tilemb{6}{27}
\tilemb{6}{28}
\tilemb{6}{29}
\tilemb{6}{30}
\tilemb{6}{31}
\tilemb{6}{32}
\tilemb{6}{33}
\tilemb{6}{34}
\tilemb{6}{35}
\tilemb{6}{36}
\tilemb{6}{37}
\tilemb{6}{38}
\tilemb{6}{39}
\tilemb{7}{39}
\tilemb{8}{39}
\tilemb{9}{39}
\tilemb{10}{39}
\tilemb{11}{39}
\tilemb{12}{39}
\tilemb{13}{39}
\tilemb{14}{39}
\tilemb{15}{39}
\tilemb{16}{39}
\tilemb{17}{39}
\tilemb{18}{39}
\tilemb{18}{40}
\tilemb{18}{41}
\tilemb{18}{42}
\tilemb{17}{42}
\tilemb{16}{42}
\tilemb{15}{42}
\tilemb{14}{42}
\tilemb{13}{42}
\tilemb{12}{42}
\tilemb{11}{42}
\tilemb{10}{42}
\tilemb{9}{42}
\tilemb{8}{42}
\tilemb{7}{42}
\tilemb{6}{42}
\tilemb{5}{42}
\tilemb{4}{42}
\tilemb{3}{42}
\tilemb{2}{42}
\tilemb{1}{42}
\tilemb{1}{41}
\tilemb{1}{40}
\tilemb{1}{39}
\tilemb{1}{38}
\tilemb{1}{37}
\tilemb{1}{36}
\tilemb{1}{35}
\tilemb{1}{34}
\tilemb{1}{33}
\tilemb{1}{32}
\tilemb{1}{31}
\tilemb{1}{30}
\tilemb{1}{29}
\tilemb{1}{28}
\tilemb{1}{27}
\tilemb{1}{26}
\tilemb{1}{25}
\tilemb{1}{24}
\tilemb{1}{23}
\tilemb{1}{22}
\tilemb{1}{21}
\tilemb{1}{20}
\tilemb{1}{19}
\tilemb{1}{18}
\tilemb{1}{17}
\tilemb{1}{16}
\tilemb{1}{15}
\tilemb{1}{14}
\tilemb{1}{13}
\tilemb{1}{12}
\tilemb{1}{11}
\tilemb{1}{10}
\tilemb{1}{9}
\tilemb{1}{8}
\tilemb{1}{7}
\tilemb{2}{7}
\tilemb{3}{7}
\tilemb{4}{7}
\tilemb{5}{7}
\tilemb{6}{7}
\tilemb{7}{7}
\tilemb{8}{7}
\tilemb{9}{7}
\tilemb{10}{7}
\tilemb{11}{7}
\tilemb{12}{7}
\tilemb{13}{7}
\tilemb{14}{7}
\tilemb{15}{7}
\tilemb{16}{7}
\tilemb{17}{7}
\tilemb{18}{7}
\tilemb{19}{7}
\tilemb{20}{7}
\tilemb{21}{7}
\tilemb{22}{7}
\tilemb{23}{7}
\tilemb{24}{7}
\tilemb{25}{7}
\tilemb{25}{8}
\tilemb{25}{9}
\tilemb{25}{10}
\tilemb{25}{11}
\tilemb{25}{12}
\tilemb{25}{13}
\tilemb{25}{14}
\tilemb{25}{15}
\tilemb{25}{16}
\tilemb{25}{17}
\tilemb{25}{18}
\tilemb{25}{19}
\tilemb{25}{20}
\tilemb{25}{21}
\tilemb{26}{21}
\tilemb{27}{21}
\tilemb{28}{21}
\tilemb{29}{21}
\tilemb{30}{21}
\tilemb{30}{22}
\tilemb{30}{23}
\tilemb{29}{23}
\tilemb{28}{23}
\tilemb{27}{23}
\tilemb{26}{23}
\tilemb{25}{23}
\tilemb{24}{23}
\tilemb{23}{23}
\tilemb{23}{24}
\tilemb{23}{25}
\tilemb{23}{26}
\tilemb{23}{27}
\tilemb{22}{27}
\tilemb{21}{27}
\tilemb{20}{27}
\tilemb{20}{26}
\tilemb{20}{25}
\tilemb{20}{24}
\tilemb{19}{24}
\tilemb{18}{24}
\tilemb{17}{24}
\tilemb{17}{25}
\tilemb{17}{26}
\tilemb{17}{27}
\tilemb{16}{27}
\tilemb{15}{27}
\tilemb{14}{27}

\diay{15}{21}
\diay{16}{21}
\diay{17}{21}
\diay{18}{21}
\diay{19}{21}
\diay{20}{21}
\diay{21}{21}
\diay{21}{20}
\diay{21}{19}
\diay{21}{18}
\diay{21}{17}
\diay{21}{16}
\diay{21}{15}
\diay{21}{14}
\diay{21}{13}
\diay{20}{13}
\diay{19}{13}
\diay{18}{13}
\diay{17}{13}
\diay{16}{13}
\diay{15}{13}
\diay{14}{13}
\diay{13}{13}
\diay{12}{13}
\diay{11}{13}
\diay{11}{14}
\diay{11}{15}
\diay{12}{15}
\diay{13}{15}
\diay{14}{15}
\diay{15}{15}
\diay{16}{15}
\diay{17}{15}
\diay{18}{15}
\diay{19}{15}
\diay{19}{16}
\diay{19}{17}
\diay{19}{18}
\diay{19}{19}
\diay{19}{19}
\diay{18}{19}
\diay{17}{19}
\diay{16}{19}
\diay{15}{19}
\diay{14}{19}
\diay{13}{19}
\diay{12}{19}
\diay{11}{19}
\diay{10}{19}
\diay{9}{19}
\diay{8}{19}
\diay{7}{19}
\diay{6}{19}
\diay{6}{18}
\diay{6}{17}
\diay{6}{16}
\diay{6}{15}
\diay{6}{14}
\diay{6}{13}
\diay{6}{12}
\diay{6}{11}
\diay{7}{11}
\diay{8}{11}
\diay{9}{11}
\diay{10}{11}
\diay{11}{11}
\diay{12}{11}
\diay{13}{11}
\diay{14}{11}
\diay{15}{11}
\diay{16}{11}
\diay{17}{11}

\diay{14}{32}
\diay{15}{32}
\diay{16}{32}
\diay{17}{32}

\path [dotted, draw, thin] (0,0) grid[step=0.2cm] (57,49);

\fill (14.5,21.5) circle (0.16);
\node (D) at (14,23) {$P_{\hid}$};

\fill (14.5,21.5) circle (0.16);
\node (D) at (14,23) {$P_{\hid}$};

\fill (14.5,32.5) circle (0.16);
\node (D) at (14,34) {$\leftFDeux_{|\leftFDeux|-1}$};

\fill (14.5,7.5) circle (0.16);
\node (D) at (14,5.7) {$\leftF_{i}$};

\fill (14.5,11.5) circle (0.16);
\node (D) at (14,9.7) {$\leftFDeux_{i'}-\vu$};

\fill (6.5,19.5) circle (0.16);
\node (D) at (4.5,19.5) {$L_{u}$};

\node (D) at (19,11.5) {$?$};

\node (D) at (19,32.5) {$?$};


\fill (14.5,27.5) circle (0.16);
\node (D) at (17.8,29) {$P_{\second{\goB}}=\leftF_{|\leftF|-1}$};


%

\end{tikzpicture}
\caption{Proof of Lemma \ref{combi:half:neq:Four} Part (1/2):  the path $P_{0,\ldots,\hid-1}\leftF$ from Figure \ref{fig:Final:Combining:back:running}. The path $\leftF$ is in blue and its interior is in light blue.
The yellow diamonds represent $\leftFDeux_{1,\ldots,|\leftFDeux|-1}-\vu$ from the notations \ref{notation:hidden}. The two paths $L_{1,\ldots,|L|-1}$ and $L'_{1,\ldots,|L'|-1}-\vu$ are identical until they split at $L_u=L'_u-\vu$. Here, we study the case where $\leftF$ turn right of $\leftF_0(\leftFDeux_{1,\ldots,|\leftFDeux|-1}-\vu)$. Then, $\leftFDeux_{u+1,\ldots,|\leftFDeux|-1}-\vu$ is inside $\inter{\leftF}$: it cannot leave by crossing through the border since $w<w'$ and it cannot intersect with $\leftF_{u+1,\ldots,|\leftF|-1}$ otherwise there is a contradiction with the minimality of $\leftF$. Nevertheless, there exists $0\leq i' \leq |\leftFDeux|-1$ such that $\leftFDeux_{i',\ldots,|\leftFDeux|-1}$ is a positive upward arc with $y_{\leftFDeux_{i'}-\vu}\leq y_{P_{\hid}}<y_{P_{\second{\goB}}}<y_{\leftFDeux_{|\leftFDeux|-1}-\vu}$. 
In this example, one can remark that $\leftFDeux_{i',\ldots,|\leftFDeux|-1}$ must intersect with $\leftF_{i,\ldots,|\leftF|-1}$. The other possibility is that $\leftFDeux_{i',\ldots,|\leftFDeux|-1}$ dominates $P_{\main{\goB},\ldots,\second{\goB}}$ contradicting that this arc is dominant by applying Lemma \ref{combi:half:equal:Three}.}
\label{fig:Final:Combining:back:running:LemFour}
\end{figure}

If $\leftFDeux_{1,\ldots,u+1}-\vu$ turns right of $\leftF$ see Figure \ref{fig:Final:Combining:back:running:LemFourBis}. Consider $u\leq v \leq |\leftF|-1$ such that $(\leftF_{1,\ldots,v}+\vu)$ is the largest prefix of $(\leftFDeux+\vu)$ to be inside $\inter{\leftFDeux}$. Then, $Q_{0,\ldots,\hidDeux}(\leftF_{1,\ldots,v}+\vu)$ is producible. If $\leftF_{u+1,\ldots,v}+\vu$ intersects with $\leftFDeux_{u+1,\ldots,|\leftFDeux|-1}$, then it is possible to contradict the fact $\leftFDeux$ is a minimal hole. If $v<|\leftF|-1$ and if $\leftF_{1,\ldots,v}+\vu$ crosses the frontier of the hole $\leftFDeux$ then there exist $0\leq i \leq j \leq \leftF$ such that $\leftF_{i,\ldots,j}+\vu$ is a positive arc of glue column $\Bfour(c')$ with $\min\{y_{\leftF_{i'}+\vu},y_{\leftF_{j'}+\vu}\} \leq y_{\leftFDeux_{0}} <  y_{\leftFDeux_{0}}+w < \max\{y_{\leftF_{i'}},y_{\leftF_{j'}}\}$. Using Lemma \ref{combi:half:equal:Three} and by proceeding as in the previous paragraph, we obtain a contradiction that either $P$ is not simple or $P_{\main{\goB}, \ldots, \second{\goB}}$ is not a dominant arc. Then $v=|\leftF|-1$ and the lemma is true. 

\input{./tikz/Final/CombiningBackup/DifLemma4Bis}

\end{proof}

Now, we study how the paths $\leftFDeux-\vu$ and $F=SP_{a,\ldots,\second{\goB}}$ interact.

\begin{lemma} \label{combi:half:neq:Five}
Consider the notations \ref{notation:hidden}. Consider the path $F=SP_{a,\ldots,\second{\goB}}$ there exists $0\leq v' \leq |\leftFDeux|-1$ such that:
\begin{itemize}
\item $P_{0,\ldots,\hid}(\leftFDeux_{0,\ldots,v'+1}-\vu)$ is producible;
\item there exists $0 \leq v \leq |F|-1$ such that $\leftFDeux_{v'}-\vu=F_v$;
\item $F_{0,\ldots,v}(\leftFDeux_{v'+1}-\vu)$ turns right of $F$;
\end{itemize}
\end{lemma}

\begin{proof}
First of all, since $\leftF_{|\leftF|-1}=P_{\second{\goB}}=F_{|F|-1}$ then there exists such $0\leq b \leq |F|-1$ such that $F_{0,\ldots,b}$ is the shortest prefix of $F$ to intersect with $\leftF$. Note that $b \geq |S|-1$ since $S$ is not inside the interior of the hole $P_{\main{\goB},\ldots,\second{\goB}}$. Let $0\leq b' \leq |\leftF|$ such that $\leftF_{b'}=F_b$.

Consider $0\leq u \leq |\leftF|$ such that $\leftF_{0,\ldots,u}$ is the largest common prefix between $\leftF$ and $\leftFDeux-\vu$. For the sake of contradiction suppose that $u<b$, see Figure \ref{fig:Final:Combining:back:running:LemFive}. By Lemma \ref{combi:half:neq:Four}, $\leftFDeux_0(\leftF_{1,\ldots,u+1}+\vu)$ turns left of $\leftFDeux$ and $\leftF_{1,\ldots,|\leftF|-2}+\vu$ is inside the interior of the hole $\leftFDeux$. Then, $Q_{0,\ldots,\hidDeux}(\leftF_{1,\ldots,|\leftF|-2}+\vu)$ is producible. Moreover, the tile $\leftF_{b'}+\vu=F_b+\vu$ belongs to $\leftF_{u+1,\ldots,|\leftF|-2}+\vu$ and is inside $\inter{\leftFDeux}$. From this tile, $\rev{F_{0,\ldots,b}+\vu}$ can grow as long as it stays inside $\inter{\leftFDeux}$. Since $x_{F_0+\vu}=\Bfour(c'')+0.5$ then $\rev{F_{0,\ldots,b}+\vu}$ must leave $\inter{\leftFDeux}$. Since $w_{\rev{F_{0,\ldots,b}+\vu}}>c'$ then $\rev{F_{0,\ldots,b}+\vu}$ can leave $\inter{\leftFDeux}$ only by intersecting with $\leftFDeux_{u+1,\ldots,|\leftFDeux|-1}$. This is a contradiction of the minimality of $\leftFDeux$. 


\input{./tikz/Final/CombiningBackup/DifLemma5}

Then $u\geq b'$ and  $\leftF_{b'}=\leftFDeux_{b'}-\vu=F_b$. Now, we have four cases to deal with, see Figure \ref{fig:Final:Combining:back:running:LemFive}. The first one is when $F_{0,\ldots,b}(\leftFDeux_{b'+1}-\vu)$ turns right of $F$. In this case, the lemma is true.

\input{./tikz/Final/CombiningBackup/DifLemma5Bis}

The second case is when $\leftFDeux_{b'+1}-\vu=F_{b+1}$. In this case, let $b'<v'\leq |\leftFDeux|-1$ such that $\leftFDeux_{b',\ldots,v'}$ is the largest common prefix between $\leftFDeux_{b',\ldots,|\leftFDeux|-1}$ and $F_{b,\ldots,|F|-1}$. Let $0\leq v \leq b$ such that $F_{v}=\leftFDeux_{v'}$. If $F_{0,\ldots,v}(\leftFDeux_{v'+1}-\vu)$ turns right of $F$ then the lemma is true. Otherwise, by Lemma \ref{combi:half:neq:Four} $F_{0,\ldots,v}(\leftFDeux_{v'+1}-\vu)$ cannot turn left of $\leftF$. Then $\leftFDeux_{v'+1}-\vu$ is inside an area delimited by $F_{b,\ldots,|F|-1}$ (which is a subpath of $P_{a,\ldots,b}$) and $\leftF_{b',\ldots,|\leftF|-1}$, see Figure \ref{fig:Final:Combining:back:running:LemFive}. to leave this area, $\leftFDeux_{v'+1,\ldots,|\leftFDeux|-1}-\vu$ must cross trough $F_{b,\ldots,|F|-1}$ and the Lemma is true.

The third case is when $\leftFDeux_{b'+1}=F_{b-1}$. In this case let $b'<v'\leq |\leftFDeux|-1$ such that $\leftFDeux_{b',\ldots,v'}$ is the largest common prefix between $\leftFDeux_{b',\ldots,|\leftFDeux|-1}$ and $\rev{F_{b,\ldots,0}}$. Let $0\leq v \leq b$ such that $F_{v}=\leftFDeux_{v'}$. If $F_{0,\ldots,v}(\leftFDeux_{v'+1}-\vu)$ turns right of $F$ then the lemma is true. Otherwise, $\leftFDeux_{0,\ldots,v'}(F_{v-1}+\vu)$ turns left of $\leftFDeux$. Then, as in the proof of Lemma \ref{combi:half:neq:Four},  it is possible to grow $\rev{F_{v-1,\ldots,0}}+\vu$ from $Q_{0,\ldots,\hidDeux}\leftFDeux_{0,\ldots,v'}$ inside the interior of $\leftFDeux$. Since $F_0+\vu$ is on column $\Bfour(c')+0.5$ then $\rev{F_{v-1,\ldots,0}}+\vu$ will intersect with $\leftFDeux_{v'+1,\ldots,|\leftFDeux|-1}$ which contradicts the minimality of the hole $\leftFDeux$.

The fourth case is when $F_{0,\ldots,b}(\leftFDeux_{b'+1}-\vu)$ turns left of $F$. We remind that $\leftF_{0,\ldots,b}(\leftFDeux_{b'+1}-\vu)$ turns right of $\leftF$ by Lemma \ref{combi:half:neq:Four}. In this case, we can contradict the minimality of the hole $\leftFDeux$ as in the previous paragraph.

%
%
%
\end{proof}

We can combine Lemma \ref{combi:half:neq:One} and Lemma \ref{combi:half:neq:Five} to obtain the desired result.

\begin{lemma}
\label{combi:half:neq:Six}
Consider the notations \ref{notation:hidden}, the path $F=SP_{a,\ldots,\second{\goB}}$ and a path $B$ with $e_B< \Bfour(\colDeux)$ and such that $Q_{0,\ldots,\proS-1}B$ is producible. There exists $0 \leq v' < \proS$ and $0 < v < |F|-1$ such that $(Q_{v',\ldots,\proS-1}-\vu)(B-\vu)$ is a subassembly of $\uniterm$ and $Q_{v'}=F_v$. Moreover, the path $(B-\vu)$ is inside $\areaRes{\second{\goB}}$ and does not intersect with $FR_{1,\ldots,|R|-1}$.
\end{lemma}

\begin{proof}
From Lemma \ref{combi:half:neq:One} and Lemma \ref{combi:half:neq:Five}, we can deduce that there exist $0\leq v'' \leq |\leftFDeux|-1$ and $0\leq v \leq |F|-1$ such that $P_{0,\ldots,\hid}(\leftFDeux_{1,\ldots,|\leftFDeux|-1}-\vu)$ is producible and $F_{0,\ldots,v}\leftFDeux_{v''+1}$ turns right of $F$, see Figure \ref{fig:Final:Combining:back:running:LemSix}. Then, all intersections between $\leftFDeux$ and $F$ are agreements. From $\leftFDeux_{|\leftFDeux|-1}-\vu=Q_{\proS}-\vu$, we can try to assemble $\rev{Q_{0,\ldots,\proS}}-\vu$ until it intersect with $F$. 

Since $\glueQ{\proS-1}{\proS}$ is a ``backup'' protected glue, there exists $0\leq  m < \proS$ such that $Q_{m,\ldots,\proS}$ is a positive dominant upward arc of $Q$ on glue column $c'$.  By Lemma \ref{combi:half:neq:One}, $\glueQ{\proS-1}{\proS}-\vu$ is inside $\areaRes{\goB}$ and on glue column $c$ and we have $y_{Q_{\proS}-\vu}>y_{P_{\second{\goB}}}$. By definition of $\vu$, we have $y_{Q_{m}-\vu} \leq y_{Q_{\hidDeux}-\vu} <  y_{P_{\second{\goB}}}$. We also have $w_{Q_{m,\ldots,\proS}}<\Bfour(c)$ and then  the intersection between $Q_{m,\ldots,\proS}-\vu$ and $F$ exists. Let $m\leq v' \leq p$ such that $Q_{v',\ldots,\proS}$ is the shortest suffix of $Q_{m,\ldots,\proS}$ to intersect with $F$.

If $Q_{v'} \neq F_v$ then we can extract form $\asm{Q_{v',\ldots,\proS}}\cup \asm{F_{v,\ldots,|F|-1}}$ a path which contradicts Lemma \ref{combi:half:neq:One}. Then, $P_{0,\ldots,\hid}(\leftFDeux_{1,\ldots,v''-1}-\vu)(Q_{v,\ldots,\proS}-\vu)$ is producible. If a path $B$ is such that $Q_{0,\ldots,\proS-1}B$ is producible then all intersections between $B$ and $\leftFDeux$ are agreements. Then by Lemma \ref{combi:half:neq:One}, $P_{0,\ldots,\hid}(\leftFDeux_{1,\ldots,v''-1}-\vu)(Q_{v,\ldots,\proS-1}-\vu)(B-\vu)$ is a producible path, $B-\vu$ is inside $\areaRes{\second{\goB}}$ and does not intersect with $FR_{1,\ldots,|R|-1}$ .

\input{./tikz/Final/CombiningBackup/DifLemma6}

\end{proof}

Now, we have the second half of the final Lemma \ref{combi:half:end} of this section.

\subsubsection{Conclusion for combining two protected ``backup'' glues}\label{sec:combining:half:conc}

Now, that we have dealt with the two cases, we can conclude this section and solve Lock \ref{lock:shieldone}.

\begin{lemma}
\label{combi:half:end}
Consider the notations \ref{notation:hidden}. If there exists an index $\proNDeux \geq \proSLast{\colDeux}$ such that the $\glueQ{\proNDeux}{\proNDeux+1}$ is a protected glue (either a ``main'' or ``backup'' one) of the northern (resp. southern) decomposition of $Q$ into dominant arcs then either there exists $\goBUn < i < \decoUn$ such that $\second{i}$ is protected or there exists $\goBUn < i \leq \decoUn$ such that $\main{i}$ is protected.
\end{lemma}

\begin{proof}
Using Lemma \ref{combi:half:equal:Three} and Lemma \ref{combi:half:equal:ThreeBis} if $w=w'$ or Lemma \ref{combi:half:neq:Six} if $w\neq w'$, it is possible to follow the same steps as in Subsection \ref{sec:combining:full} to find a new protected tile with the required properties.
\end{proof}

\subsection{Conclusion}\label{sec:conc}

\begin{theorem}
\label{the:end}
Consider a directed tile assembly system $(T,\sigma,1)$, if its terminal assembly $\uniterm$ is finite then its horizontal width and vertical height are bounded by $\BfinalTheorem$.
\end{theorem}

\begin{proof}
For the sake of contradiction, suppose that $e_\uniterm>\Bfinal$. For each glue column $\Bone\leq c \leq \Btwo$, we consider a canonical path $\cano{c}$. Note that there are $4|T|+1$ of these paths. For each path $\cano{c}$, we consider the index of the last protected glue of its northern decomposition $\proNLast{c}$ and the index of the last protected glue of its southern decomposition $s_{c}$. The path $\cano{c}$ is said to be be protected from the south if $s_{c} \leq \proNLast{c}$ and protected from the north if $s_{c} \geq \proNLast{c}$. In the special case where $s_{c}=\proNLast{c}$, we consider that this path is protected from the north and from the south.

Among these paths, either at least $2|T|+1$ of them are protected from the north or at least $2|T|+1$ of them are protected from the south. Without loss of generality, we consider the later case. Consider the category of the last protected glue of the southern decomposition of these paths: at least $|T|+1$ of these glues are in the category ``main''  or at least $|T|+1$ of these glues are in the category ``back-up''. Then, we can conclude that there exist $\Bone\leq c<c'\leq \Btwo$ such that $\cano{c}$ and $\cano{c'}$ are protected from the south, and their last protected glue of their southern decomposition have the same category and type. Thus, by Lemma \ref{combi:full:end} (if the category of these two glues is ``main'') or \ref{combi:half:end} (if the category of these two glues is ``back-up''), we obtain the contradiction that $s_{c}$ is not the last protected glue of the southern decomposition of $\cano{c}$. 

The value $w_\uniterm$ can be bounded similarly and then the horizontal width of $\uniterm$ can be bounded by $\BfinalTheorem$. The same reasoning can be done for the vertical height. 
\end{proof}

\subsection*{Acknowledgments}
We would like to thank Pierre-Étienne Meunier and Damien Woods. Some of their unpublished results are presented here. In particular, some key points of the proof of Lemma \ref{Uturn:LinearBound} were discussed with them but were never finished and written before. 

\bibliographystyle{plain}
\bibliography{dir-lin-bound}

\end{document}